\documentclass[10pt,journal,compsoc]{IEEEtran}
\ifCLASSOPTIONcompsoc
  \usepackage[nocompress]{cite}
\else
  \usepackage{cite}
\fi
\ifCLASSINFOpdf
\else
\fi
\IEEEoverridecommandlockouts 

\usepackage{lineno,hyperref}
\modulolinenumbers[5]
\usepackage{tikz}
\usepackage{graphicx}
\usepackage{subfigure}  
\usepackage[ruled,vlined]{algorithm2e} 
\usepackage{amsmath,amssymb}
\usepackage{epstopdf}
\usepackage{textcomp}
\usepackage{color}
\usepackage{diagbox}
\usepackage{amsthm,amsmath,amssymb}
\usepackage{mathrsfs}

\newcommand{\tabincell}[2]{\begin{tabular}{@{}#1@{}}#2\end{tabular}}

\usepackage{amsthm}
\newtheorem{theorem}{Theorem}

\begin{document}
%
\title{Gaussian Curvature Filter on 3D Meshes}
%
%
%
%

\author{Wenming~Tang\textsuperscript{1},
        Yuanhao~Gong\textsuperscript{1},
        Kanglin~Liu\textsuperscript{2},
        Jun~Liu\textsuperscript{3},
		Wei~Pan\textsuperscript{4},
        Bozhi~Liu\textsuperscript{5},
        and Guoping~Qiu\textsuperscript{*}
\IEEEcompsocitemizethanks{\IEEEcompsocthanksitem \textsuperscript{1} Wenming~Tang and Yuanhao~Gong equally contributed to this work.\protect\\
\IEEEcompsocthanksitem Wenming~Tang, Yuanhao~Gong, Kanglin~Liu, Jun~Liu and Bozhi~Liu are with the College of Information Engineering, Shenzhen University,  Guangdong Key Laboratory of Intelligent Information Processing, Shenzhen Institute of Artificial Intelligence and Robotics for Society, Shenzhen, China.\protect\\
E-mail: tangwenming@szu.edu.cn, gong@szu.edu.cn, max\_liu@szu.edu.cn, ljunbit@126.com, Liu@szu.edu.cn.
\IEEEcompsocthanksitem Wei~Pan  is with the School of Mechanical \& Automotive Engineering,
South China University of Technology,  Department of Research and Development,
OPT Machine Vision Tech Co., Ltd, Jinsheng Road, Chang’an, Dongguan 523860, Guangdong, China.\protect\\E-mail: vpan@foxmail.com.
\IEEEcompsocthanksitem Guoping~Qiu (Corresponding author) is with the College of Information Engineering, Shenzhen University,  Guangdong Key Laboratory of Intelligent Information Processing, Shenzhen Institute of Artificial Intelligence and Robotics for Society, Shenzhen, China, and also with the School of Computer Science, University of Nottingham, Nottingham NG8 1BB, U.K.\protect\\E-mail: guoping.qiu@nottingham.ac.uk.
}
}

\IEEEtitleabstractindextext{%
	\begin{abstract}
	Minimizing the Gaussian curvature of meshes can play a fundamental role in 3D mesh processing. However, there is a lack of computationally efficient and robust Gaussian curvature optimization method. In this paper, we present a simple yet effective method that can efficiently reduce Gaussian curvature for 3D meshes. We first present the mathematical foundation of our method. Then, we introduce a simple and robust implicit Gaussian curvature optimization method named Gaussian Curvature Filter (GCF). GCF implicitly minimizes Gaussian curvature without the need to explicitly calculate the Gaussian curvature itself. GCF is highly efficient and this method can be used in a large range of applications that involve Gaussian curvature. We conduct extensive experiments to demonstrate that GCF significantly outperforms state-of-the-art methods in minimizing Gaussian curvature, and geometric feature preserving soothing on 3D meshes. GCF program is available at https://github.com/tangwenming/GCF-filter.
	\end{abstract}

}

\maketitle

\IEEEdisplaynontitleabstractindextext

%
\IEEEpeerreviewmaketitle

\IEEEraisesectionheading{\section{Introduction}\label{sec:introduction}}

\par Among Various representations of 3D models, triangular meshes are perhaps the most popular. Triangular meshes usually contain two parts. The first part is a set of vertices, representing 3D spatial locations of the surface. The other part is a set of triangular faces that indicate the connectivity between vertices. With the topological information of adjacent vertices, triangular meshes can represent the geometric details of a surface. 

\par Automatic 3D mesh generation has made great progress in the past few years. There are several ways to generate triangular meshes such as, interactive design from CAD software, 3D scanning, end-to-end generation and so on. In 3D scanning, a scanner can automatically obtain the 3D coordinates of the surface to produce a high-quality 3D mesh~\cite{kriegel2012next}. In end-to-end methods, the 2D images are used to train a neural network, which generates the corresponding 3D mesh~\cite{wang2018pixel2mesh}.
\begin{figure}[htbp]
	\centering
	\includegraphics[height=0.6\linewidth]{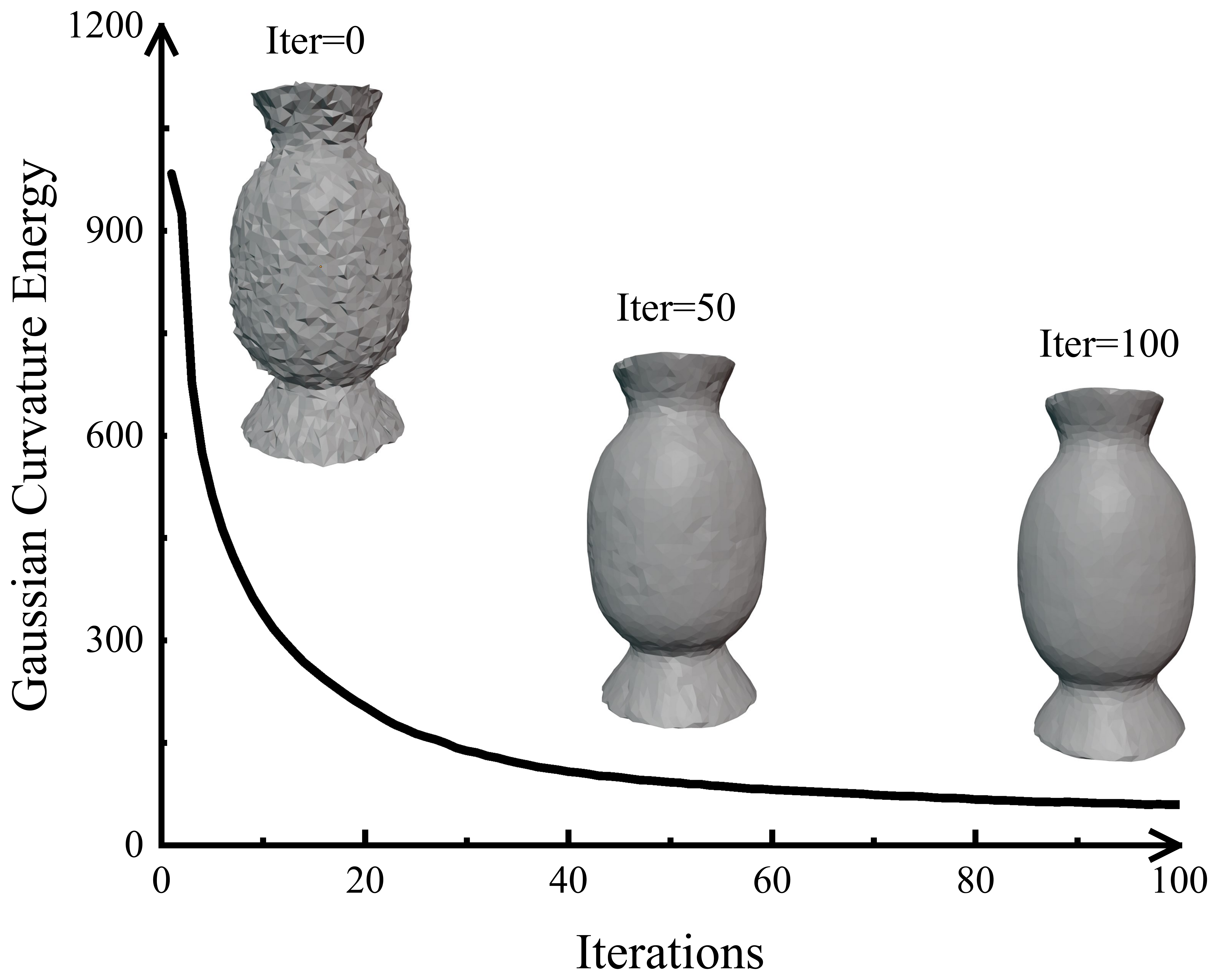}
	\caption{\label{Fig.first} Gaussian curvature filter on the noisy Vase mesh.}
\end{figure}

\par Unfortunately, the 3D mesh obtained through these technologies is often noisy. As a result, the obtained 3D mesh cannot be directly used in practice. For this reason, smoothing methods for 3D meshes become indispensable. In the literature, various methods have been developed. These methods can be categorized into three types: optimization~\cite{fleishman2003bilateral,jones2003non,he2013mesh,wang2014decoupling,wei2014bi,yadav2017mesh}, training~\cite{wang2016mesh,arvanitis2018feature,wei2018learning}, and filtering~\cite{sun2007fast,zheng2011bilateral,zhu2013coarse,zhang2015guided,lu2015robust,li2018non}. Optimization-based methods for mesh smoothing need some manually set parameters, which often need to be optimized iteratively to satisfy the assumed regularization. Jones et al.~\cite{jones2003non} proposed a method that captures the smoothness of a surface by defining local first-order predictors. Lei and Schaefer~\cite{he2013mesh} proposed an $L_0$ minimization method that maximizes the flat regions of the model and removes noise while preserving sharp features. This method is very practical for models with rich flat features. Such methods are generally time consuming and sometimes do not converge for certain given parameter values~\cite{wang2016mesh}. Training-based methods need to provide sufficient training data by learning the mapping relationship between the noisy model and the ground truth model. The trained network can achieve model denoising and feature preservation similar to the noise distribution of the training model. However, the disadvantage of the training-based approach is that it is difficult to find sufficient models to train the network parameters, to have good generalization capabilities for different meshes and different noise levels. Filter-based methods are implemented based on mutual constraints between vertex normals and face normals. Zheng et al.~\cite{zheng2011bilateral} treated the vertex position update as a quadratic optimization problem based on the two normal fields. Sun et al.~\cite{sun2007fast} proposed a two stages denoising method. The first is model patch normal filtering: the weighted average of the neighboring face normals of each face is used, and then the vertices are updated according to the filtered face normals. Those methods use global or local statistics, and may also rely on several manually set parameters. Setting different parameters is conducive to better filtering a specific model, but it also brings inconvenience to use.

\par Curvature is an important geometric feature of surfaces. It is often used as an important tool for surface analysis and processing. The literature has reported that using curvature features on 3D mesh surfaces can achieve good results in surface fairing~\cite{zhao2006triangular}. They designed a diffusion equation whose diffusion direction depends on the mean curvature normal, and its magnitude is a defined function of Gaussian curvature. Michael et al.~\cite{eigensatz2008curvature} proposed a 3D geometry processing framework to achieve 3D mesh filtering and editing by utilizing the curvature distribution of the surface. Gaussian curvature is a specific type of curvature.  It is an intrinsic measurement of surfaces. It has been applied to images and 3D meshes~\cite{desbrun1999implicit,gong2017curvature,lee2005noise}.

\par A 3D mesh that contains noise has higher Gaussian curvature in absolute value than its corresponding noise-free model. Therefore, reducing the curvature energy can smooth or denoise the meshes~\cite{gong2017curvature,gong2015spectrally}. Based on this observation, we can formulate the problem of noise removal on 3D mesh as that of reducing the Gaussian curvature.

\par However, minimizing Gaussian curvature is challenging. It is traditionally carried out by Gaussian curvature flow~\cite{zhao2006triangular}. This method requires the explicit computation of Gaussian curvature. Although Gabriel Taubin~\cite{taubin1995estimating} proposed a method of explicitly estimating the Gaussian curvature on closed manifold meshes, in order to optimize the Gaussian curvature of a model, it is not the best choice to explicitly calculate the Gaussian curvature of each vertex. Another problem with Gaussian curvature flow is that the time step has to be small to ensure numerical stability~\cite{zhao2006triangular}. As a result, such geometric flow is time consuming. These two issues hamper the application of Gaussian curvature on meshes.

\par In order to solve the above-mentioned problems in 3D mesh processing, we propose a simple, easy to implement, and robust filter that can efficiently minimize Gaussian curvature. Our method can effectively remove noise and preserve geometric features as illustrated in Fig. ~\ref{Fig.first}. Different from most existing methods that have many free parameters, our method has only one. The contributions of this paper are as follows: 
\begin{itemize}
	\item We propose a simple and robust implicit Gaussian curvature optimization method, which we call Gaussian Curvature Filter(GCF). GCF does not need to explicitly calculate the Gaussian curvature.
	\item We developed a Gaussian curvature optimization algorithm using a 1-ring neighborhood to preserve the model's original geometric features and simultaneously optimize the Gaussian curvature. 
	\item Our algorithm has only one parameter - the number of iterations. Our method is more robust than existing methods and outperforms state-of-the-art qualitatively and quantitatively.
\end{itemize}
	
	\section{RELATED WORK}
	Before presenting our method that minimizes Gaussian curvature energy and preserves the geometric features, we mainly discuss the related work from three aspects: the Gaussian curvature in image processing, Gaussian curvature in mesh processing, and finally, mesh smoothing capability and feature preservation property.
\subsection{Gaussian curvature in image processing}
Researchers have made great progress in image smoothing using the geometric properties of Gaussian curvature in past decades. Lee et al.~\cite{lee2005noise} design a Gaussian curvature-driven diffusion equation for image noise removal. This method can maintain the boundary and some details better than the mean curvature.  Jidesh et al.~\cite{jidesh2012fourth} proposed Gaussian curvature to guide image smoothing of fourth-order partial differential equations (PDE). It works for image smoothing and maintains curved edges, slopes, and corners. 

\par The calculation formula of Gaussian curvature is generally complicated and has some numerical issues. To overcome these issues, researchers found simple filters to optimize Gaussian curvature. Gong et al.~\cite{gong2013local} proposed a locally weighted Gaussian curvature as a regularized variational model and designed a closed-form solution. It has achieved excellent results in image smoothing, smoothing, texture decomposition and image sharpening. They further proposed an optimization method for regularizers based on Gaussian curvature, mean curvature and total variation~\cite{gong2017curvature}. These pixel local filters can be used to efficiently reduce the energy of the entire model, thus significantly reduce computational complexity because there is no need to explicitly calculate Gaussian curvature itself. 

\subsection{Gaussian curvature in mesh processing}
\par In 3D geometry, researchers have found some approximation methods to calculate Gaussian curvature for discrete meshes~\cite{peng2003estimating,surazhsky2003comparison}. Michael et al.~\cite{eigensatz2008curvature} proposed a framework for 3D geometry processing that provides direct access to surface curvature to facilitate advanced shape editing, filtering, and synthesis algorithms. This algorithm framework is widely used in geometric processing, including smoothing, feature enhancement, and multi-scale curvature editing. There have been some works~\cite{stein2018developability,rabinovich2018discrete, ion2020shape,sellan2020developability} about the application of Gaussian curvature in fied of developable 3D meshes. For example, Oded et al.~\cite{stein2018developability} used a variational approach that drives a given mesh toward developable pieces separated by regular seam curves. The partial developability of a mesh makes the mesh convenient for industrial manufacturing. Similarly, the real-time nature of the algorithm needs to be greatly improved. There has been several works based on curvature flow~\cite{zhao2006triangular,kazhdan2012can,crane2013robust} in the field of 3D mesh optimization. For example, Zhao et al.~\cite{zhao2006triangular} applied Gaussian curvature flow to mesh fairing. They designed a diffusion equation whose evolution direction relies on the normal and the step size is a manually defined function of Gaussian curvature. The corner and edge features of the mesh are preserved during fairing. However, the Gaussian curvature of each vertex is explicitly calculated, and the computation complexity is too high.

\par In Gaussian curvature optimization methods, the Gaussian curvature flow is the most classical one. The Gaussian curvature flow method relies on high-precision Gaussian curvature calculations. It also requires the time step size to be small for ensuring numerical stability. If the step size is set too large, the algorithm may be unstable. If the step size is too small, the convergence speed is slow~\cite{zhao2006triangular}. How to set a reasonable step size is an open problem.

\subsection{Mesh smoothing and feature preservation}
There are many types of 3D mesh smoothing and feature preservation methods. Here, we mainly discuss the optimization-based and filter-based state-of-the-art methods.

\par Optimization-based methods achieve global optimization constrained by the priors of the ground truth geometry and noise distribution. He et al.~\cite{he2013mesh} proposed a $L_0$ minimization based method that achieves smoothing by maximizing the plane area of the model. In a model with rich planar features, this method can preserve some sharp geometric features during the smoothing. Wang et al.~\cite{wang2014decoupling} implement smoothing and feature preservation in two steps. First, the global Laplace optimization algorithm is used to denoise, and then an L1-analysis compressed sensing optimization is used to recover sharp features.

\par Filter-based methods are commonly implemented by moving the vertex position along the vertex or face normal. In~\cite{zhu2013coarse,zhang2015guided,lu2015robust}, researchers perform the smoothing and feature preservation by moving the vertex position of the model. The vertex is moved along the normal direction. And the moving step size is an empirical parameter. Lu et al.~\cite{lu2015robust} constructs geometric edges by extracting geometric features of the input model, and iteratively optimizes vertex positions for smoothing and feature preservation by guiding the geometric edges. In~\cite{lu2015robust,yadav2017mesh,lu2017efficient,zhang2015guided,li2018non}, iterative optimization of the face normal is used as a guide for smoothing and feature preservation. Li et al.~\cite{li2018non} present a non-local low-rank normal filtering method. Smoothing and feature preservation of synthetic and real scan models are achieved by guided normal patch covariance and low-rank matrix approximation.

\par Most of existing smoothing and feature preservation algorithms have the following problems: 1) Excessive dependence on the \textit{a priori} assumptions~( for example, edge and corner features), thus resulting in many parameters to be manually set; 2) It is difficult to find the optimal parameters;  3) Based on the method of face normal filtering, the original feature is easily damaged while smoothing texture-rich models.
	
	\section{ Gaussian Curvature Filter on Mesh}
	\label{sec:GCF}
In this section, we show a simple iterative filter that can efficiently reduce Gaussian curvature for meshes. Meanwhile, our method preserves geometric features of the input mesh during the optimization process. We show a mathematical theory behind this filter, which guarantees to reduce Gaussian curvature.
\subsection{Variational Energy}
In many applications, reducing Gaussian curvature is usually imposed by following variational model ( see~\cite{gong2015spectrally}, page 131, formula 6.2):
\begin{equation}
\label{eq:total}
\arg\min_{\{v'_i\}}\sum_{i=1}^{N}[\frac{1}{2}(v_i-v_i')^2+\lambda |K(v_i')|]\,,
\end{equation} where $\{v_i\}$ are the input vertices $N$ is the number of vertices, $v_i'$ is the desired output, $K(v_i')$ is the Gaussian curvature at $v_i'$ and $\lambda>0$ is a scalar parameter that usually is related to noise level. The first quadratic term measures the similarity between the input and the output. The second term measures the Gaussian curvature energy of the output mesh. The main challenge in this model is how to efficiently minimize the Gaussian curvature. The definition for a ``discrete Gaussian curvature'' on a triangle mesh is via a vertex’s angular deficit~\cite{meyer2003discrete}:
\begin{equation}
\label{eq:Gaussian curvature}
K(v_i') = (2\pi - \sum_{j\in N(i)} \theta_{ij})/A_{ N(i)},
\end{equation} where $N(i)$ are the triangles incident on vertex $i$ and $\theta_{ij}$ is the angle at vertex $i$ in triangle $j$, $A_{ N(i)}$ is the sum of areas of the $N(i)$ triangles~\cite{meyer2003discrete}. 

According to ~\cite{guoliang2013geometric} theorem $1.15$, the Gaussian curvature energy (GCE) is defined as
\begin{equation}
\label{eq:GC}
{\cal E}_{\mathrm{GC}}(v_i')=\sum_{i=1}^{N}|K(v_i')|\,.
\end{equation} This energy measures the developability of the mesh $\{v_i'\}$. Different from the Eq.~\ref{eq:total}, this energy does not consider the similarity between the output $\{v_i'\}$ and input $\{v_i\}$. Therefore, only minimizing Gaussian curvature energy does not preserve the geometric features of the input mesh during the optimization. We will discuss how to minimize Gaussian curvature and preserve geometric features in our method.

When ${\cal E}_{\mathrm{GC}}=0$, it is clear that $K=0$ everywhere on the surface. Such surface is called a developable surface, which can be mapped to a plane without any distortion. That is why it is called ``developable''. Reducing the Gaussian curvature on the surface is trying to make the surface developable. Developable surfaces can be easily manufactured and produced in industry~\cite{stein2018developability}. This is one reason that minimizing Gaussian curvature is an important topic. Although we know that developable surfaces are very useful, this is not the focus of this paper. This paper is not to obtain a developable surface, but to design an implicit optimization of the Gaussian curvature to achieve smoothing and feature preservation of the model.
\subsection{Mathematical Foundation}
\label{sec:math}
For any developable surface ${\cal S}$ (Gaussian curvature is zero everywhere on the surface), we denote ${\cal TS}$ as its tangent space. We have following theorem:
\begin{theorem}\label{the-proj}
	$\forall \vec{x}\in {\cal S}$, $\forall \epsilon>0$, $\exists \vec{x}_0\in {\cal S},0<|\vec{x}-\vec{x}_0|<\epsilon$, s.t. $\vec{x_0}\in {\cal TS}(\vec{x})$.
\end{theorem}

\begin{proof}
	Let $\vec{x}=\vec{r}(u,v)\in {\cal S}$, where~$\vec{r}=(x,y,z)\in R^3$, and~$(u,v)$ is the parametric coordinate. Since ${\cal S}$ is developable, $\vec{r}(u,v)$ can be represented as $\vec{r}(u,v)=\vec{r}_A(u)+v\vec{r}_B(u)$ ~\cite{pottmann2009computational}, where $ \vec{r}_A(u)$ is the directrix and $\vec{r}_B(u)$ is a unit vector. Let $\vec{x}_0=\vec{r}(u,v_0)\in {\cal S}$, where $v_0=v+\epsilon$ and $\epsilon\neq 0$, then $\vec{x}_0=\vec{r}_A(u)+(v+\epsilon)\vec{r}_B(u)$. For two arbitrary scalars $\alpha_1$ and $\alpha_2$, the tangent plane at $\vec{x}$ is
	\begin{equation}\label{eq:plane}
	{\cal TS}(\vec{x}) = \vec{r} +  \alpha_1\frac{\mathrm{d}\vec{r}}{\mathrm{d}u} +  \alpha_2 \frac{\mathrm{d}\vec{r}}{\mathrm{d}v}=\vec{r} +  \alpha_1\frac{\mathrm{d}\vec{r}}{\mathrm{d}u} +  \alpha_2 \vec{r}_B(u)\,.
	\end{equation} 
	Because of Eq.~\ref{eq:plane}, $\vec{x}_0$ is on the plane that passes $\vec{x}$ and is spanned by the two vectors $\frac{\mathrm{d}\vec{r}}{\mathrm{d}u}$ and $\vec{r}_B$. Therefore, $\vec{r}_0\in {\cal TS}(\vec{x})$.
\end{proof}

This theorem indicates that for any point $\vec{x}$ on a developable surface there must be another neighbor point $\vec{x_0}$ that lives on its tangent plane. This conclusion is the theoretical foundation for our method.

\begin{figure}
	\includegraphics[width=0.8\linewidth]{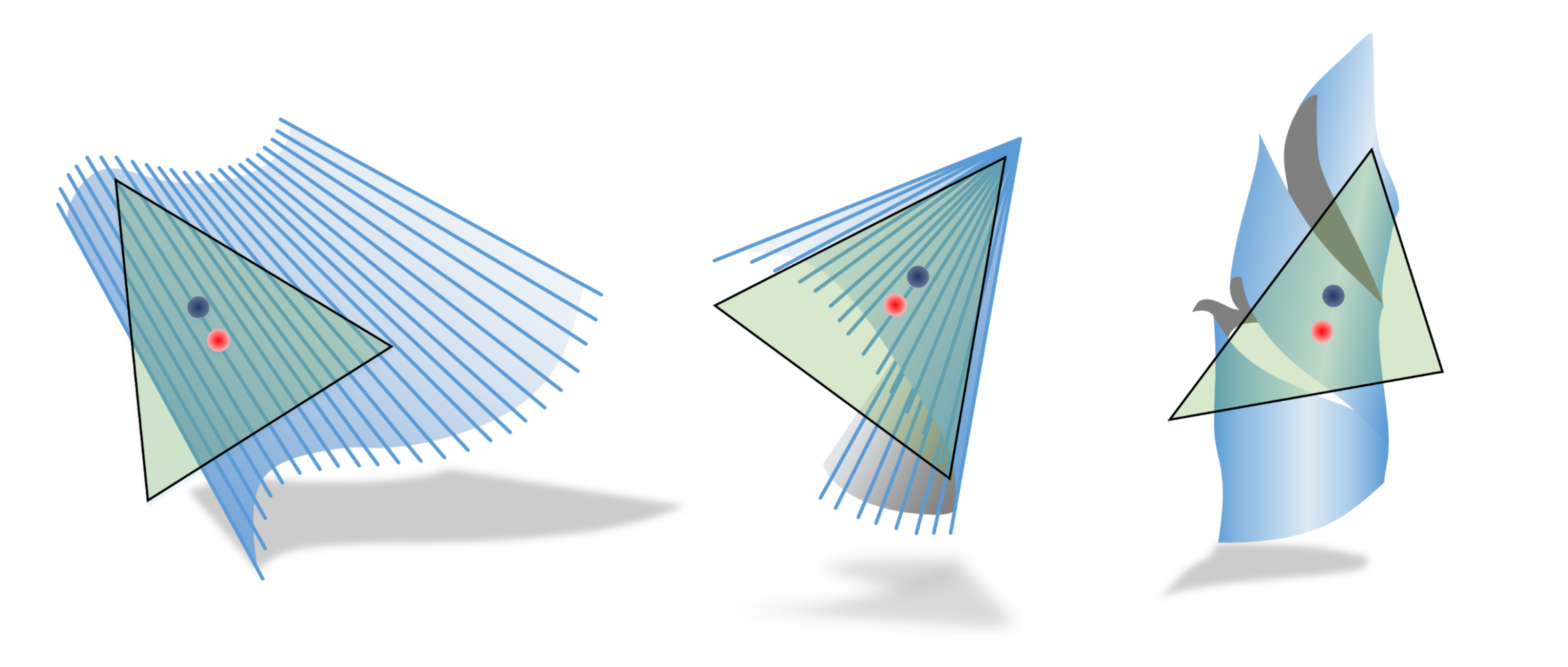}
	\centering
	\setlength{\arraycolsep}{0pt}
	\[
	\begin{array}{ccc}
	\makebox[2.9cm]{\small (a) Cylindrical surface}&
	\makebox[3.2cm]{\small (b) Conical surface}&
	\makebox[3.2cm]{\small (c) Tangent surface}\\
	\end{array}
	\]
	\caption{Theorem~\ref{the-proj} on three types of developable surfaces.}
	\label{fig:theorem}
\end{figure}

This theorem can be verified on developable surfaces. In mathematics, it is already known that there are only three types of developable surface: cylinder, cone and tangent developable. As shown in Fig.~\ref{fig:theorem}, for any point $\vec{x}$ (red point) on such surface, there is another point $\vec{x_0}$ (blue point) that lives on its tangent plane (green triangle ).

This theorem can also be explained from another point of view. In differential geometry, Gaussian curvature of a vertex on a surface is the product of the principal curvature $\kappa_1$ and $\kappa_2$ at the vertex. That is $K =\kappa_1 \kappa_2$.~In literature ~\cite{gong2015spectrally} chapter $6.1.2$, it has been proved that minimizing a principal curvature in the Gaussian curvature of a vertex is to minimize the Gaussian curvature of the vertex for the 2D discrete images. More specifically, we have following relationship: 
\begin{equation}\label{eq:curvature}
\kappa_1\kappa_2=0 \Longleftrightarrow \min\{|\kappa_1|,|\kappa_2|\}=0.
\end{equation} This result is stronger than Theorem~\ref{the-proj} because it tells where $\vec{x_0}$ should be. Theorem~\ref{the-proj} ~and Formula ~\ref{eq:curvature}~can tell us that Gaussian curvature can be minimized without calculating principal curvature.

\par Although this theory is for continuous surfaces, it is still valid for discrete triangular meshes. And all numerical experiments in this paper have confirmed its validity. The only issue on the meshes is that $\vec{x}_0$ is not necessarily a vertex on the mesh (even not on the mesh). However, we can use one of the 1-ring neighboring vertexes to approximate $\vec{x}_0$. Such approximation works well for practical applications as confirmed in this paper. The procedure to find this vertex will be explained in Section~\ref{sec:vmd}.

\par In this paper, we adopt Theorem~\ref{the-proj} and apply it on mesh processing. According to Theorem~\ref{the-proj}, we can reduce the Gauss curvature of the vertex by moving its position such that one of its neighbors falls on its tangent plane. If the Gaussian curvature at the vertex is zero, then the moving distance is zero because one of its neighbors already lives on its tangent plane, see Fig.~\ref{Fig.hl_compare}. Otherwise, the absolute value of Gaussian curvature is high before the movement. After the movement, the processed vertex is closer to a developable surface. Therefore, the Gaussian curvature is reduced. 

\subsection{Discrete Neighborhood on Meshes}

Based on Theorem~\ref{the-proj}, there is always a neighbor point $\vec{x_0}$ that lives on the tangent plane of $\vec{x}$ for any point $\vec{x}$ on a developable surface. However, on the discrete mesh, this point $\vec{x_0}$ is not necessarily a vertex in real applications. To overcome this issue, we take all the 1-ring topology neighborhood vertices as possible candidates and finally adopt only one as an approximation to $\vec{x_0}$. Although such approximation introduces some numerical error, it simplifies the way to find $\vec{x_0}$ on triangular meshes. Our numerical experiments confirm that such approximation works well on triangular meshes in practical applications, see Fig.~\ref{fig:com_gcf06_100}~and~\ref{fig:com_denosing and_features}.

\subsection{ Our Method}
Our method can be roughly divided into two stages. The first part is to classify all the vertices of a mesh according to their neighborhood relationship, so as to ensure that a certain vertex is moved in the local area and its neighborhood is fixed. In this paper, we call it the Greedy Domain Decomposition Algorithm, or GDD for short. The second part is the vertex update algorithm. The vertex update is performed according to the normal direction and the minimum absolute distance.

\begin{figure}[tbp]
	\centering
	\subfigure[]{
		\includegraphics[width=3.1cm]{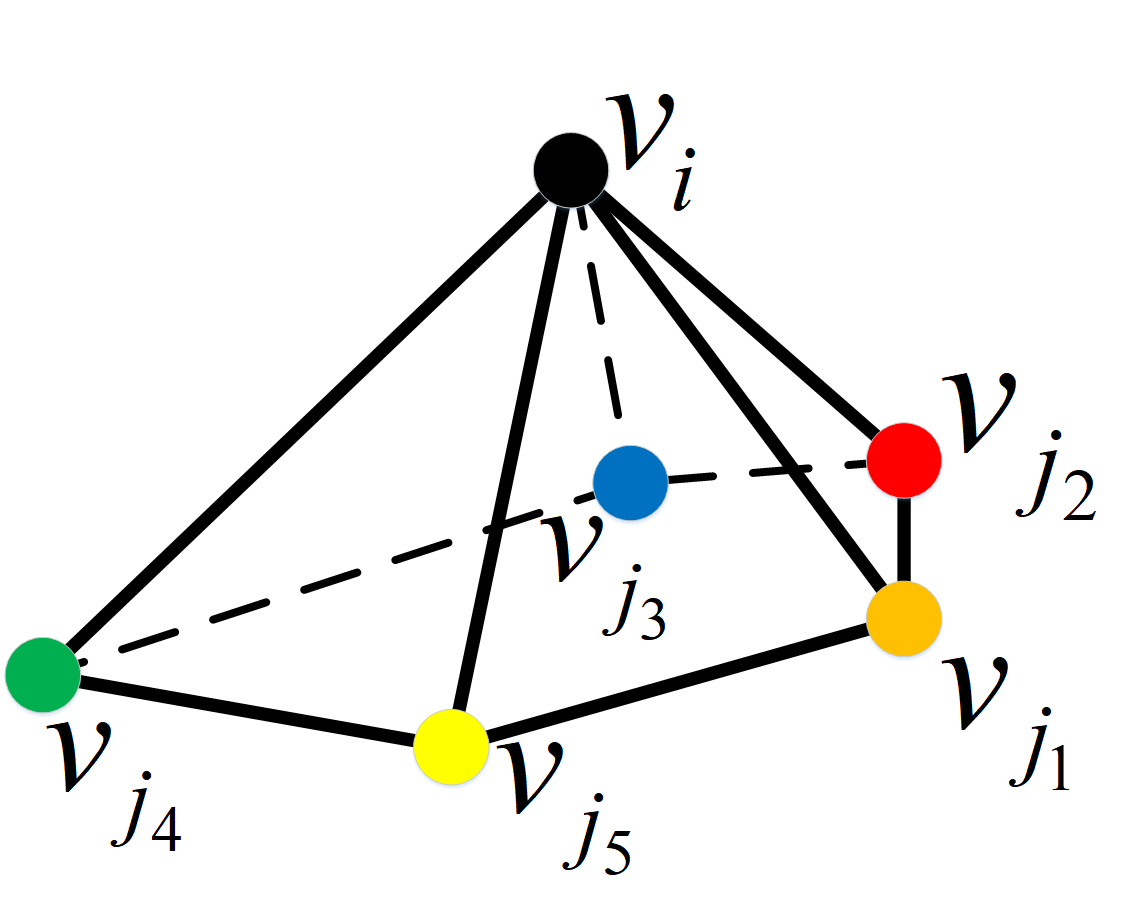}}
	\subfigure[]{
		\includegraphics[width=5.3cm]{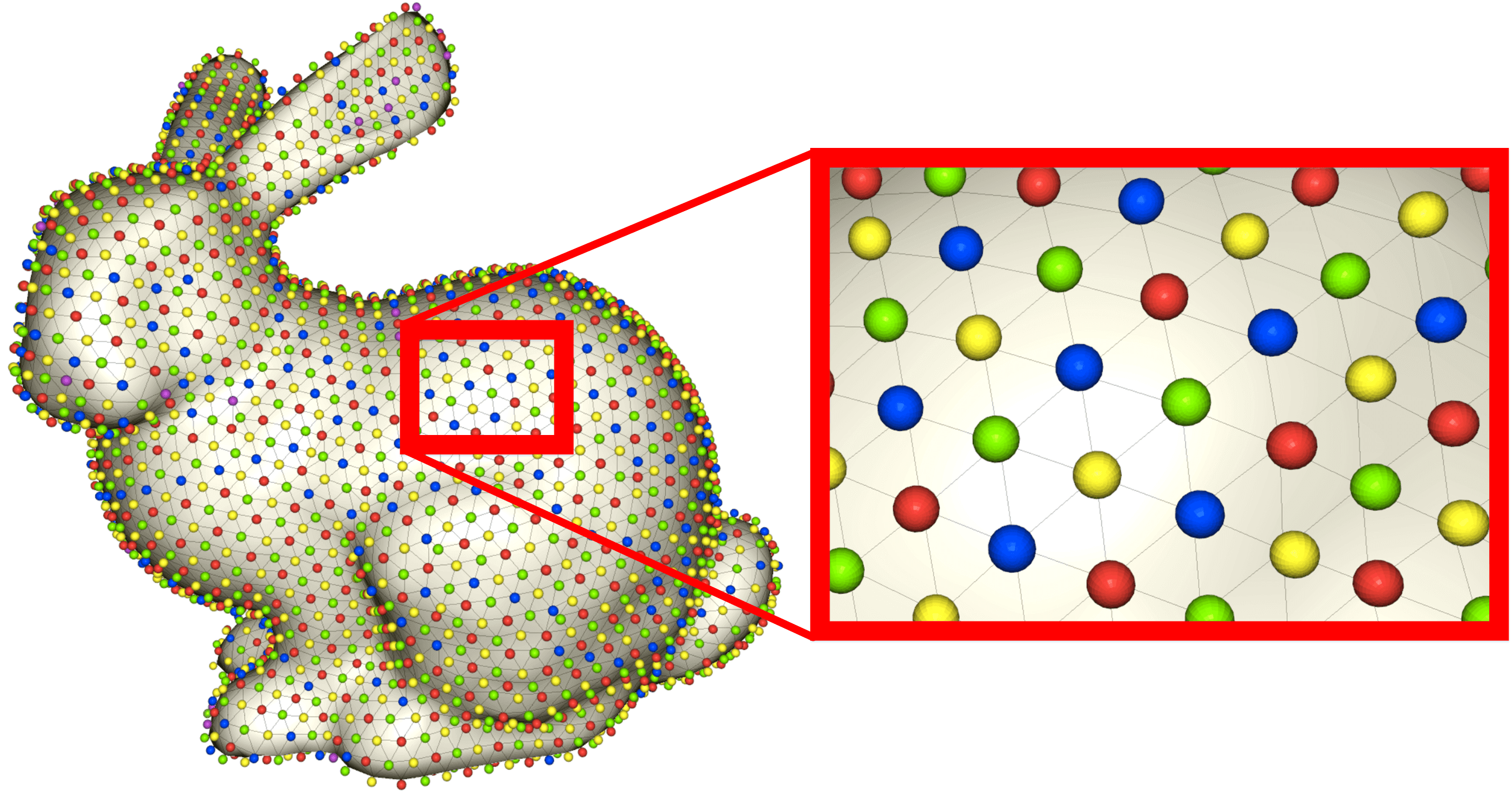}}
	\caption{\label{Fig.com_detail} The 1-ring neighborhood and the domain decomposition result on the Stanford bunny mesh. (a) is the basic structure of the 1-ring topology neighborhood. (b) is the result of the domain decomposition algorithm on the Stanford bunny.}
\end{figure}

\begin{algorithm} 
	\SetAlgoNoLine  
	\caption{GREEDY DOMAIN DECOMPOSITION} 
	\label{algorithms:DOMAIN DECOMPOSITION}
	\KwIn{Vertices V=\{$v_1$,~$v_2$,~...,~$v_n$\},
		Neighbor list P=\{P$_1$,~P$_2$,~...,~P$_n$\}} 
	
	Initialization each vertex color $C_i=0,i=1,\cdots,n$\;
	
	\For{$each~ v_i~ \in~ V$}
	{ 
		Through the greedy algorithm, until the color of the vertex $v_i$ is different from the vertex set of its neighborhood $P_i$, the maximum value k of the color $C_i$ obtained is the number of domain sets $D_k$, and each domain set only contains vertices of the same color.
		%
		%
		%
		%
		%
		%
		%
		%
	}
	\KwOut{vertex domain set \{D$_1$,~D$_2$,~...,~D$_k$\}} where all vertices in D$_i$ have the same color label.
\end{algorithm} 

\subsubsection{Greedy Domain Decomposition}
A 1-ring neighborhood of a triangular mesh is usually composed of the similar structure as Fig.~\ref{Fig.com_detail} $(a)$. The local shape structure consists of a vertex $v_i$ and its neighborhood vertex set $P_i$=\{$v_{j_1}, ..., v_{j_5}$\}. Vertex and neighborhood vertices are connected by edges. 



\par Implementation of this algorithm is described in {\bf Algorithm ~\ref{algorithms:DOMAIN DECOMPOSITION}}, where  $V= \{v_1, v_2, ..., v_n\}$ is the set of all vertices of a mesh, $P= \{P_1, P_2, ..., P_n\}$ is the set of neighborhood points of each vertex, and $D =\{D_1, D_2, ..., D_n\}$ is the color label of each vertex after greedy domain decomposition (Algorithm ~\ref{algorithms:DOMAIN DECOMPOSITION}).

\par The advantages of greedy domain decomposition for mesh vertices are as follows: First, it can ensure that the vertex moves while the neighboring vertices do not move. Second, all vertices are divided into several independent sets. Therefore, each set can move independently. This mechanism can speed up the convergence of our algorithm, as shown in Fig.~\ref{Fig.the convergence average slope}. This is another reason why our algorithm is faster than the Gaussian curvature flow~\cite{zhao2006triangular} (see table~\ref{table1}). We only give the numerical convergence rate in the experiment. Here we define the average convergence slope (${\cal_{ACS}}$) as:
	\begin{equation}\label{eq:ACS}
	{\cal_{ACS}} = \frac{1}{M*(N-2)}\sum_{j=1}^M\sum_{t=2}^{N-1} \frac{lg(\frac{||{{{\cal E}_{\mathrm{GC}}}_{j}}^{{(t+1)}}-{{{\cal E}_{\mathrm{GC}}}_{j}}^{(t)}||_\infty}{||{{{\cal E}_{\mathrm{GC}}}_{j}}^{(t)}-{{{\cal E}_{\mathrm{GC}}}_{j}}^{(t-1)}||_\infty})}{lg(t+1)-log(t)}\,.
	\end{equation}
where, M represents the number of meshes, and N represents the number of iterations of each mesh. ${\cal E}_{\mathrm{GC}}$ is Gaussian curvature energy. 
The result of the greedy domain decomposition of a triangle mesh by {\bf Algorithm ~\ref{algorithms:DOMAIN DECOMPOSITION}} is shown in Fig. ~\ref{Fig.com_detail} $(b)$. We can see that each vertex color is different from the color of its neighborhood vertices, and all the vertices of the bunny are independently divided into several sets.

\begin{figure}[htbp]
	\centering
	\subfigure[one normal projection]{
		\includegraphics[width=3.5cm]{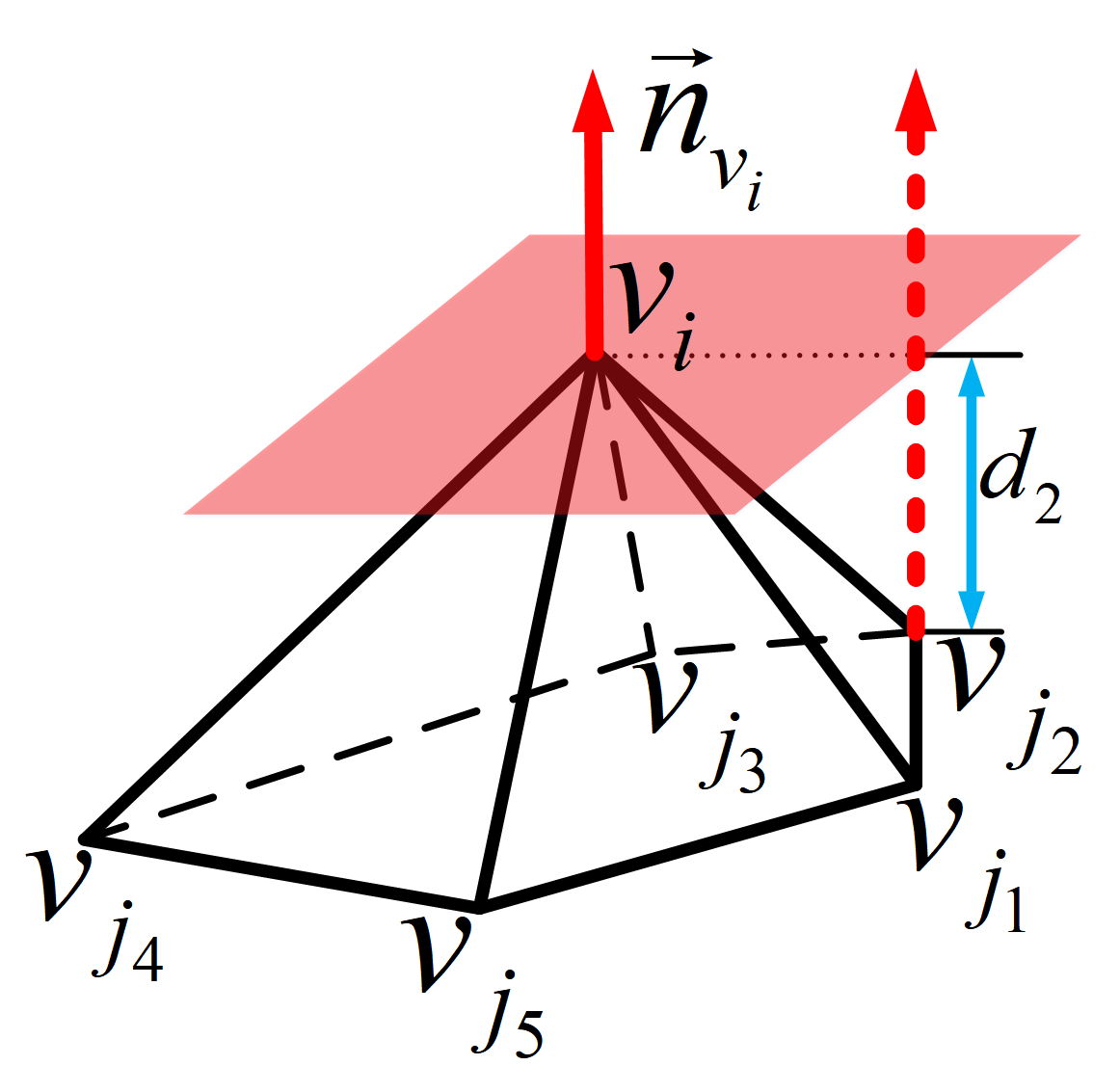}}
	\subfigure[multi normal projection]{
		\includegraphics[width=3.5cm]{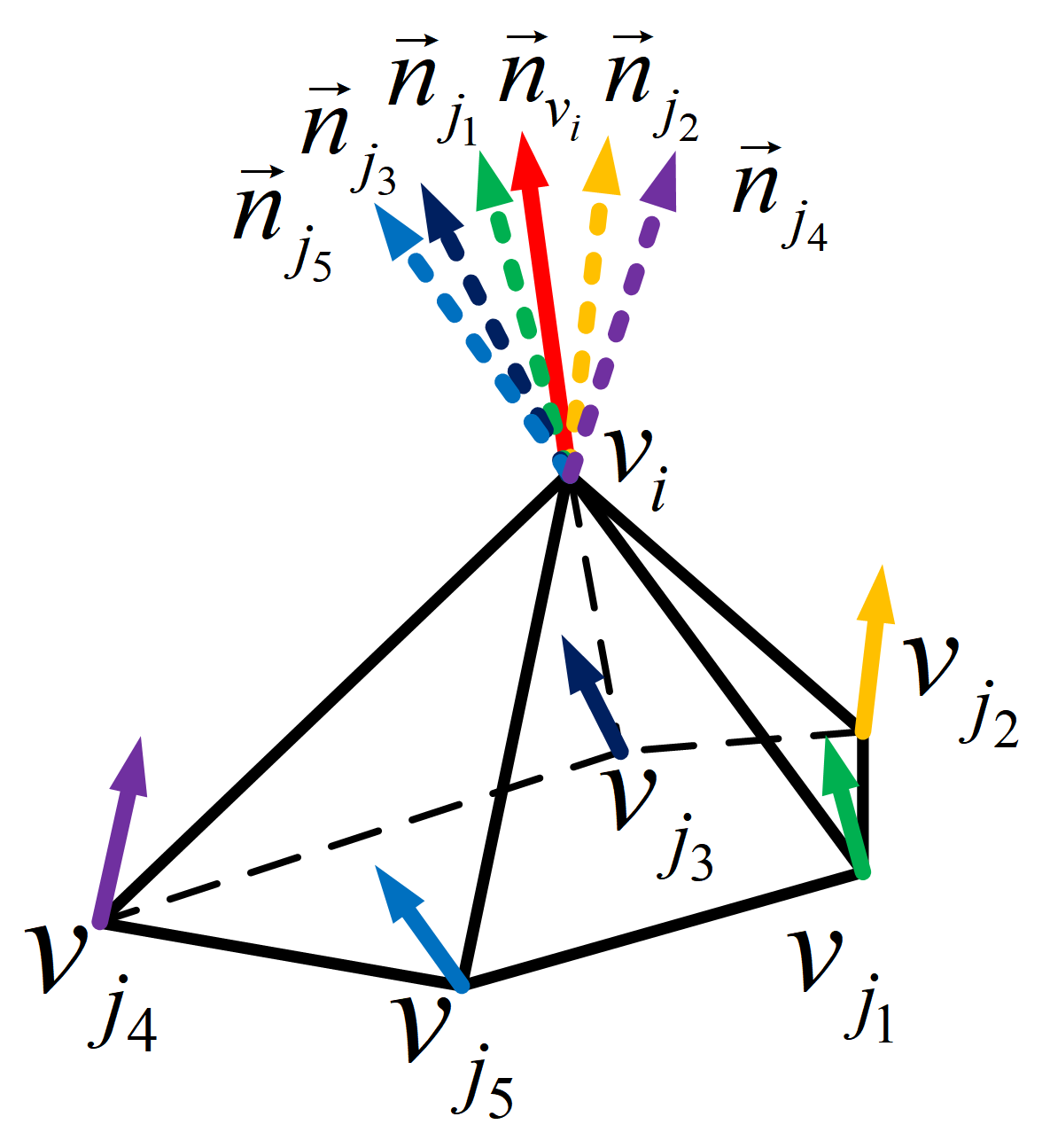}}
	\subfigure[vertex update]{
		\includegraphics[width=4.0cm]{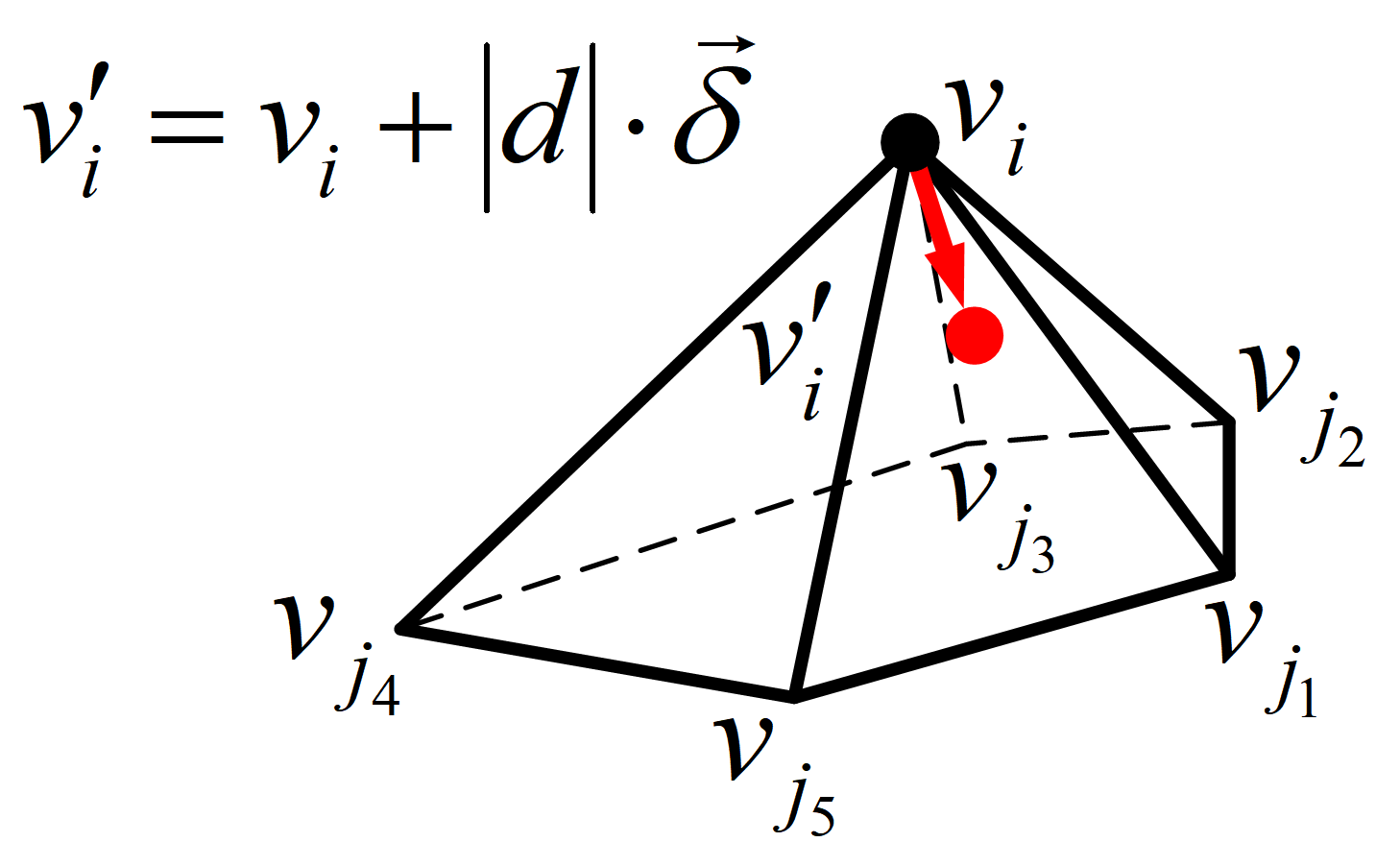}}
	\caption{Projection strategy and the vertex update. (a) And (b) are one and multi normal projection strategy. (c) shows the vertex movement direction and amplitude.}
	\label{Fig.projection_and_update}
\end{figure}

\subsubsection{Vertex Moving Direction}
The differential coordinate of the i-th vertex $v_i = (x_i, y_i, z_i) \in V = \{v_1, v_2, ..., v_n\}$ is the difference between the absolute coordinates of $v_i$ and the center of mass of its immediate neighbors $P_i = \{v_{j_1}, v_{j_2}, ..., v_{j_m}\}$ in the mesh~\cite{sorkine2006differential}, i.e.
\begin{equation}\label{eq:differential}
\vec{\delta}_{v_i} = (\delta_{x_i}, \delta_{y_i}, \delta_{z_i}) = v_{i} - \frac{1}{m}\sum_{v_{j}\in P(v_i)} v_{j}\,.
\end{equation}
\par In differential geometry, the direction of the differential coordinate vector $\vec{\delta}_{v_i}$  approximates the normal direction of the local area~\cite{taubin1995signal}. Following these works, we use the $\vec{\delta}$ unit vector Eq.~\ref{eq:Laplace} (reverse normal direction) as the moving direction. 
\begin{equation}\label{eq:Laplace}
\vec{\delta}=\frac{-\vec{\delta}_{v_i}}{\|\vec{\delta}_{v_i}\|}\,.
\end{equation}
\subsubsection{Vertex Moving Distance}
\label{sec:vmd}
After the greedy domain decomposition, we update each independent vertex set separately. During each iteration,  we have to find the moving direction of the vertex and also the corresponding moving distance. We propose a multi normal projection strategy for computing the moving distance, As shown in Fig. ~\ref{Fig.projection_and_update}. we optimize the Gaussian curvature of $v_i$ by moving the vertex $v_i$. So we need to calculate the magnitude of the movement of the vertex $v_i$. It is shown in Fig.~\ref{Fig.projection_and_update} $(a)$ that multiple edges \{$\overrightarrow{v_{i}v}_{j_1}$,$\overrightarrow{v_{i}v}_{j_2}$...,$\overrightarrow{v_{i}v}_{j_5}$\} composed by vertex $v_i$ and neighborhood vertices \{$v_{j_1}$,$v_{j_2}$...$v_{j_5}$\} are projected to the vertex $v_i$ unit normal vector $\vec{n}_{v_i}$ to calculate the distance set d=\{$d_1$, $d_2$, ..., $d_5$\}. 
\par As shown in Fig.~\ref{Fig.projection_and_update} $(b)$, we compute the normal $\vec{n}_{vi}$~(Eq.\ref{eq:nvi})~\cite{zhao2006accurate} of the vertex $v_i$ and then the normal of each neighborhood vertices. 
	\begin{equation}\label{eq:nvi}
	\vec{n}_{v_i}=\sum_{j\in F_v(i)} A_j\vec{n}_{Fj},
	\end{equation}
	where  $A_j$ is the corresponding face area and $\vec{n}_{Fj}$is the number of the $j$th face normal. $F_v(i)$ is the number of faces in the $i$th vertex-ring. It should be noted that the neighborhood vertex normal is the cross-product unit vector of the two edges of the neighborhood vertices. More specifically, the $\vec{n}_{j_k}$ unit vector in Fig.~\ref{Fig.projection_and_update} $(b)$ is given by~(Eq.\ref{eq:cross})
\begin{equation}\label{eq:cross}
\vec{n}_{j_k}=\frac{\overrightarrow{v_{j_k}v}_{j_{(k-1)}}\times \overrightarrow{v_{j_k}v}_{j_{(k+1)}}}{\|\overrightarrow{v_{j_k}v}_{j(_{k-1)}}\times \overrightarrow{v_{j_k}v}_{j_{(k+1)}}\|}.
\end{equation}
We calculate the projection distance of the unit vector set \{$\vec{n}_{v_i}$, $\vec{n}_{j_1}$, ..., $\vec{n}_{j_5}$\} from the vertex of the 1-ring neighborhood structure. Then each edge \{$\overrightarrow{v_{i}v}_{j_1}$, $\overrightarrow{v_{i}v}_{j_2}$, ..., $\overrightarrow{v_{i}v}_{j_5}$\} has a projection to each unit vector \{$\vec{n}_{v_i}$, $\vec{n}_{j_1}$, ..., $\vec{n}_{j_5}$\}.  As a result, we have all possible projection distances \{$d_1$, $d_2$, ..., $d_n$\} (in this example $n=5\times6=30$). We choose the smallest absolute value in this set $|d|$ as the moving amplitude of vertex  $v_i$.

In summary, the minimal moving vertex distance for $v_i$ $i \in \{1, ..., n\}$ is computed
\begin{equation}
d = 
\begin{array}{l}
\min\limits_{k}\{\left|<\{\vec{N}\},~\overrightarrow{v_{i}v}_{j_k}>\right|\},\{\vec{N}\}=\{\vec{n}_{v_i}, \vec{n}_{j_1}, ..., \vec{n}_{j_m}\}\\
\end{array}
\end{equation}
where $<,~>$ is the standard inner product, $~k=1, ..., m;$. Such minimal distance corresponds to a neighboring vertex $v_{j_k}$. And this vertex is selected as an approximation to $\vec{x}_0$ as described in Theorem~\ref{the-proj}.

\begin{figure}[htbp]
	\centering
	\includegraphics[width=8.0cm]{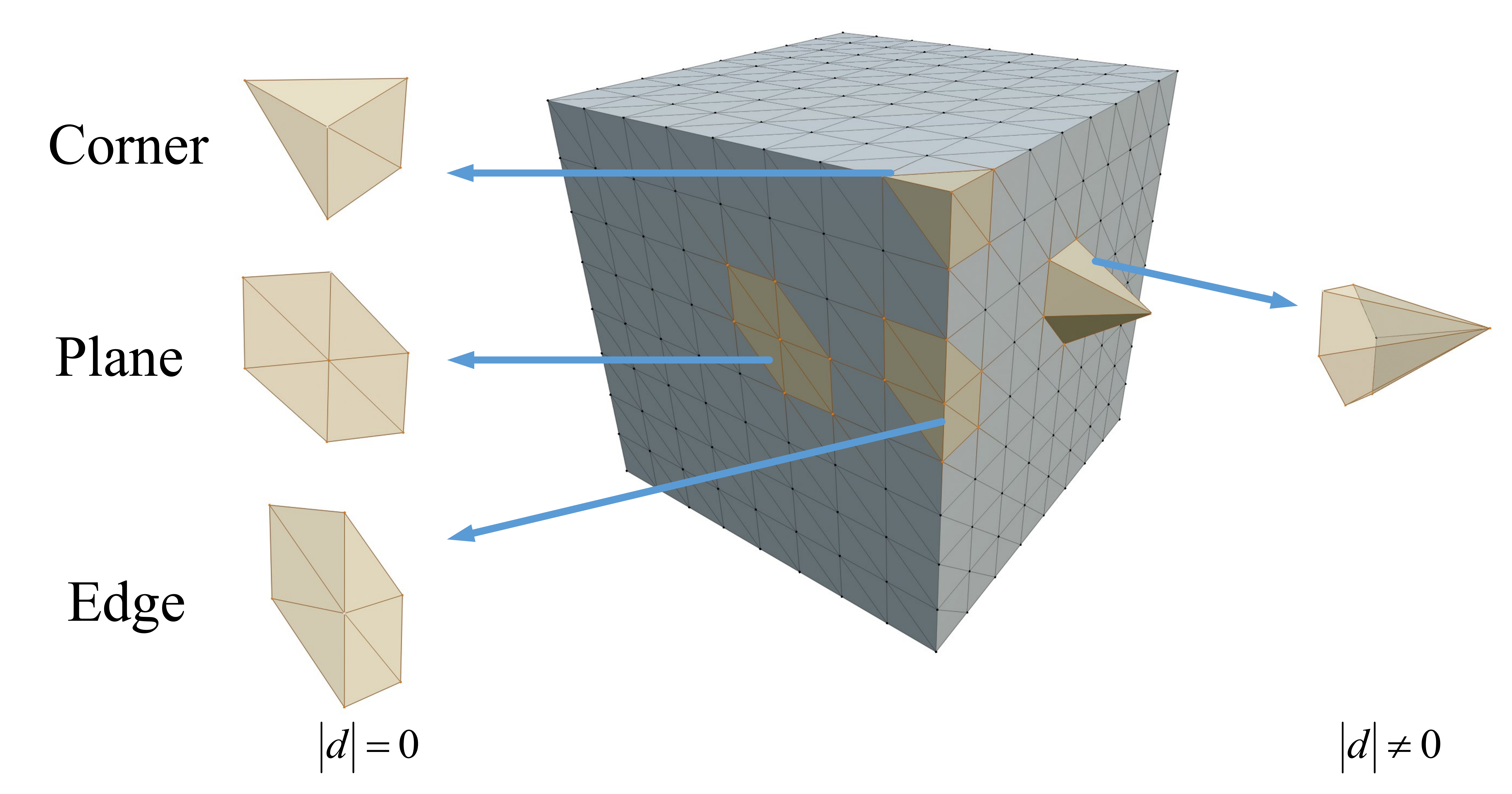}
	\caption{\label{Fig.constrain} Multi normal projection strategy distance.}
\end{figure}
\par Utilizing the multi normal projection strategy as shown in Figure ~\ref{Fig.projection_and_update} $(b)$ ensures that geometric features are preserved  when optimizing the Gaussian curvature. This property is important for the vertices at the corner, edge, and plane geometry, as shown in Fig.~\ref{Fig.constrain}. Since the vertices are contained in the above geometrical features, the minimum absolute value projection distance obtained by the projection strategy of Fig.~\ref{Fig.projection_and_update} $(b)$ is $|d|= 0$. 
\begin{equation}
	\exists~\vec{n}\in \{\vec{N}\},~ \vec{n}~\bot~\overrightarrow{v_{i}v}_{jk}~\Rightarrow~ \exists~k, <\{\vec{N}\}, ~\overrightarrow{v_{i}v}_{jk}>~=~0.
\end{equation}
Therefore, the vertex moving distance 0 and the spatial position is preserved. This kind of projection strategy makes our algorithm have strong smoothing and feature preservation capability for different noise levels. Not surprisingly, the performance of feature preservation in the noise-free model is still robust, as shown in Fig.~\ref{fig:com_gcf06_100}~(Vase).

\subsubsection{Vertex Update Algorithm}
Through the above computation, we obtain the minimum projection distance of the vertex $v_i$, see Fig. ~\ref{Fig.projection_and_update} $(c)$. {\bf algorithm~\ref{algorithms:VERTEX UPDATE}:} The vertex update then is given by

\begin{equation}
v'=v+|d|\cdot \vec{\delta}\,.
\end{equation}
It is worth noting that our algorithm is fundamentally different from the existing Laplace method. The comparative experimental results are shown in Fig.~\ref{Fig.hl_compare}.
\begin{algorithm} 
	\SetAlgoNoLine  
	\caption{GAUSSIAN CURVATURE FILTER ON MESH} 
	\label{algorithms:VERTEX UPDATE}
	\KwIn{Vertex set V=\{$v_1$,~$v_2$,~...,~$v_n$\}, 
		Vertex normal set N=\{$\vec{n_1}$,~$\vec{n_2}$,~...,~$\vec{n_n}$\},
		Neighbor List P=\{P$_1$,~P$_2$,~...,~P$_n$\},
		Domains D=\{D$_1$,~D$_2$,~...,~D$_k$\}, IterationNumber} 
	
	\For{$i = 0 \to IterationNumber-1$}
	{ 
		\For{$j = 0 \to k-1$}
		{
			//each vertex in the same domain
			
			\For{$t = 0 \to D_k[j].size()-1$}
			{
				$index=D_k[j][t]$\;
				$current =  V[index]$\;
				\eIf{current on boundary}
				{
					$v'[index] = v[index]$\;
				}
				{
					Compute $\vec{\delta}$ by $Eq.~\ref{eq:Laplace}$\;
					Compute $ProjN$ by $Eq.~\ref{eq:cross}$\;
					put $N[index]$ into $ProjN$\;
					//find the min projection distance\;
					$MIN = + \infty$\;

					\For{$p = 0 \to P[index].size()-1$}
					{
						\For{$r=0 \to ProjN.size()-1$}
						{
							$ProjectDistance = abs((P[index][p]-current)*ProjN[r])$\;
							\If{$ProjectDistance<MIN$}
							{
								$MIN = ProjectDistance$\;
							}
						}
					}
					// vertex update\;
					$v'[index] = v[index]+\vec{\delta}*MIN$\;
				}
				
			}
		}
	}
	\KwOut{Vertices V'=\{$v'_1$,~$v'_2$,~...,~$v'_n$\}} 
\end{algorithm} 
	
	\section{Experiments}
	In this section, we perform several experiments to show two properties of our method: minimizing the Gaussian curvature and smoothing with feature preservation. In the aspect of Gaussian curvature optimization, we choose ~\cite{zhao2006triangular} and ~\cite{stein2018developability} for comparative experiments. In the aspect of  geometric feature preserving noise removal, there are many methods, including optimization based methods ~\cite{jones2003non} and ~\cite{he2013mesh}, and filter based methods ~\cite{fleishman2003bilateral}, ~\cite{sun2007fast}, ~\cite{zheng2011bilateral},~\cite{zhang2015guided}, and~\cite{li2018non}. We compare our approach with these methods with respect to these two aspects.     
\subsection{Minimize Gaussian Curvature}
To show the property of minimizing Gaussian curvature, we compare our method with two approaches. The first one is the classical Gaussian curvature flow method~\cite{zhao2006triangular}. We compare them by processing two commonly used but representative meshes Max Planck head and Vase. Moreover, we also show the performance of both methods on these meshes when adding some random noise. 

\par We further compare our method with a recent approach from SIGGRAPH 2018 by Stein et. al~\cite{stein2018developability}. We compare both methods on noise-free and noisy meshes respectively. We will discuss the results in later sections.
\par We chose Gaussian curvature energy (GCE, see formula~\ref{eq:GC}), mean square angle error (MSAE), the maximum and average distance from vertex to vertex ($\mathscr{D}$$_{max}$, $\mathscr{D}$$_{mean}$)~\cite{zhang2015guided}, and the Gaussian curvature distribution curve Kullback-Leibler divergence KLD for quantitative evaluation with other comparison methods. The definition of MSAE comes from previous work~\cite{shen2004fuzzy, zhao2018robust, pan2020hlo}~:
	\begin{equation}\label{eq:MSAE}
	MSAE=E\lbrack\angle(n_p,n_o)\rbrack.
	\end{equation}
	Where $E$ is the expectation operator and $\angle(n_p,n_o)$ is the angle between the processed normal $n_p$ and the original normal $n_o$. 

\subsubsection{Parameter settings}
In \cite{zhao2006triangular}, the algorithm has five parameters to be manually adjusted, $k$, $\rho$, $\beta$, $\epsilon$, $\alpha$. For the specific meaning of each parameter, see Eq. 2 in \cite{zhao2006triangular}. According to the author's suggestion, we set $\beta=2$, $\epsilon=0.001$, and $\alpha=0.0005$. See table~\ref{table1} for detailed parameters. It is important to note that this parameter has to set to different values for different models, or different noise levels of the same model. The author also mentioned the importance of an implicit step size algorithm. 

\par However, our algorithm only has one parameter, i.e., the iteration number. How to choose the number of iterations, the number of GCF iterations only depends on the noise level of the model. The higher the noise level, the more iterations are required.

\begin{figure*}[htbp]
	\centering
	\includegraphics[height=0.07\linewidth]{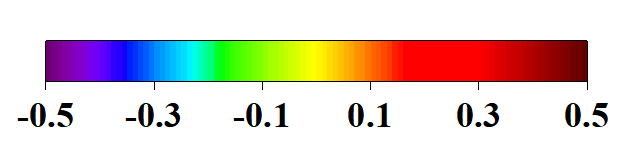}\\
	\centering
	\setlength{\arraycolsep}{0pt}
	\makebox[\textwidth][c]{$
		\begin{array}{c|c}	
		\includegraphics[height=0.3\linewidth]{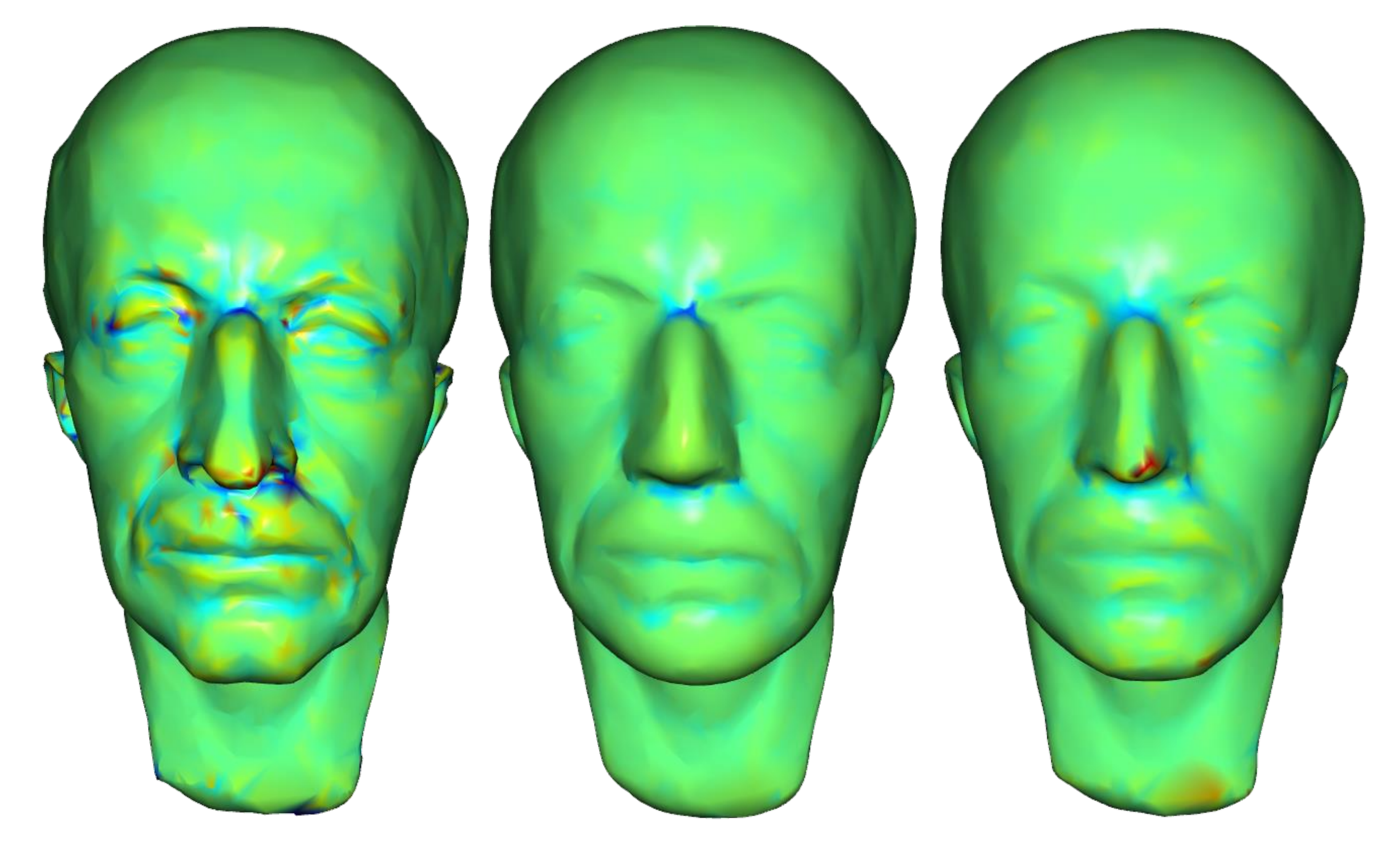}&
		\includegraphics[height=0.3\linewidth]{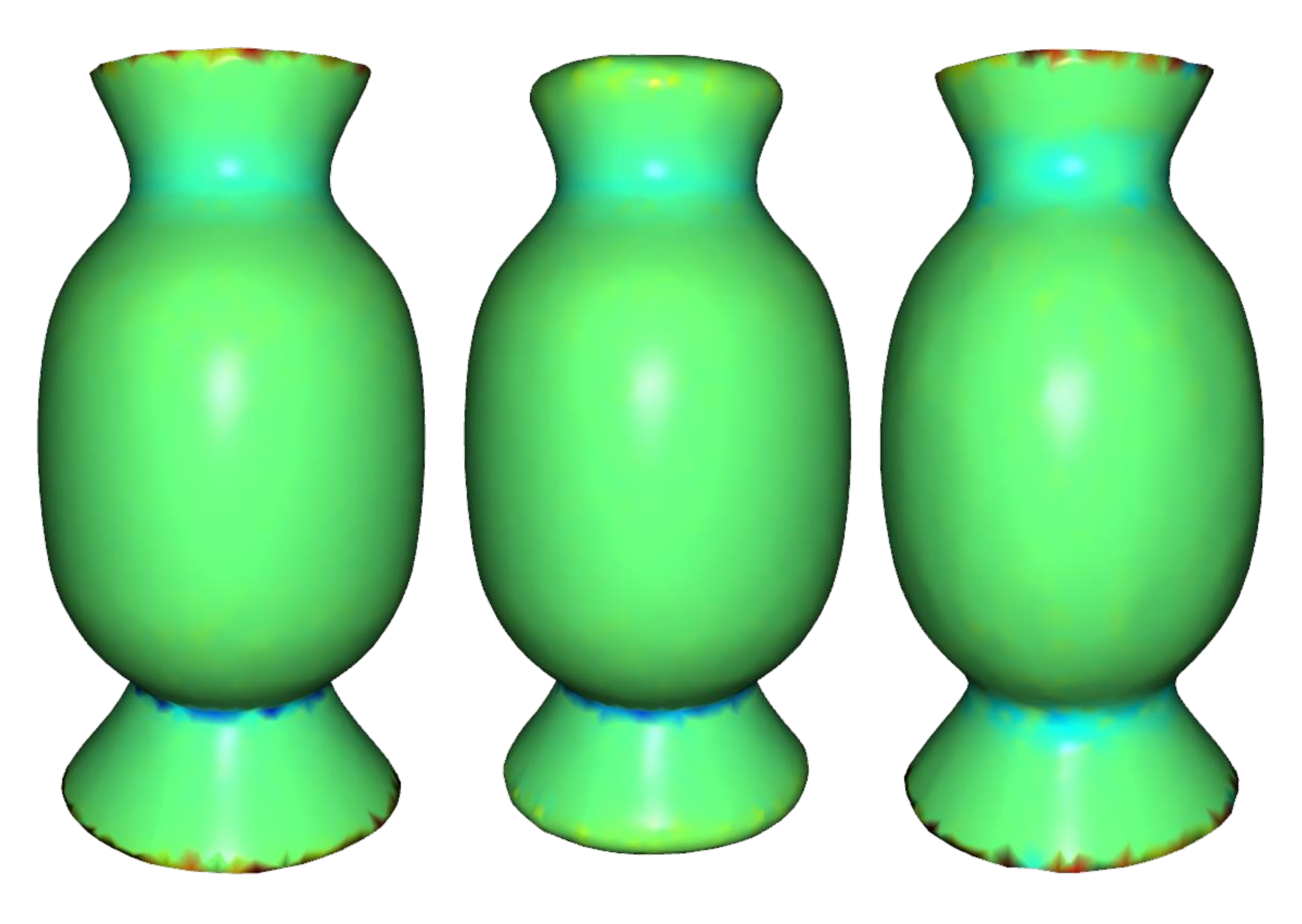}\\
		\begin{array}{ccc}
		\makebox[0.3cm]{\small Input}&
		\makebox[3.0cm]{\small GC Flow}&
		\makebox[0.3cm]{\small Ours}\\
		\end{array}&
		\begin{array}{ccc}
		\makebox[0.3cm]{\small Input}&
		\makebox[3.0cm]{\small GC Flow}&
		\makebox[0.3cm]{\small Ours}\\
		\end{array}\\
		\textrm{(a1) Max Planck (noise-free)}&\textrm{(b1) Vase (noise-free)}\\
		\setlength{\arraycolsep}{0pt}
		\begin{array}{cc}	
		\includegraphics[height=0.22\linewidth]{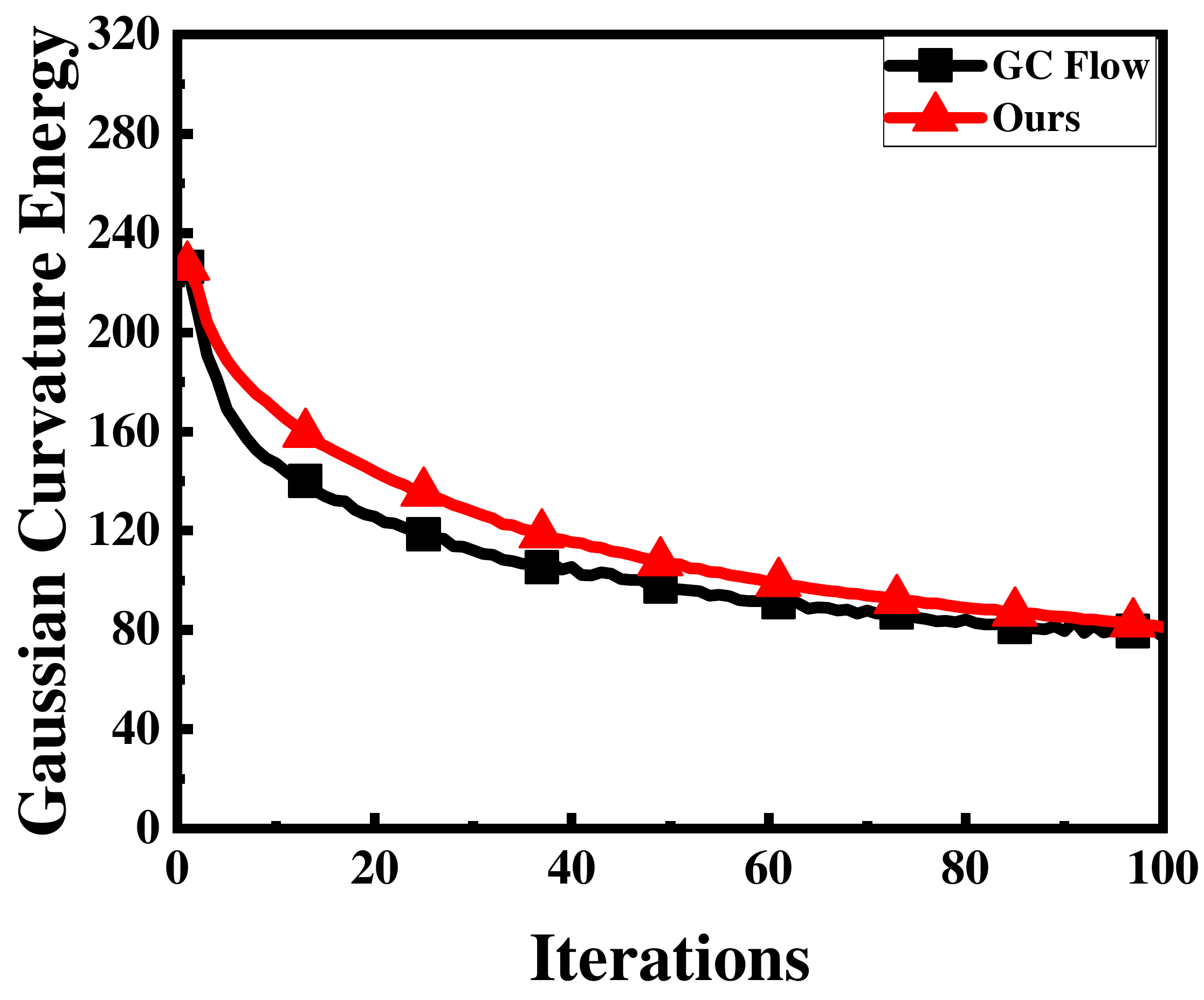}&
		\includegraphics[height=0.22\linewidth]{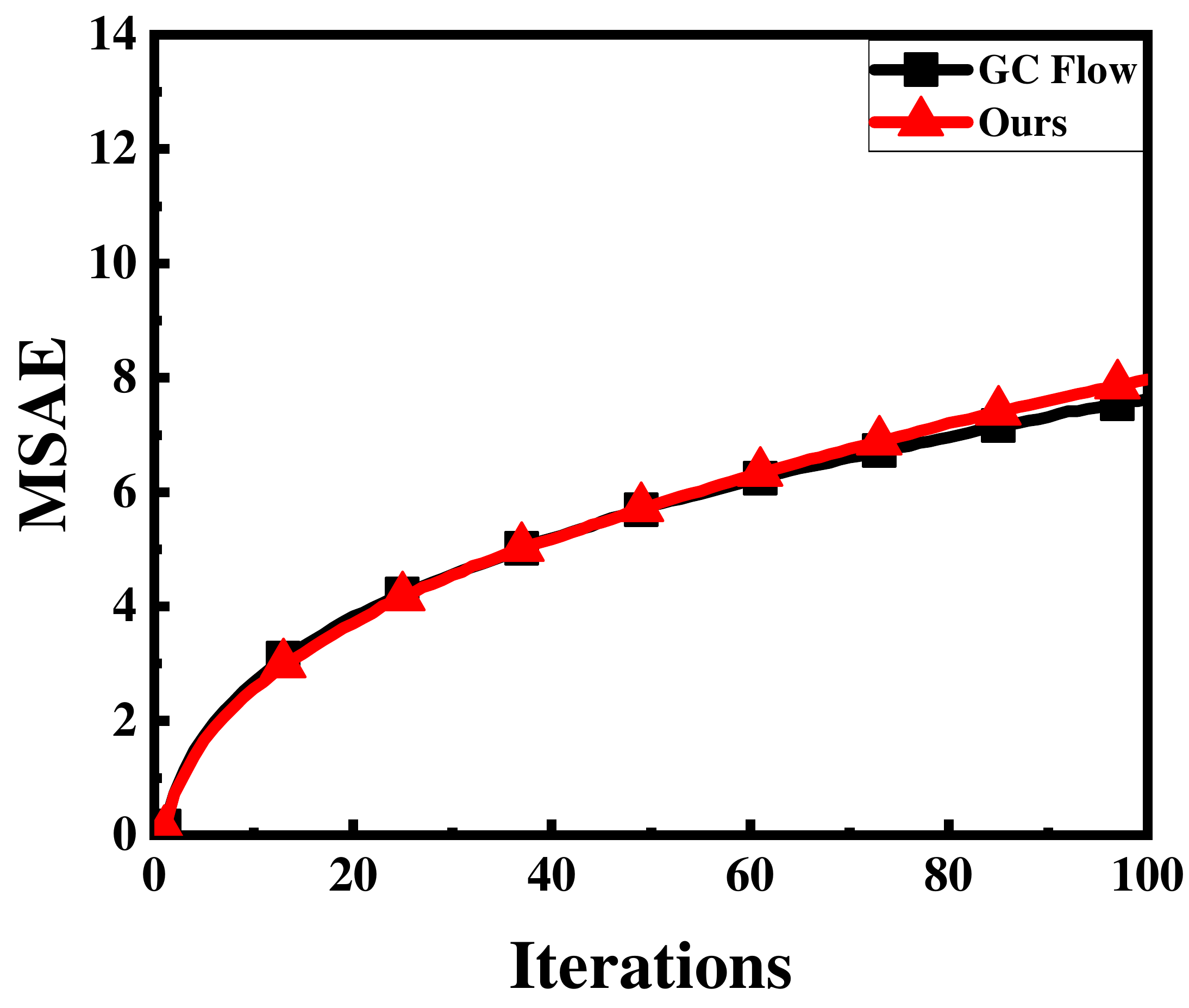}\\
		\textrm{(a2) Max Planck GCE}&\textrm{(a3) Max Planck MSAE}\\
		\end{array}&
		\setlength{\arraycolsep}{0pt}
		\begin{array}{cc}	
		\includegraphics[height=0.22\linewidth]{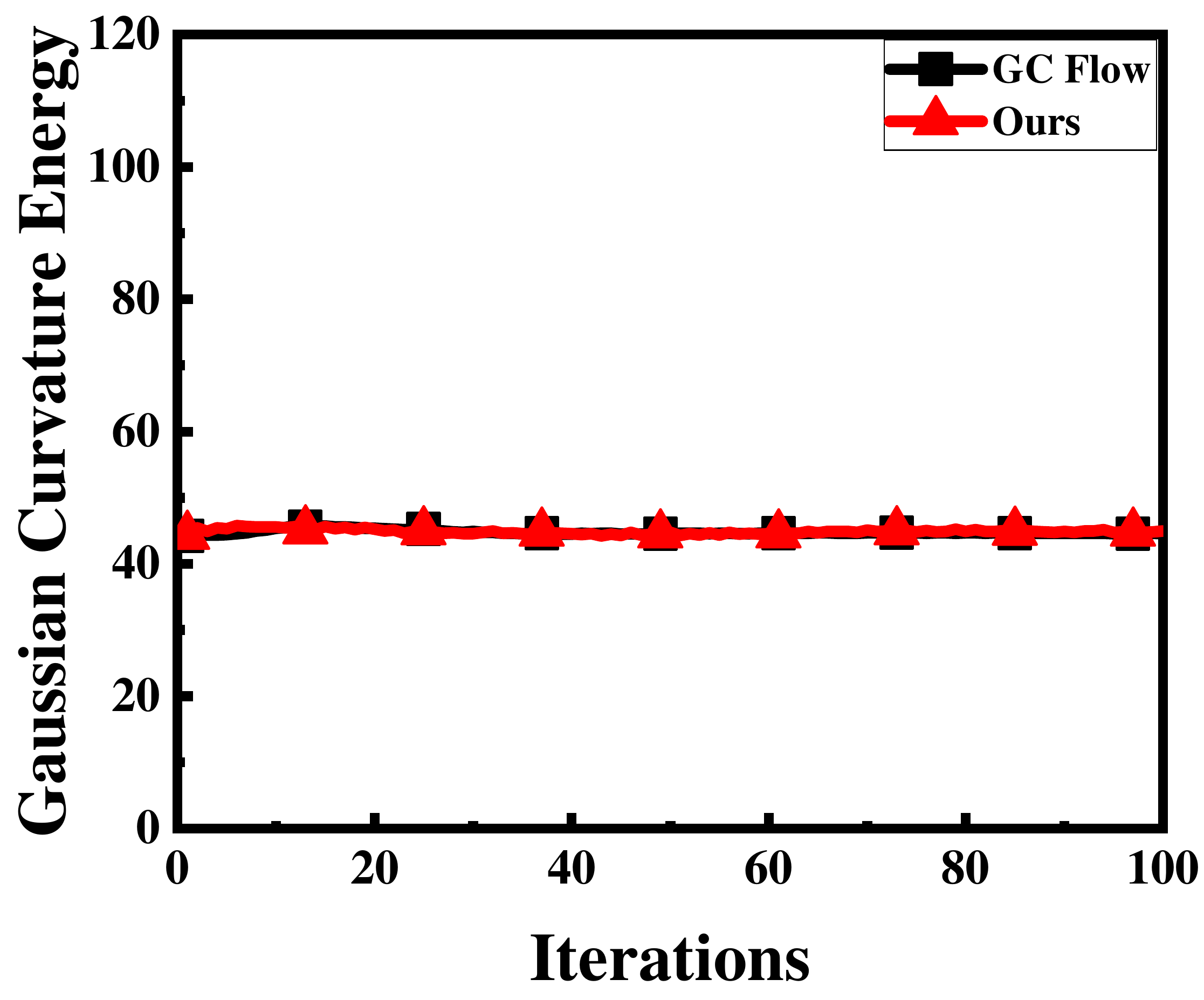}&
		\includegraphics[height=0.22\linewidth]{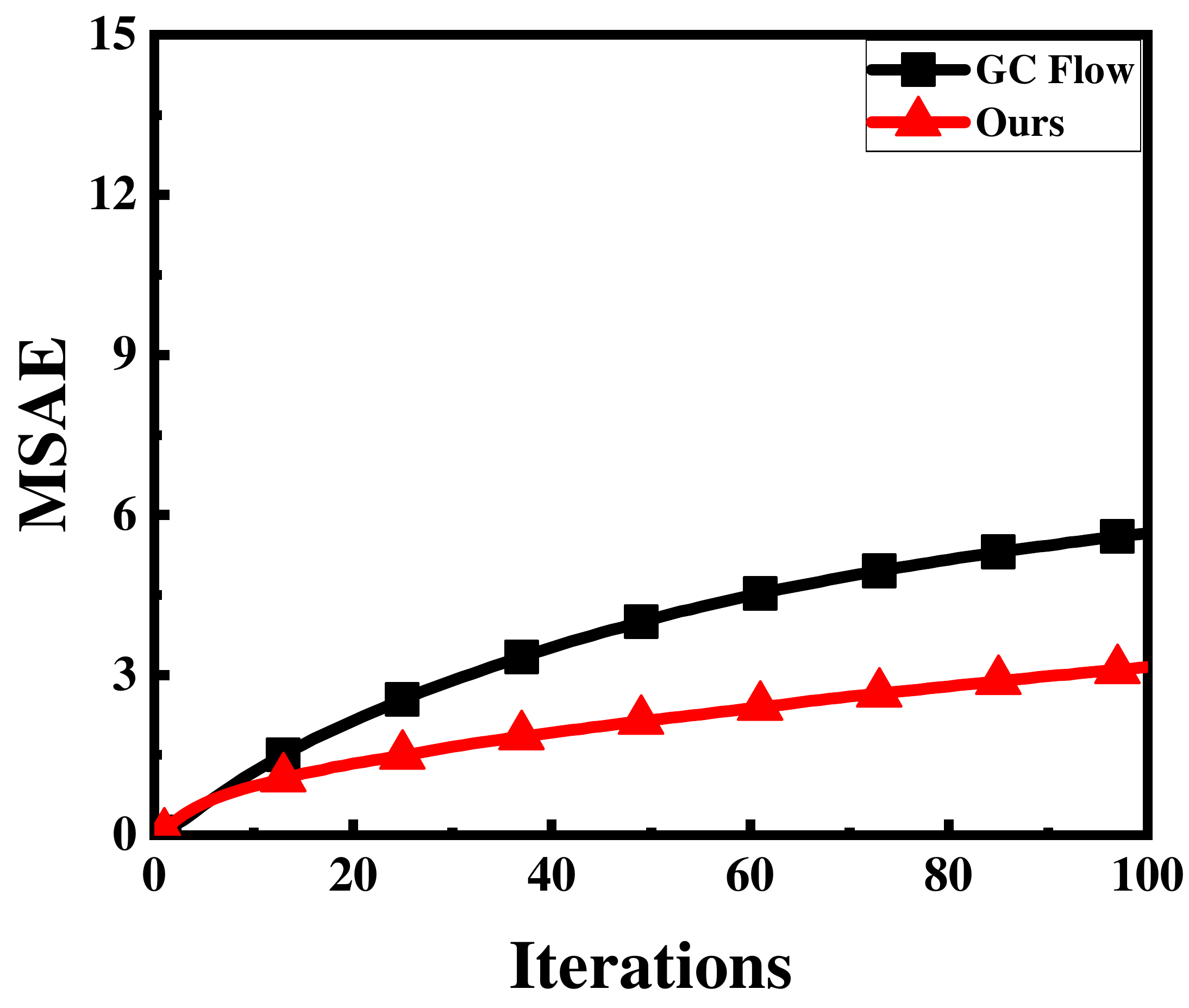}\\
		\textrm{(b2) Vase GCE}&\textrm{(b3) Vase MSAE}\\
		\end{array}\\
		\hline
		\begin{array}{c}	
		\includegraphics[height=0.3\linewidth]{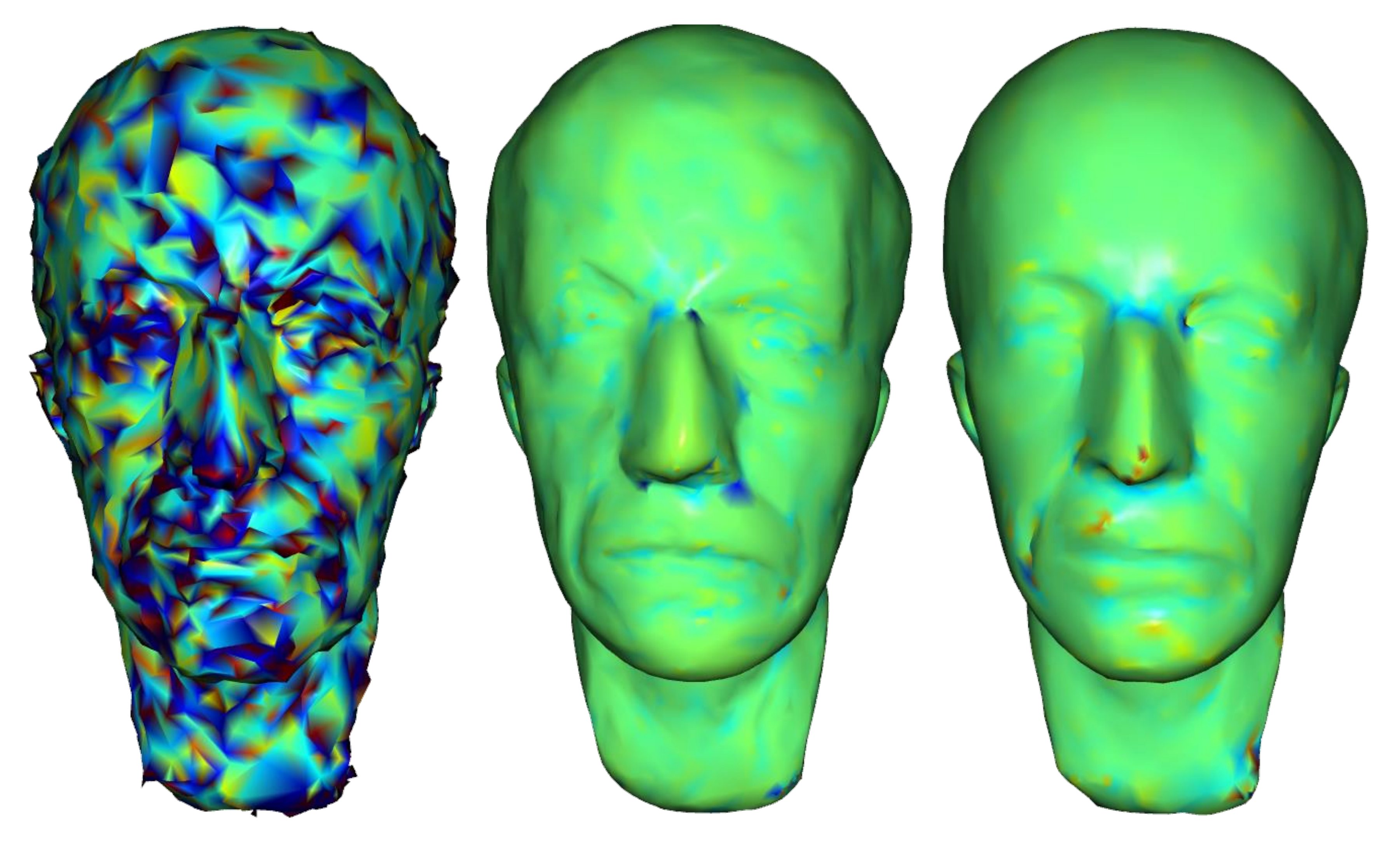}\\
		\begin{array}{ccc}
		\makebox[0.3cm]{\small Input}&
		\makebox[3.0cm]{\small GC Flow}&
		\makebox[0.3cm]{\small Ours}\\
		\end{array}\\
		\textrm{(c1) Max Planck (noise level $\sigma_{n}=0.3{e}_l$)}\\
		\end{array}&
		\begin{array}{c}	
		\includegraphics[height=0.3\linewidth]{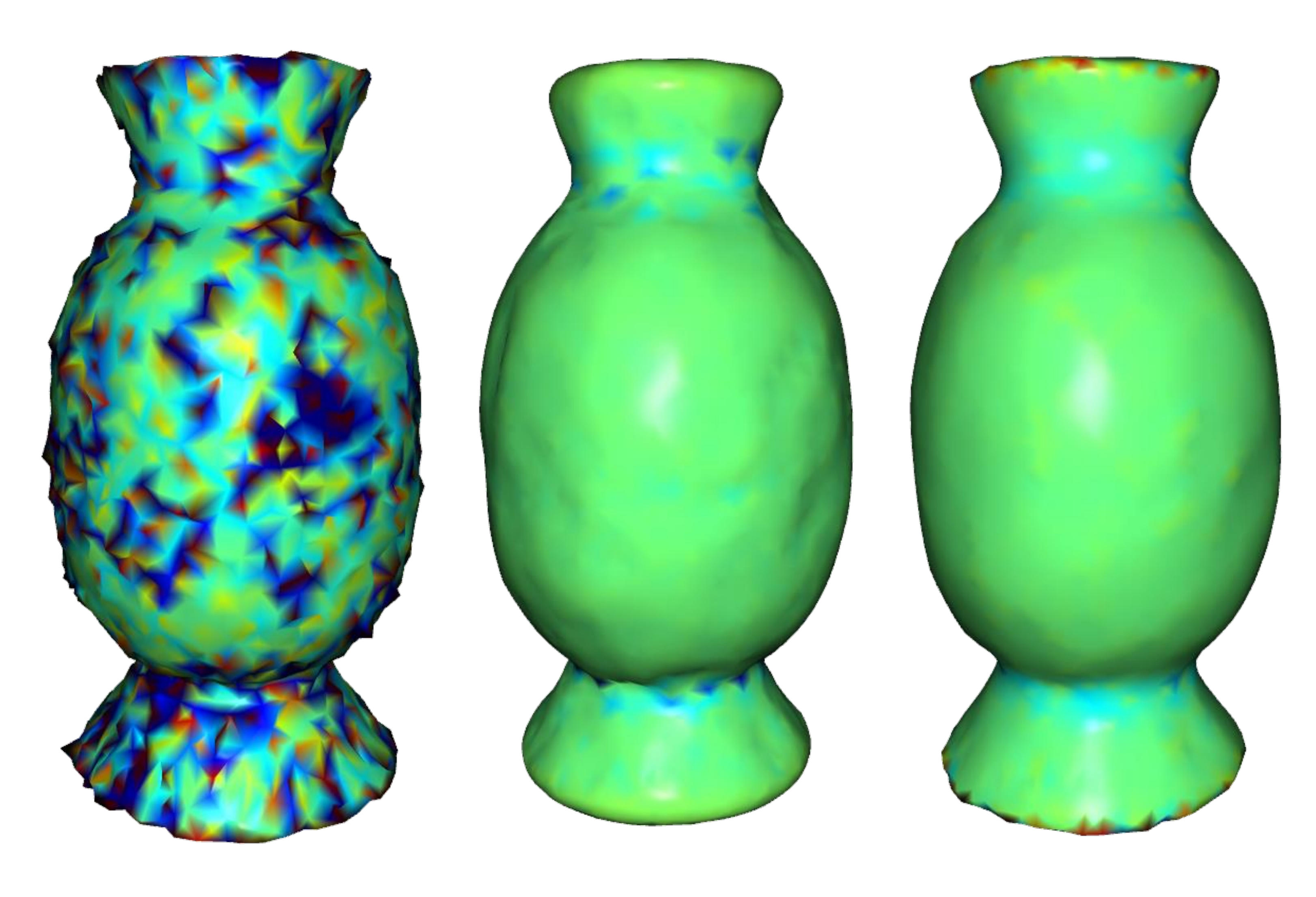}\\
		\begin{array}{ccc}
		\makebox[0.3cm]{\small Input}&
		\makebox[3.0cm]{\small GC Flow}&
		\makebox[0.3cm]{\small Ours}\\
		\end{array}\\
		\textrm{(d1) Vase (noise level $\sigma_{n}=0.3{e}_l$)}\\
		\end{array}\\
		\setlength{\arraycolsep}{0pt}
		\begin{array}{cc}	
		\includegraphics[height=0.22\linewidth]{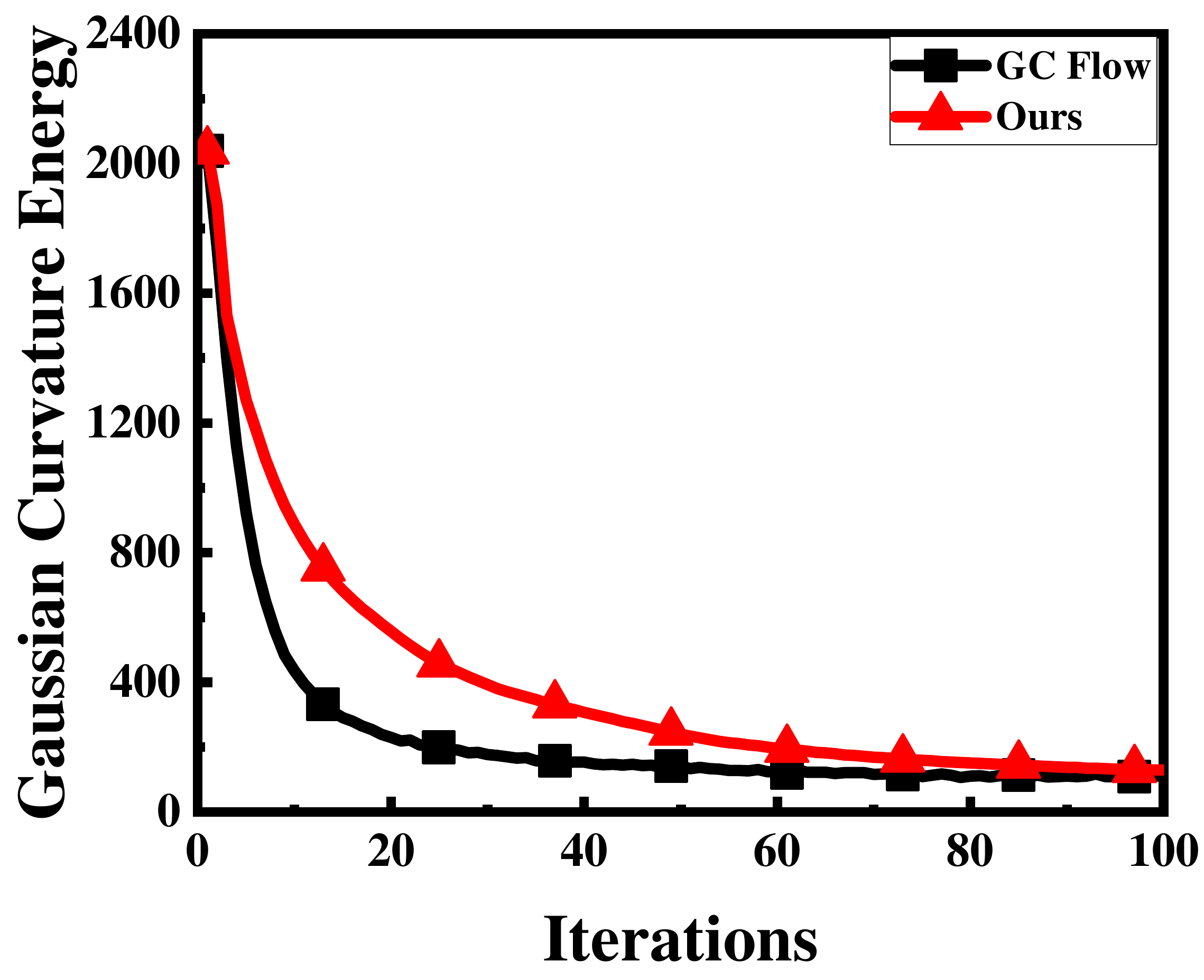}&
		\includegraphics[height=0.22\linewidth]{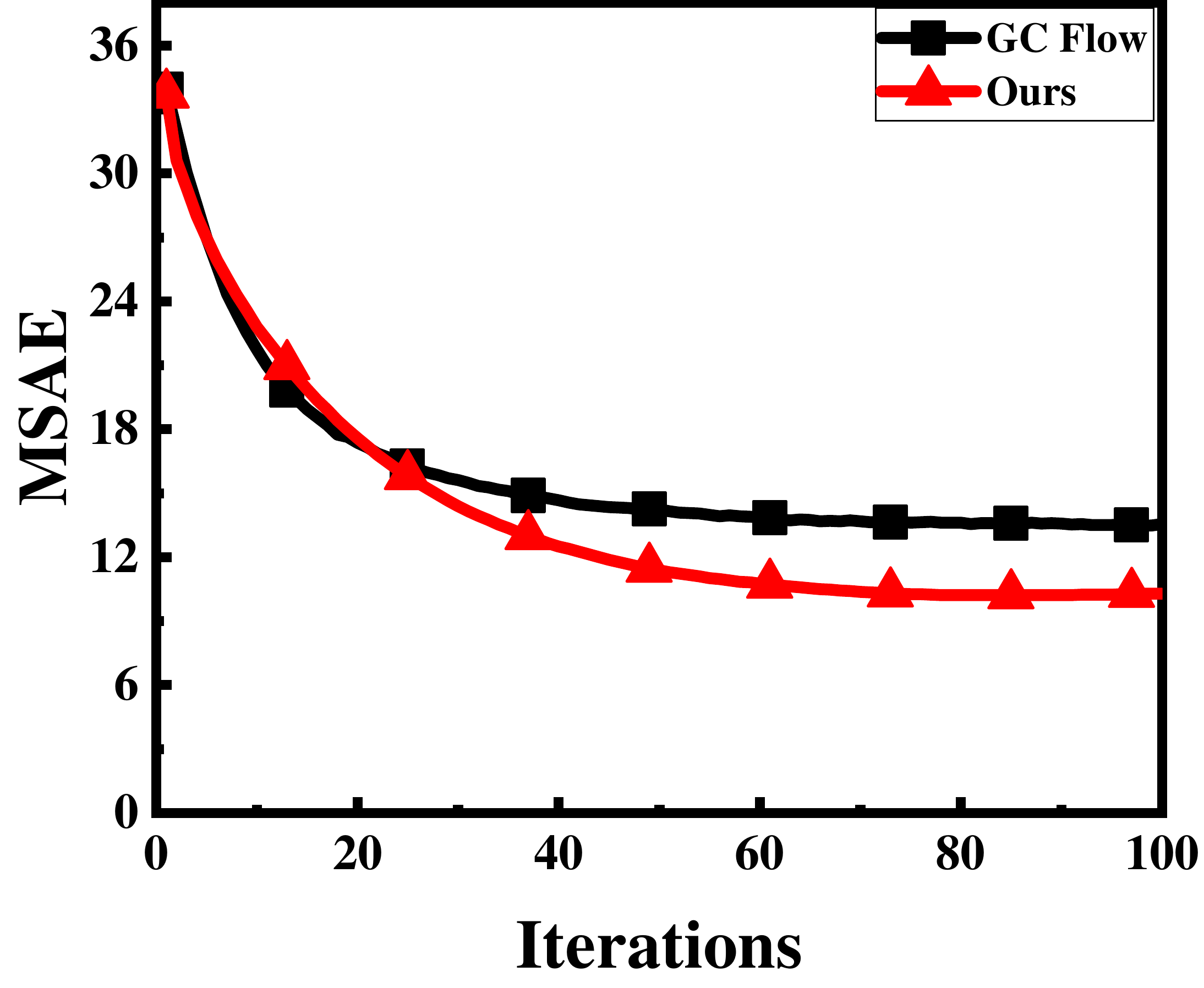}\\
		\textrm{(c2) Max Planck GCE}&\textrm{(c3) Max Planck MSAE}\\
		\end{array}&
		\setlength{\arraycolsep}{0pt}
		\begin{array}{cc}	
		\includegraphics[height=0.22\linewidth]{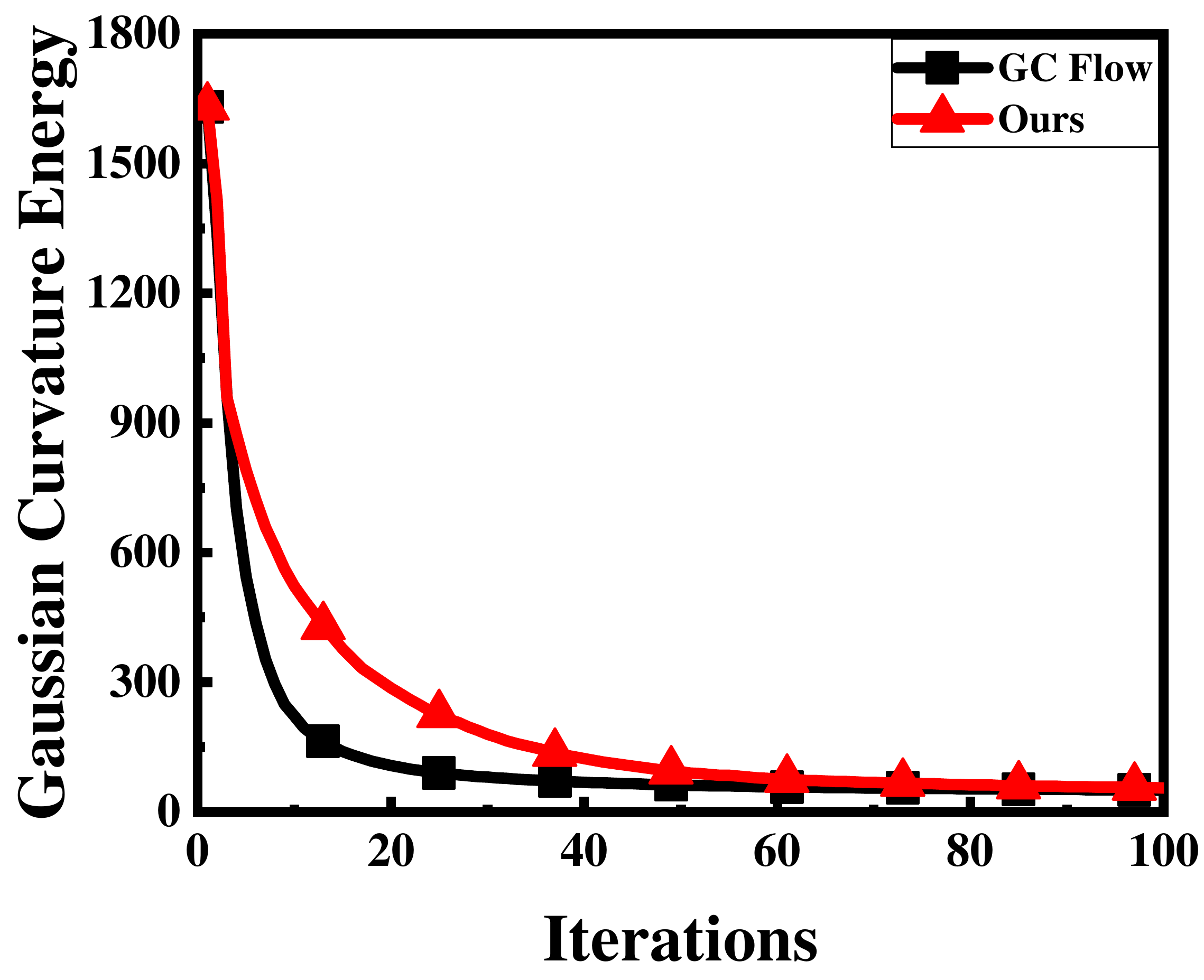}&
		\includegraphics[height=0.22\linewidth]{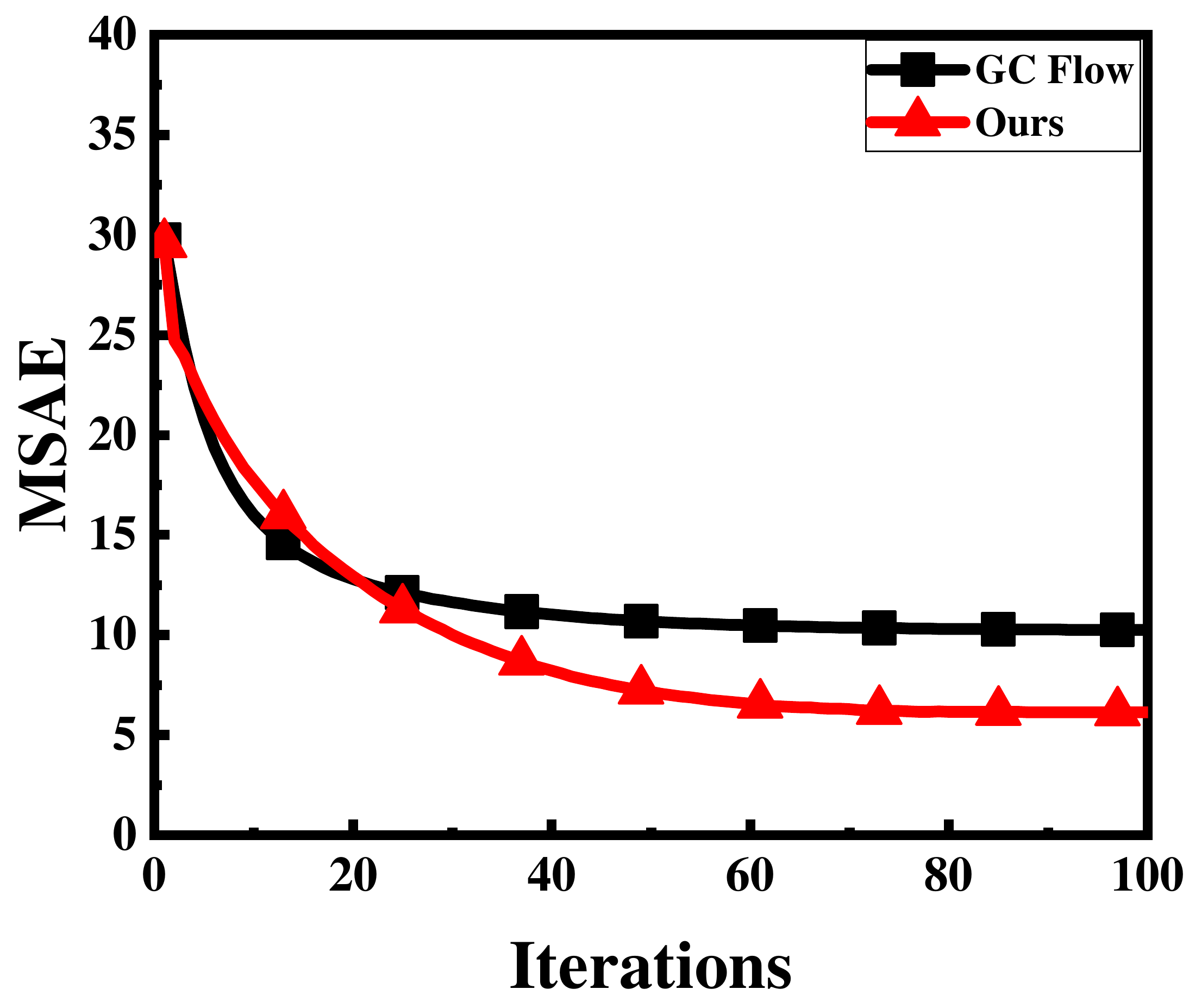}\\
		\textrm{(d2) Vase GCE}&\textrm{(d3) Vase MSAE}\\
		\end{array}\\
		\end{array}
		$}
	\caption{\label{fig:com_gcf06_100} Minimizing Gaussian curvature on noise-free meshes (top row) and noisy meshes (bottom row). In each panel, from left to right: the original, Gaussian curvature flow method, our method. GC is an abbreviation for Gaussian curvature.}
\end{figure*}

\begin{figure*}[htbp]
	\centering
	\includegraphics[height=0.07\linewidth]{graph/0_3colorbar.png}\\
	\setlength{\arraycolsep}{0pt}
	\makebox[\textwidth][c]{$
		\begin{array}{c|c}	
		\includegraphics[width=8cm]{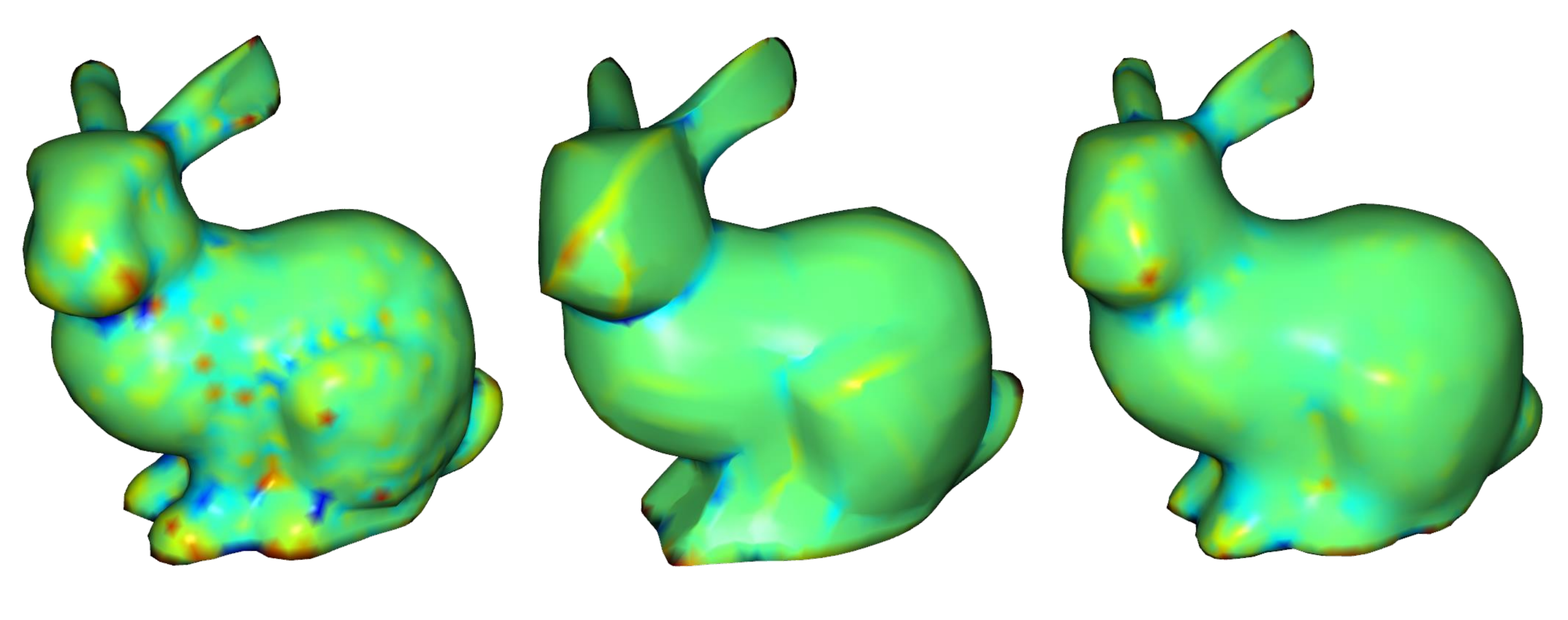}&
		\includegraphics[width=8cm]{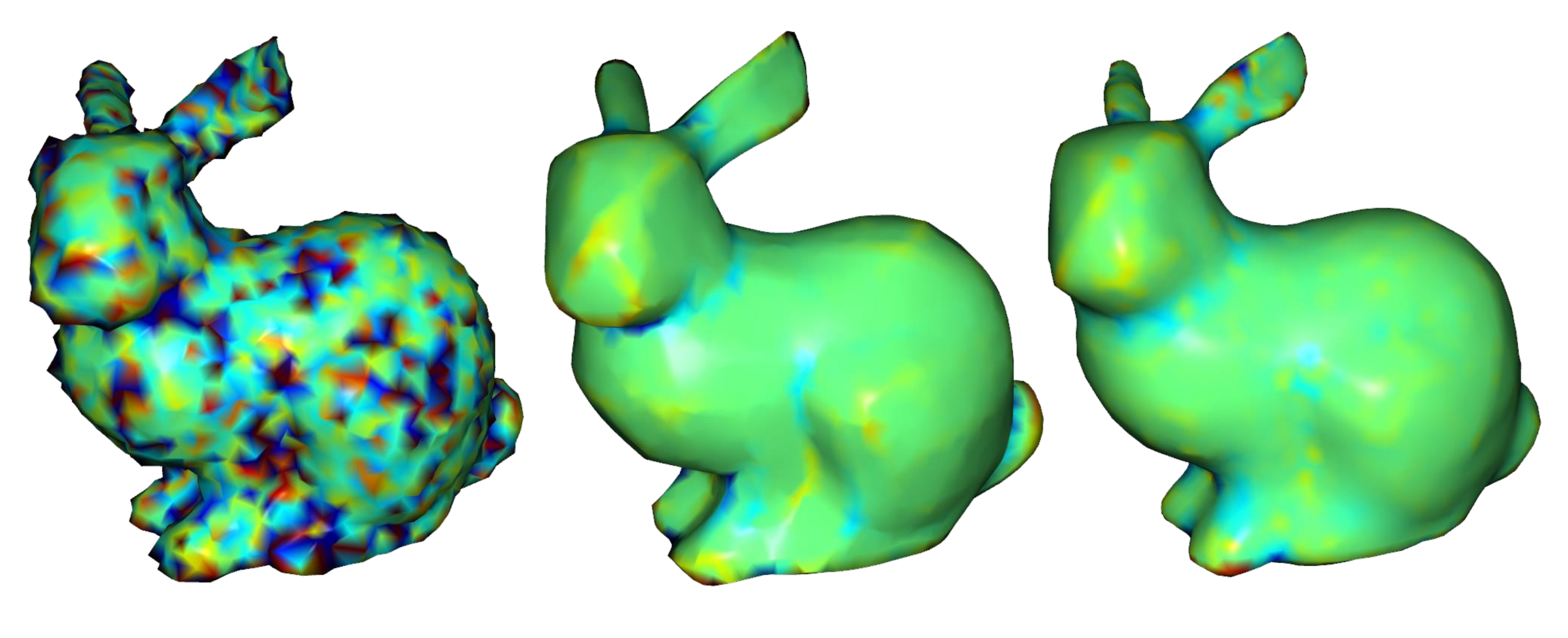}\\
		\begin{array}{ccc}
		\makebox[3.16cm]{~Input}&
		\makebox[3.16cm]{~\cite{stein2018developability}}&
		\makebox[3.16cm]{~Ours}\\	
		\end{array}&
		\begin{array}{ccc}
		\makebox[2.33cm]{~Input}&
		\makebox[2.33cm]{~\cite{stein2018developability}}&
		\makebox[2.33cm]{~Ours}\\	
		\end{array}\\
		\textrm{(a) Bunny(noise free)}&\textrm{(b) Bunny($\sigma_{n}=0.3{e}_l$)}\\
		\end{array}
		$}
	\caption{Comparison between ~\cite{stein2018developability} and our algorithm. Left: noise free case. Right: noisy case. The algorithm ~\cite{stein2018developability} converges to the comparative experimental results in the case of GCE and MSAE which are the same as GCF. The experimental results of ~\cite{stein2018developability} use the default bunny parameters provided by the author's open source. The final running time of the two algorithms is shown in Table~\ref{table3}.}
	\label{fig:crane}
\end{figure*}

\begin{figure*}[!]
	\centering
	\setlength{\arraycolsep}{0pt}
	\makebox[\textwidth][c]{$
		\begin{array}{c}
		\includegraphics[width=8.3cm]{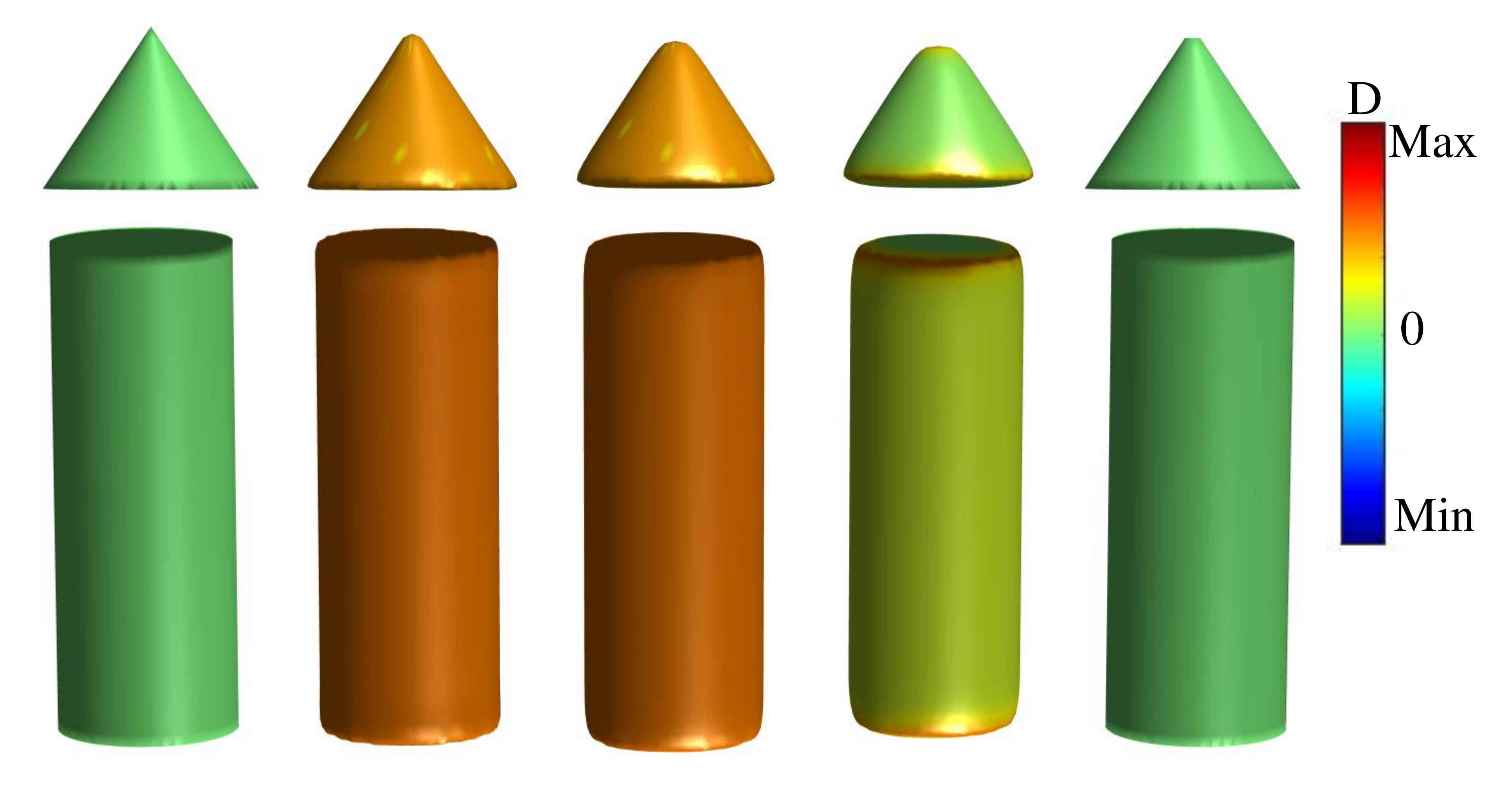}\\
		
		\begin{array}{lcccc}
			\makebox[1.5cm]{~Input}&
			\makebox[1.8cm]{~\cite{vollmer1999improved}}&
			\makebox[1.4cm]{~\cite{desbrun1999implicit}}&
			\makebox[1.4cm]{~\cite{sorkine2005laplacian}}&
			\makebox[1.9cm]{~Ours}\\
		\end{array}\\
		
		Cone~~~\left \{~~
		\begin{aligned}
			\begin{array}{ccccc}
				\makebox[0.5cm]{~$MSAE:$~}&
				\makebox[1.4cm]{~$5.50$}&
				\makebox[1.4cm]{~$9.45$}&
				\makebox[1.8cm]{~$12.92$}&
				\makebox[1.8cm]{~$\pmb{1.69}$}\\
				\makebox[0.5cm]{~$\mathscr{D}$$_{mean:}$~}&
				\makebox[1.4cm]{~$29.34$}&
				\makebox[1.4cm]{~$29.13$}&
				\makebox[1.8cm]{~$1.93$}&
				\makebox[1.8cm]{~$\pmb{0.00}$}\\
				\makebox[0.5cm]{~$\mathscr{D}$$_{max:}$~}&
				\makebox[1.4cm]{~$53.34$}&
				\makebox[1.4cm]{~$52.31$}&
				\makebox[1.8cm]{~$7.41$}&
				\makebox[1.8cm]{~$\pmb{3.10}$}\\
			\end{array}
		\end{aligned}
		\right.\\
		
		Cylinder\left \{~~
		\begin{aligned}	
			\begin{array}{ccccc}
				\makebox[0.5cm]{~$MSAE:$~}&
				\makebox[1.4cm]{~$2.90$}&
				\makebox[1.4cm]{~$4.18$}&
				\makebox[1.8cm]{~$5.32$}&
				\makebox[1.8cm]{~$\pmb{0.08}$}\\
				\makebox[0.5cm]{~$\mathscr{D}$$_{mean:}$~}&
				\makebox[1.4cm]{~$73.89$}&
				\makebox[1.4cm]{~$74.20$}&
				\makebox[1.8cm]{~$1.69$}&
				\makebox[1.8cm]{~$\pmb{0.00}$}\\	
				\makebox[0.5cm]{~$\mathscr{D}$$_{max:}$~}&
				\makebox[1.4cm]{~$158.99$}&
				\makebox[1.4cm]{~$160.56$}&
				\makebox[1.8cm]{~$5.57$}&
				\makebox[1.8cm]{~$\pmb{0.33}$}\\
			\end{array}
		\end{aligned}
		\right.\\
		
	\end{array}$}
\caption{Comparison between our algorithm and three common Laplacian smoothing algorithms on the model of cone and cylinder. The number of iterations of our method is 40, and the others default to 10.}
\label{Fig.hl_compare}
\end{figure*}
\begin{figure}[htbp]
\centering
\subfigure[GCF without GDD~(${\cal_{ACS}}:-1.45$)]{
	\includegraphics[width=5cm]{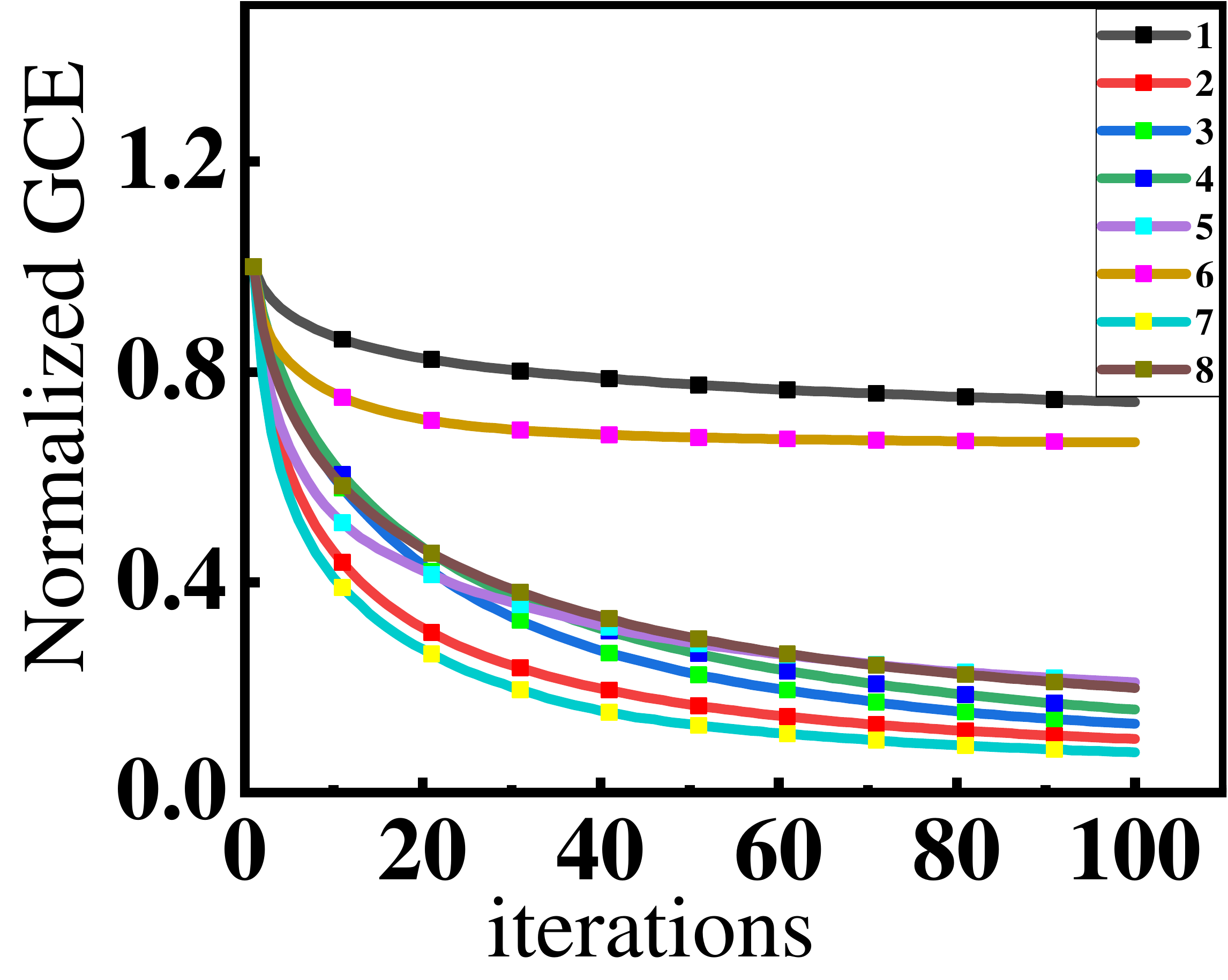}}
\subfigure[GCF with GDD~(${\cal_{ACS}}:\pmb{-1.70}$)]{
	\includegraphics[width=5cm]{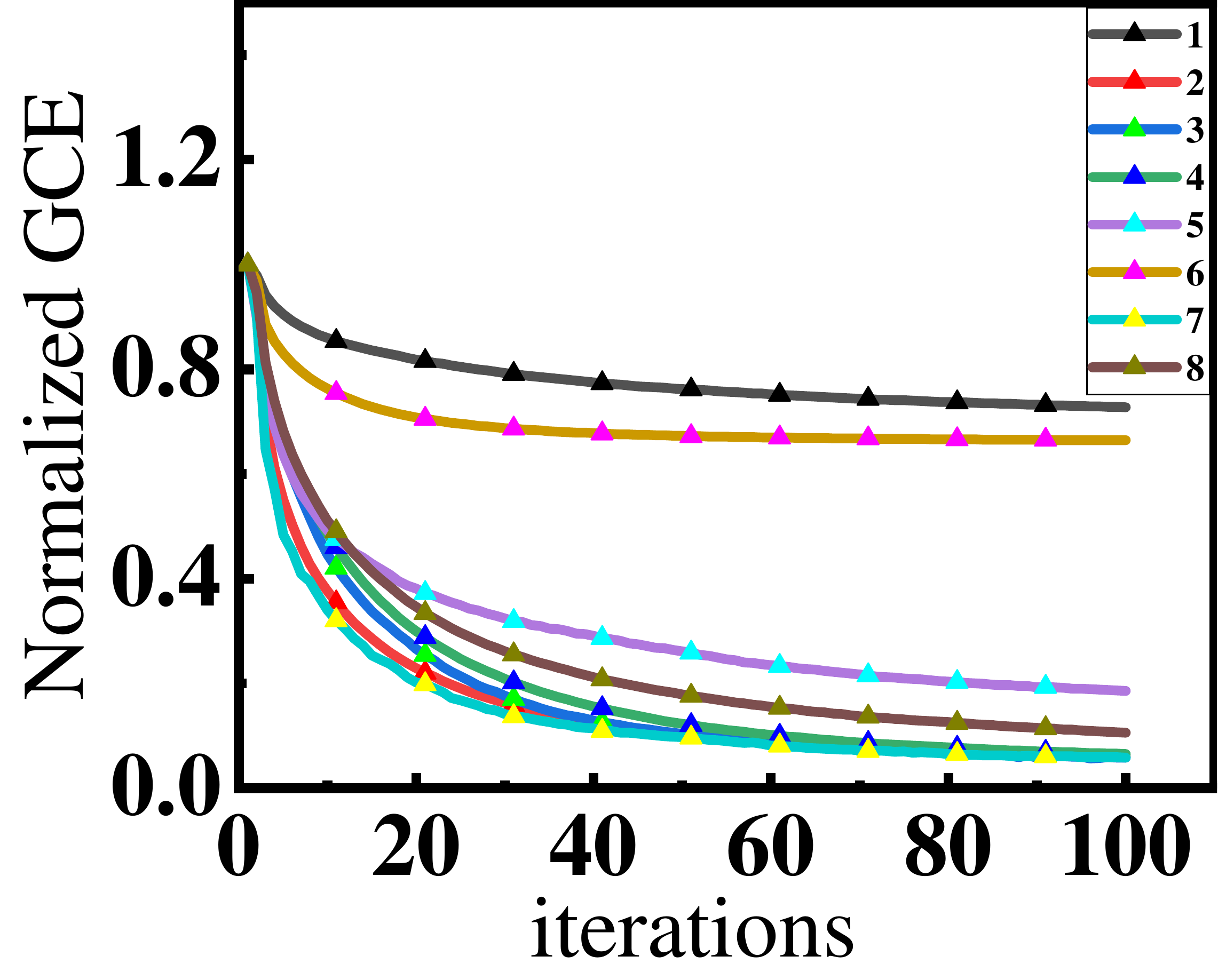}}
\caption{ The influence of GDD on the convergence of GCF. We randomly selected 8 different meshes, GCF with GDD and without GDD for 100 iterations respectively. For comparison, we normalize GCE(~Divided by the largest Gaussian curvature energy value). The smaller the average convergence slope ACS~\ref{eq:ACS}, the better.}
\label{Fig.the convergence average slope}
\end{figure}	
\begin{table*}[!]
\centering
\caption{Running time(~in debug mode) comparison between the standard Gaussian curvature flow method and our method on the Max Planck ($\left|V\right|$: 5272, $\left|F\right|$: 1054) and Vase ($\left|V\right|$: 3827, $\left|F\right|$: 7650).}\label{table1}
\resizebox{1.0\textwidth}{!}{ 
\begin{tabular}{cccc}\hline
Models & Methods & Parameters  & \tabincell{l}{Time($s$):W/O GDD}\\
\hline
\tabincell{l}{Max Planck \\(noise-free)} &
\tabincell{l}{ 
	\cite{zhao2006triangular} \\
	Ours\\
}
& \tabincell{c}{($100,1.0,2,0.001,0.0005$)\\($100$)\\}
& \tabincell{c}{$163.94$/$1016.41$\\$\pmb{56.39}$/$428.58$\\}
\\ \hline
\tabincell{l}{Vase \\(noise-free)} &
\tabincell{l}{ 
	\cite{zhao2006triangular} \\
	Ours\\
}
& \tabincell{c}{($100,0.01,2,0.001,0.0005$)\\($100$)\\}
& \tabincell{c}{$100.65$/$744.75$\\$\pmb{40.34}$/$320.48$\\}
\\ \hline
\tabincell{l}{Max Planck \\($\sigma_{n}=0.3{e}_l$)} &
\tabincell{l}{ 
	\cite{zhao2006triangular} \\
	Ours\\
}
& \tabincell{c}{($100,1.0,2,0.001,0.0005$)\\($100$)\\}
& \tabincell{c}{$158.23$/$1123.56$\\$\pmb{62.41}$/$425.20$\\}
\\ \hline
\tabincell{l}{Vase \\($\sigma_{n}=0.3{e}_l$)} &
\tabincell{l}{ 
	\cite{zhao2006triangular} \\
	Ours\\
}
& \tabincell{c}{($100,0.01,2,0.001,0.0005$)\\($100$)\\}
& \tabincell{c}{$106.81$/$833.12$\\$\pmb{45.61}$/$317.00$\\}
\\
\hline           
\end{tabular}
}
\end{table*}

\begin{table*}[!]
\centering
\caption{Quantitative comparison of ~\cite{stein2018developability} on the bunny ($\left|V\right|$: 3301, $\left|F\right|$: 6598). The running time(~in debug mode) is measured in seconds.}\label{table3}
\resizebox{1.0\textwidth}{!}{ 
\begin{tabular}{ccccccc}\hline
Models & Methods & \tabincell{c}{MSAE} & \tabincell{c}{GCE}& \tabincell{c}{$\mathscr{D}$$_{mean:}$}& \tabincell{c}{$\mathscr{D}$$_{max:}$}& \tabincell{c}{Time(s)}\\
\hline
\tabincell{c}{Bunny\\(noise-free)} &
\tabincell{c}{ 
	~\cite{stein2018developability} \\
	~Ours \\
}
& \tabincell{l}{$9.68$\\$10.65$\\}
& \tabincell{l}{$71.07$\\$71.83$\\}
& \tabincell{l}{$\pmb{8.21\times10^{-3}}$\\$1.21\times10^{-2}$\\}
& \tabincell{l}{$6.07\times10^{-2}$\\$\pmb{5.75\times10^{-2}}$\\}
& \tabincell{l}{$1129.10$\\$\pmb{25.34}$\\}
\\ \hline
\tabincell{c}{Bunny\\($\sigma_{n}=0.3{e}_l$)} &
\tabincell{c}{ 
	~\cite{stein2018developability} \\
	~Ours \\
}
& \tabincell{l}{$10.10$\\$10.79$\\}
& \tabincell{l}{$100.82$\\$100.00$\\}
& \tabincell{l}{$\pmb{1.03\times10^{-2}}$\\$1.20\times10^{-2}$\\}
& \tabincell{l}{$6.75\times10^{-2}$\\$\pmb{4.00\times10^{-2}}$\\}
& \tabincell{l}{$1303.93$\\$\pmb{28.80}$\\}
\\   \hline           
\end{tabular}
}
\end{table*}	

\subsubsection{Result analysis}
In Fig.~\ref{fig:com_gcf06_100}, we use two meshes, Max Planck head and Vase. We perform the Gaussian curvature flow~\cite{ zhao2006triangular} and GCF on these meshes, respectively.

Visually, on the noise-free model: the result of~\cite{ zhao2006triangular} is not much different from ours in the color of Gaussian curvature. But in detail, for the noise-free model, we can see that on the vase model, the feature preservation of our result is more obvious. On the noisy model: Our algorithm is not only more prominent in feature preservation, but also smoother in local details.
\par Quantitatively, our algorithm can achieve the same Gaussian curvature energy as~\cite{ zhao2006triangular}, but  lower than~\cite{ zhao2006triangular} on MSAE, which also proves that our algorithm is not only minimizing Gaussian curvature energy, but also preserving ground truth features. Our algorithm has distinguishing advantage in time consumption. In Table~\ref{table1}, we can see that without adding our GDD in~\cite{ zhao2006triangular}, we are almost 20 times faster. After adding our GDD to ~\cite{ zhao2006triangular}, our algorithm is also nearly 3 times faster (the experiment is on the same computer: Intel Xeon 4 cores, 3.7Ghz, 96GB RAM).
\par The weakness of ~\cite{ zhao2006triangular} is that the temporal direction discretization currently used is explicit, the inherent problem with this approach is that explicit methods behave poorly if the system is stiff, and in order to converge to the correct solution it is necessary to use small time steps. So there are too many parameters, and different models are sensitive to the size of the parameter values.
In Fig.~\ref{fig:com_gcf06_100}, combining quantitative and visual results, we can draw the following conclusions: Firstly, we can obtain the same Gaussian curvature energy value as ~\cite{ zhao2006triangular}, which proves that we can achieve the effect of explicit calculation optimization through implicit optimization. Secondly, there is no need to explicitly perform the calculation of Gaussian curvature and the addition of GDD, which makes our algorithm gain a clear advantage in time-consuming. Thirdly, our algorithm can simulteneously smooth the mesh and preserve its features.


\par We further compare our method with the developed method~\cite{stein2018developability}. The results are shown in Fig.~\ref{fig:crane}. In this comparative experiment, we do not make a developable comparison with it, because our work is not focused on developing. We run ~\cite{stein2018developability} according to the model and default parameters given by the author until convergence. We run our algorithm to roughly the same result (GCE, MSAE), and then compare $D_{max}$, $D_{mean}$ and time-consuming between the processed model and the ground truth model. 
\par Method~\cite{stein2018developability} gathers Gaussian curvature to regular seam curves by defining different energies, and through continuous iterative optimization to achieve piecewise developable surfaces. Our method does not introduce developable energy constraints, so there is no regular seam curves. But to a certain extent, optimization of Gaussian curvature energy can be compared. From Fig.~\ref{fig:crane} and Table~\ref{table3}, we can see that our algorithm can achieve the same effect as~\cite{stein2018developability} in optimizing Gaussian curvature, and it takes less time.

\subsection{smoothing with Feature Preservation}
To evaluate GCF performs in smoothing and feature preservation, we choose seven representative state-of-the-art methods for comparison. 
\subsubsection{Parameter settings}
From the previous multiple sets of experiments, our algorithm roughly changed after 40-50 iterations. For some models, in order to better balance smoothing and feature preservation, 40 iterations are generally selected. For large noise, the number of iterations can be increased according to the noise level. Both the filter based methods and the optimization based methods use their default parameters(except for the Bunny). The detailed parameters are listed in Table~\ref{table2}. The meanings of the specific parameters can be found in the respective papers. 

\begin{figure*}[htbp]
	\centering
	\setlength{\arraycolsep}{0pt}
	\makebox[\textwidth][c]{$
		\begin{array}{cccccccccc}	
		\includegraphics[width=1.72cm]{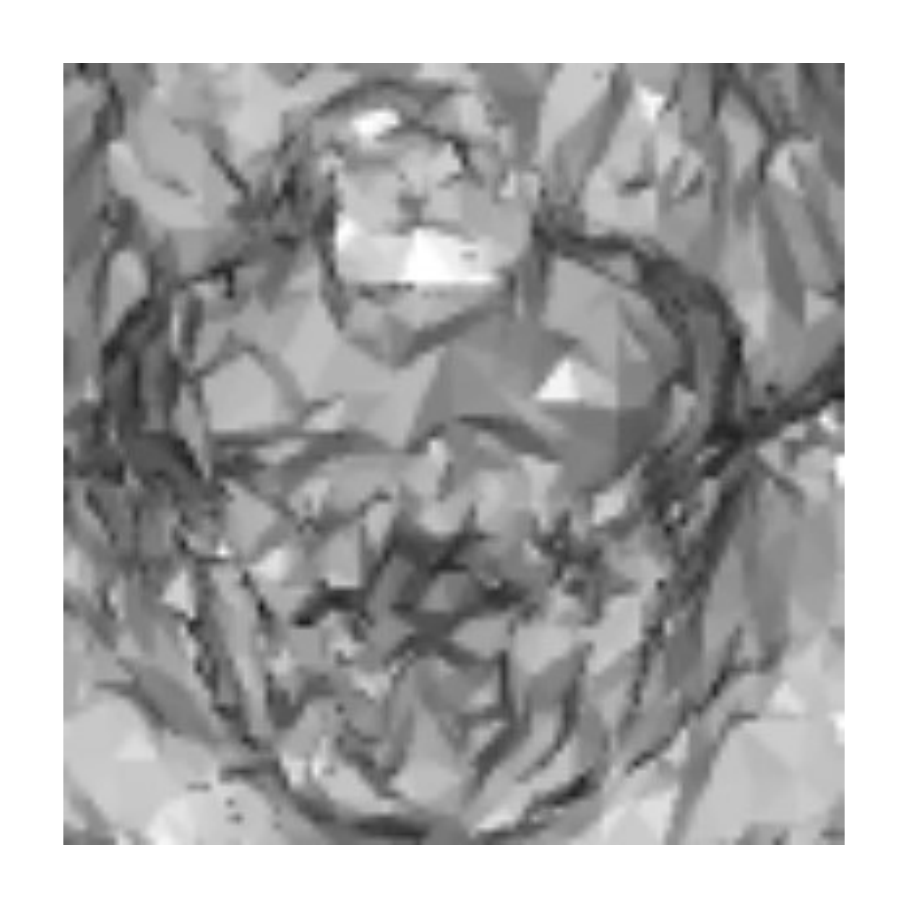}&
		\includegraphics[width=1.72cm]{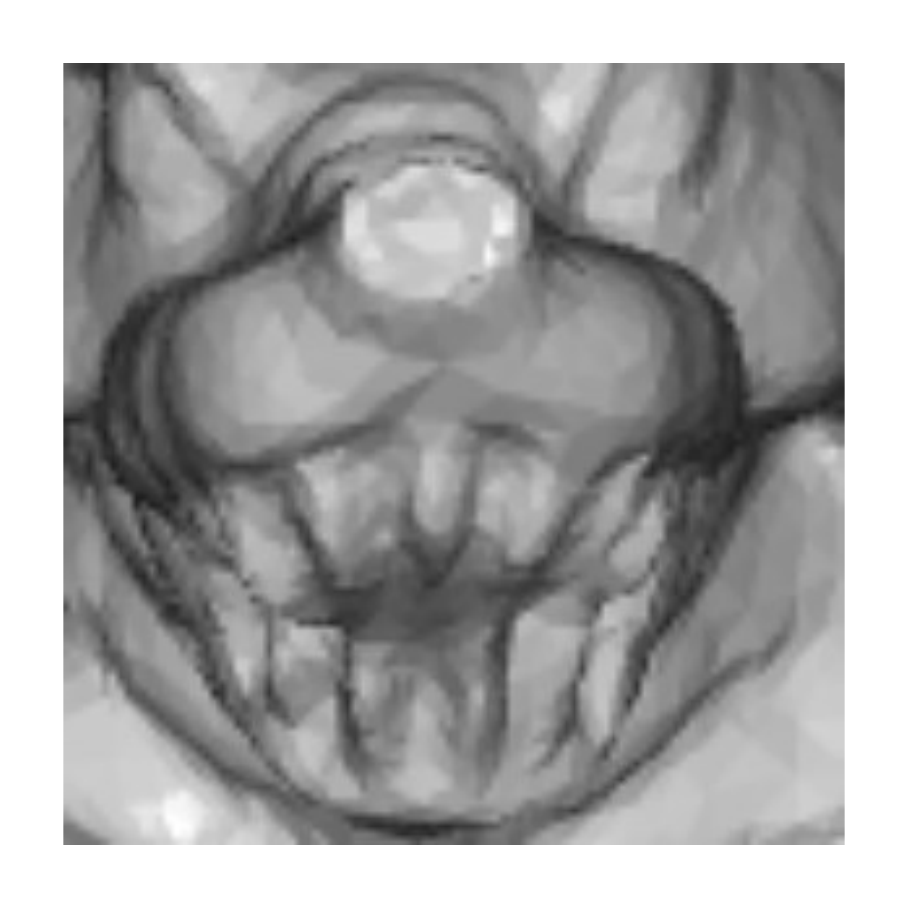}&
		\includegraphics[width=1.72cm]{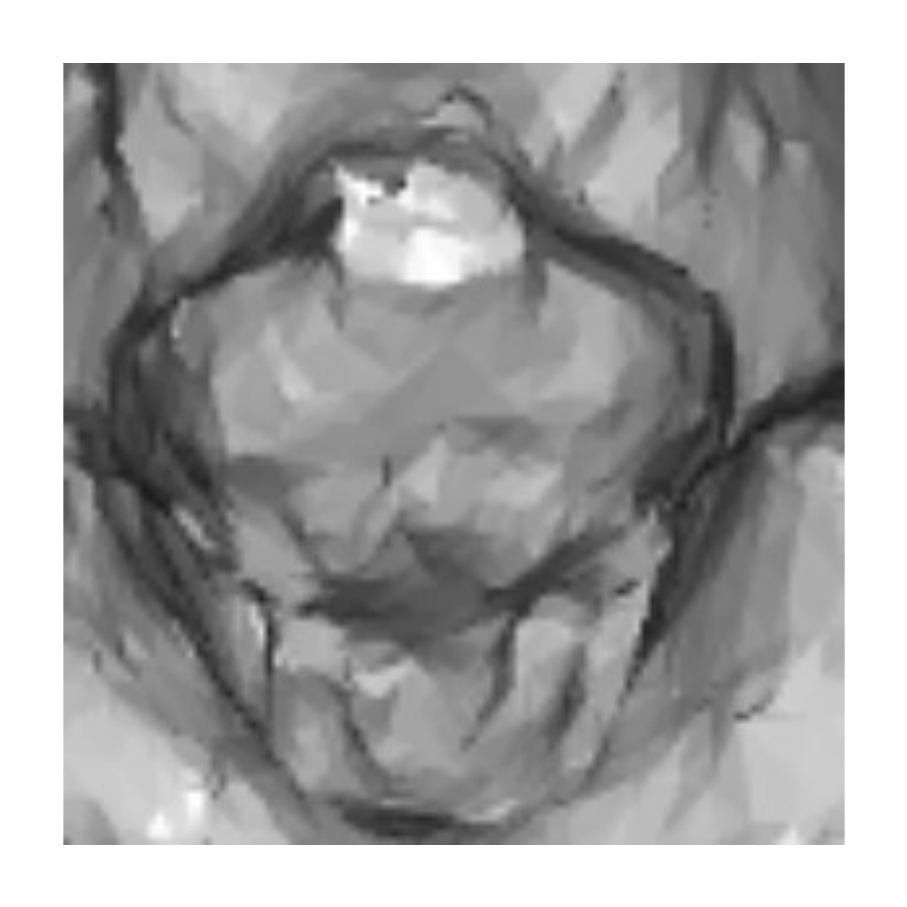}&
		\includegraphics[width=1.72cm]{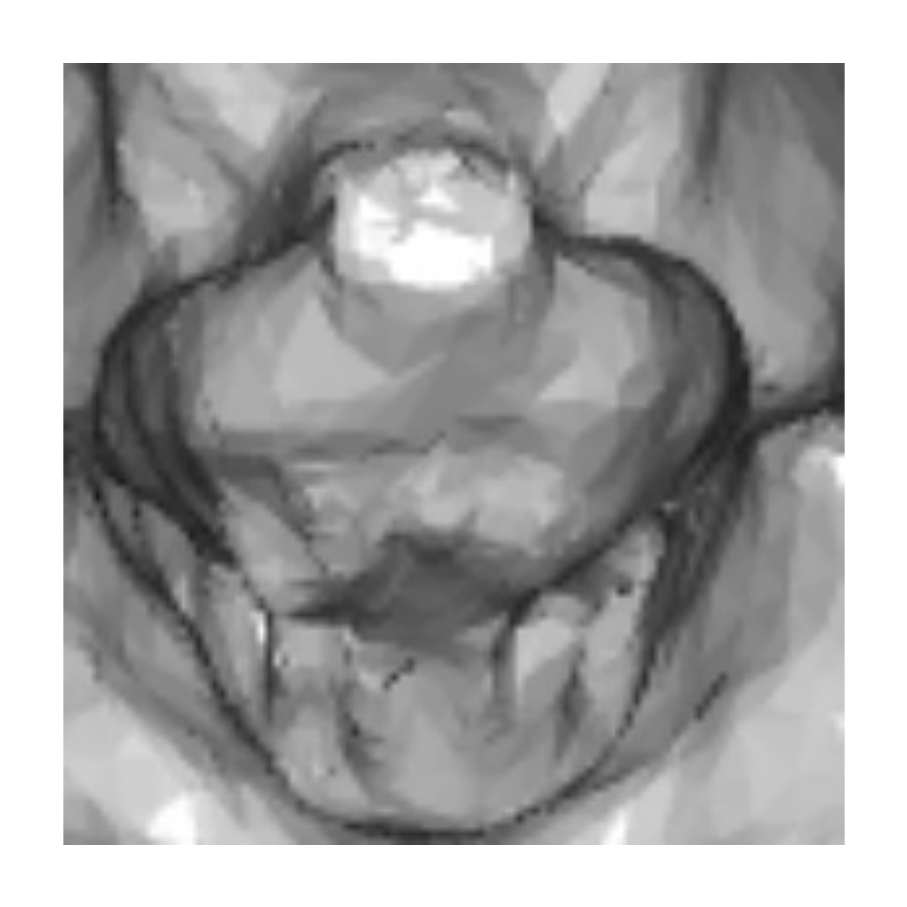}&
		\includegraphics[width=1.72cm]{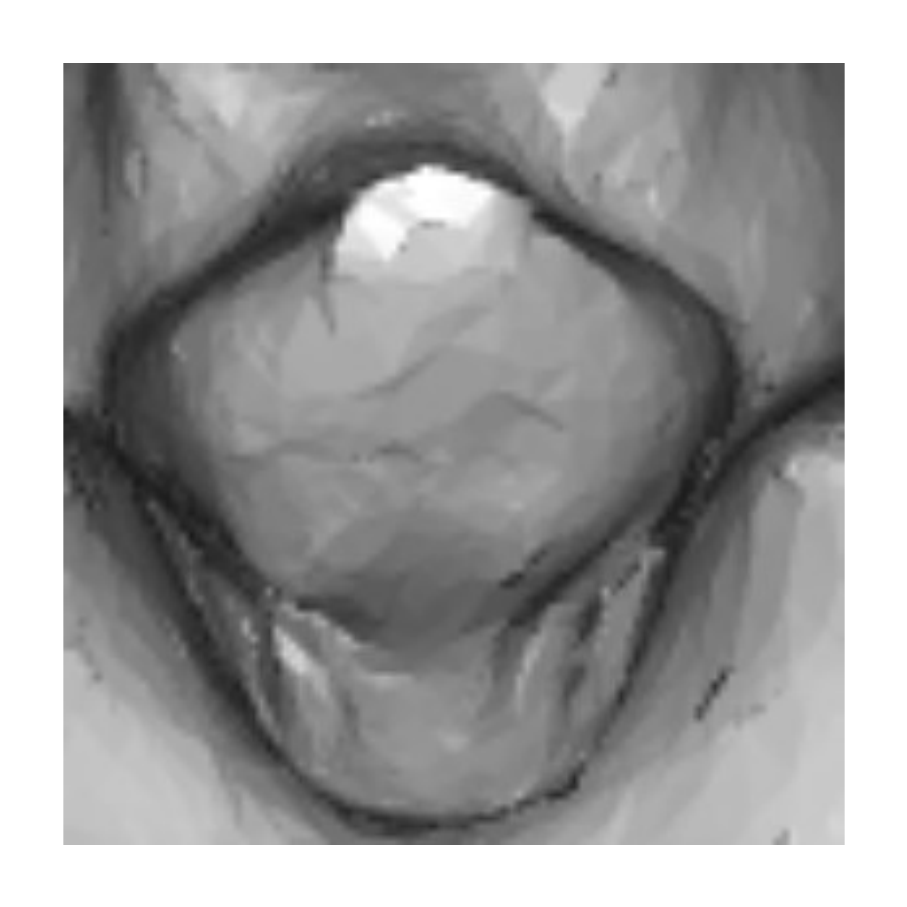}&
		\includegraphics[width=1.72cm]{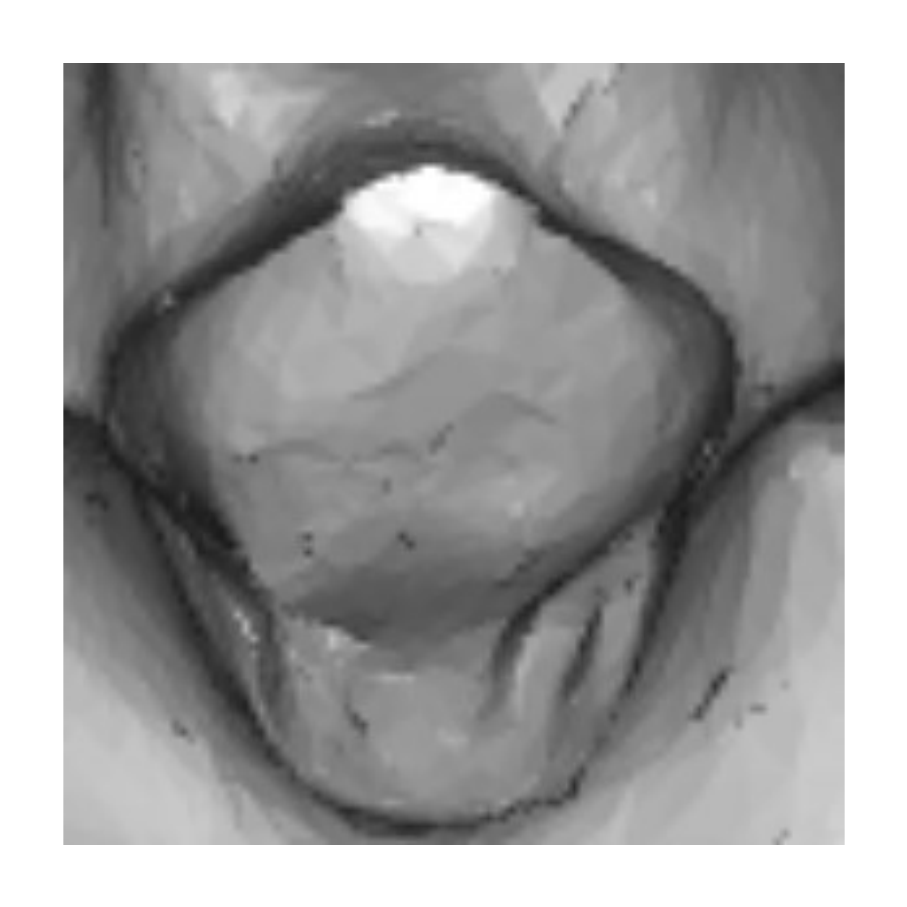}&
		\includegraphics[width=1.72cm]{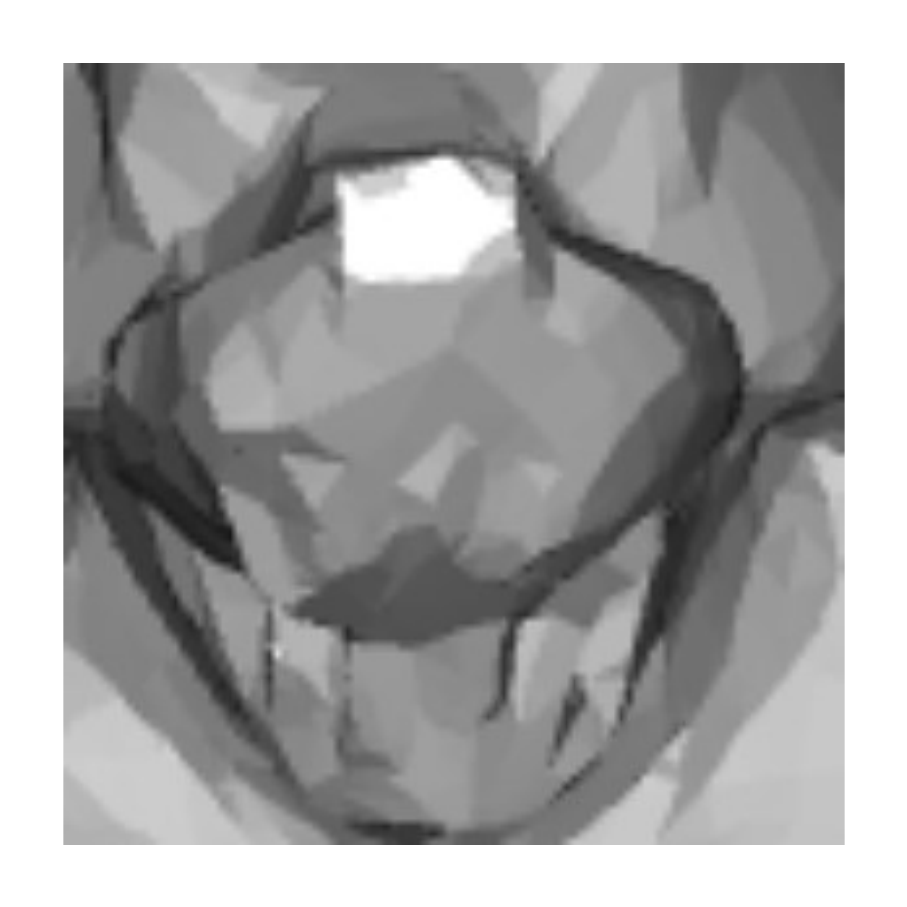}&
		\includegraphics[width=1.72cm]{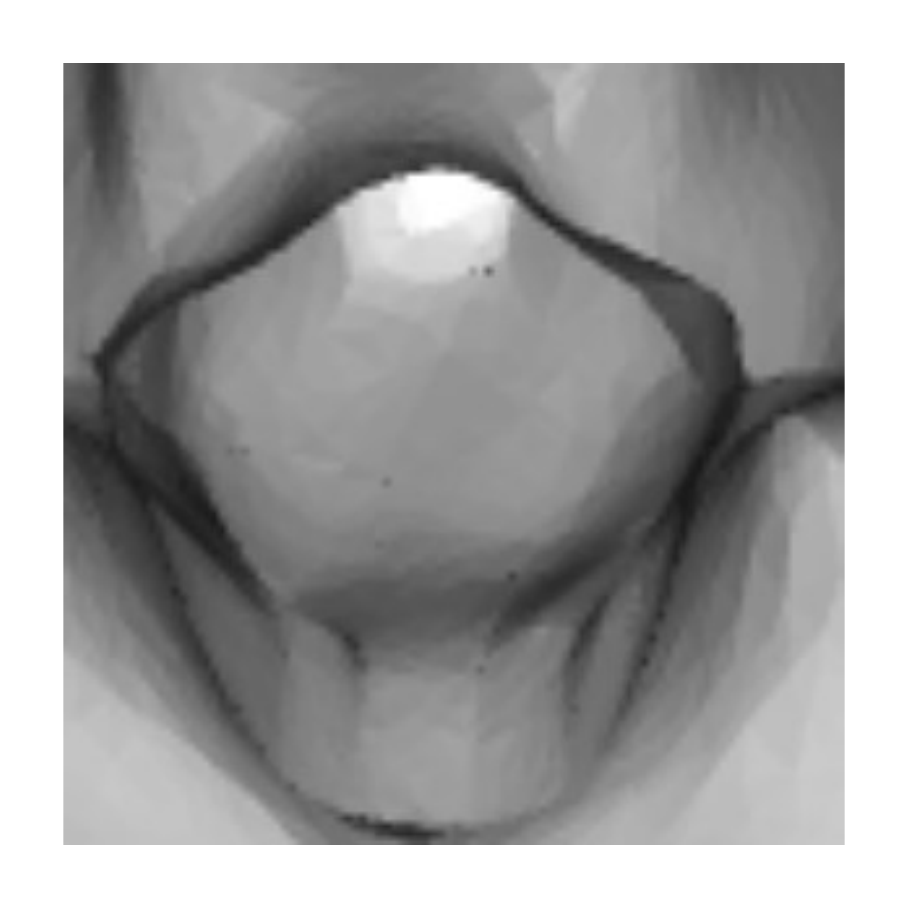}&
		\includegraphics[width=1.72cm]{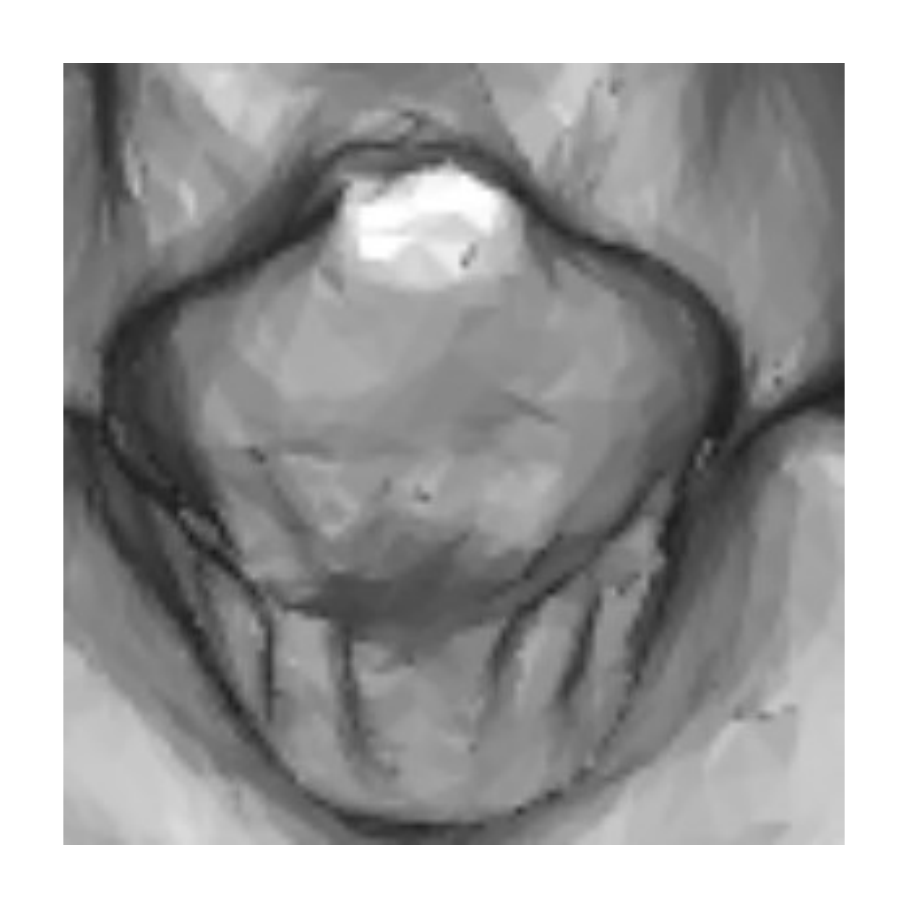}&
		\includegraphics[width=1.72cm]{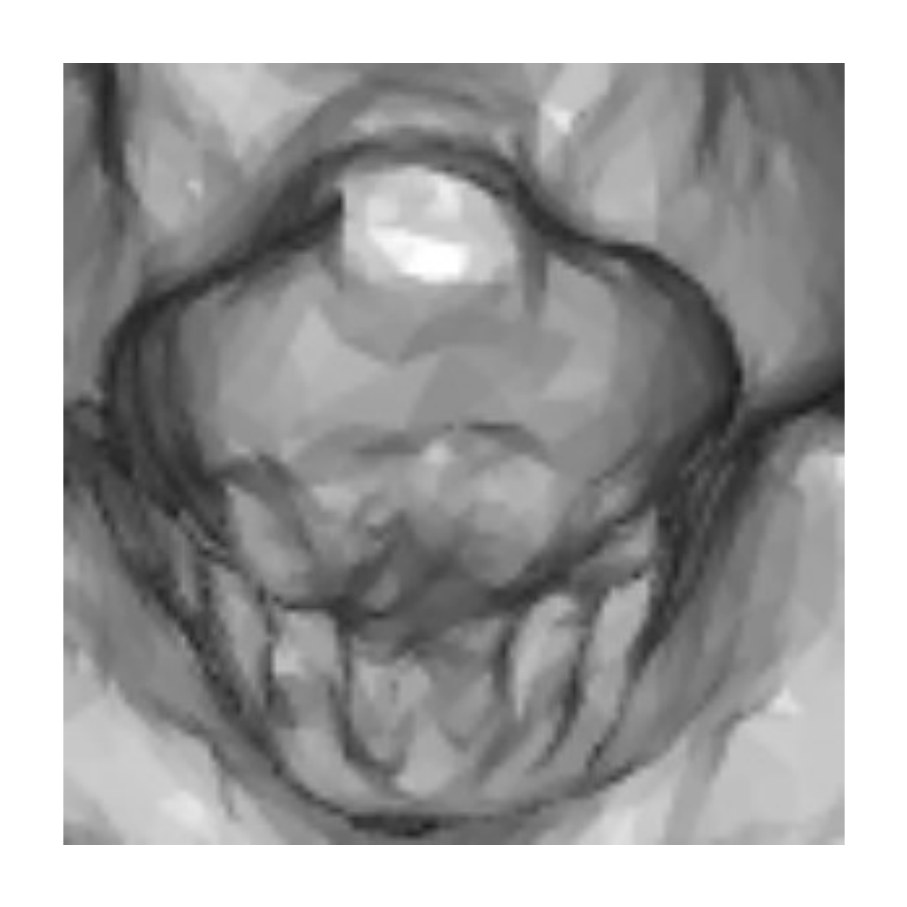}\\
		\includegraphics[width=1.72cm]{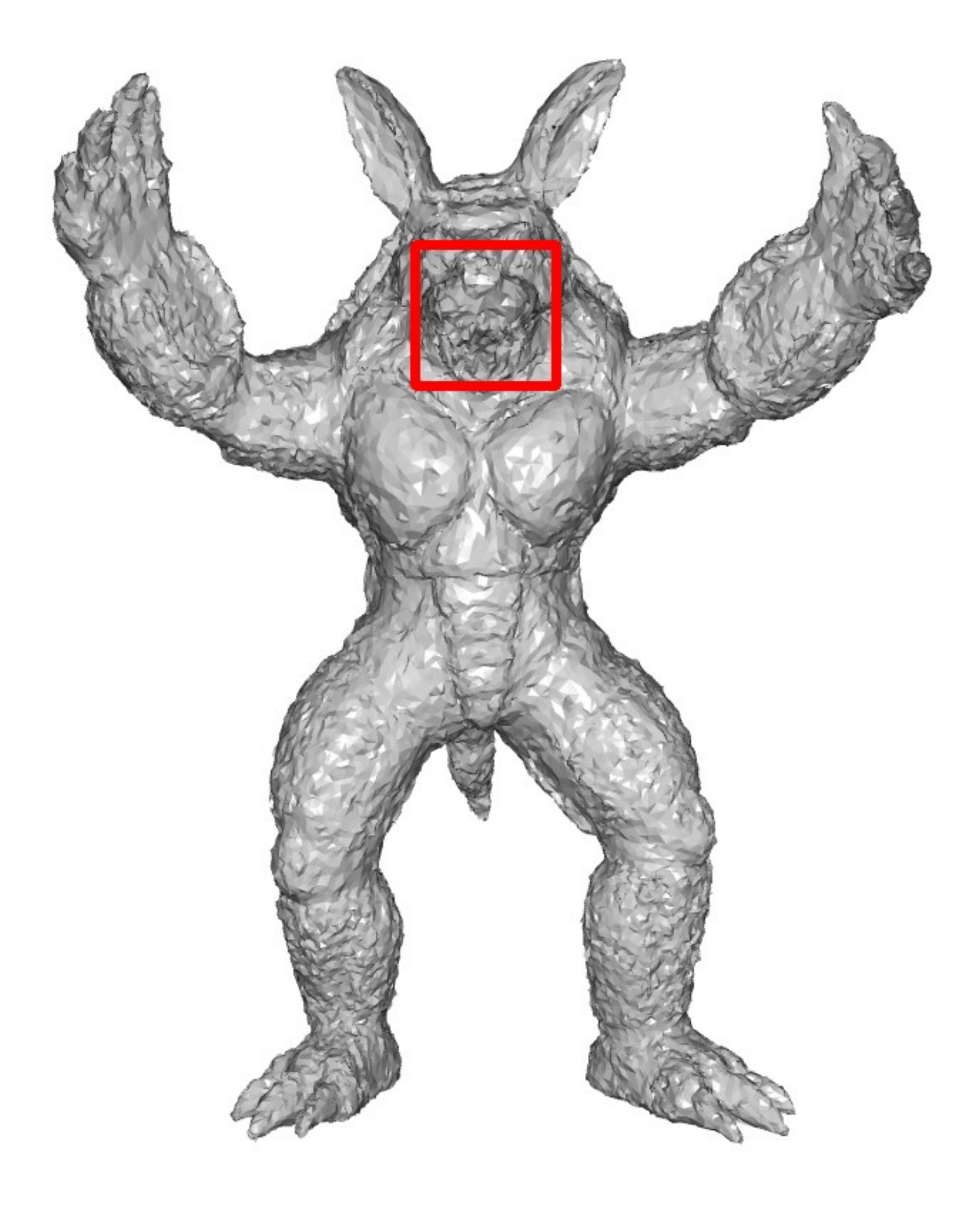}&
		\includegraphics[width=1.72cm]{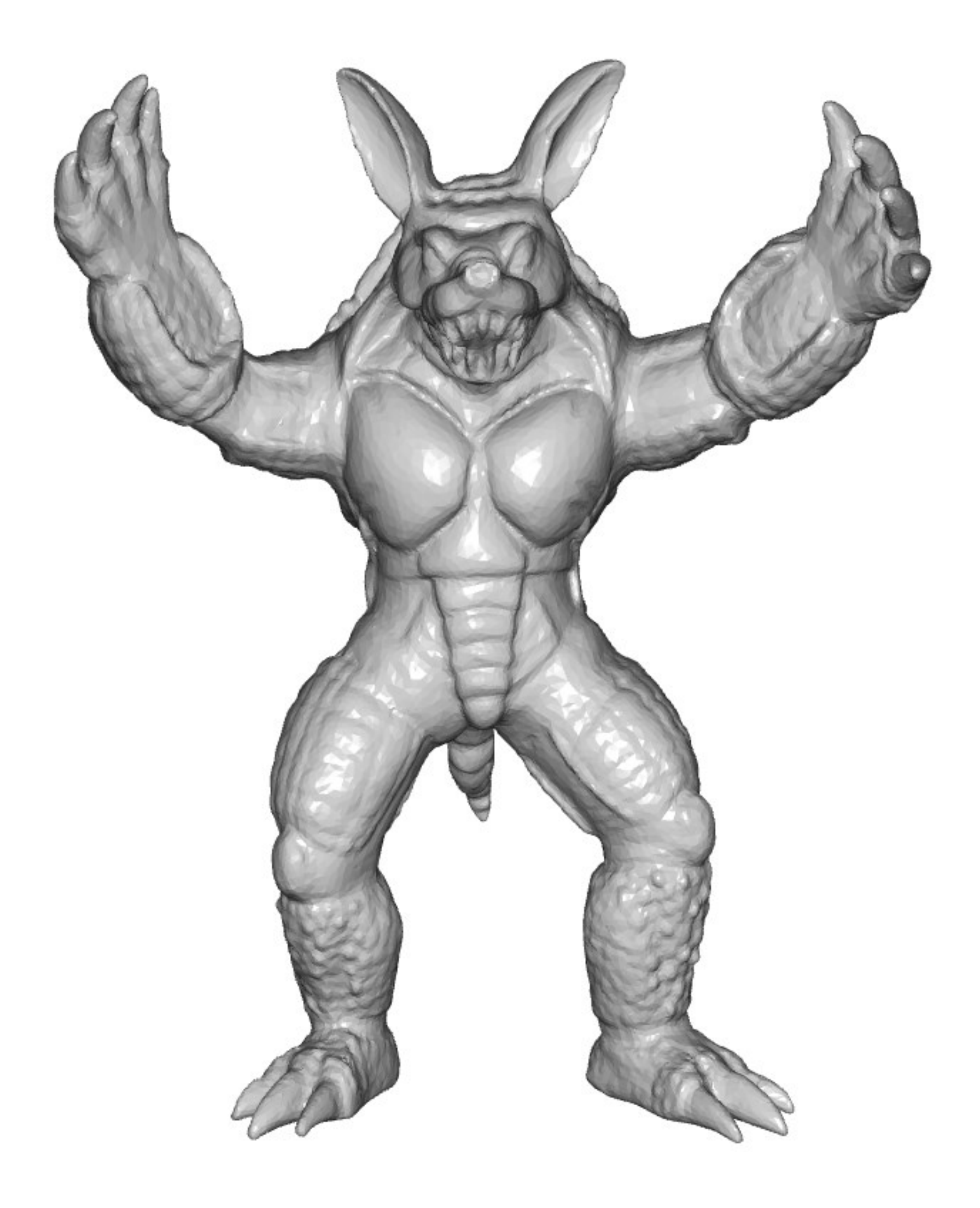}&
		\includegraphics[width=1.72cm]{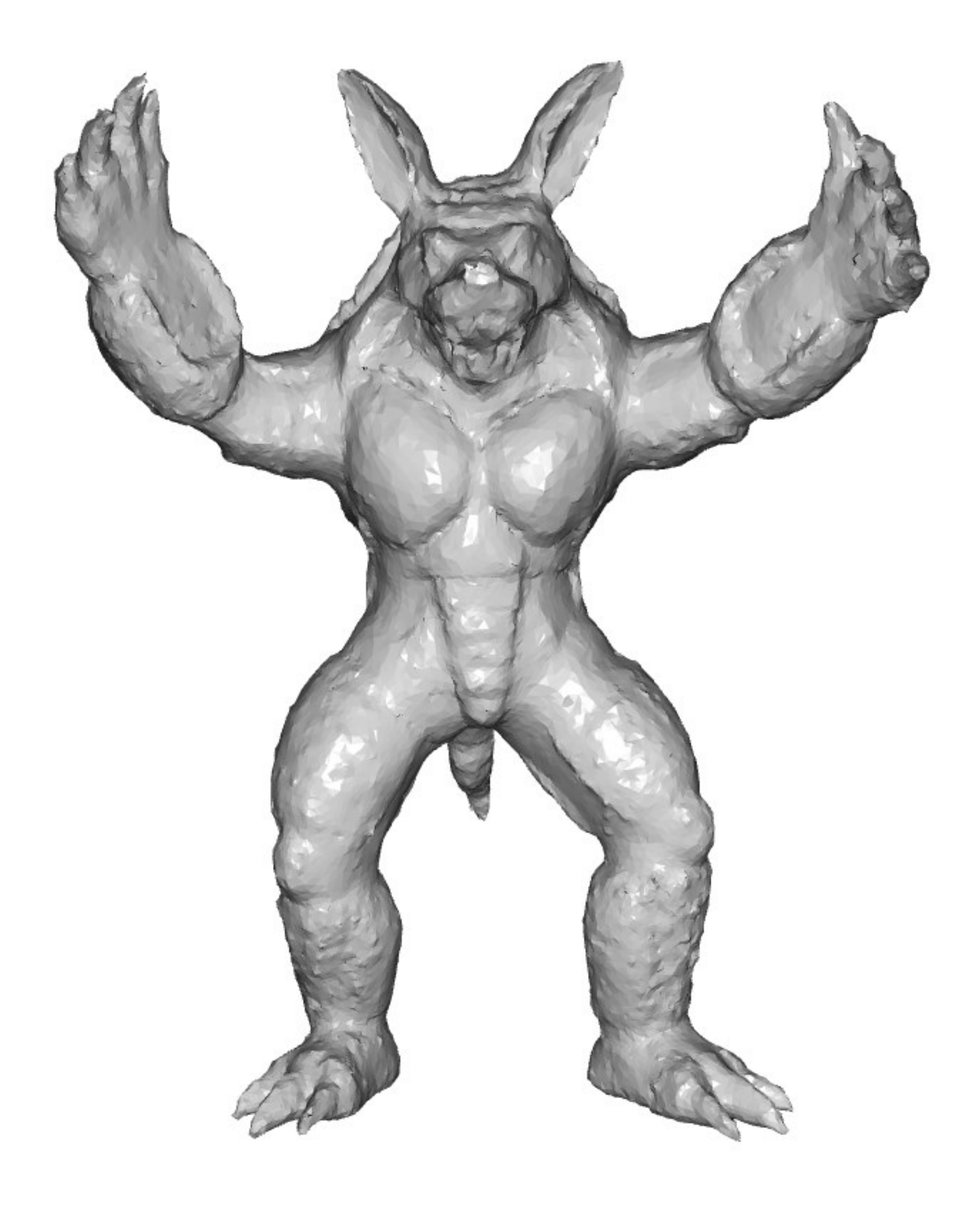}&
		\includegraphics[width=1.72cm]{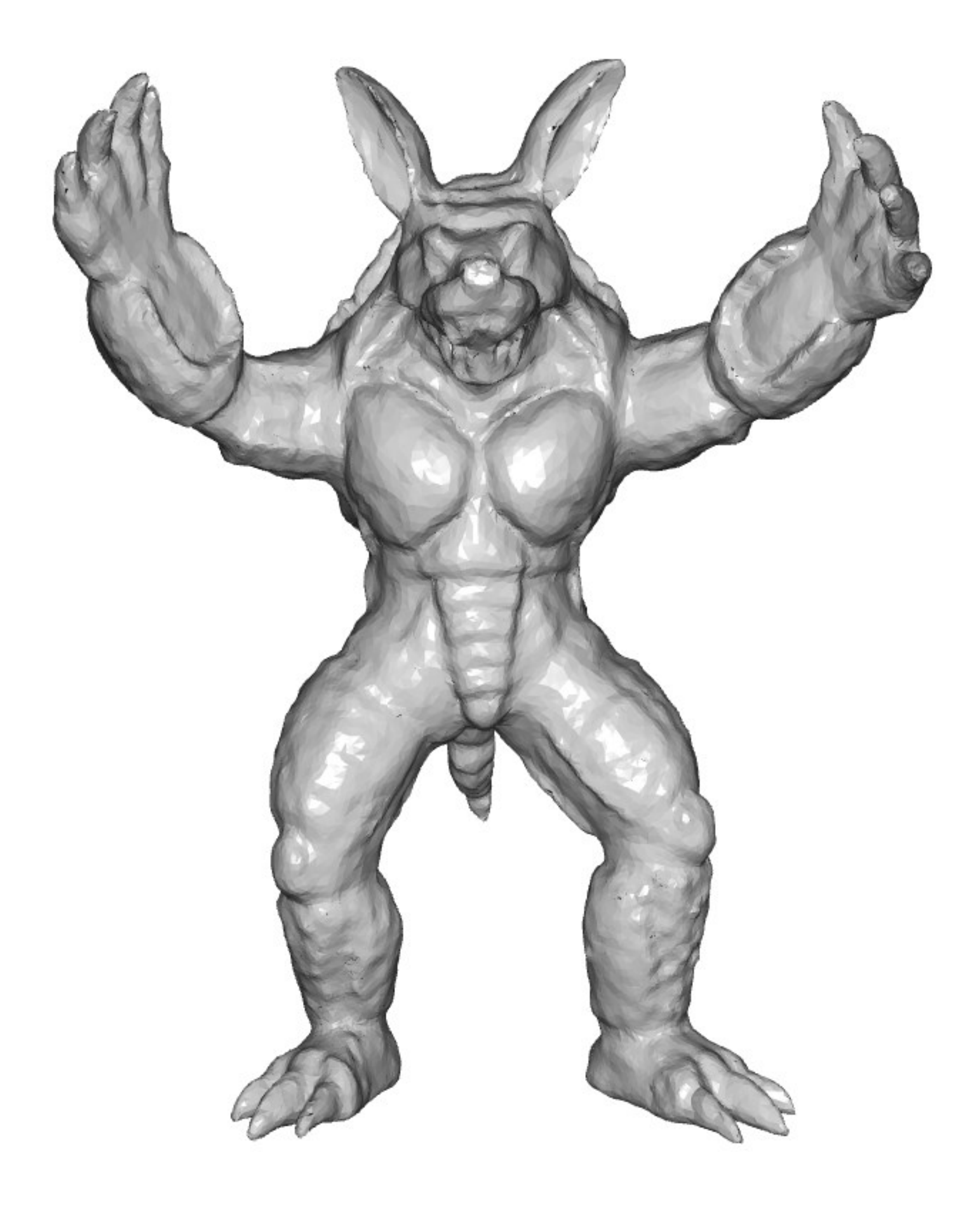}&
		\includegraphics[width=1.72cm]{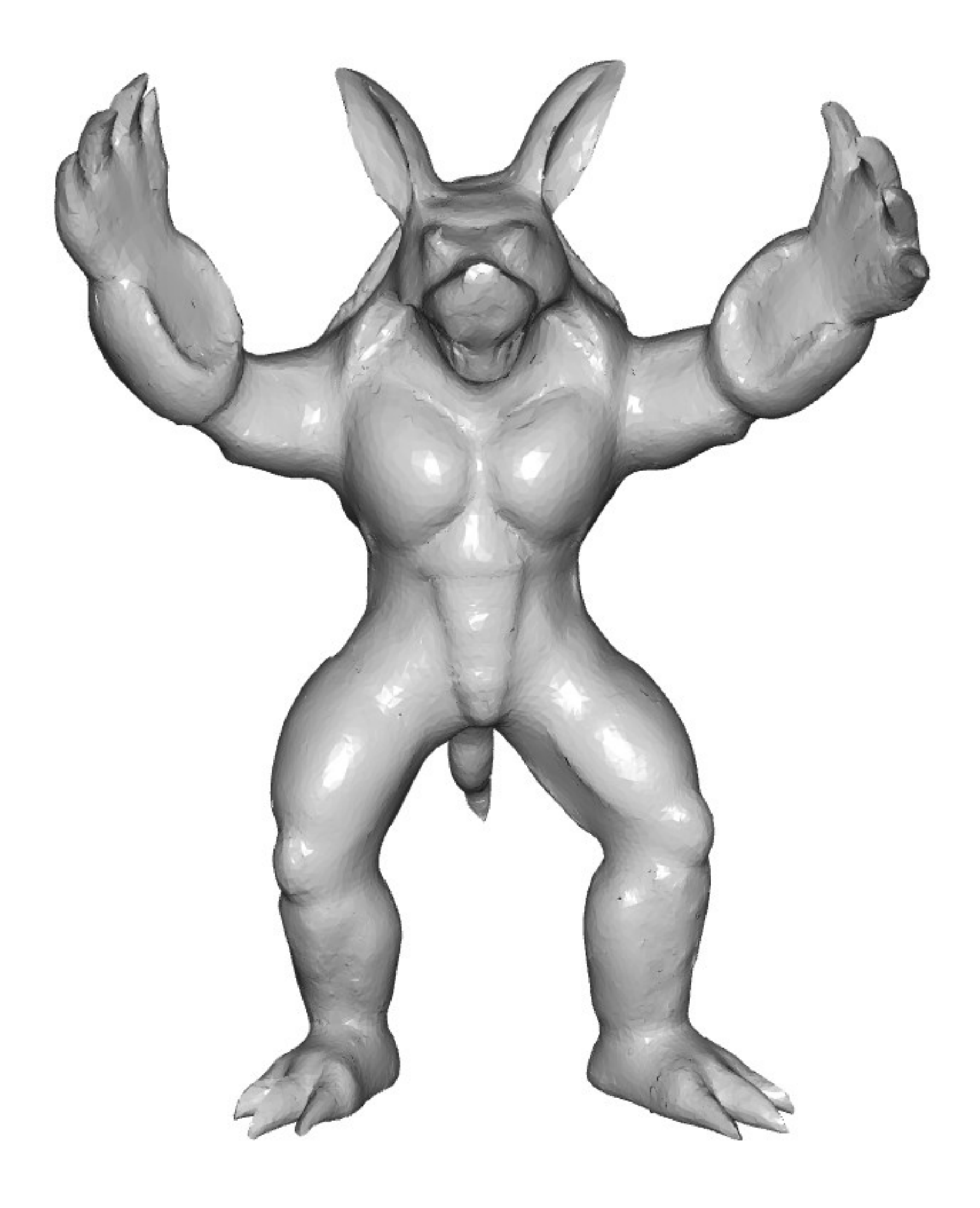}&
		\includegraphics[width=1.72cm]{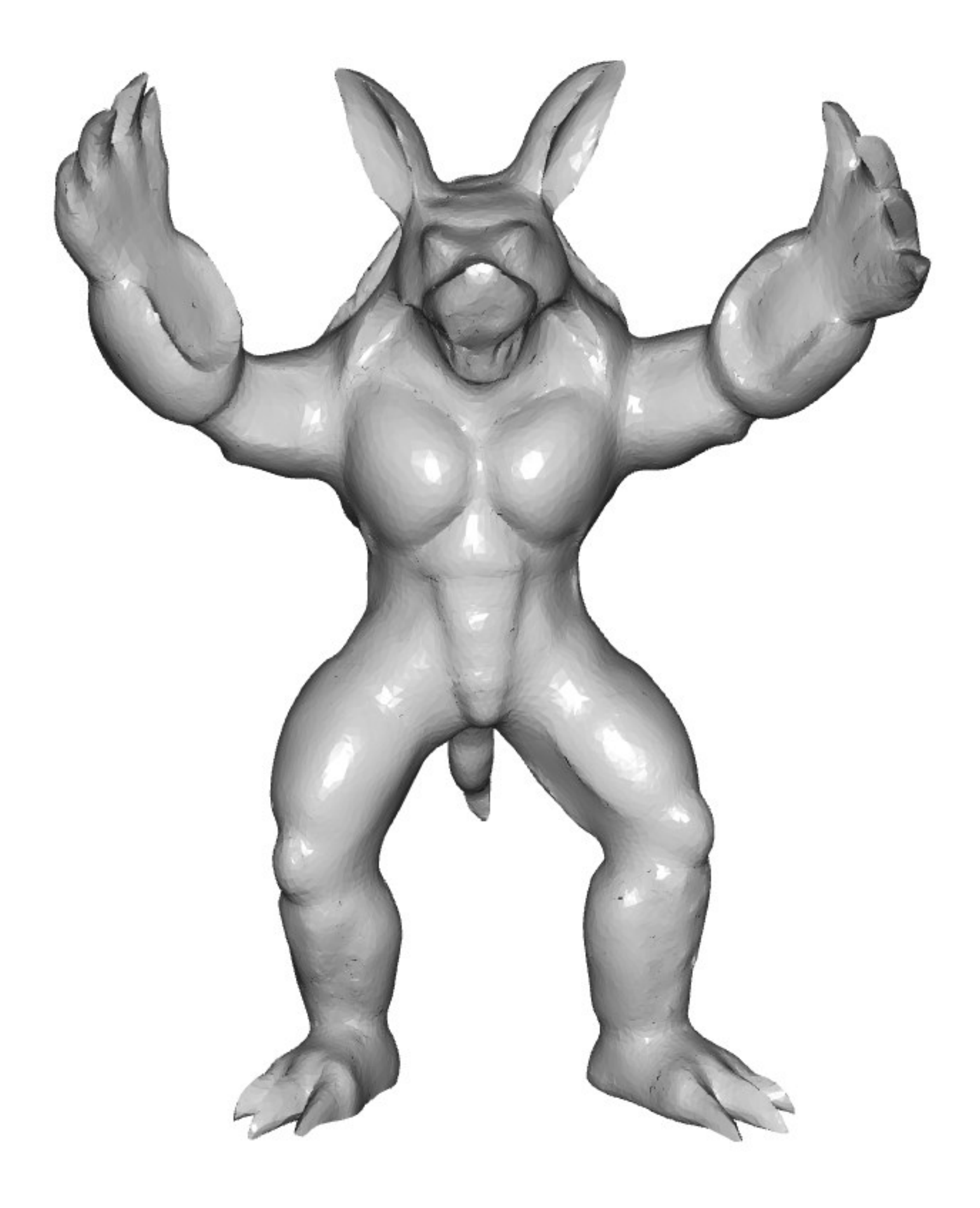}&
		\includegraphics[width=1.72cm]{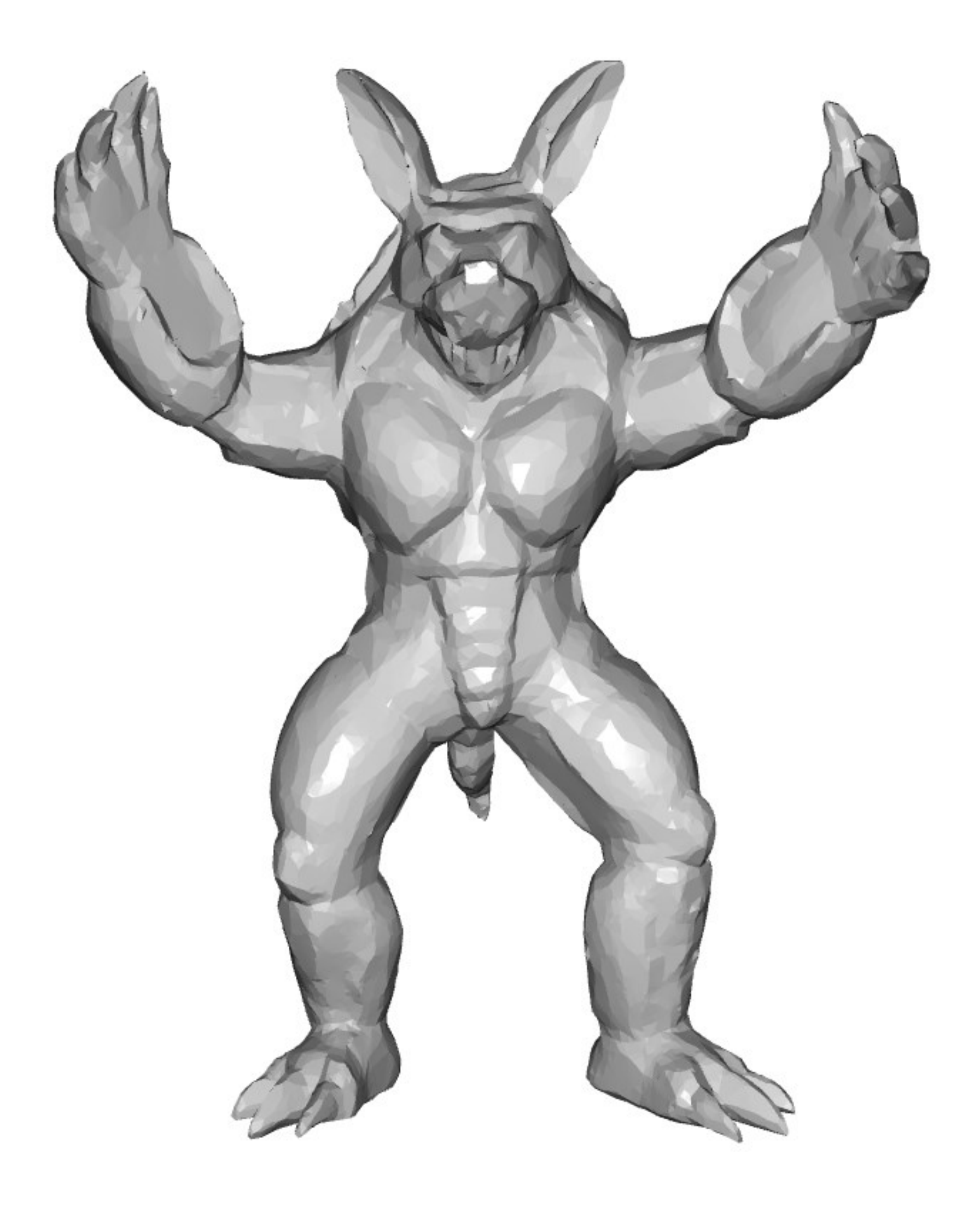}&
		\includegraphics[width=1.72cm]{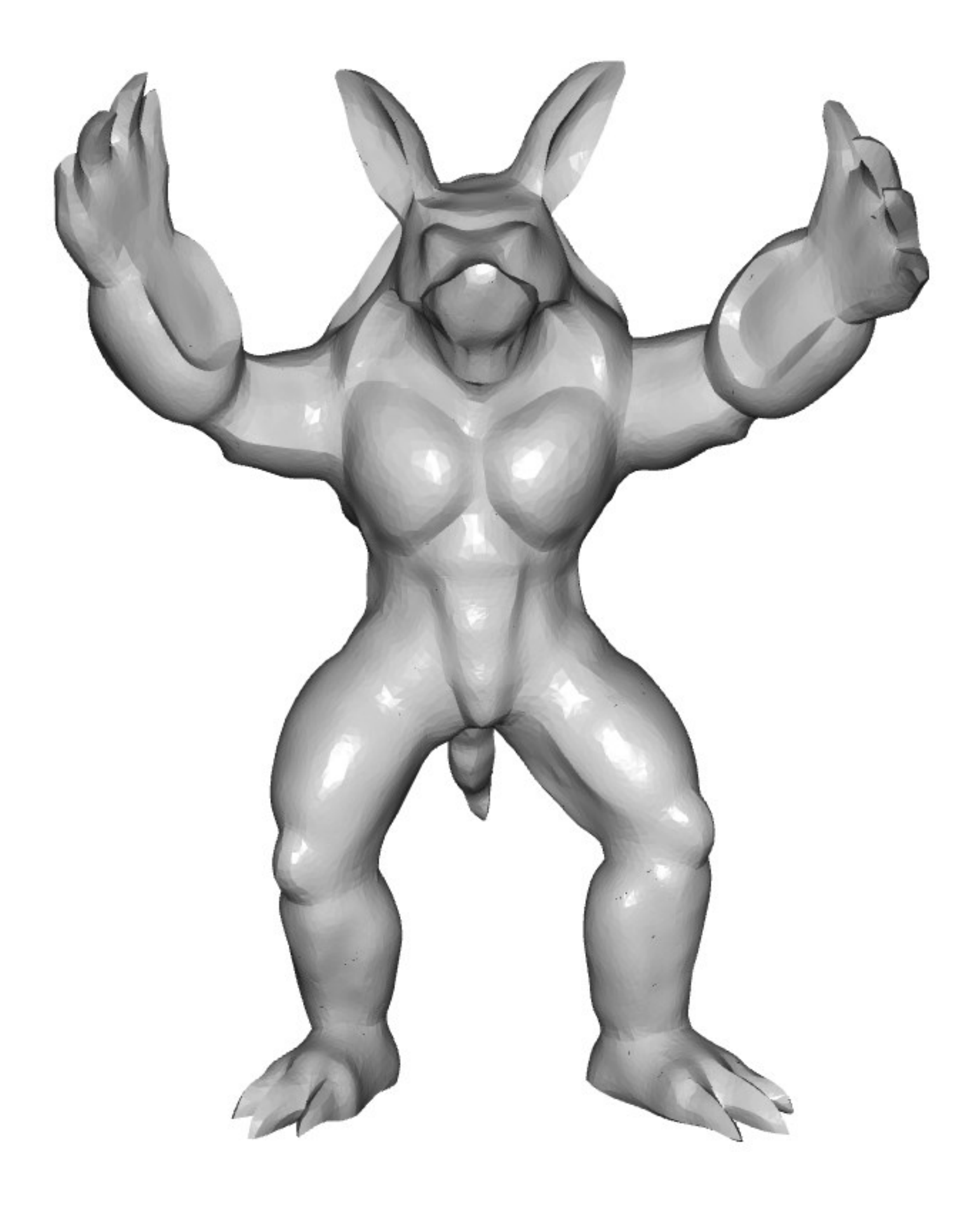}&
		\includegraphics[width=1.72cm]{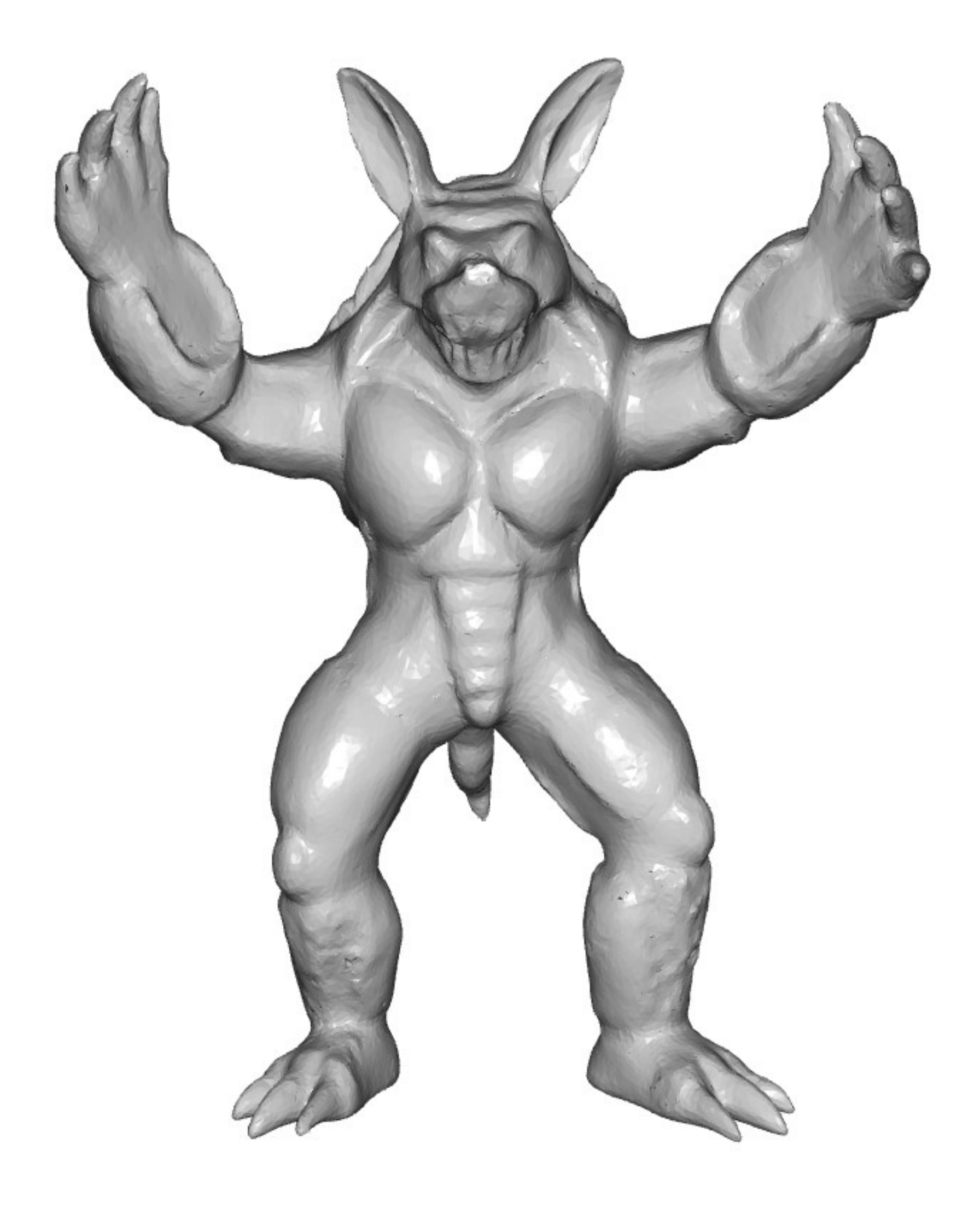}&
		\includegraphics[width=1.72cm]{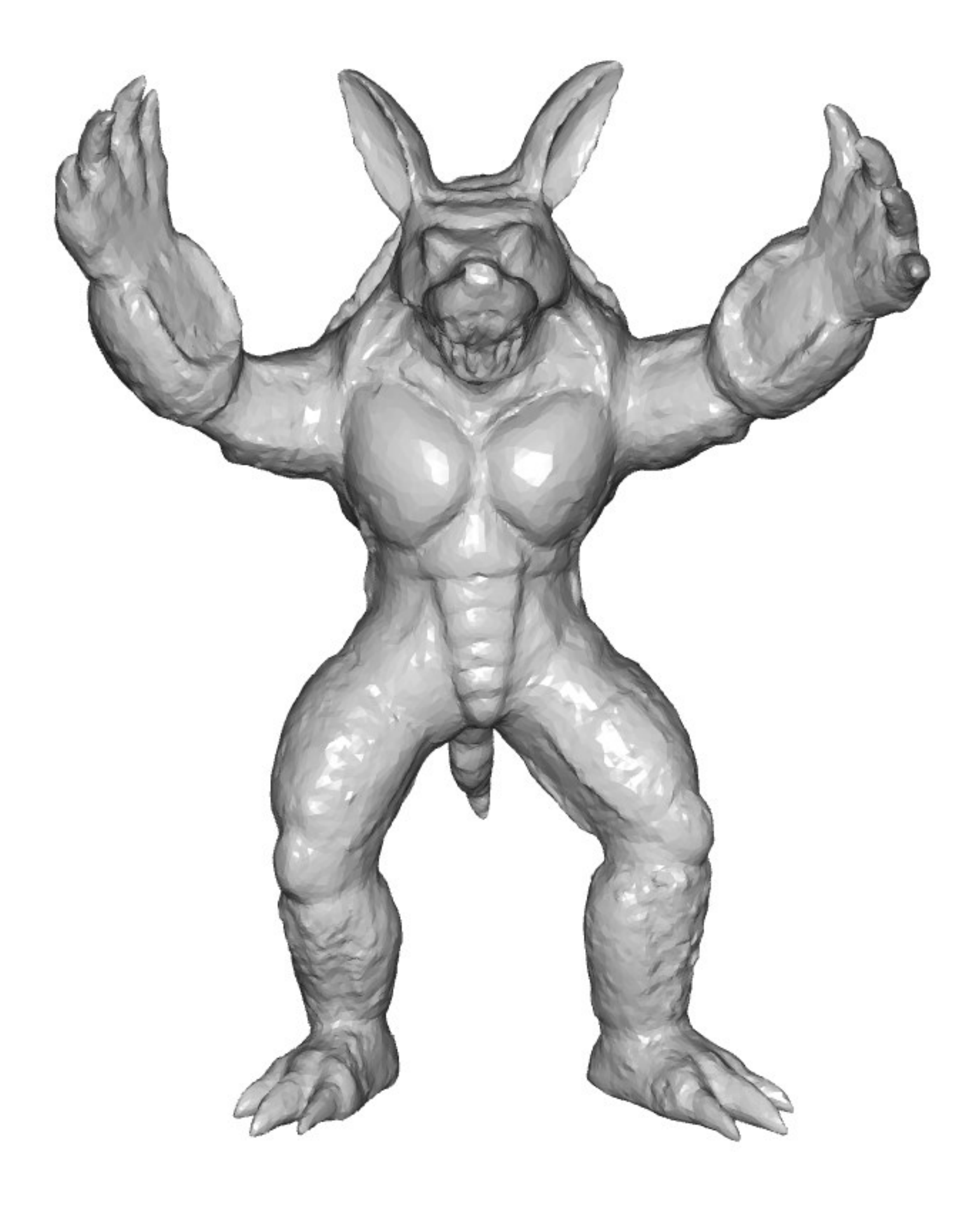}\\
		\includegraphics[width=1.72cm]{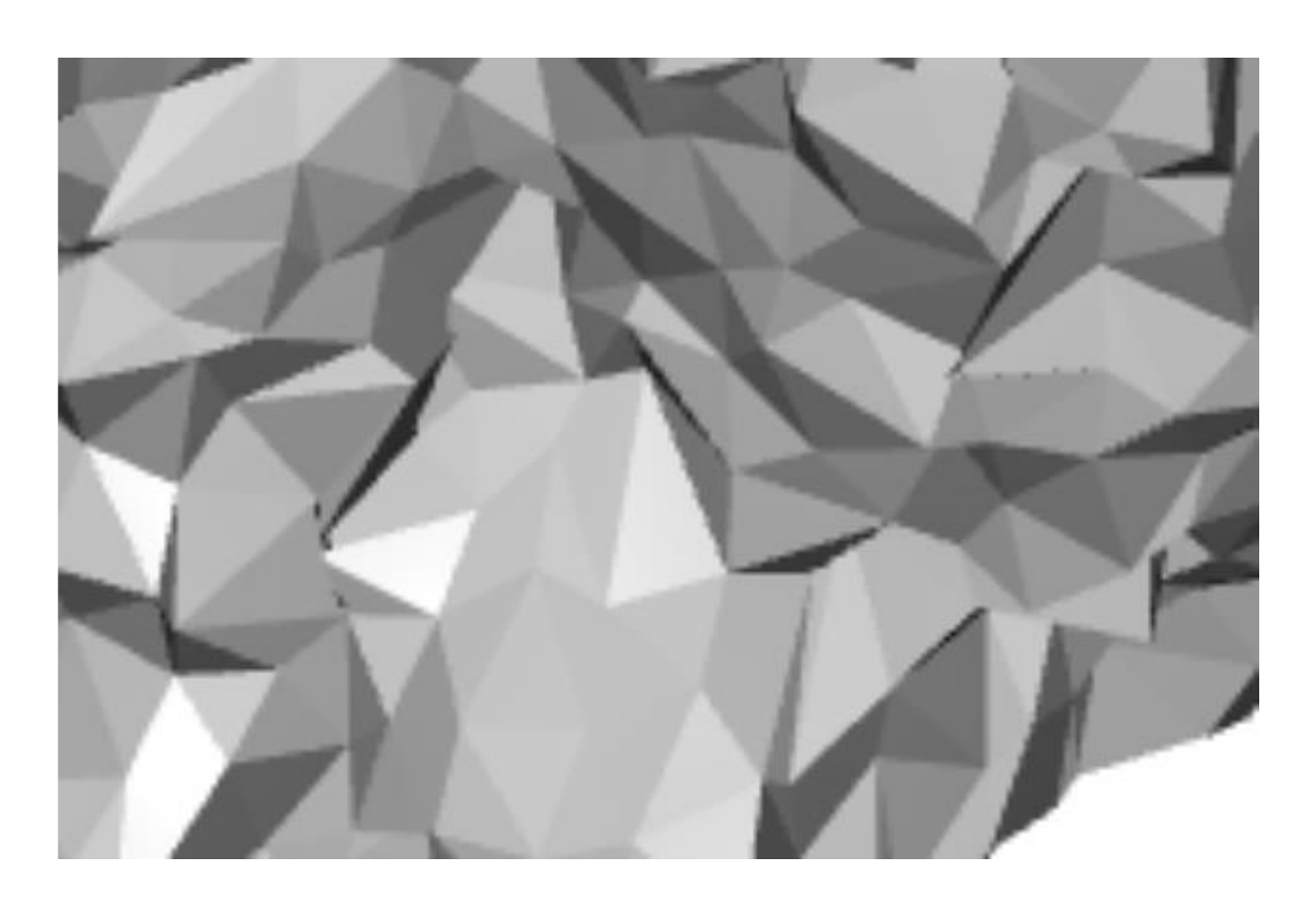}&
		\includegraphics[width=1.72cm]{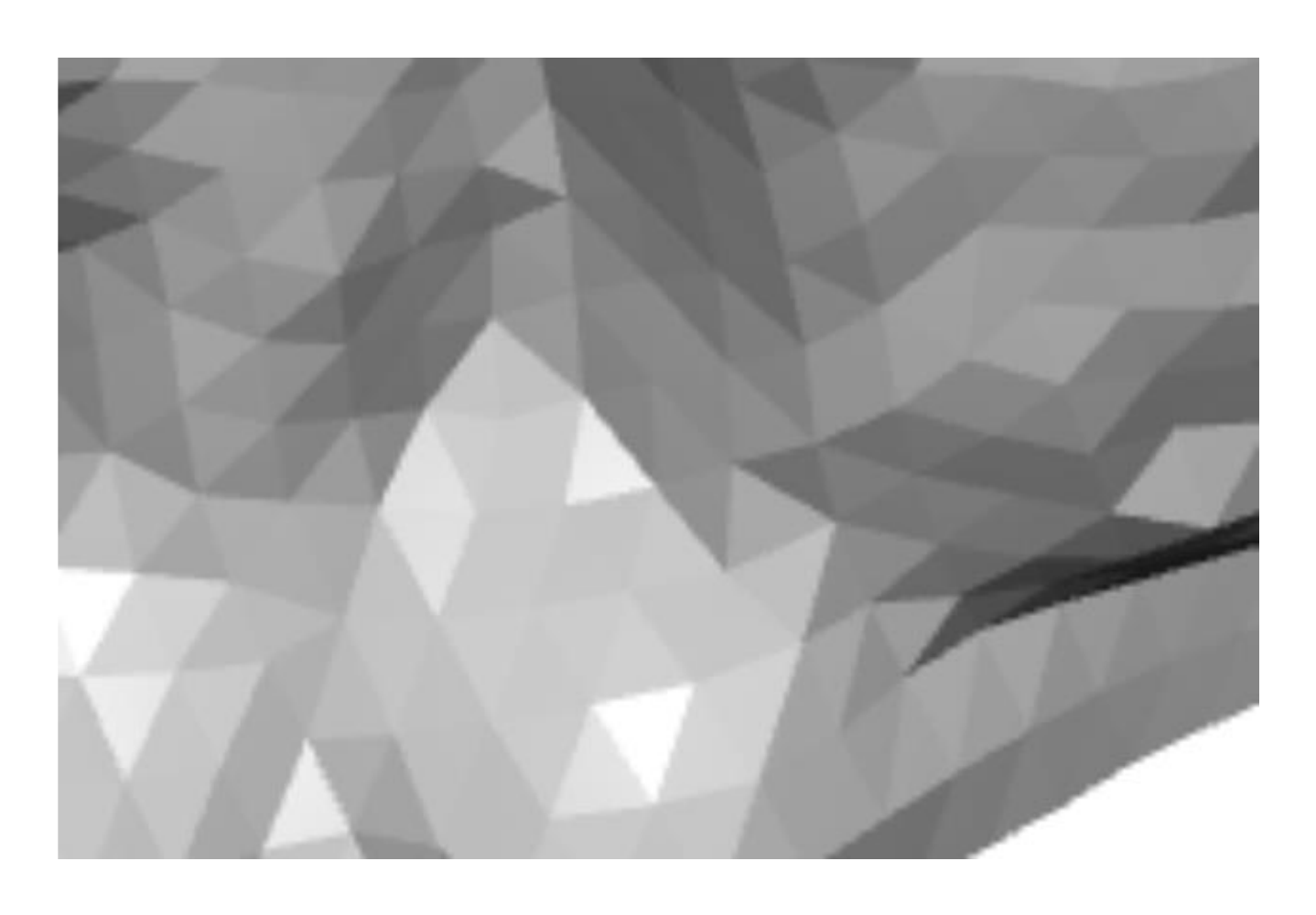}&
		\includegraphics[width=1.72cm]{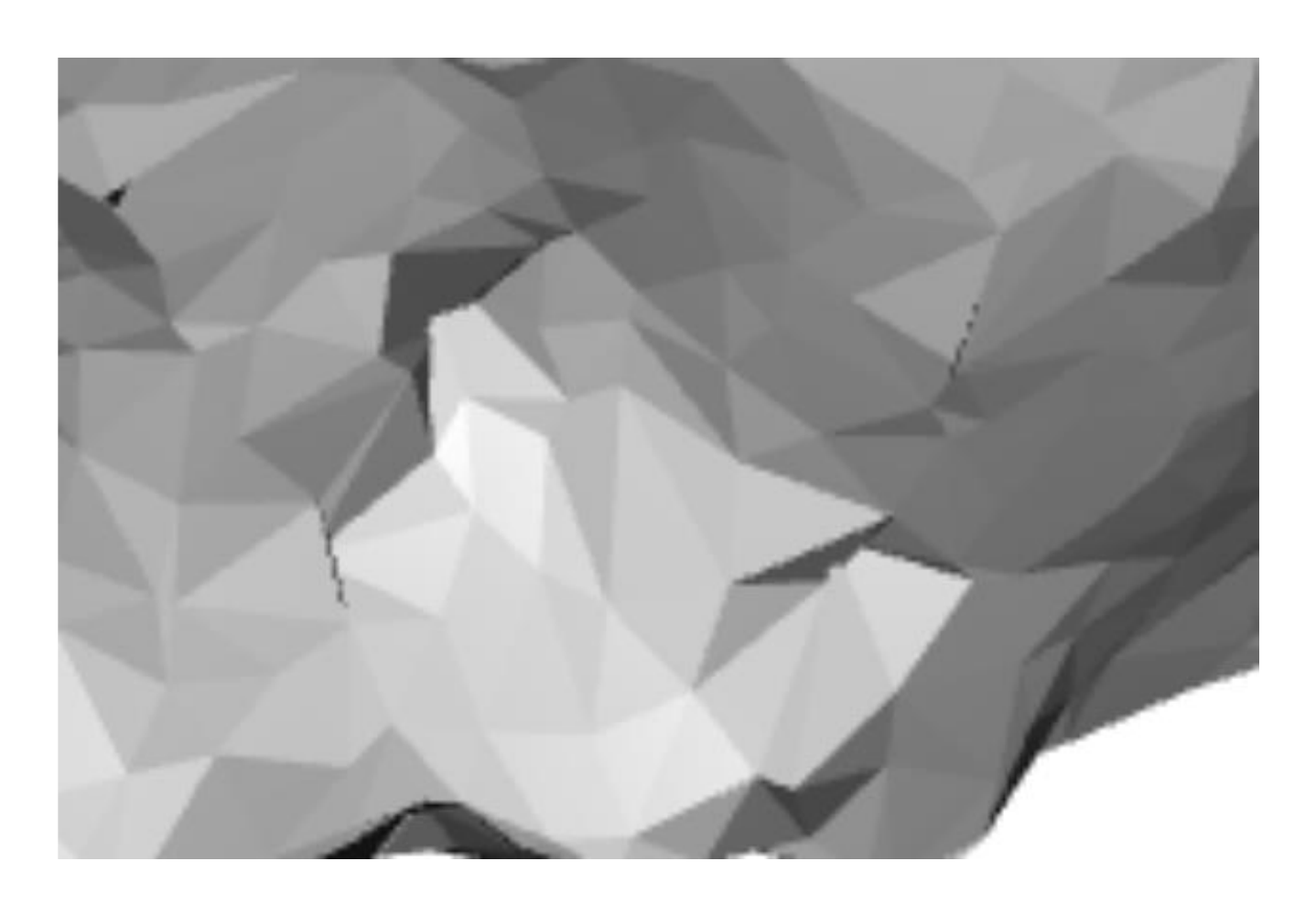}&
		\includegraphics[width=1.72cm]{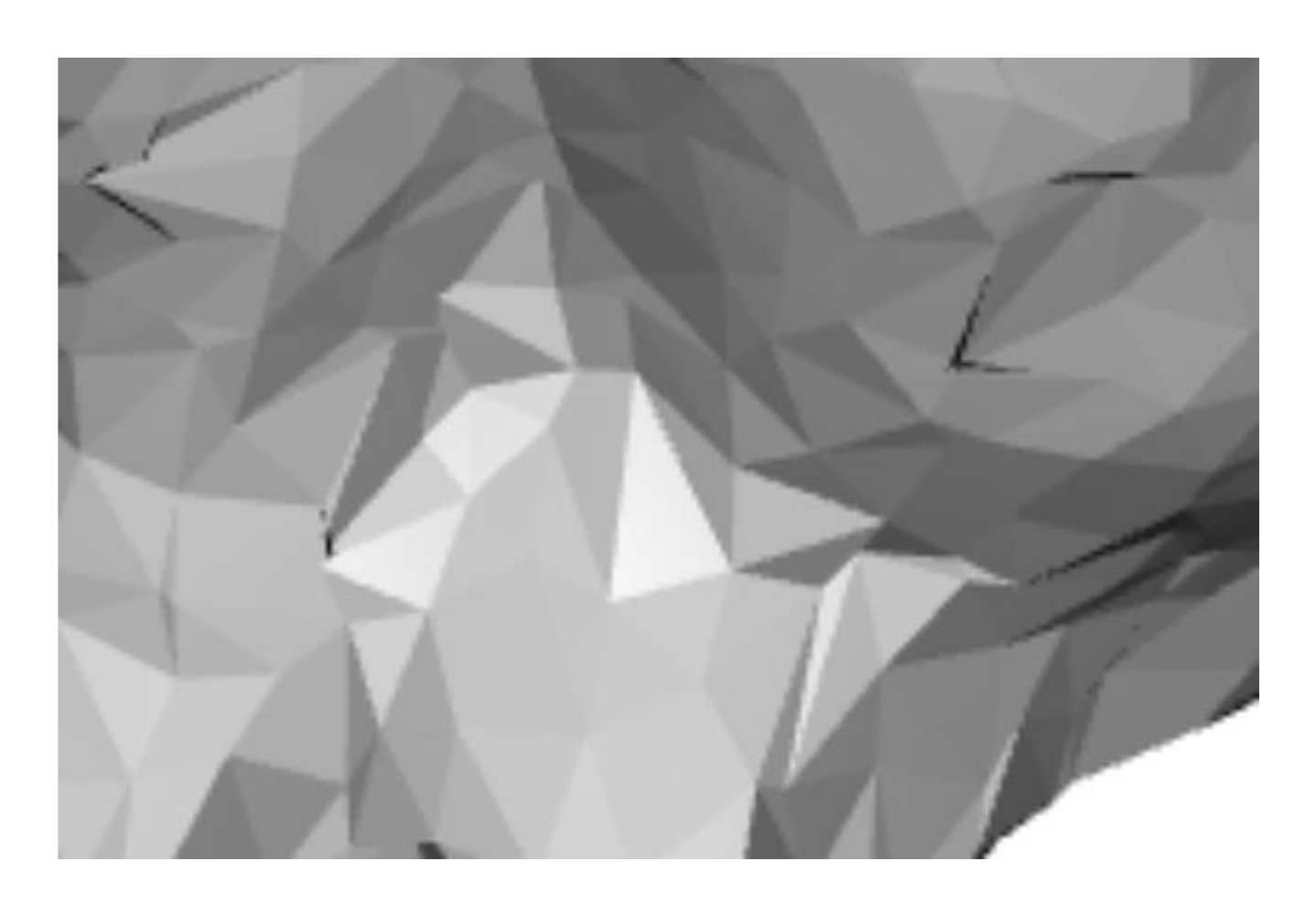}&
		\includegraphics[width=1.72cm]{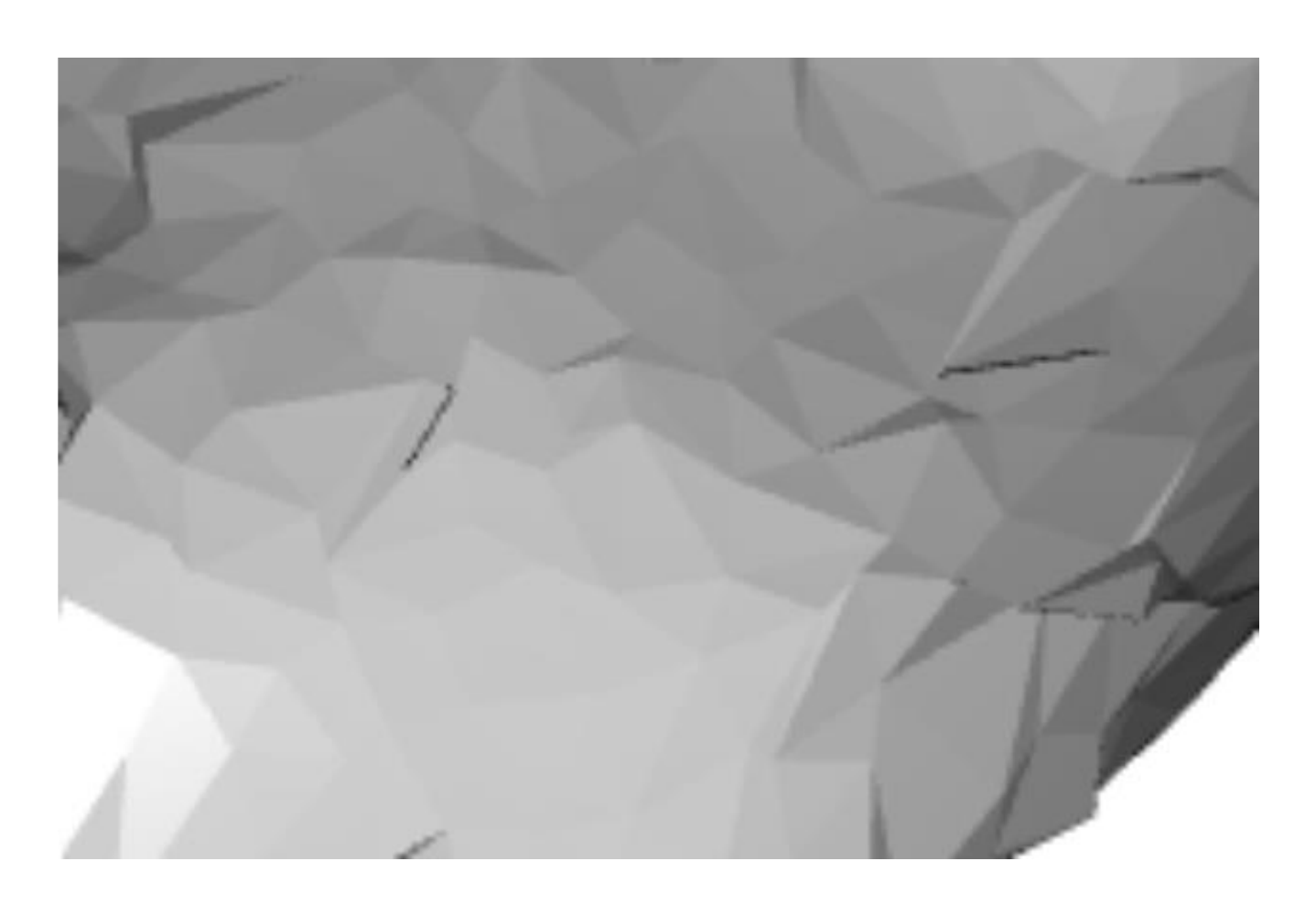}&
		\includegraphics[width=1.72cm]{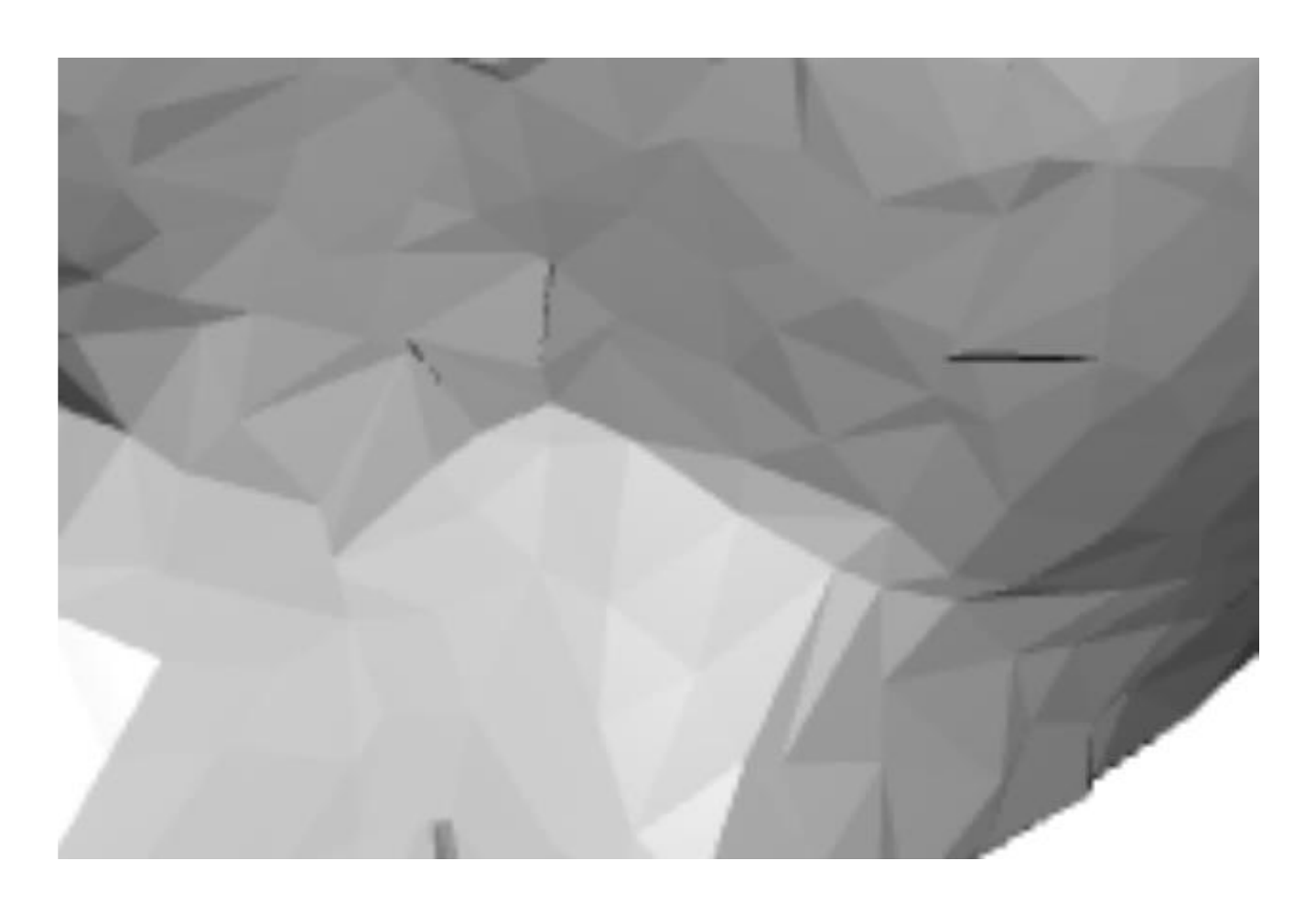}&
		\includegraphics[width=1.72cm]{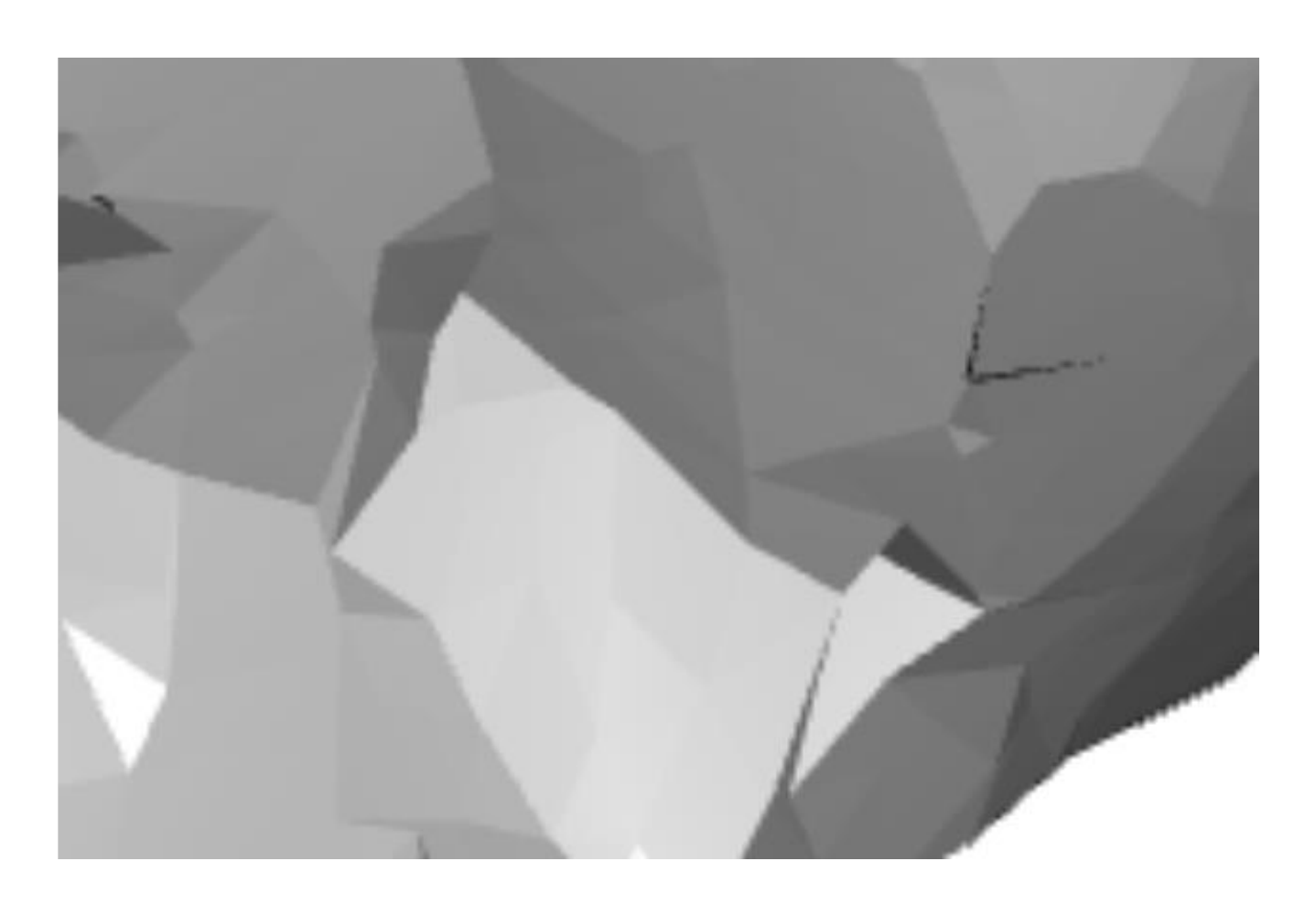}&
		\includegraphics[width=1.72cm]{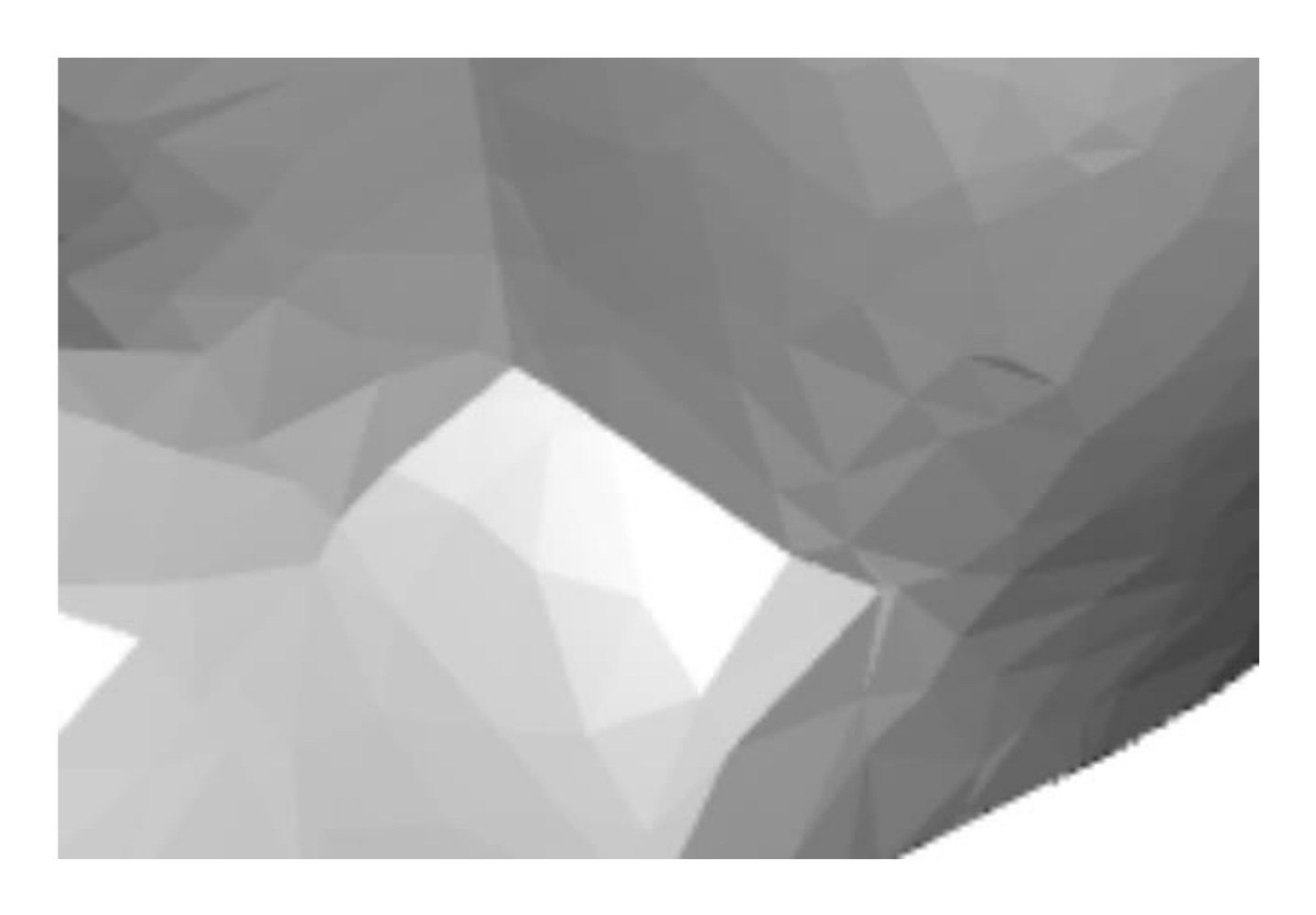}&
		\includegraphics[width=1.72cm]{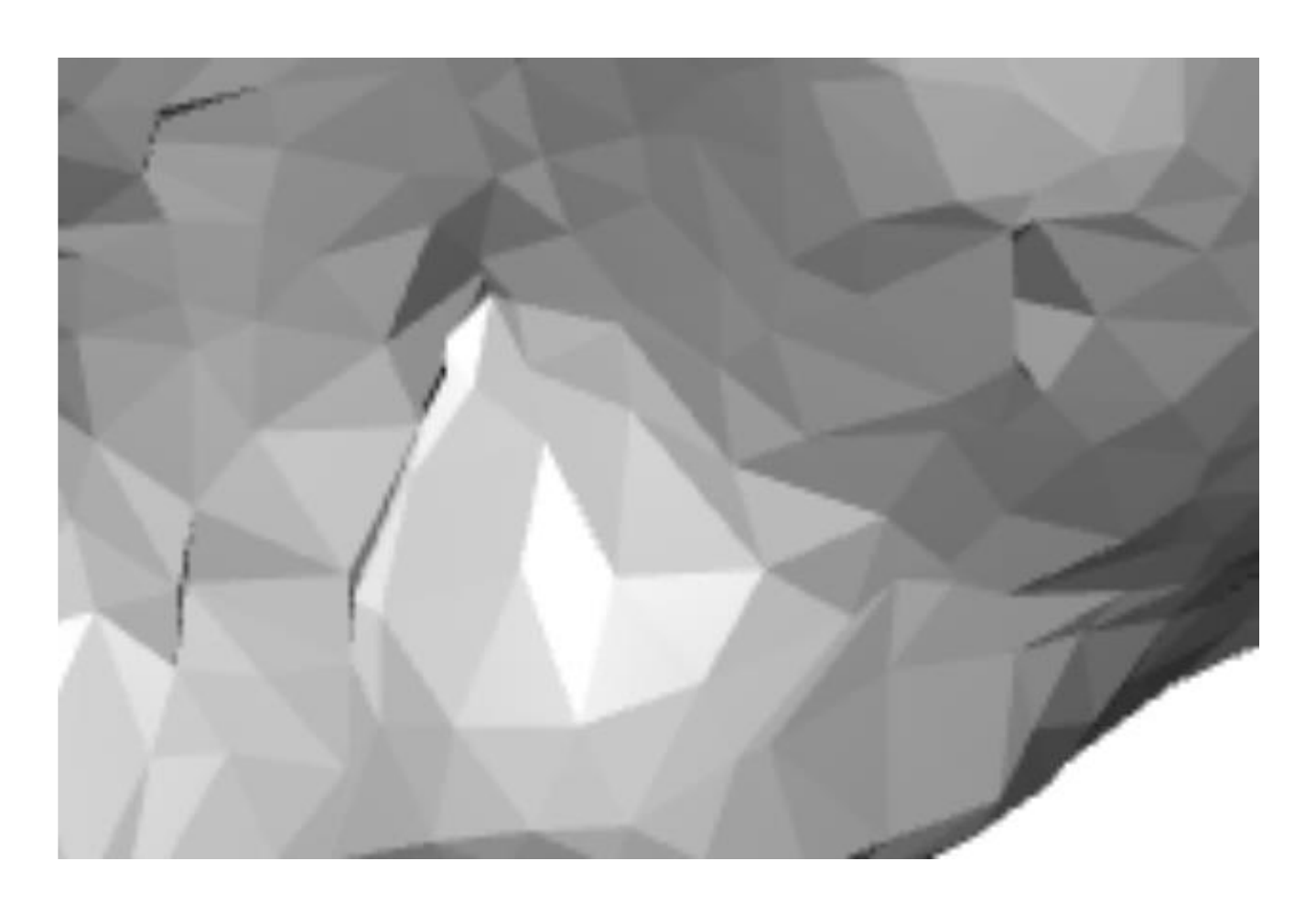}&
		\includegraphics[width=1.72cm]{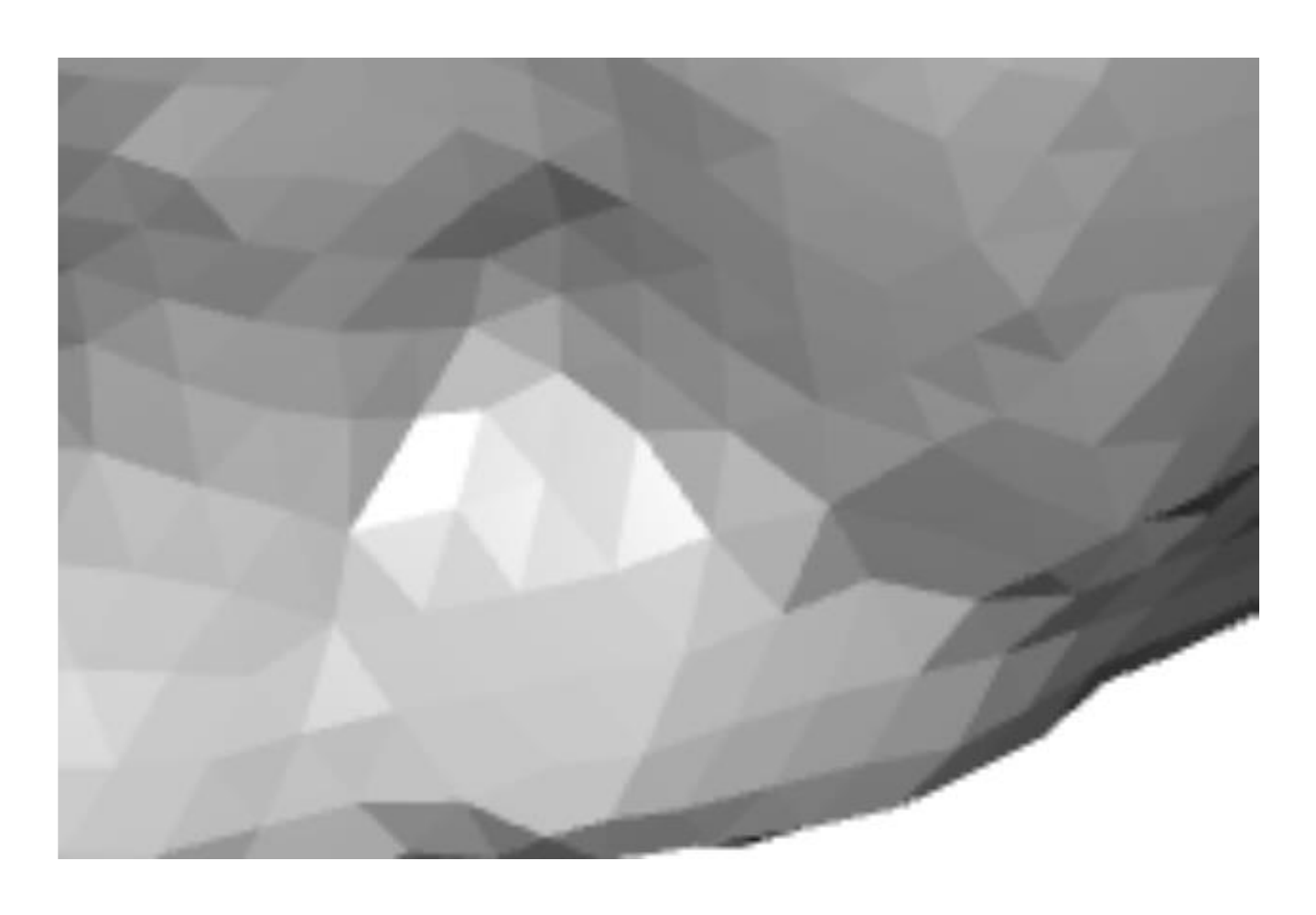}\\
		\includegraphics[width=1.72cm]{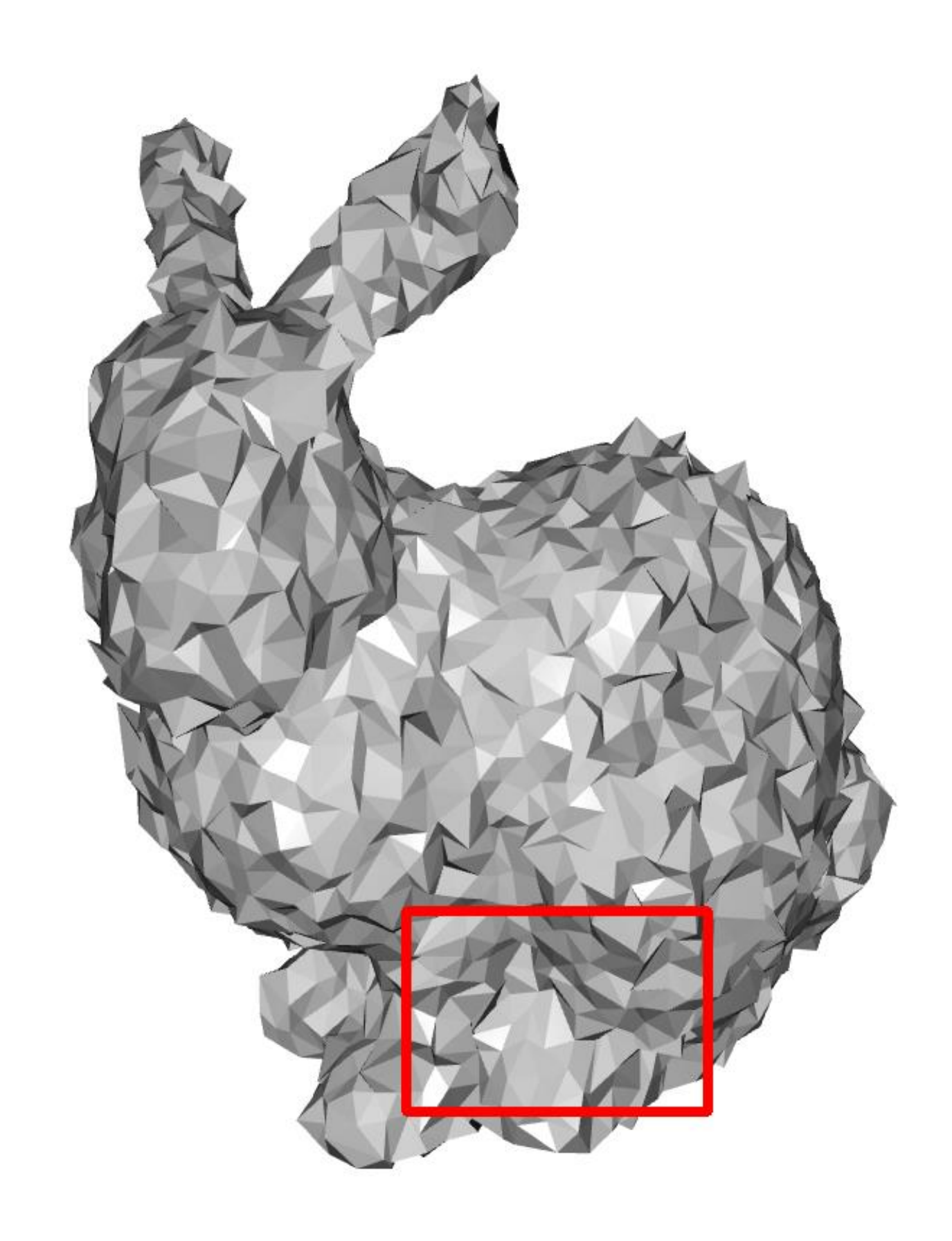}&
		\includegraphics[width=1.72cm]{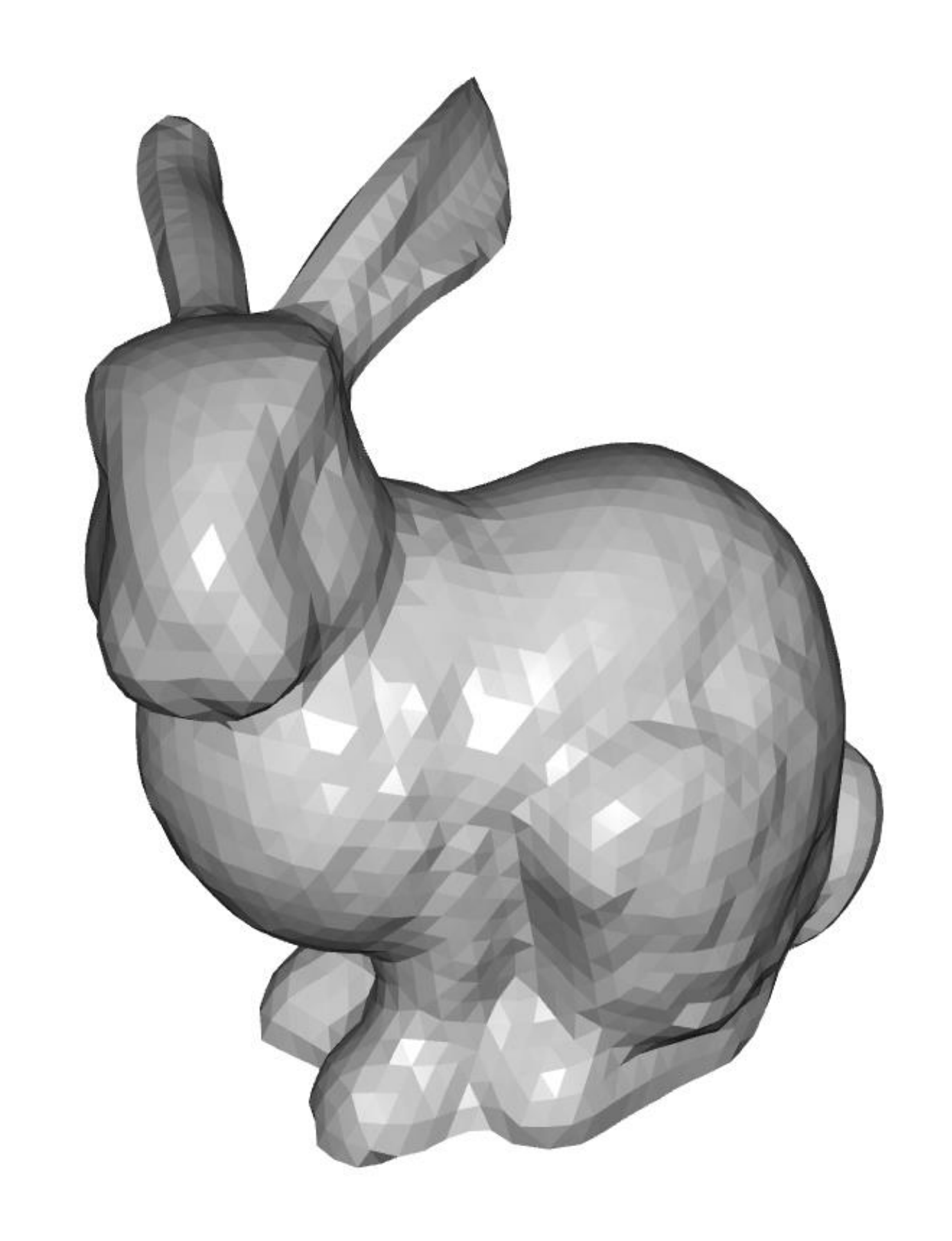}&
		\includegraphics[width=1.72cm]{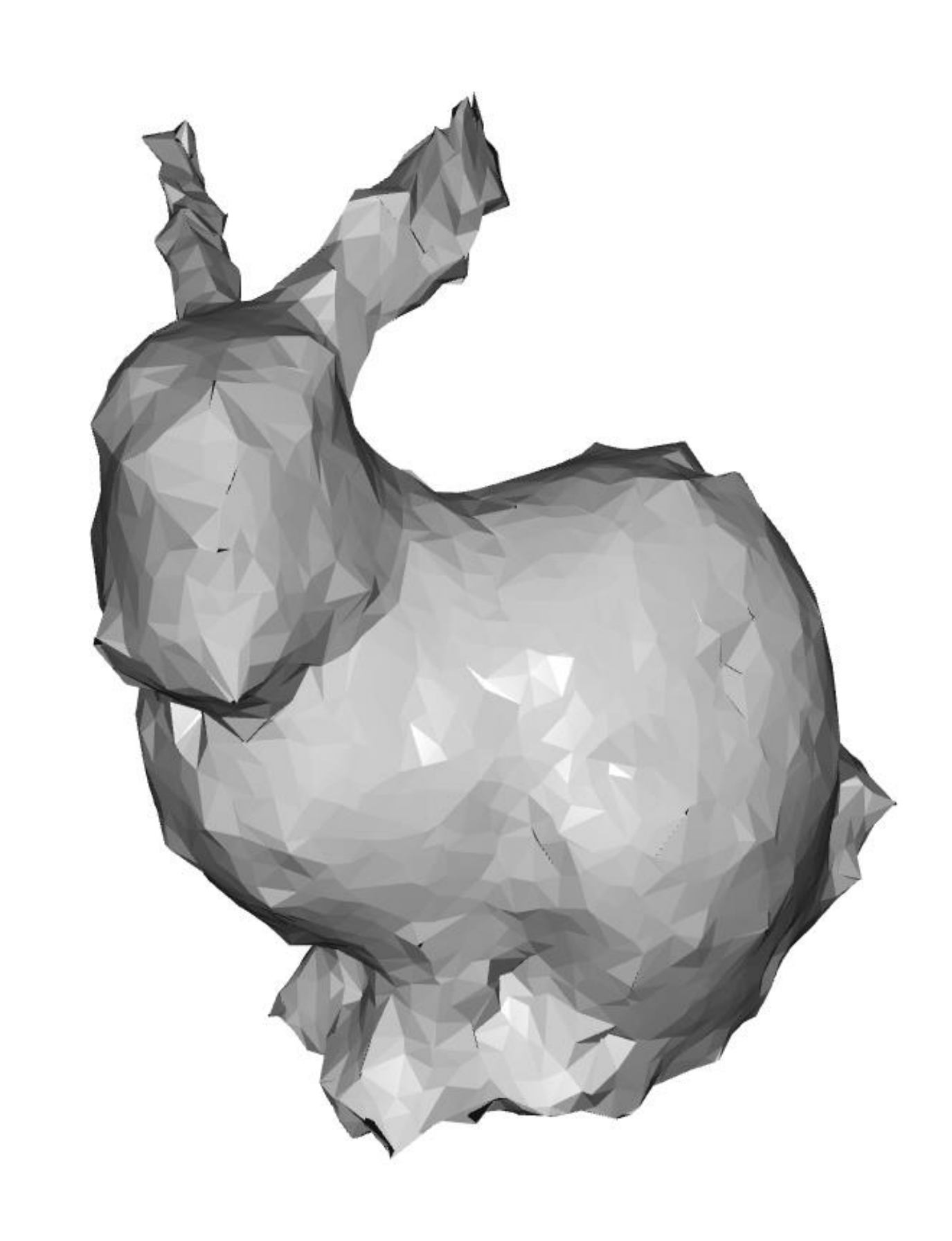}&
		\includegraphics[width=1.72cm]{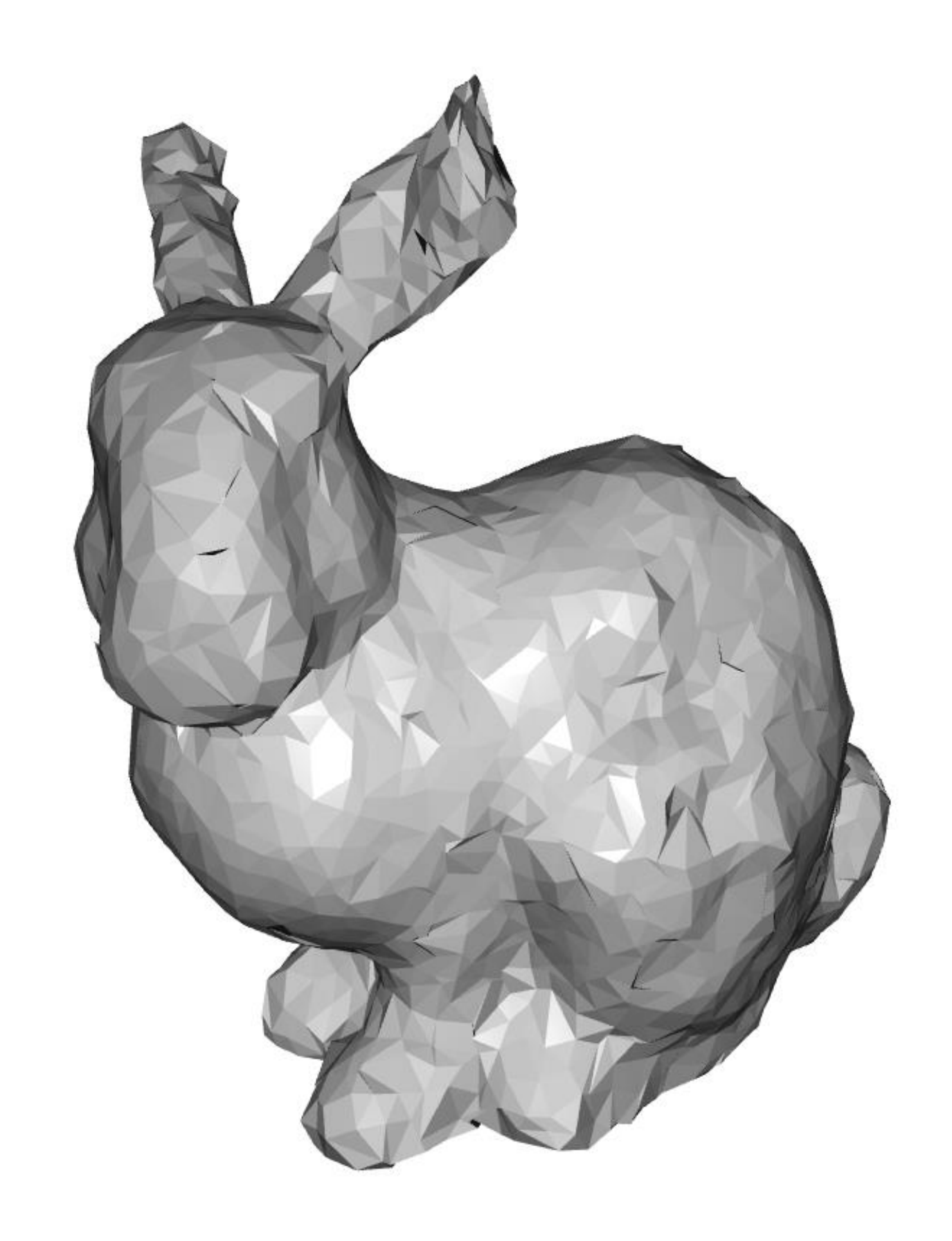}&
		\includegraphics[width=1.72cm]{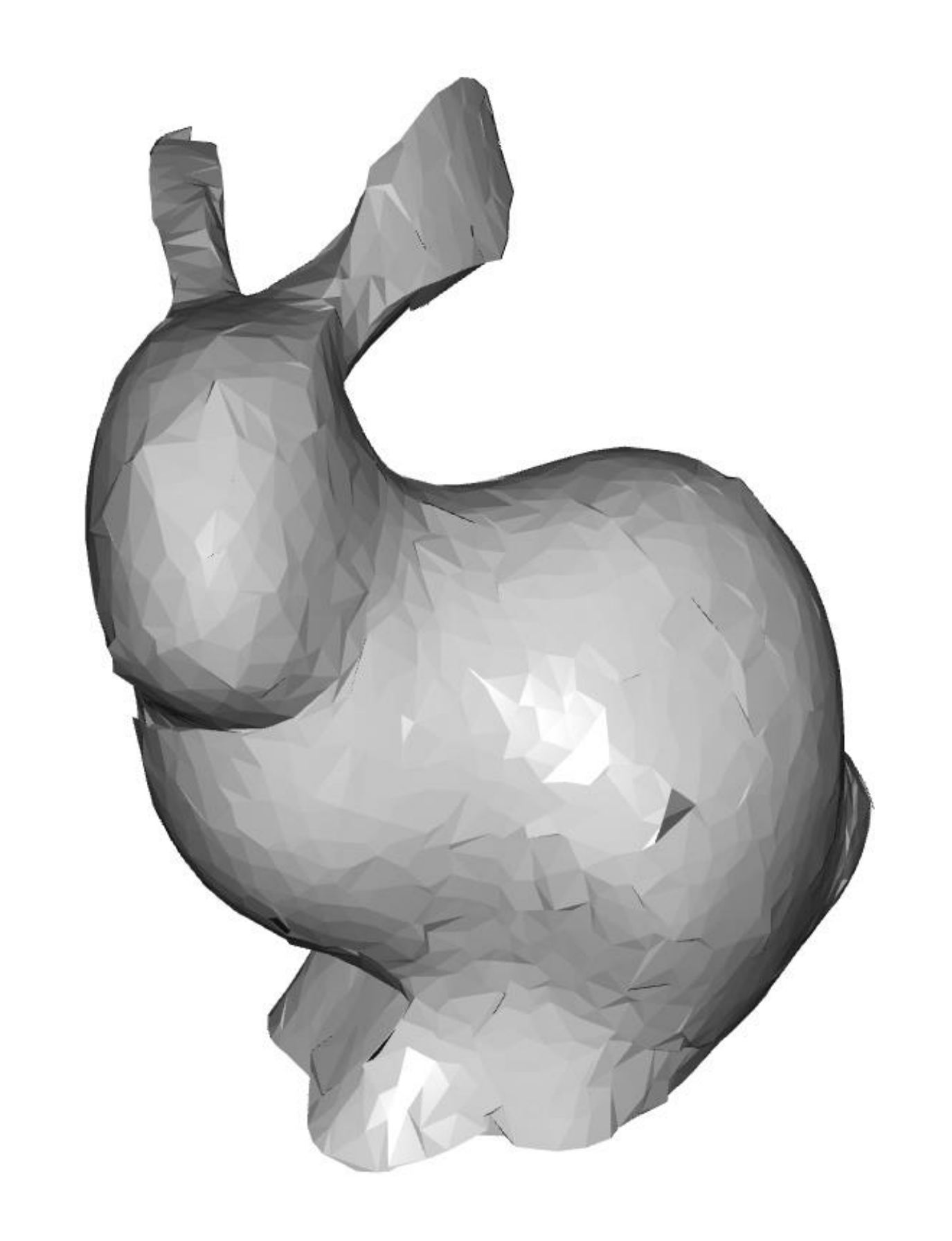}&
		\includegraphics[width=1.72cm]{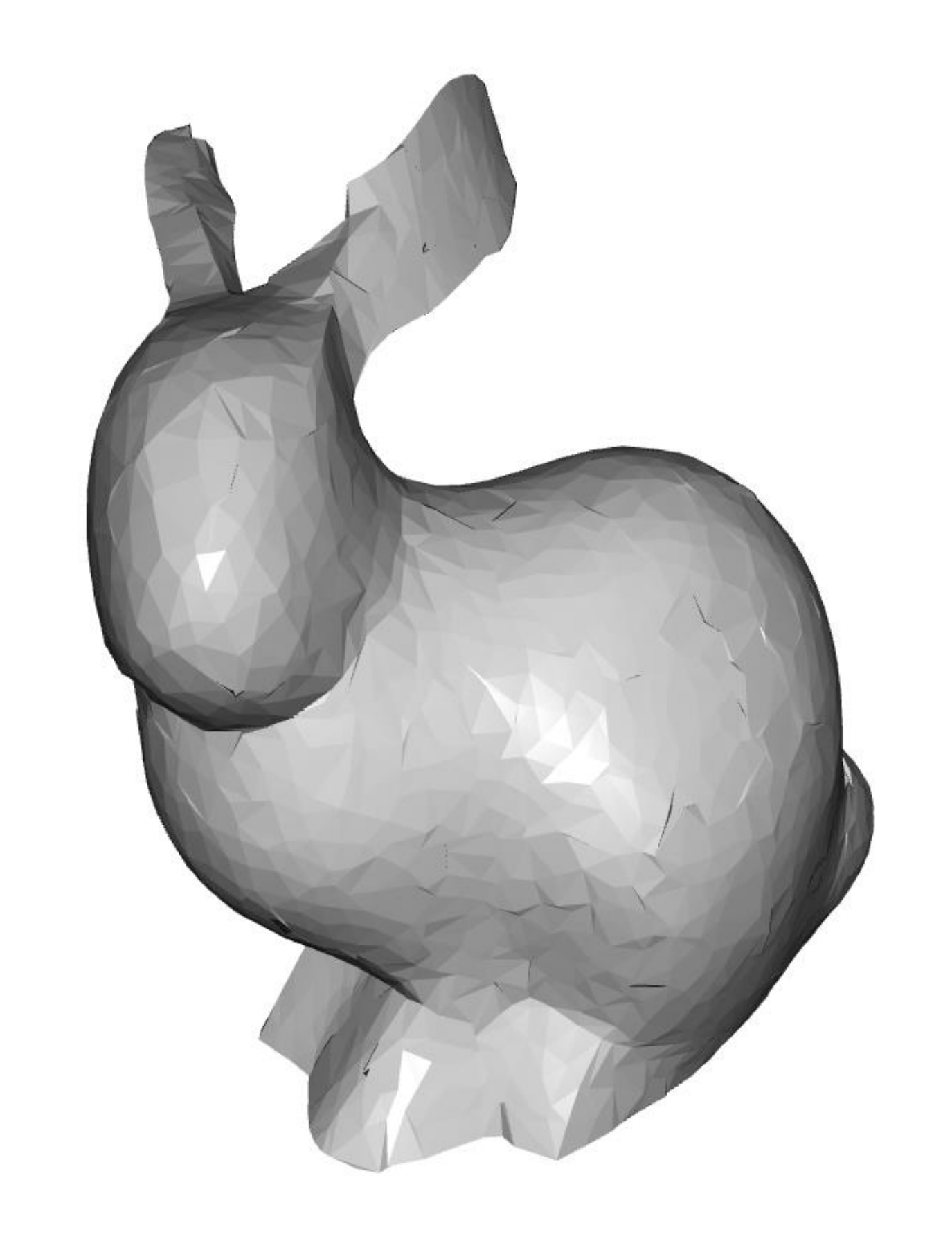}&
		\includegraphics[width=1.72cm]{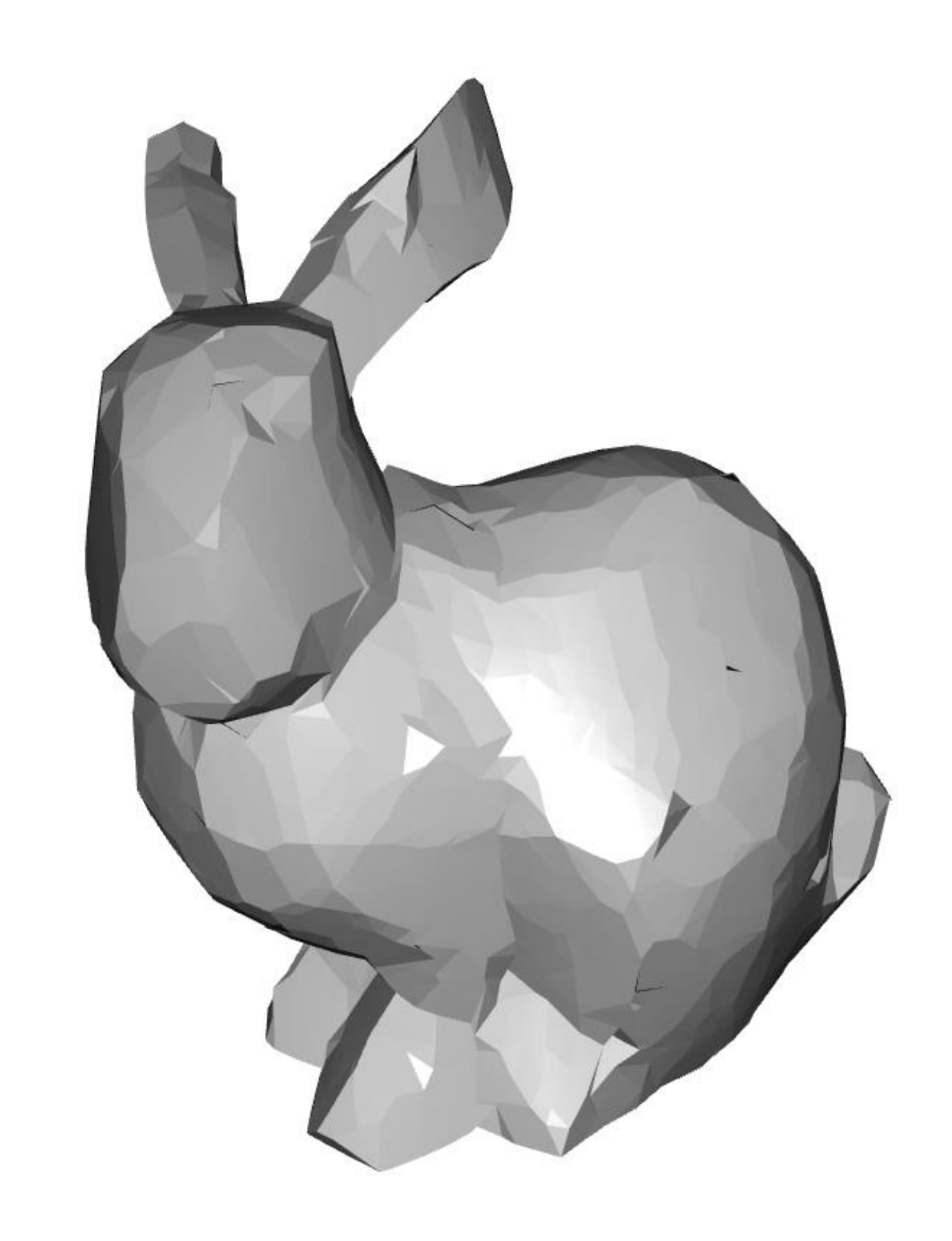}&
		\includegraphics[width=1.72cm]{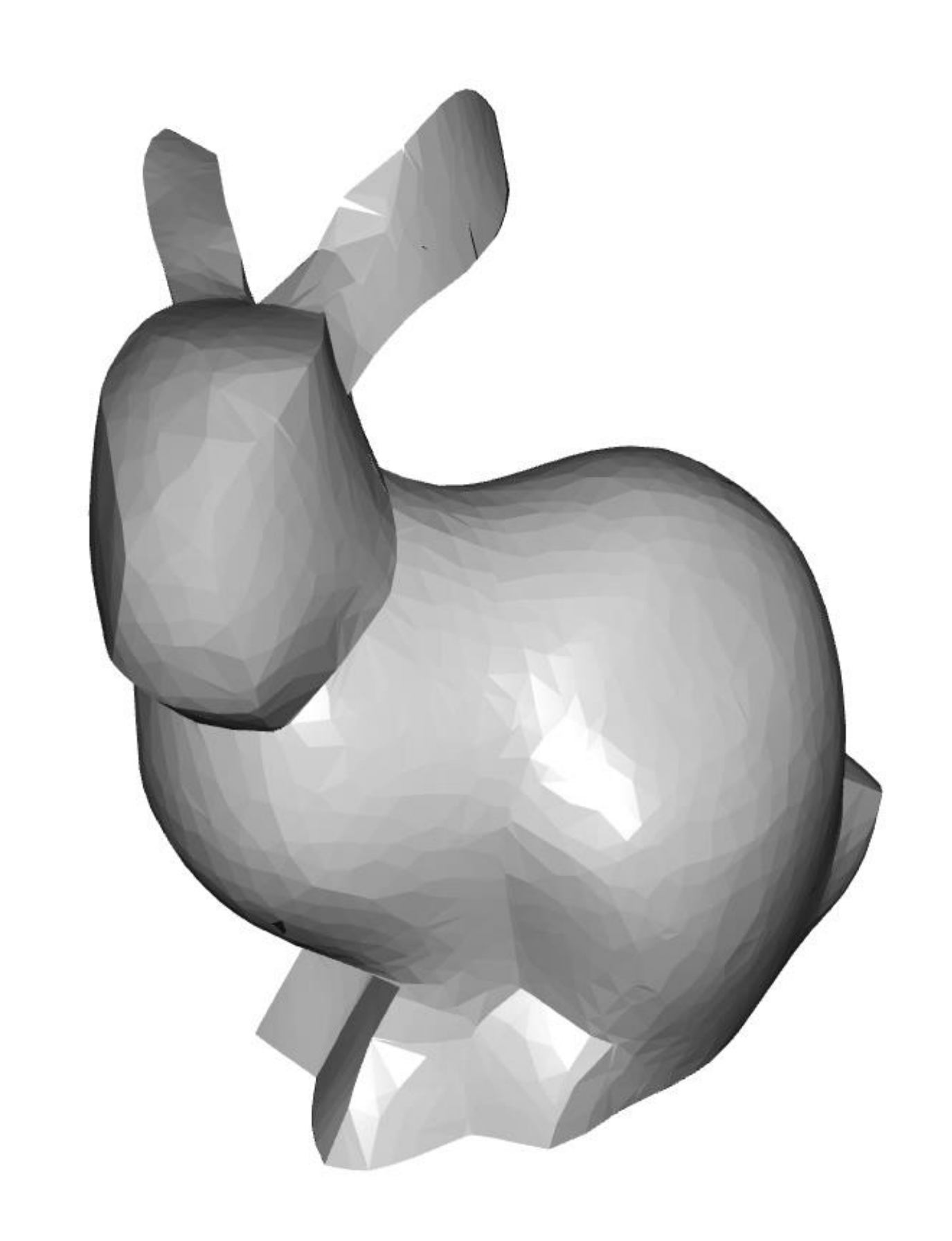}&
		\includegraphics[width=1.72cm]{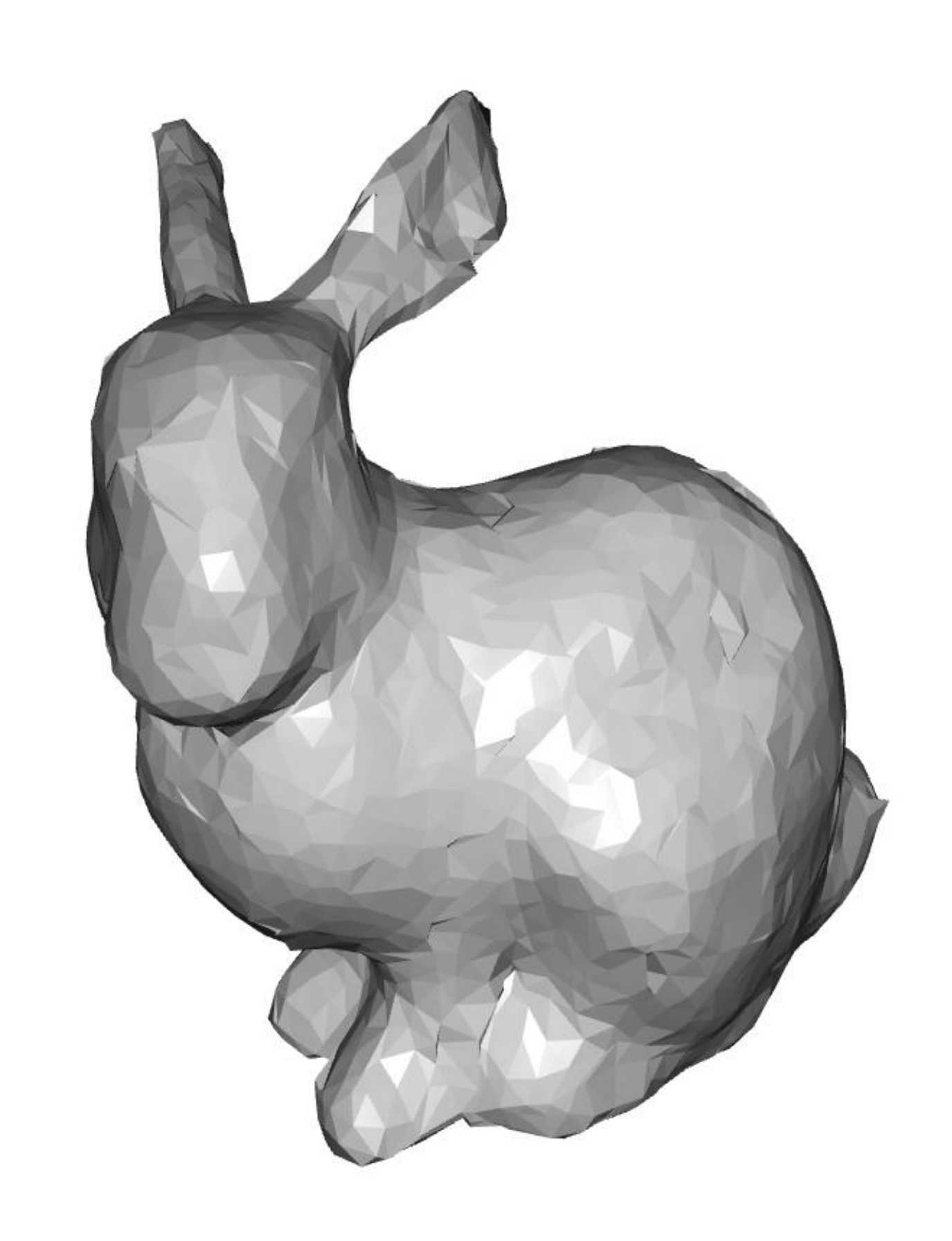}&
		\includegraphics[width=1.72cm]{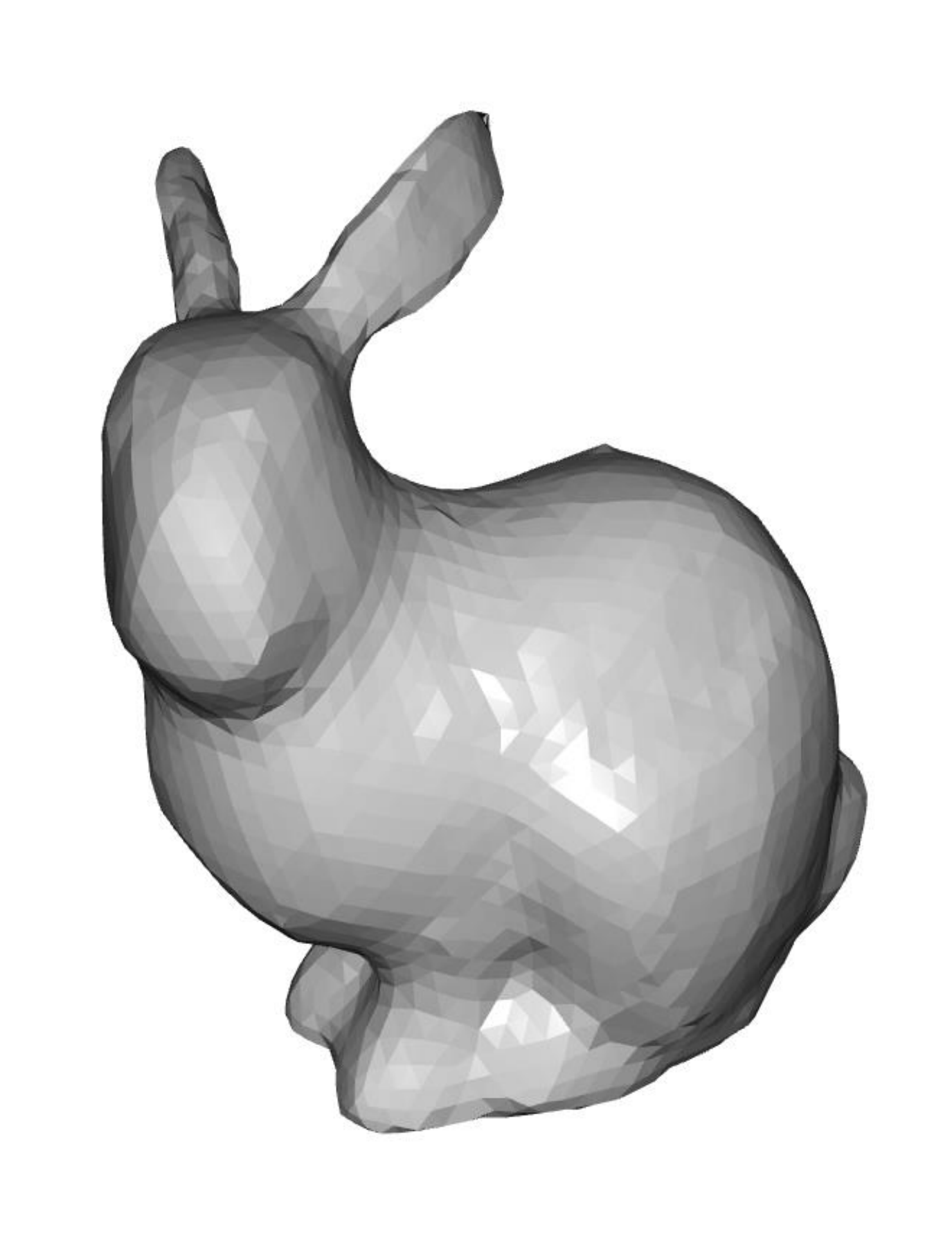}\\
		\includegraphics[width=1.72cm ,height=1.2cm]{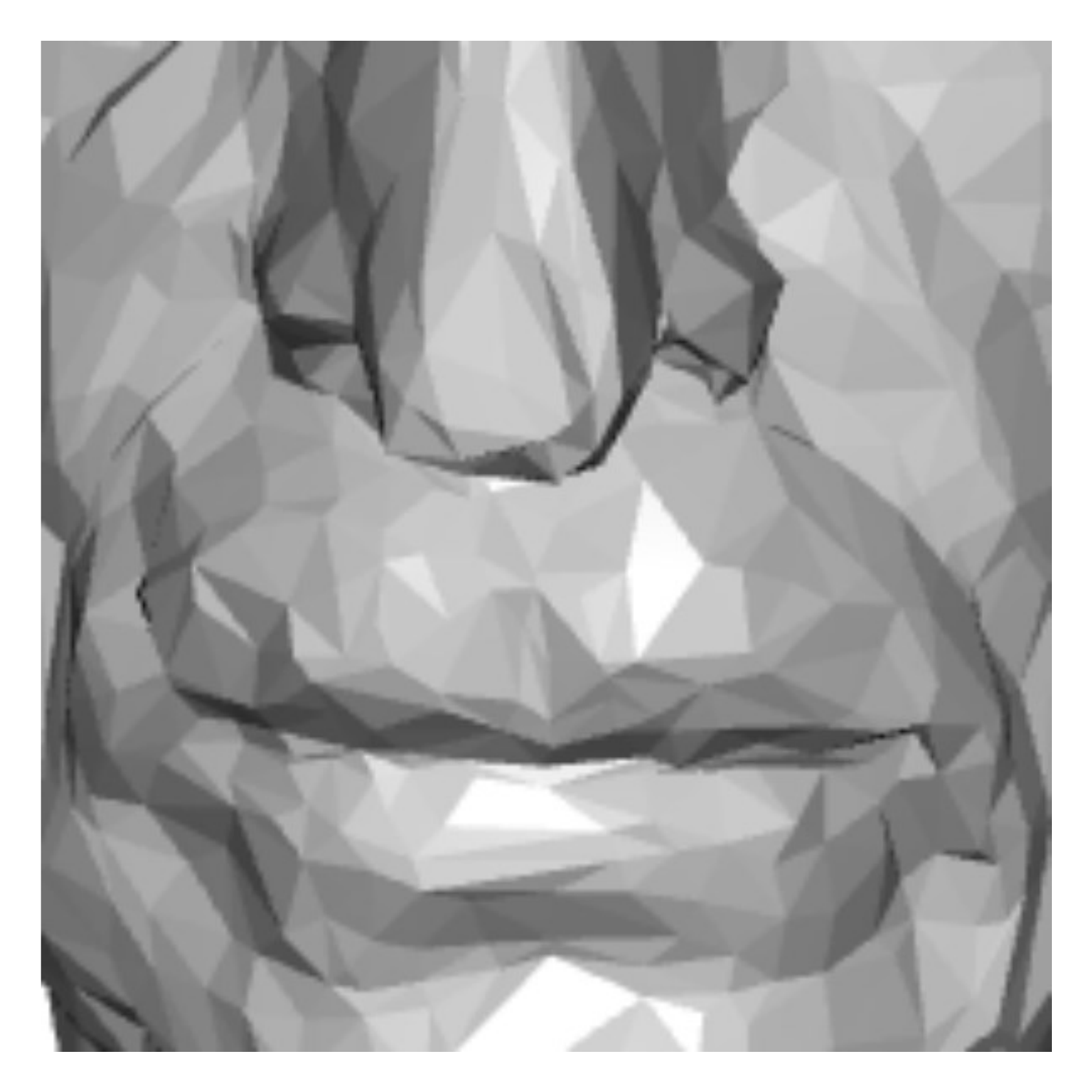}&
		\includegraphics[width=1.72cm,height=1.2cm]{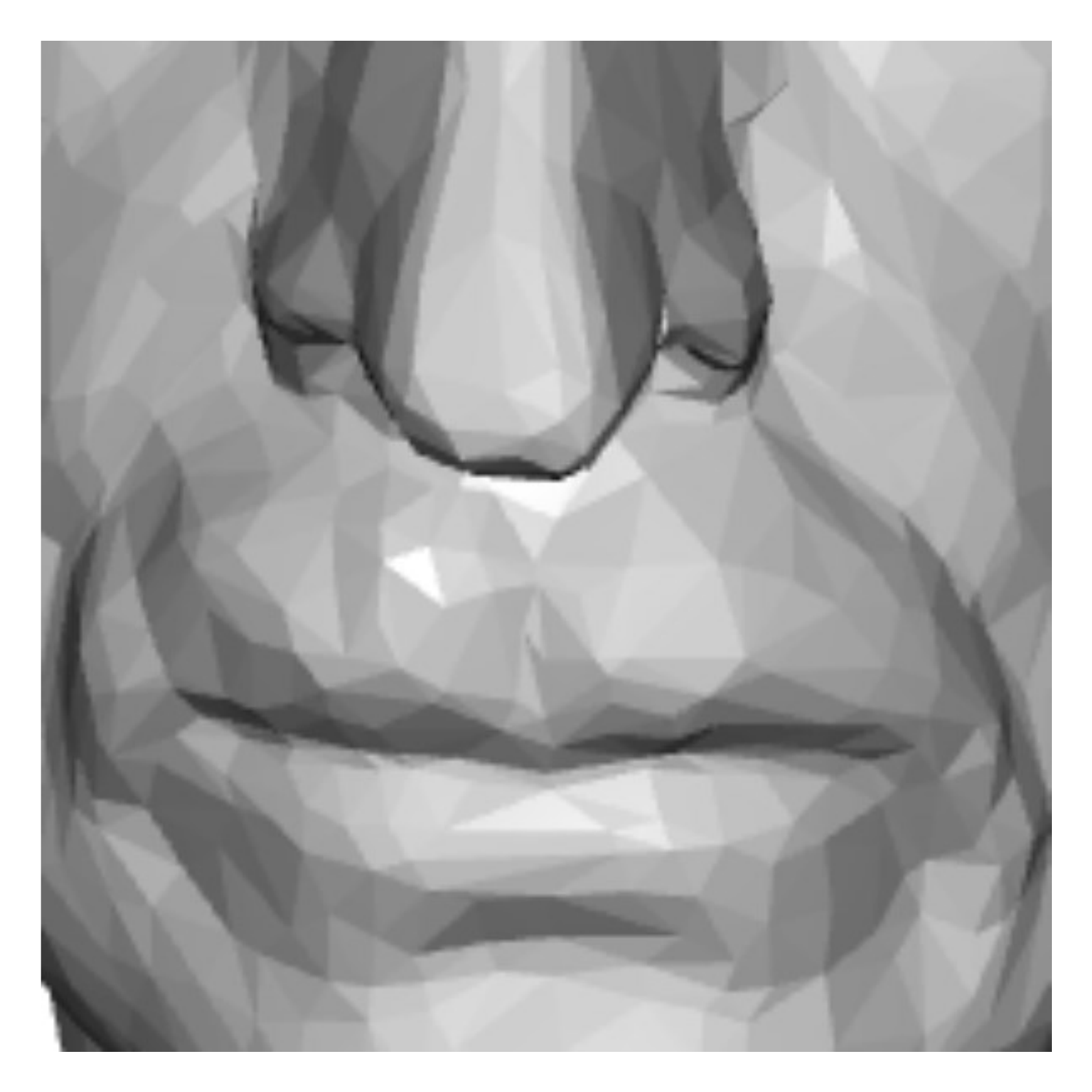}&
		\includegraphics[width=1.72cm,height=1.2cm]{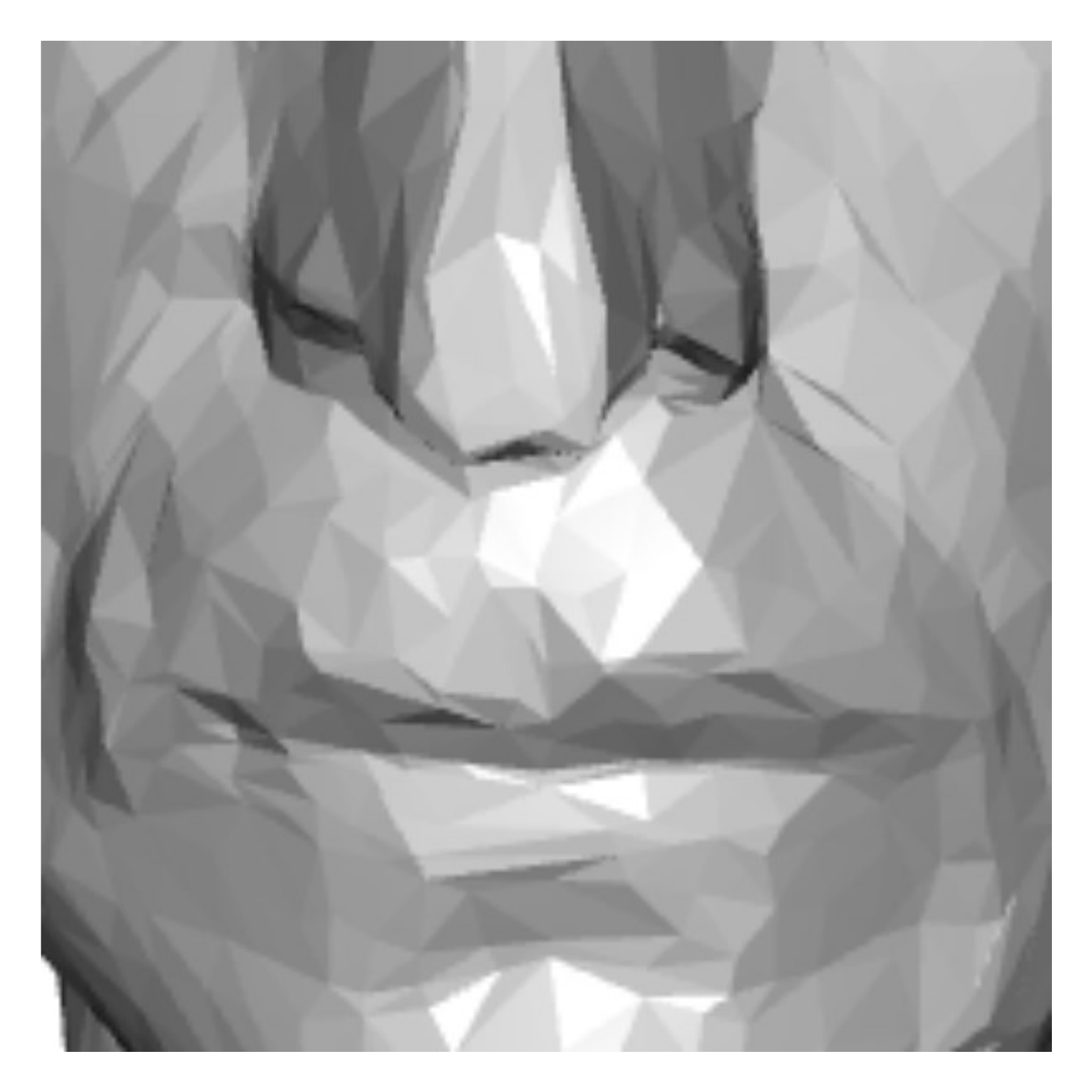}&
		\includegraphics[width=1.72cm,height=1.2cm]{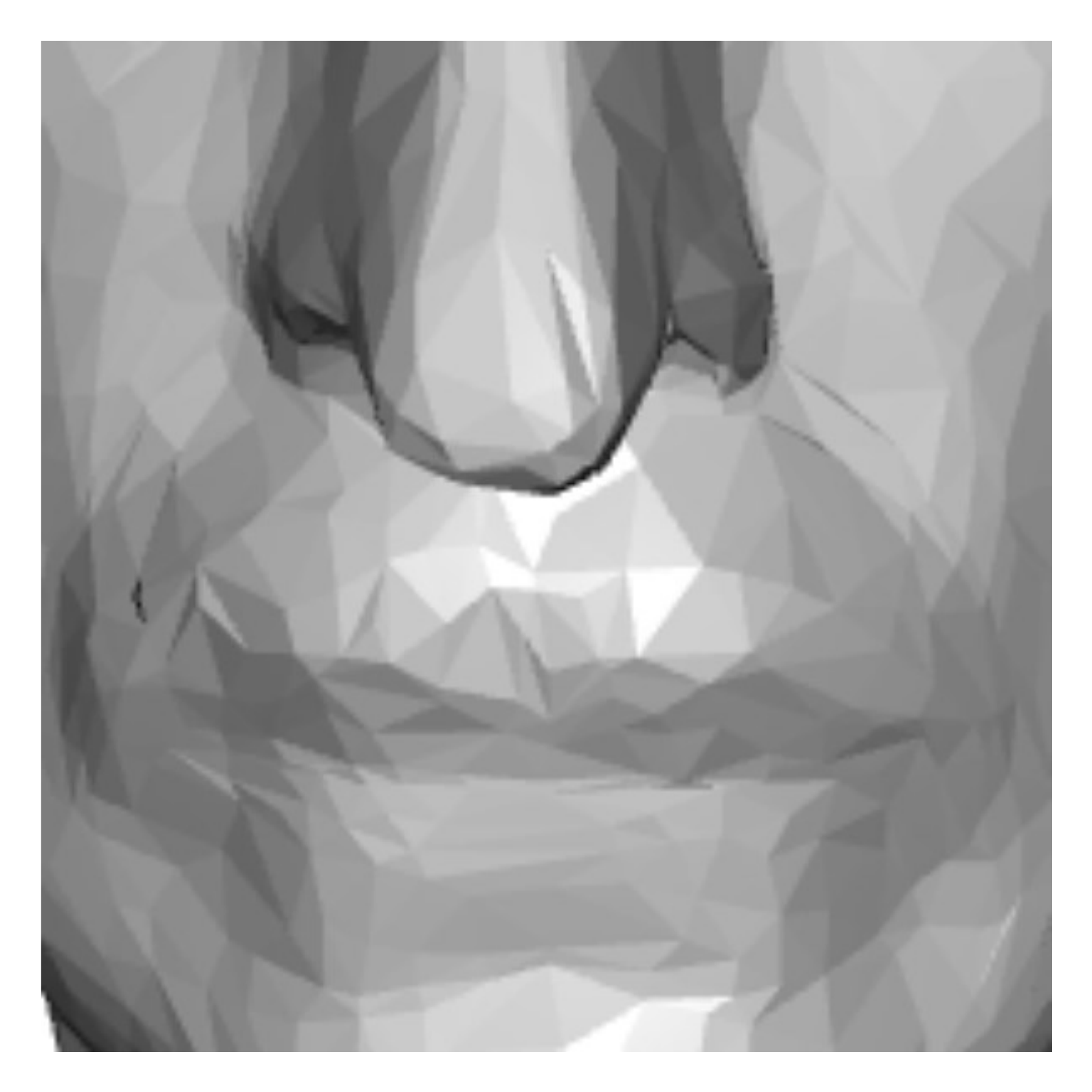}&
		\includegraphics[width=1.72cm,height=1.2cm]{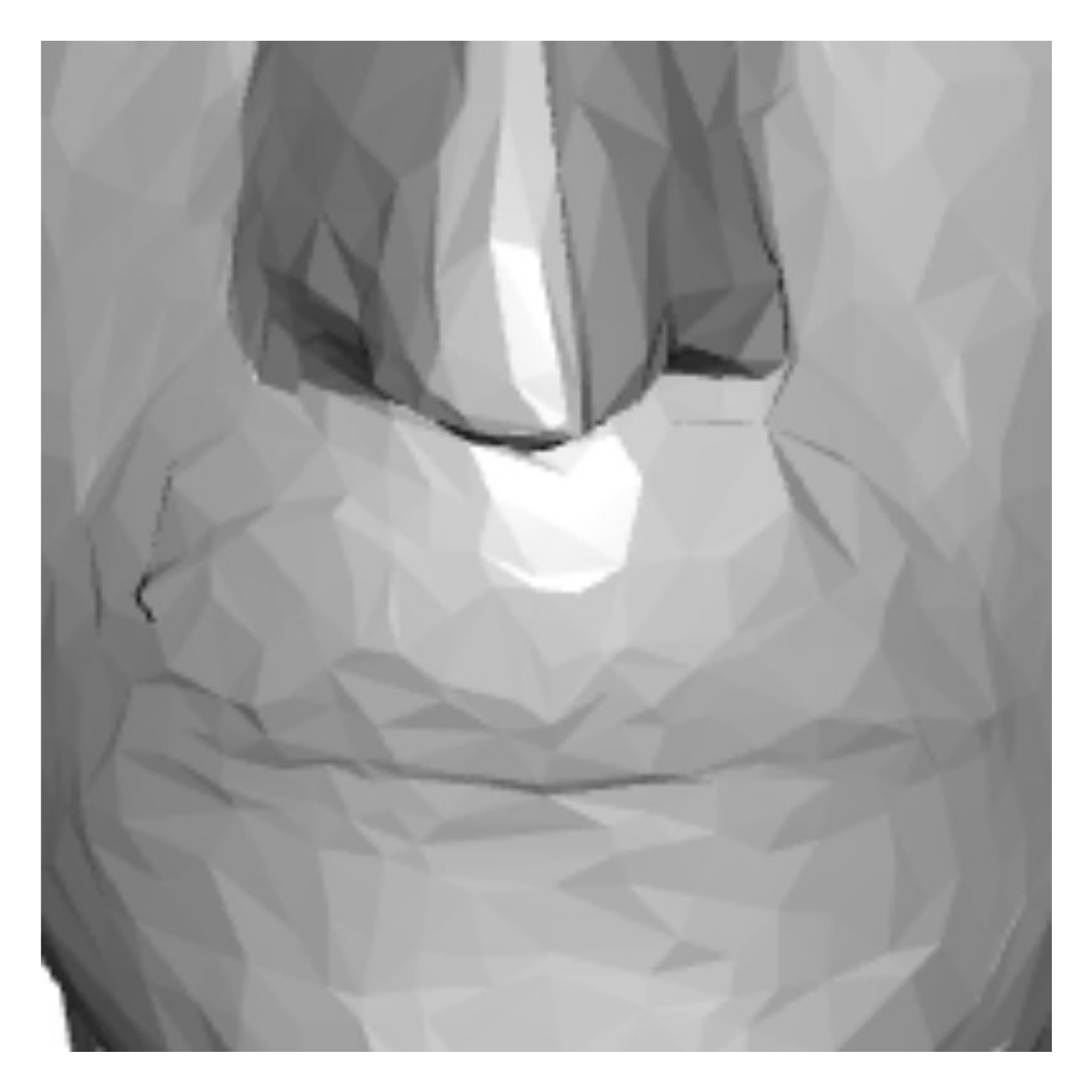}&
		\includegraphics[width=1.72cm,height=1.2cm]{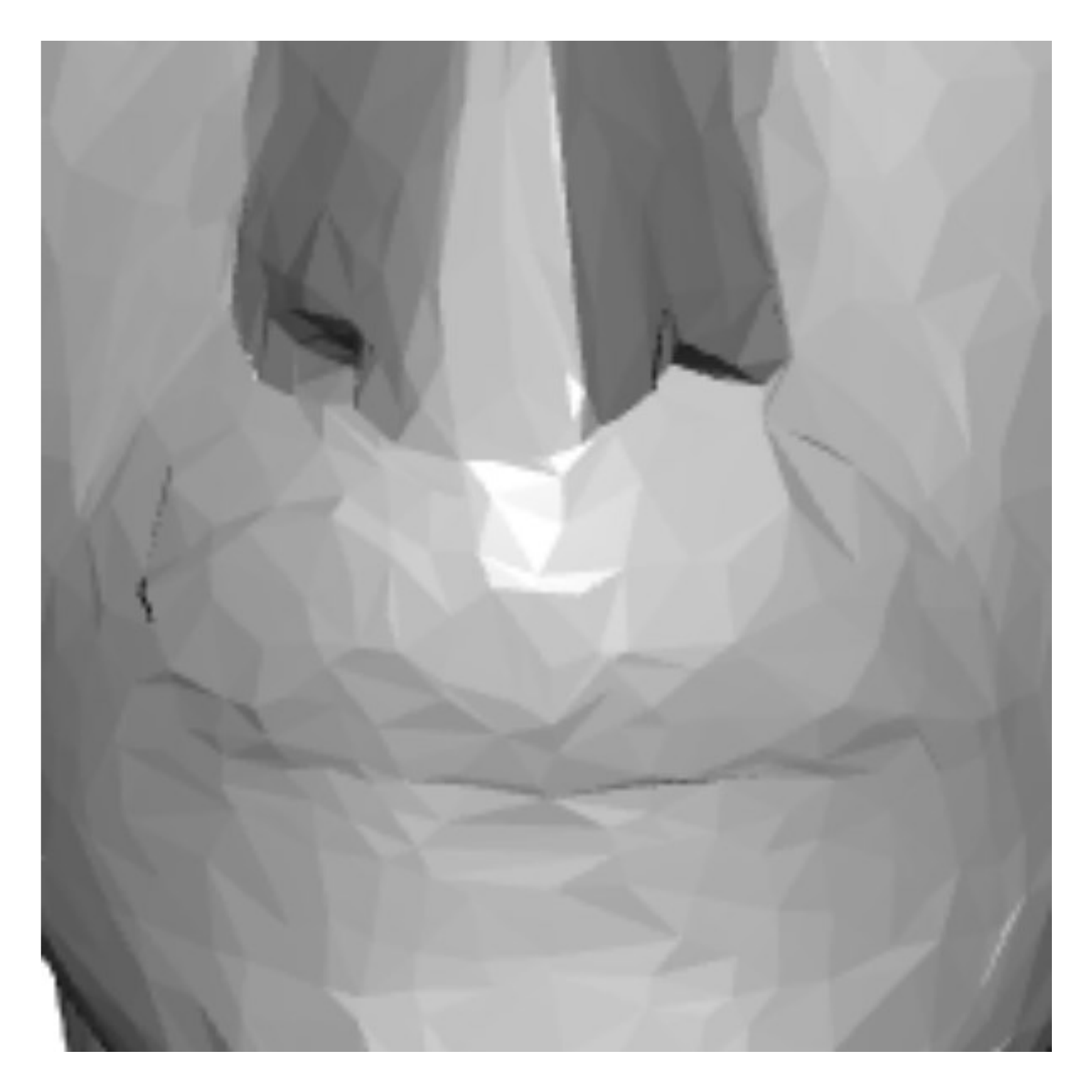}&
		\includegraphics[width=1.72cm,height=1.2cm]{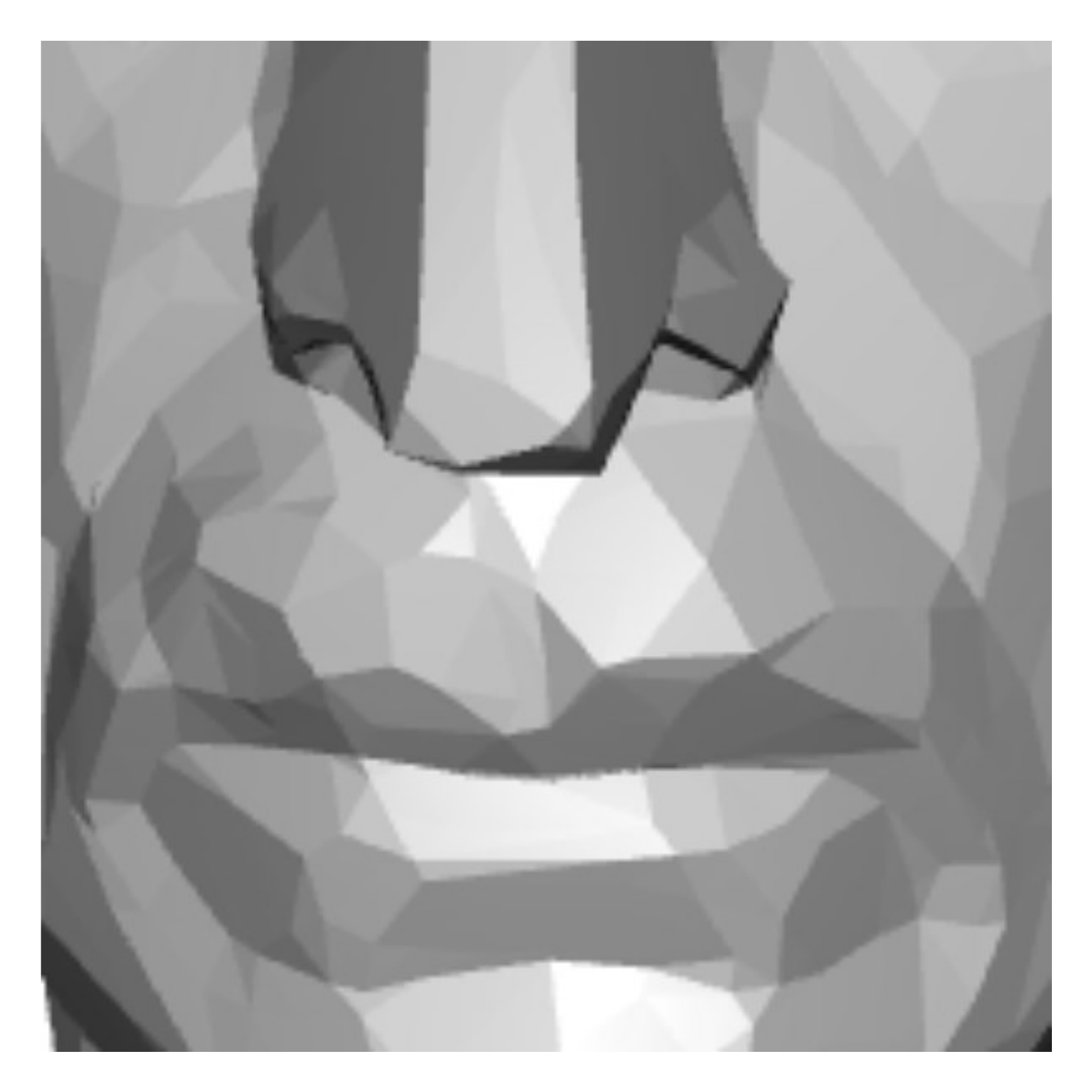}&
		\includegraphics[width=1.72cm,height=1.2cm]{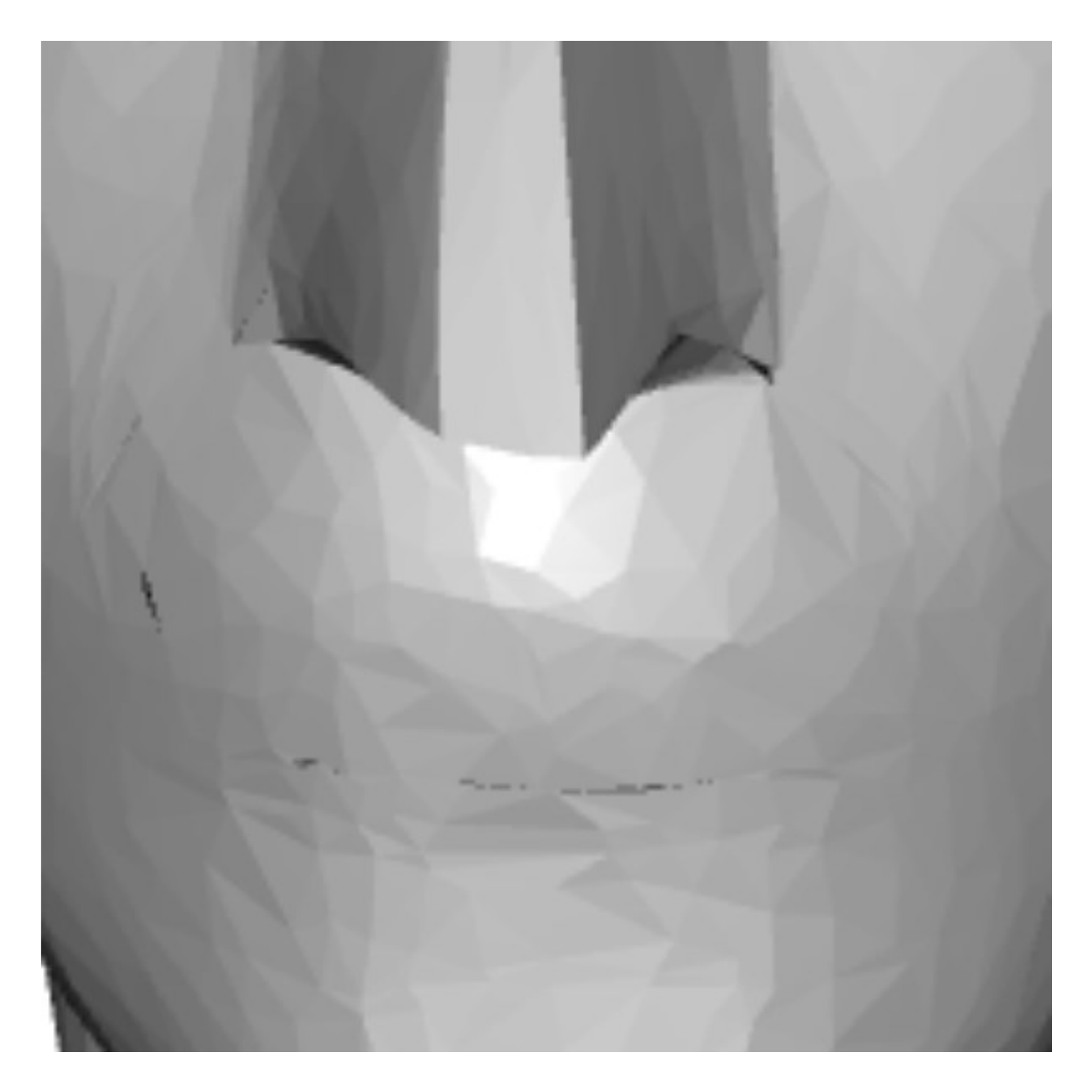}&
		\includegraphics[width=1.72cm,height=1.2cm]{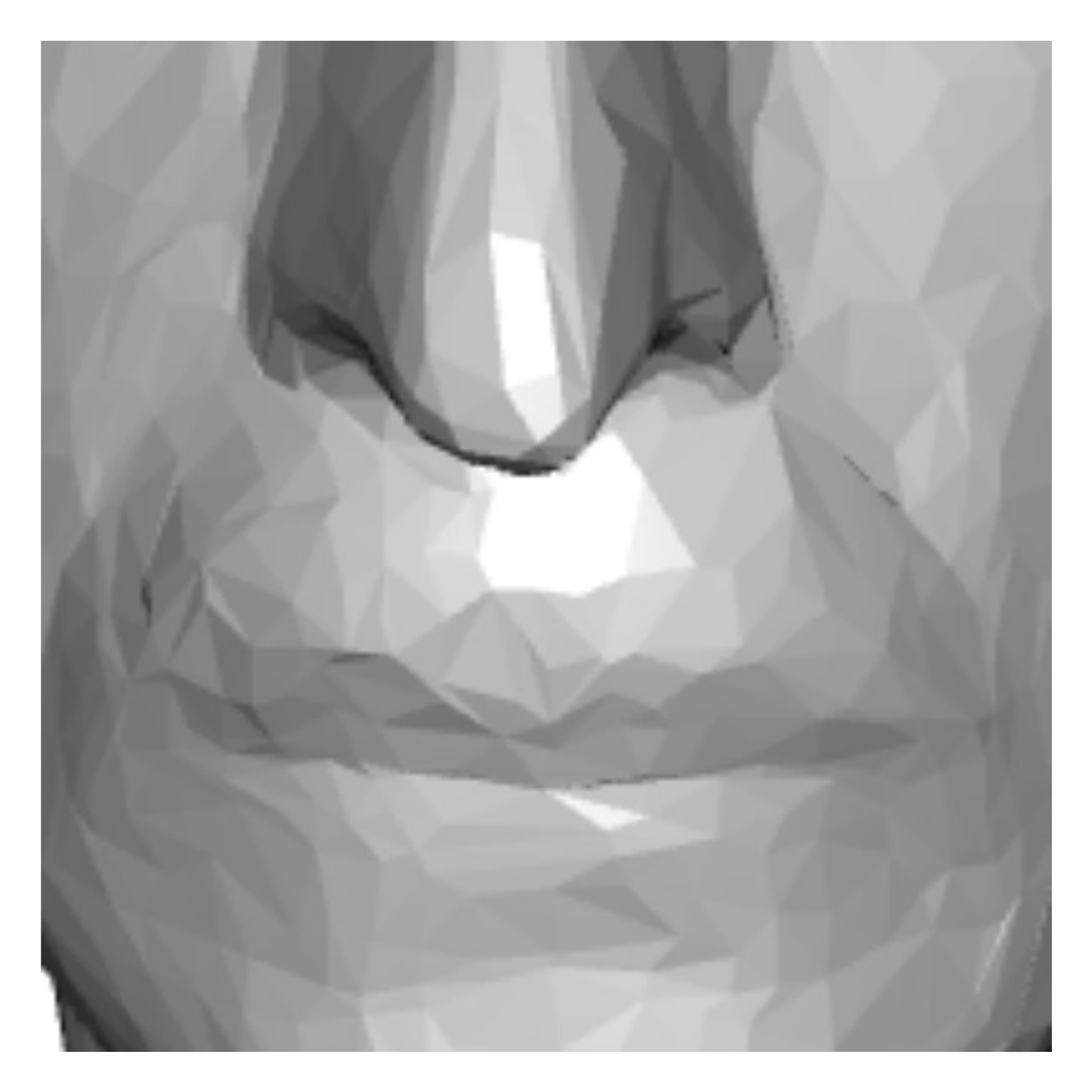}&
		\includegraphics[width=1.72cm,height=1.2cm]{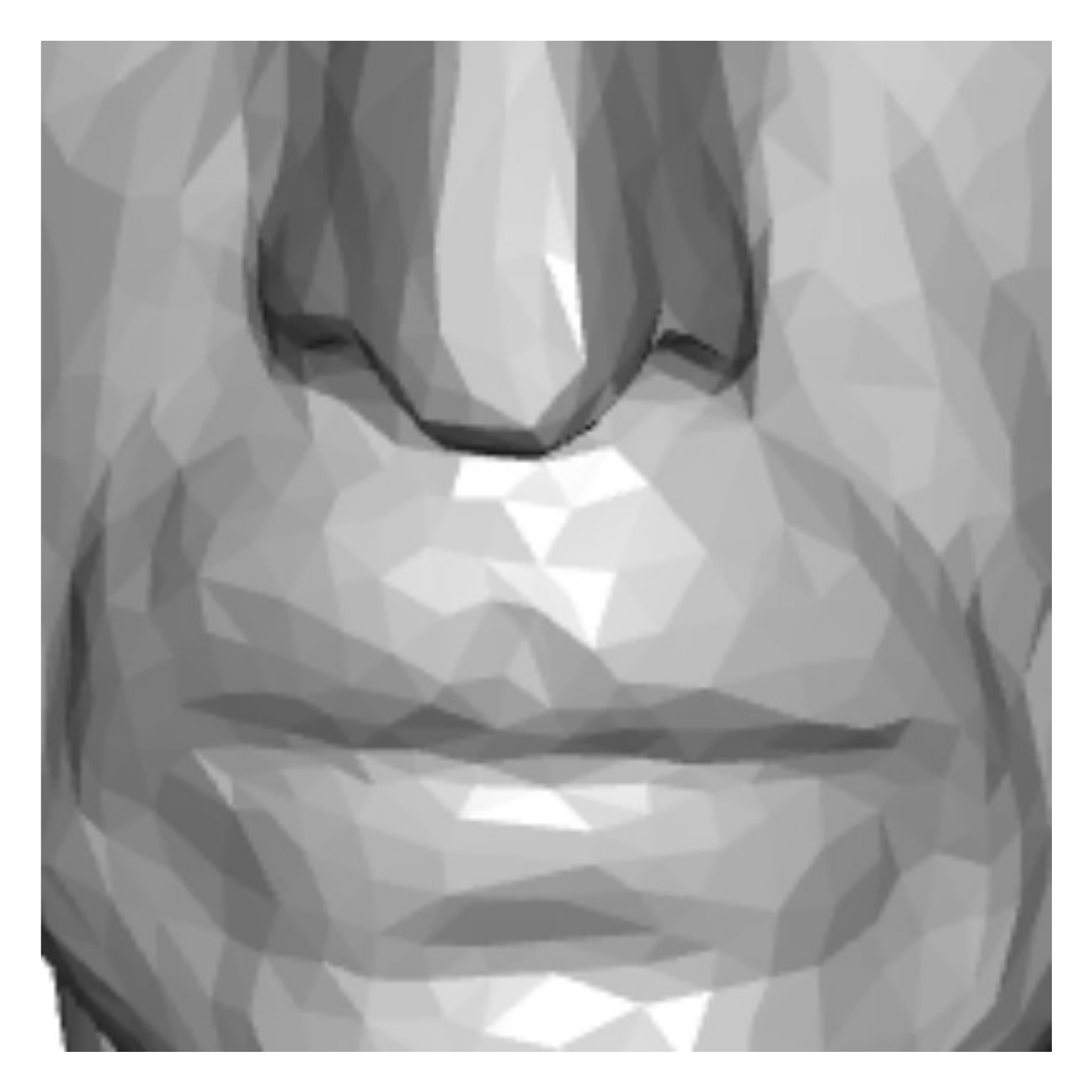}\\
		\includegraphics[width=1.72cm]{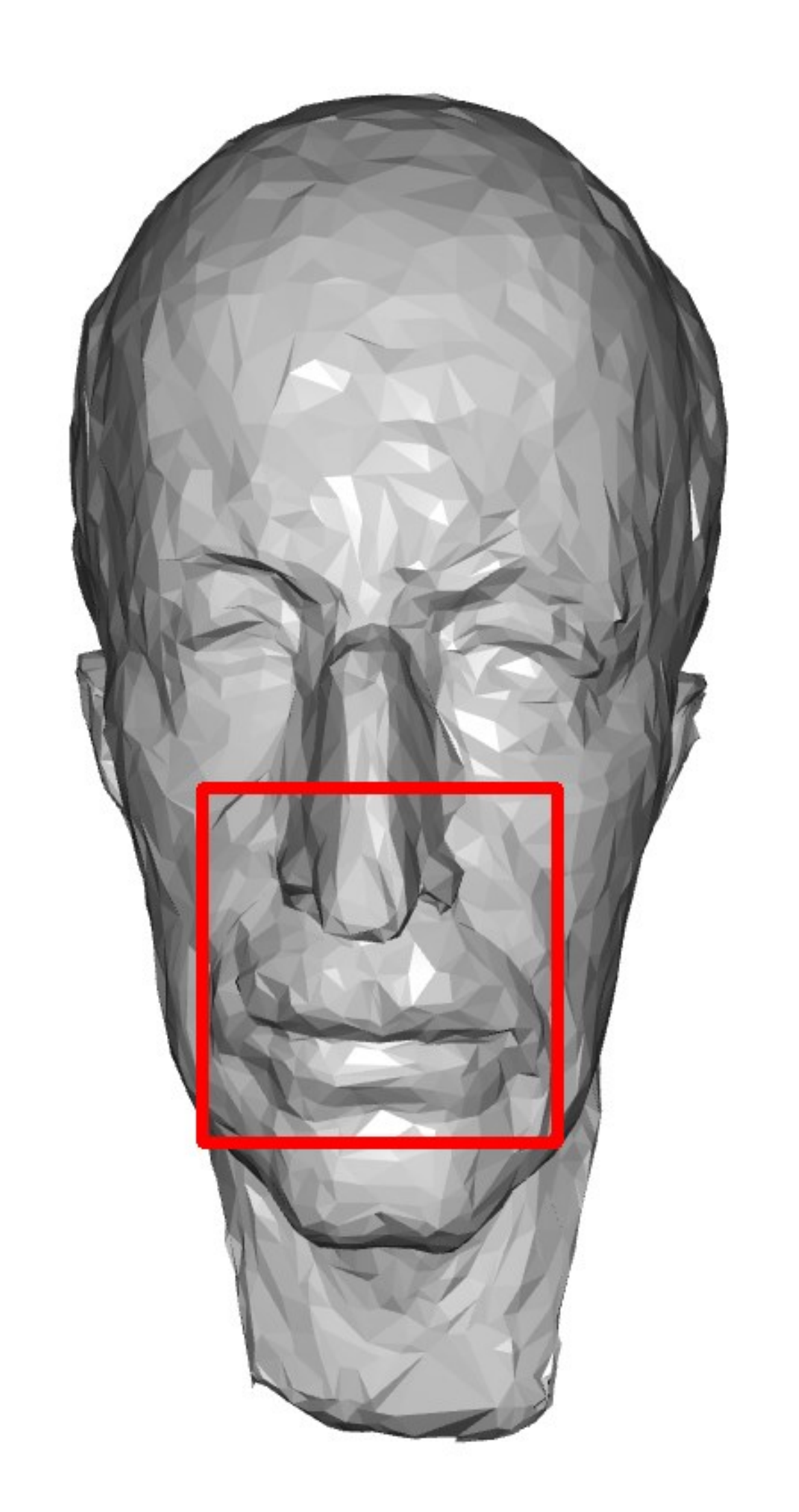}&
		\includegraphics[width=1.72cm]{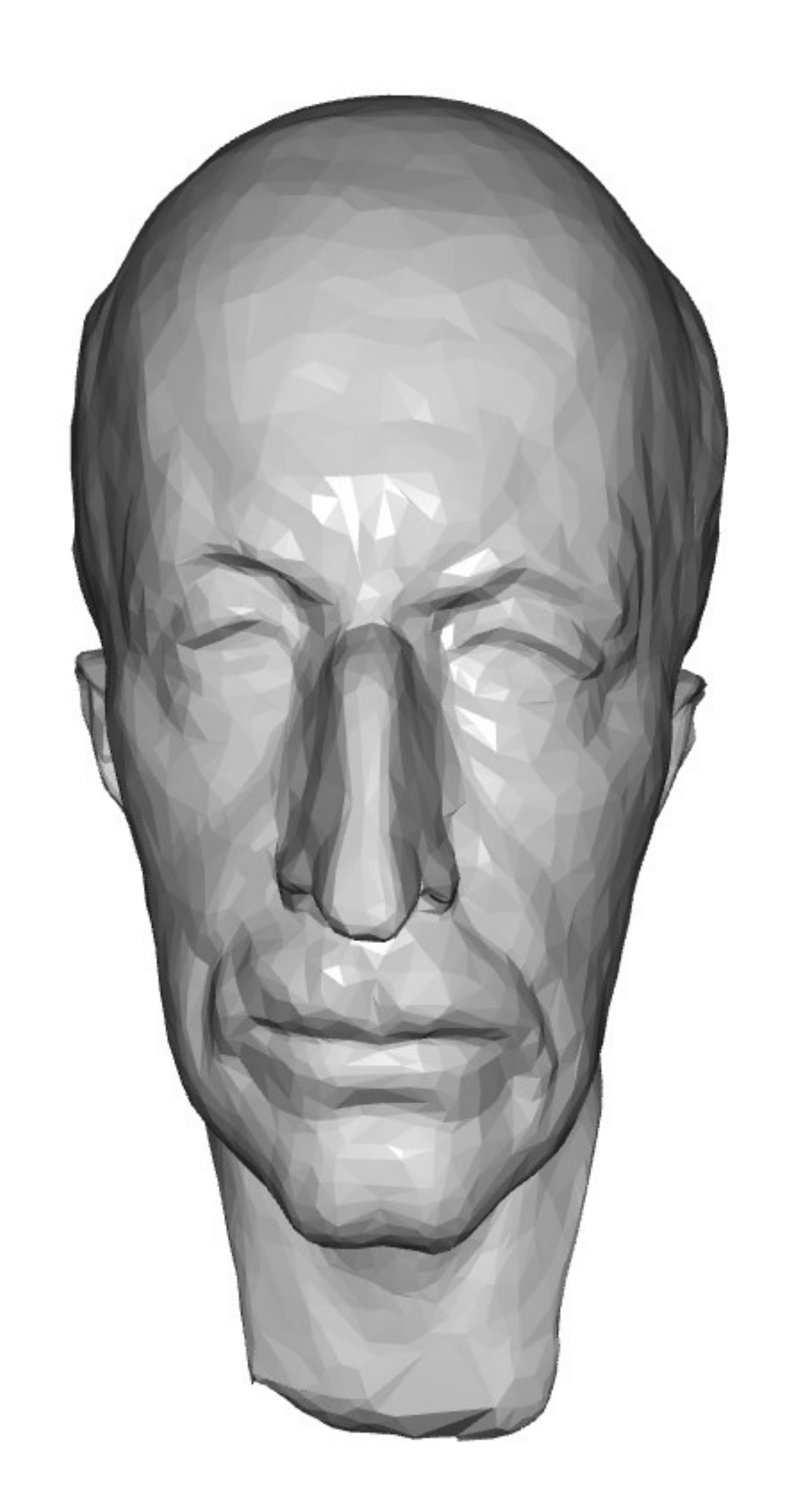}&
		\includegraphics[width=1.72cm]{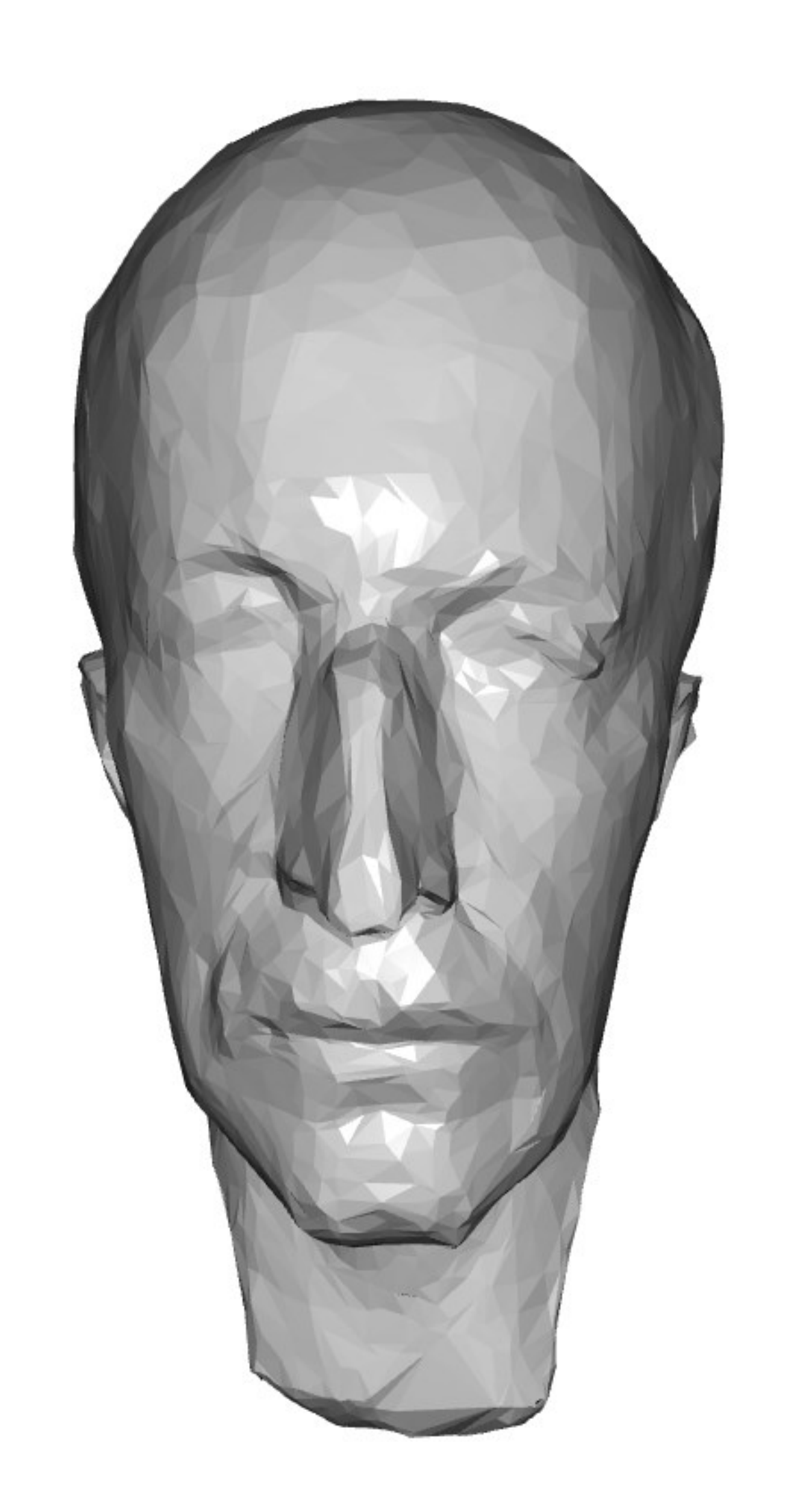}&
		\includegraphics[width=1.72cm]{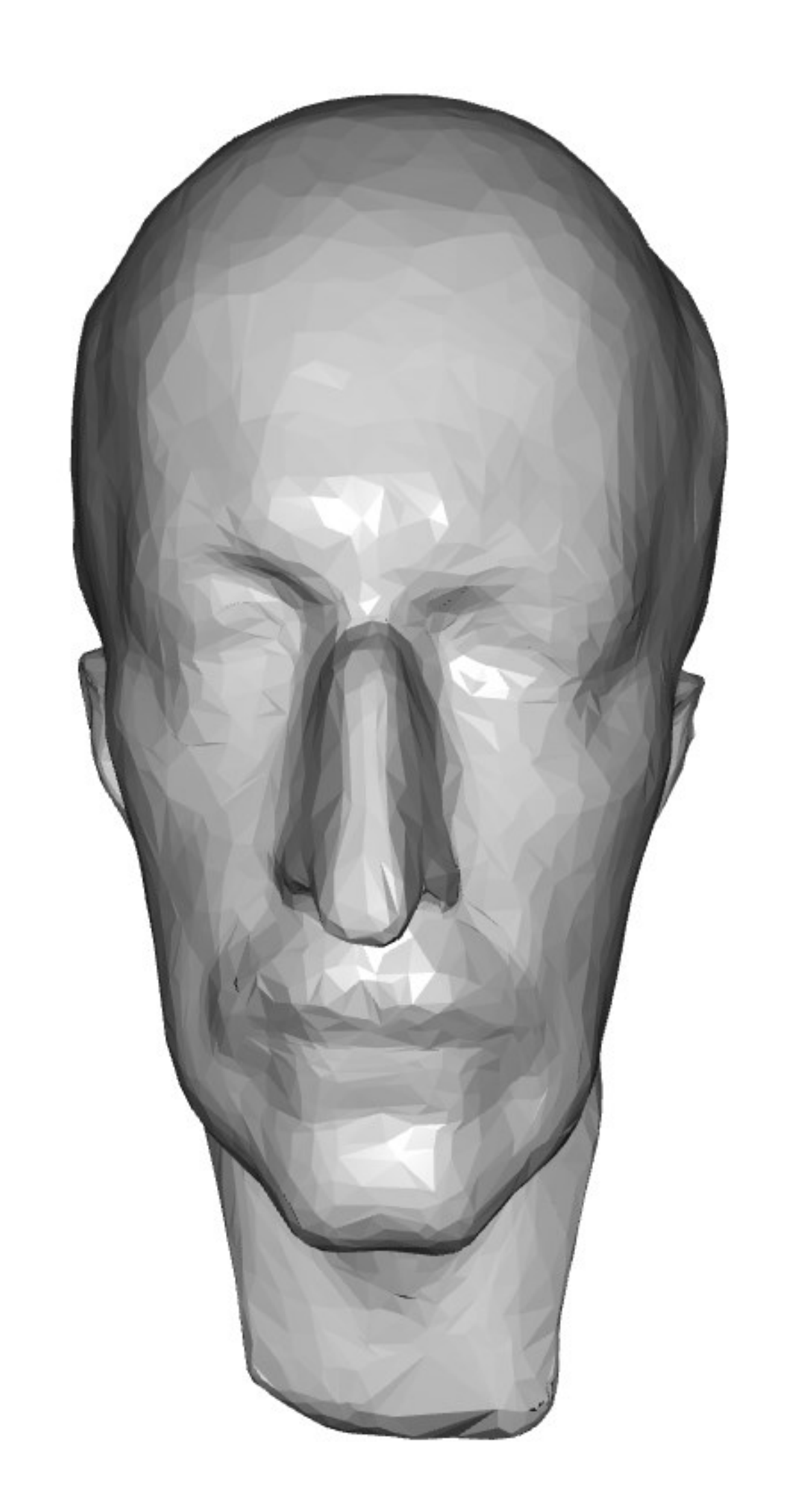}&
		\includegraphics[width=1.72cm]{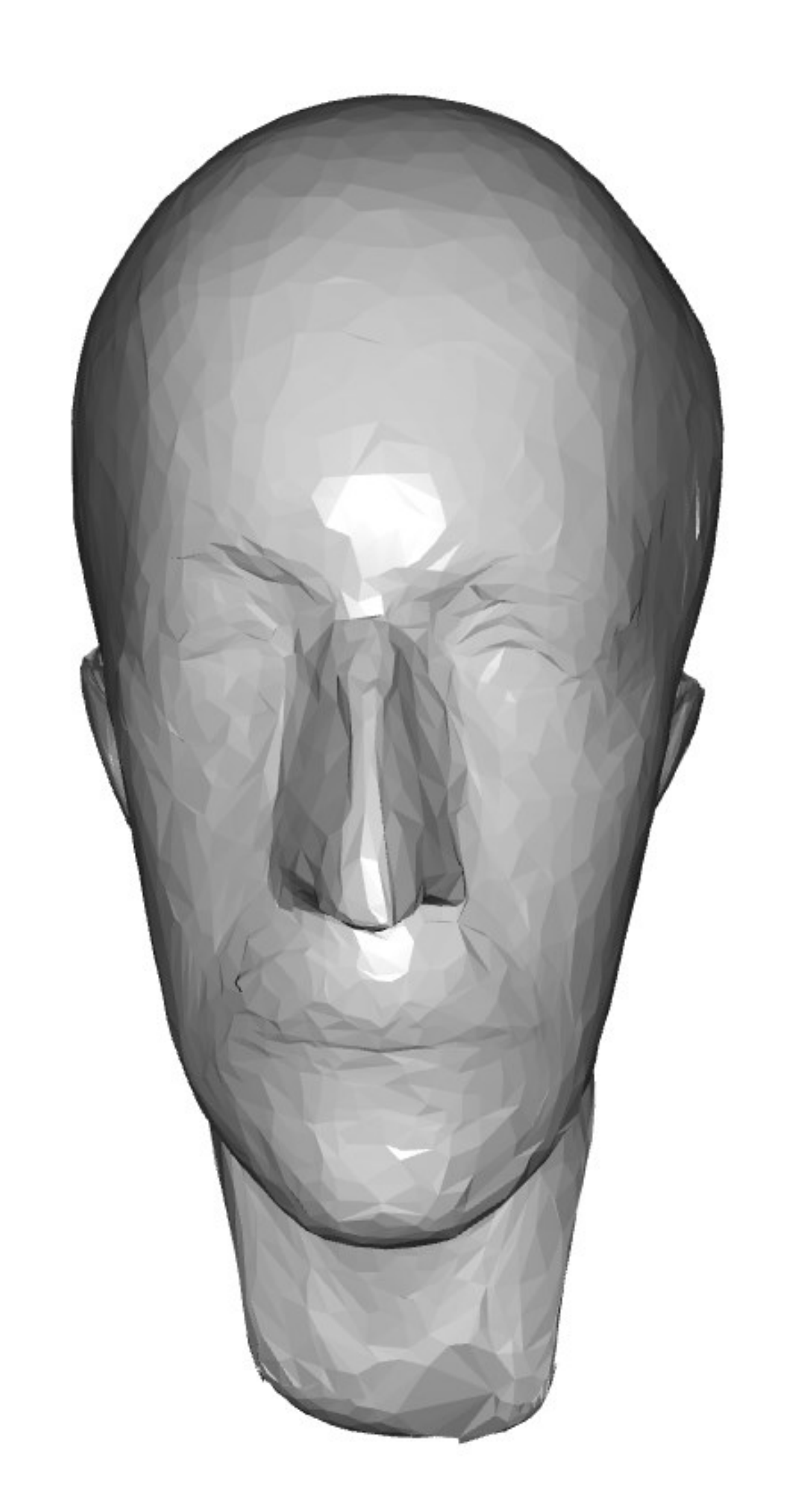}&
		\includegraphics[width=1.72cm]{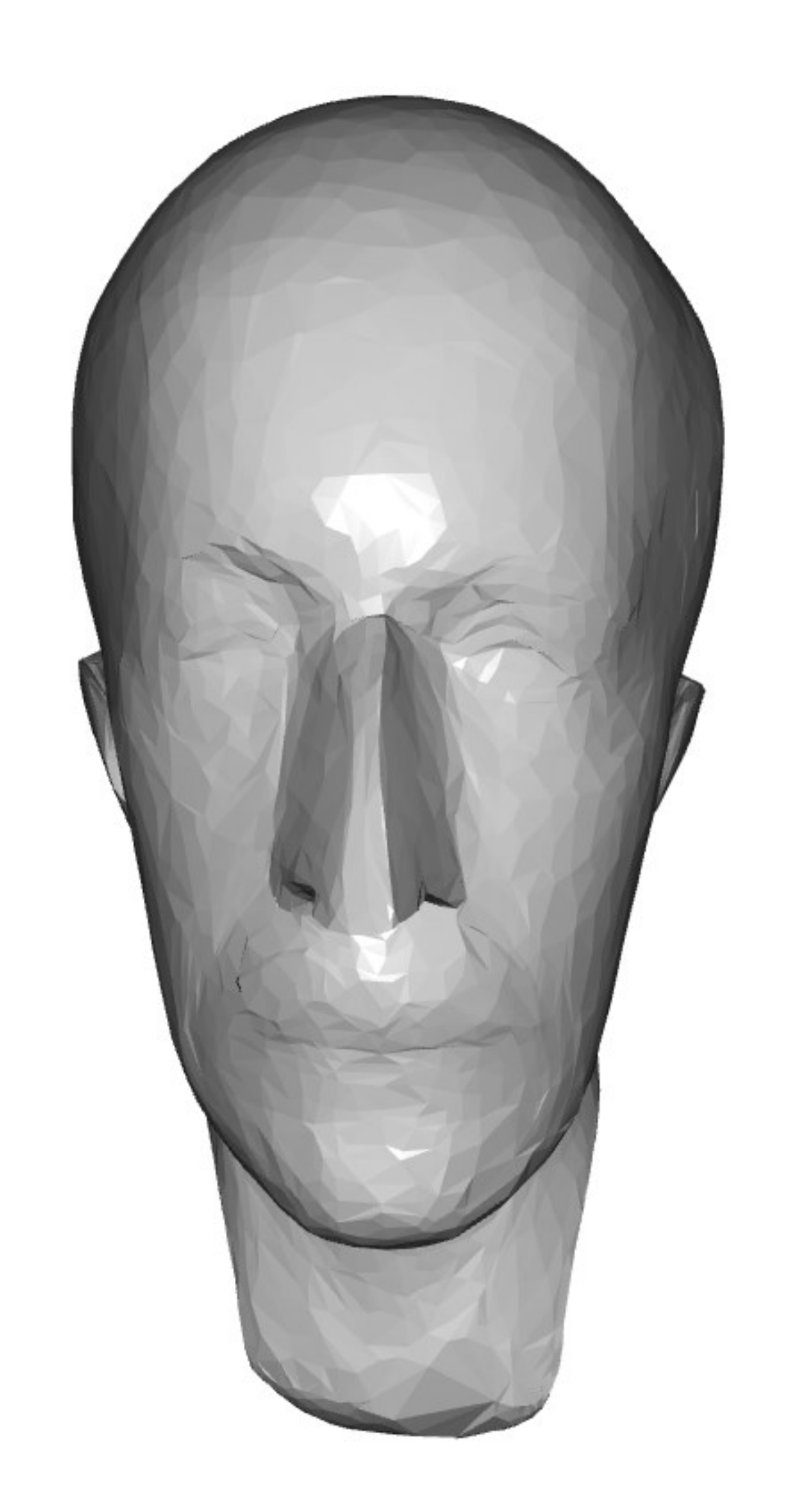}&
		\includegraphics[width=1.72cm]{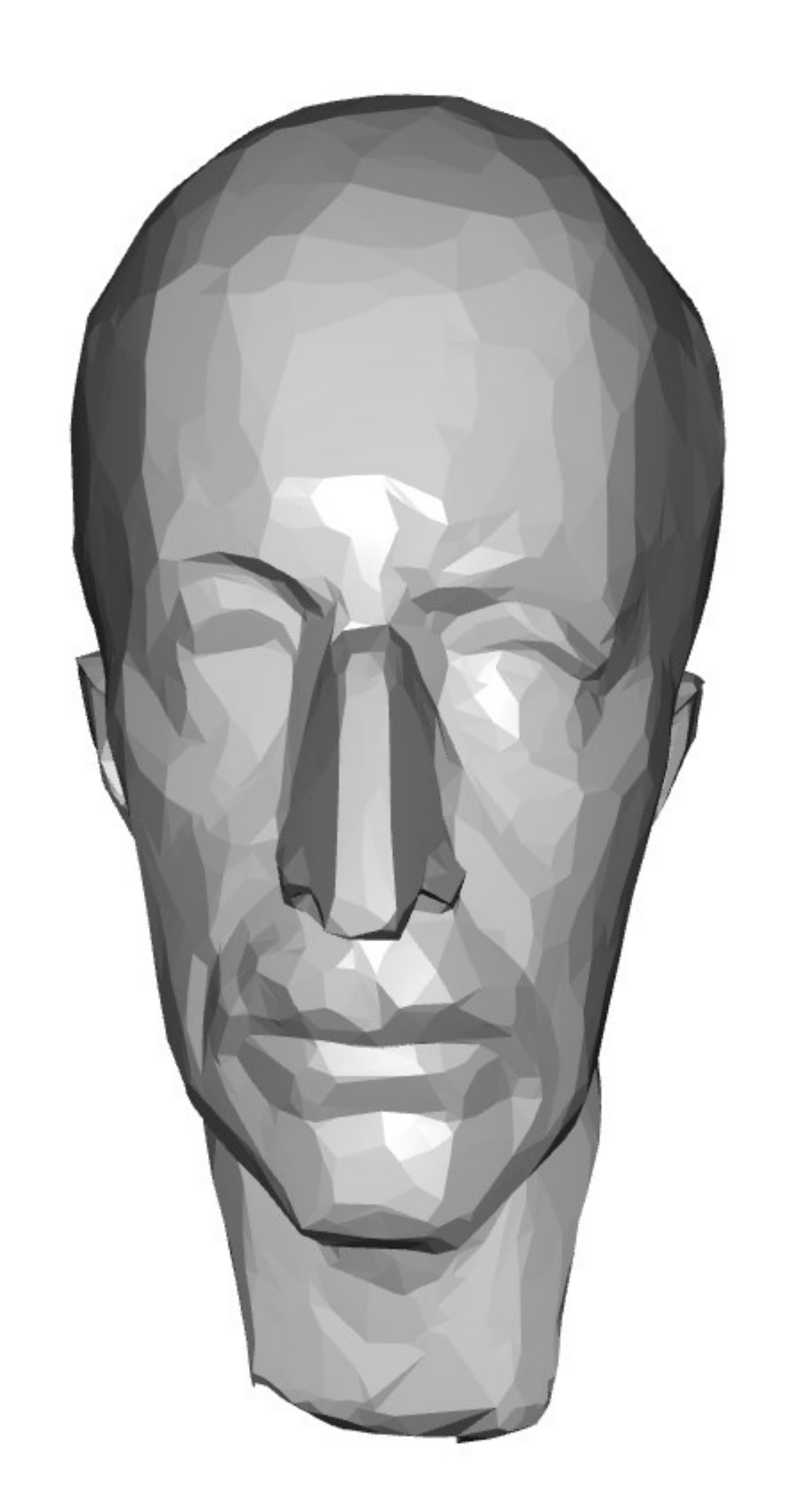}&
		\includegraphics[width=1.72cm]{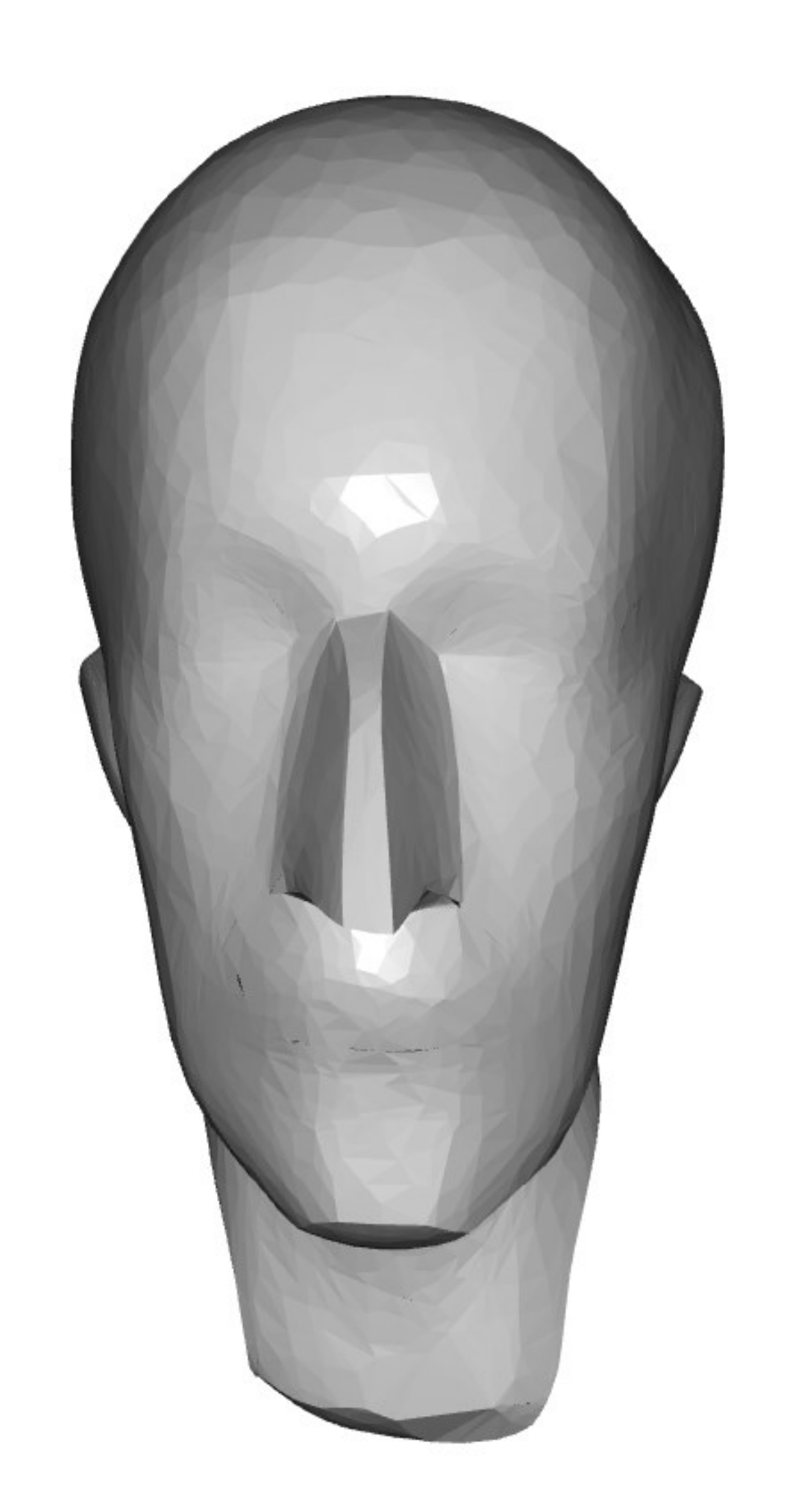}&
		\includegraphics[width=1.72cm]{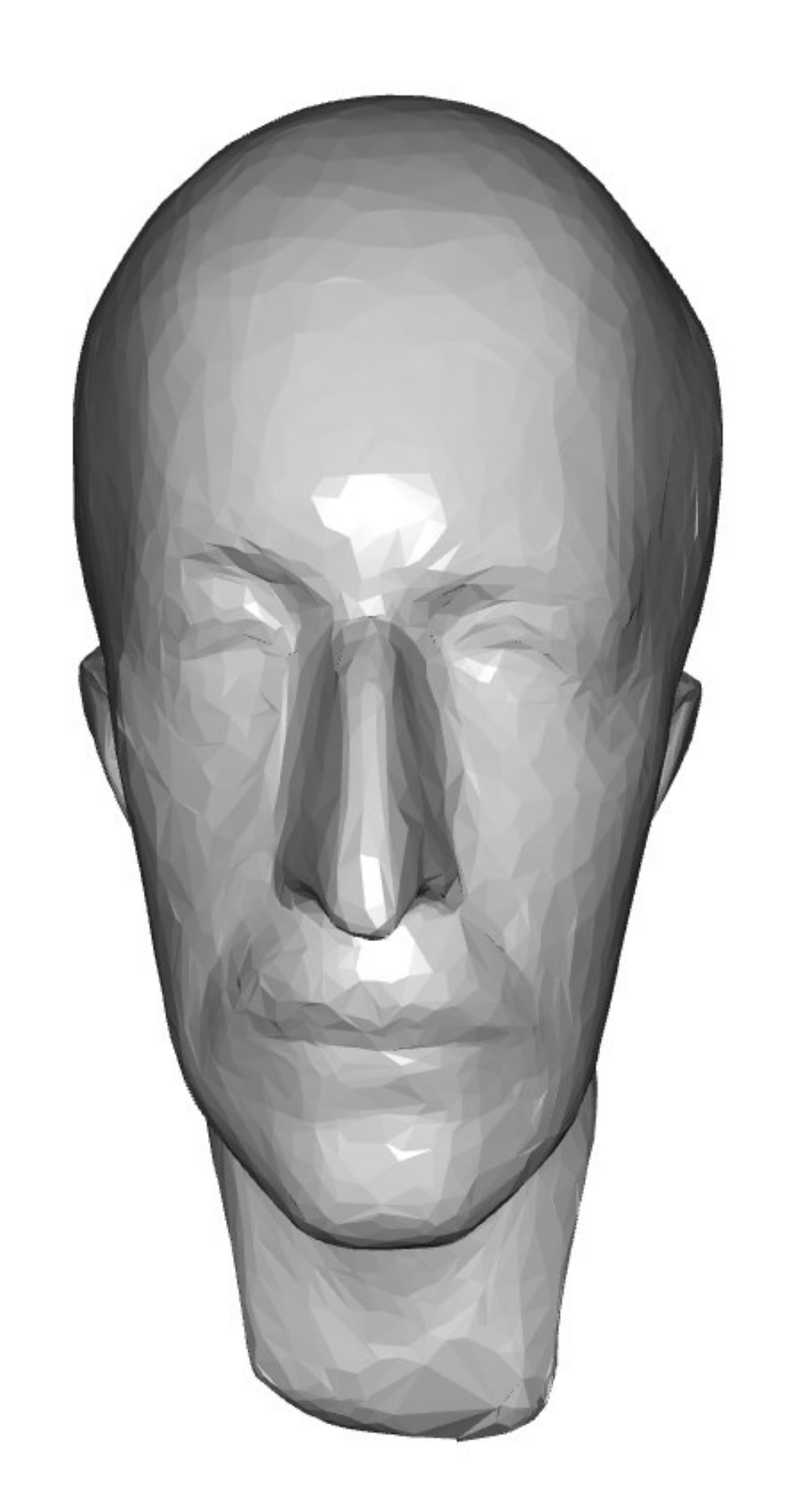}&
		\includegraphics[width=1.72cm]{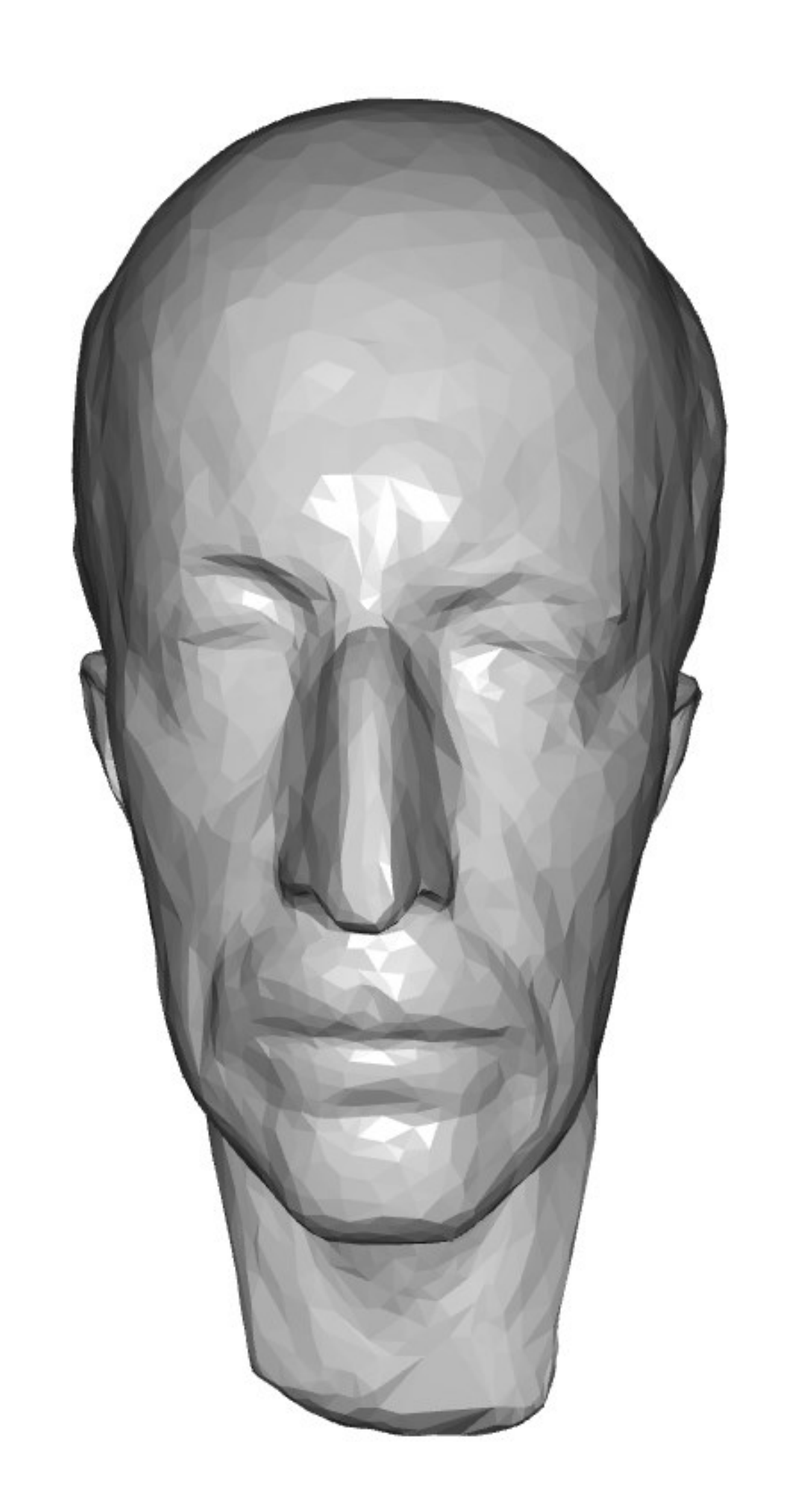}\\
		\includegraphics[width=1.72cm]{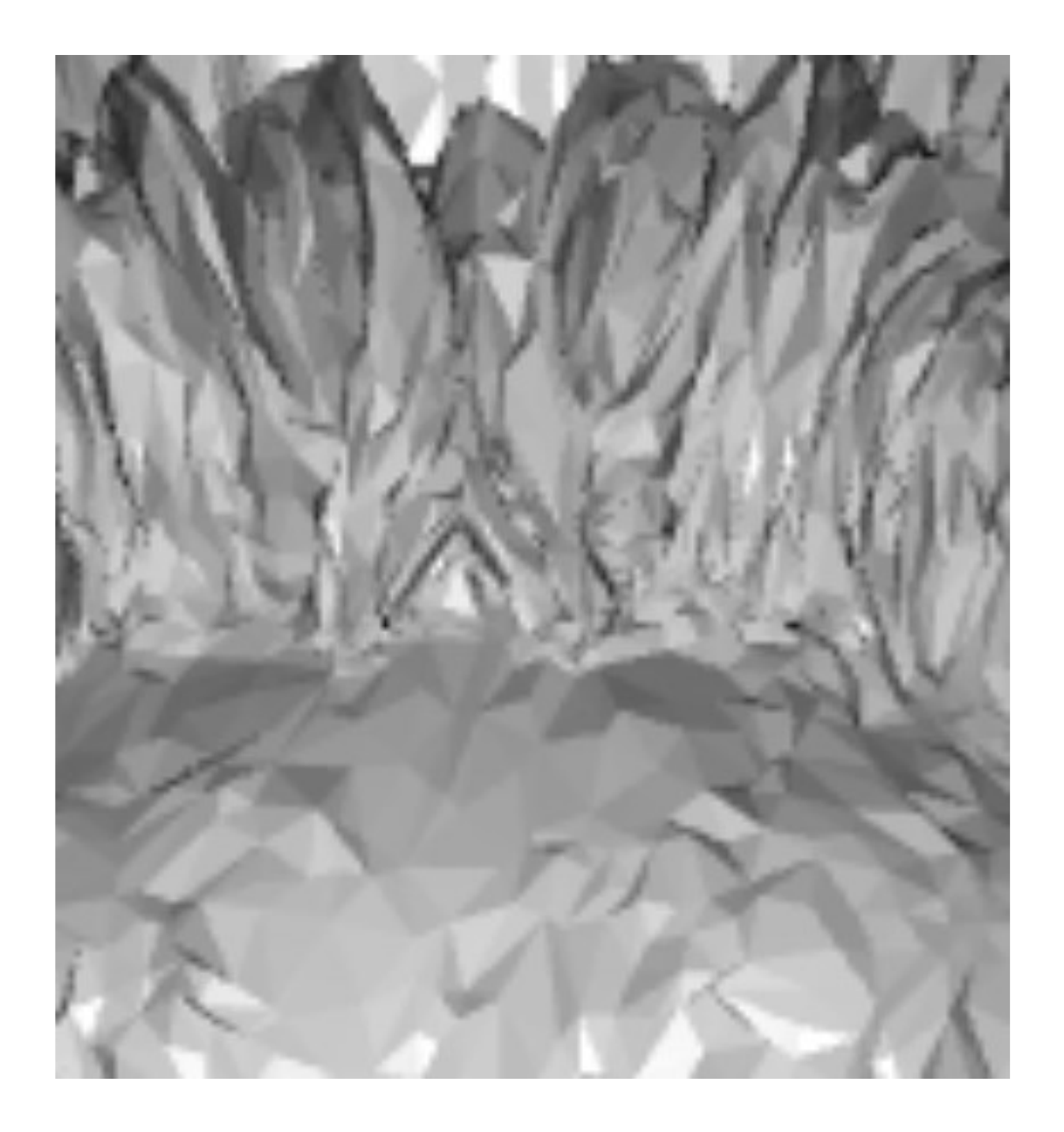}&
		\includegraphics[width=1.72cm]{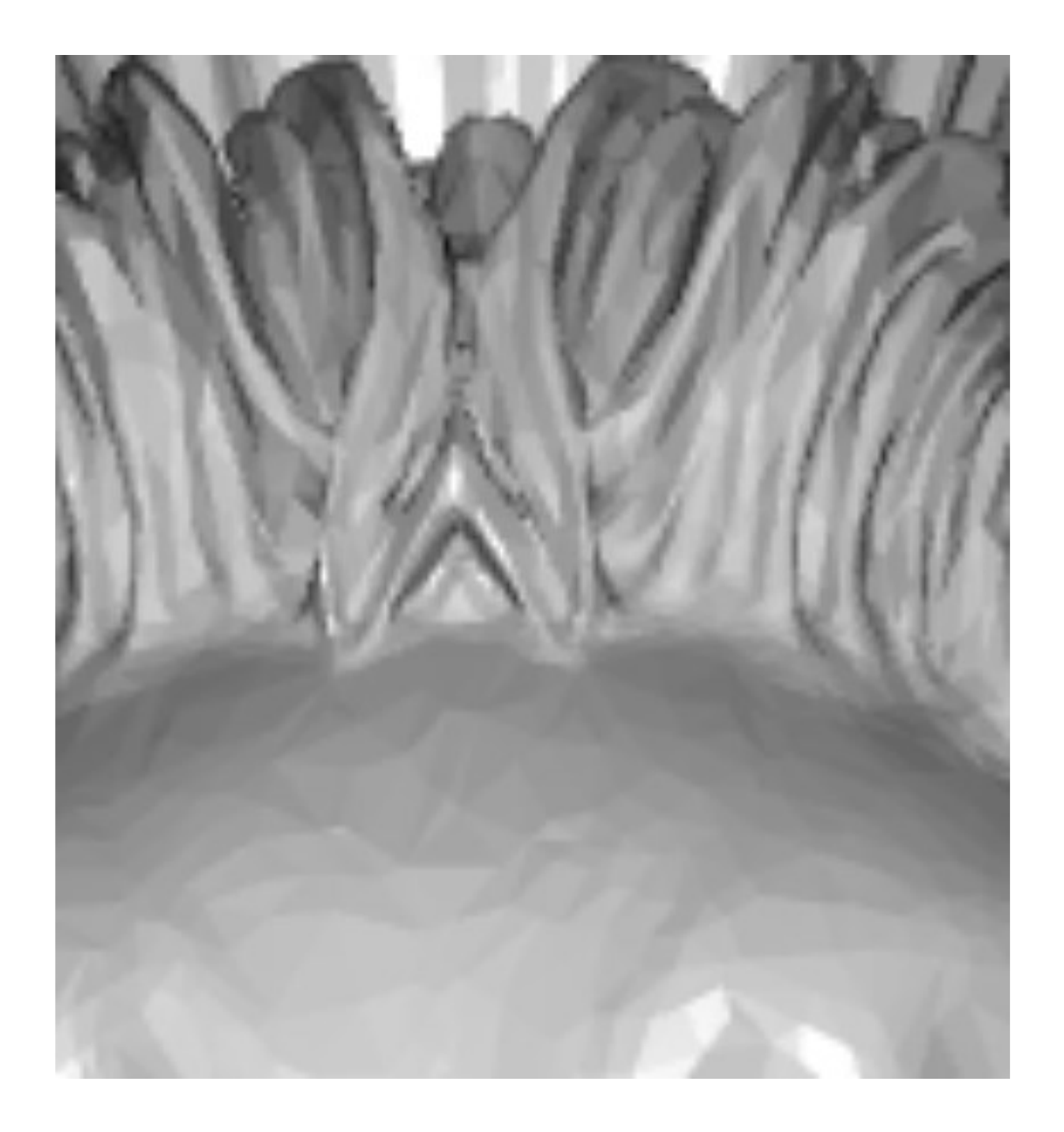}&
		\includegraphics[width=1.72cm]{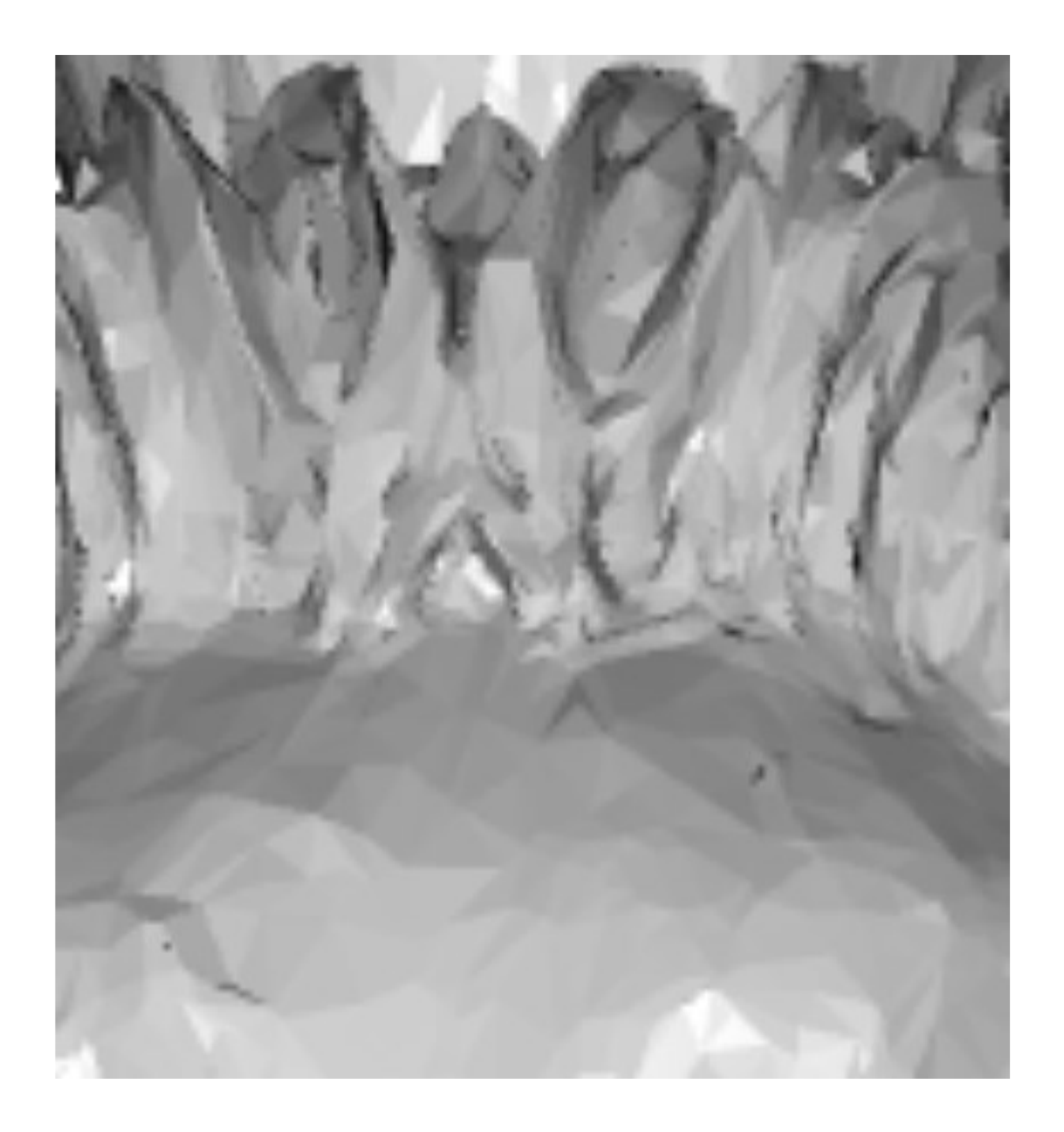}&
		\includegraphics[width=1.72cm]{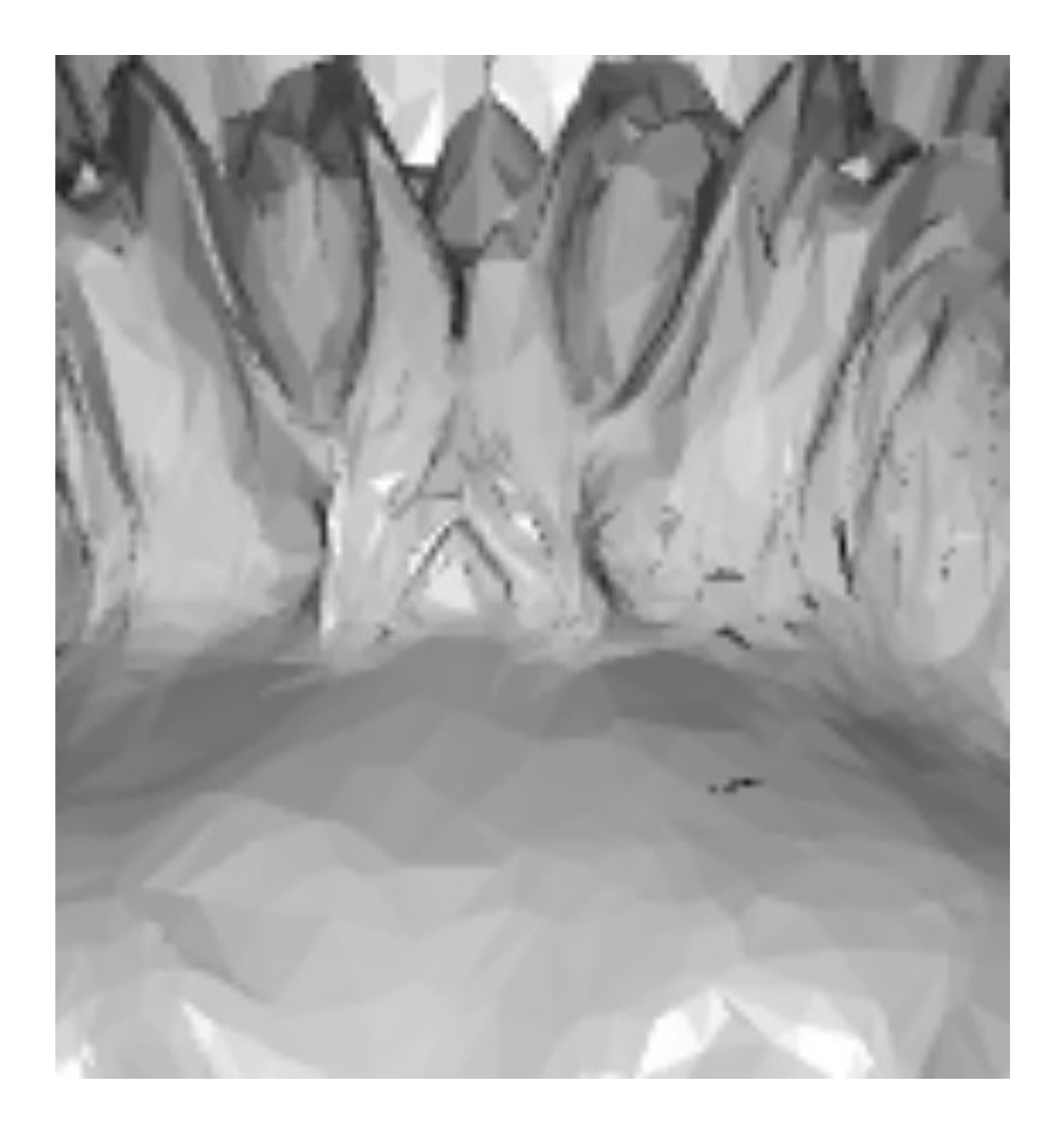}&
		\includegraphics[width=1.72cm]{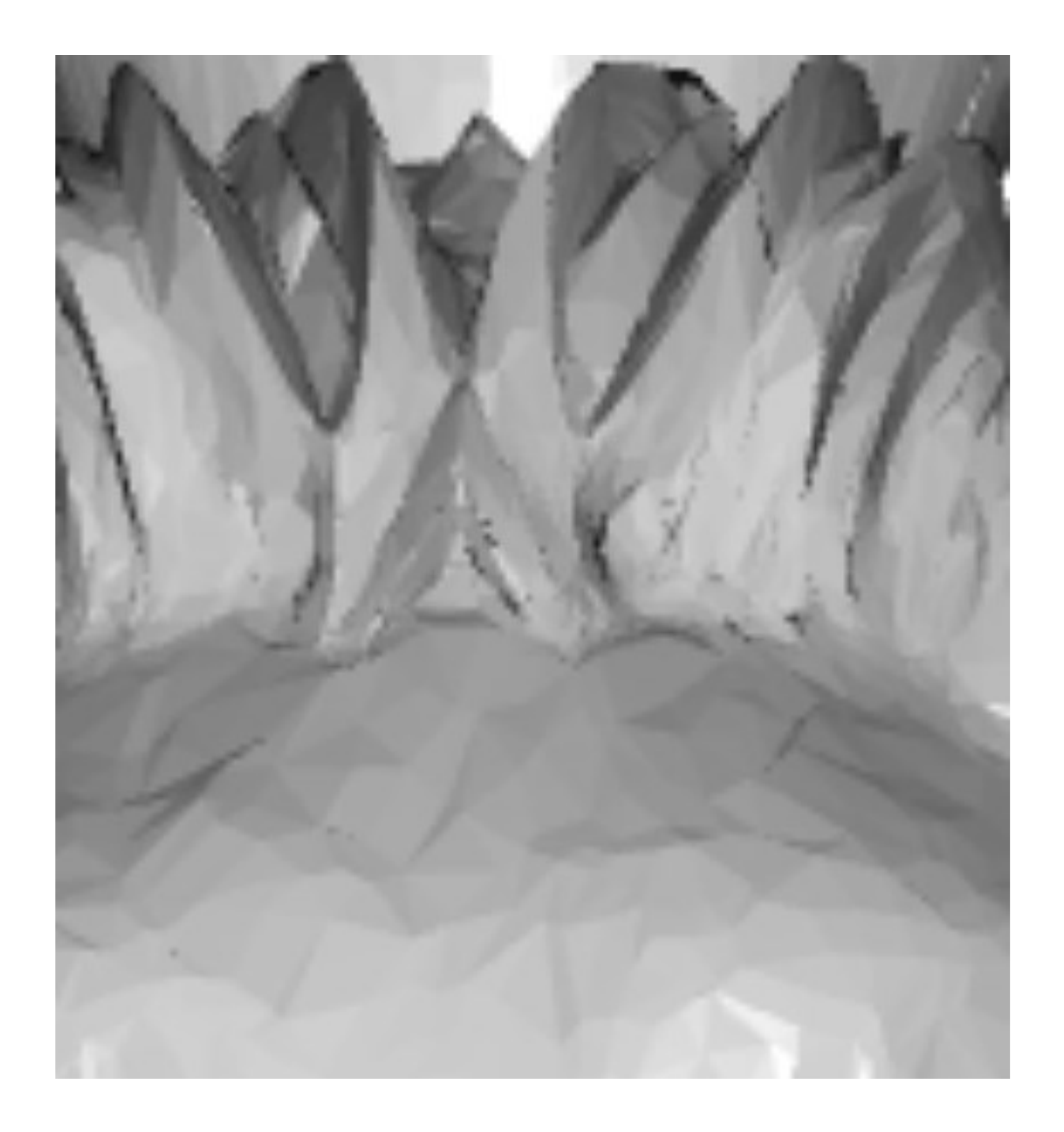}&
		\includegraphics[width=1.72cm]{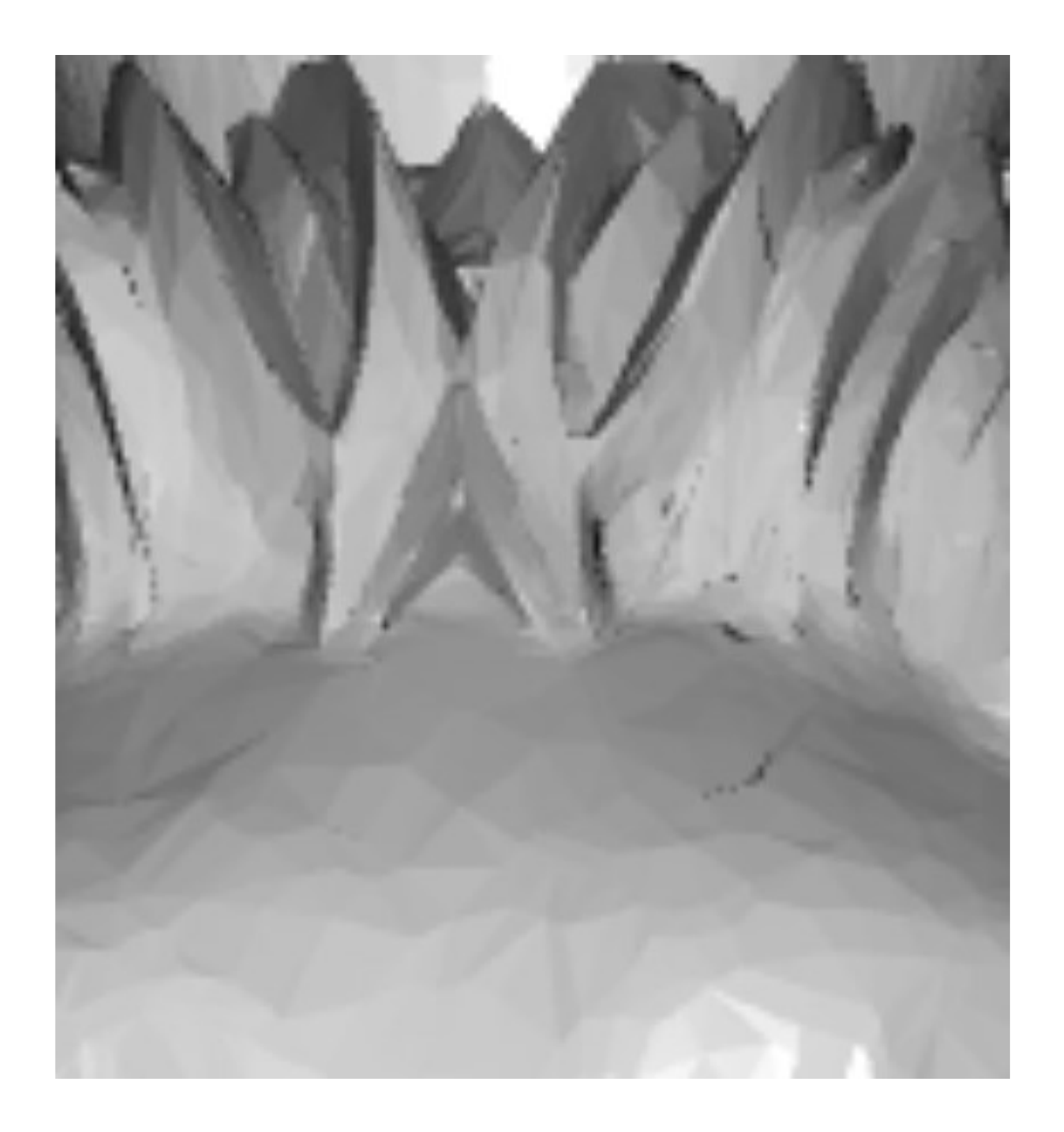}&
		\includegraphics[width=1.72cm]{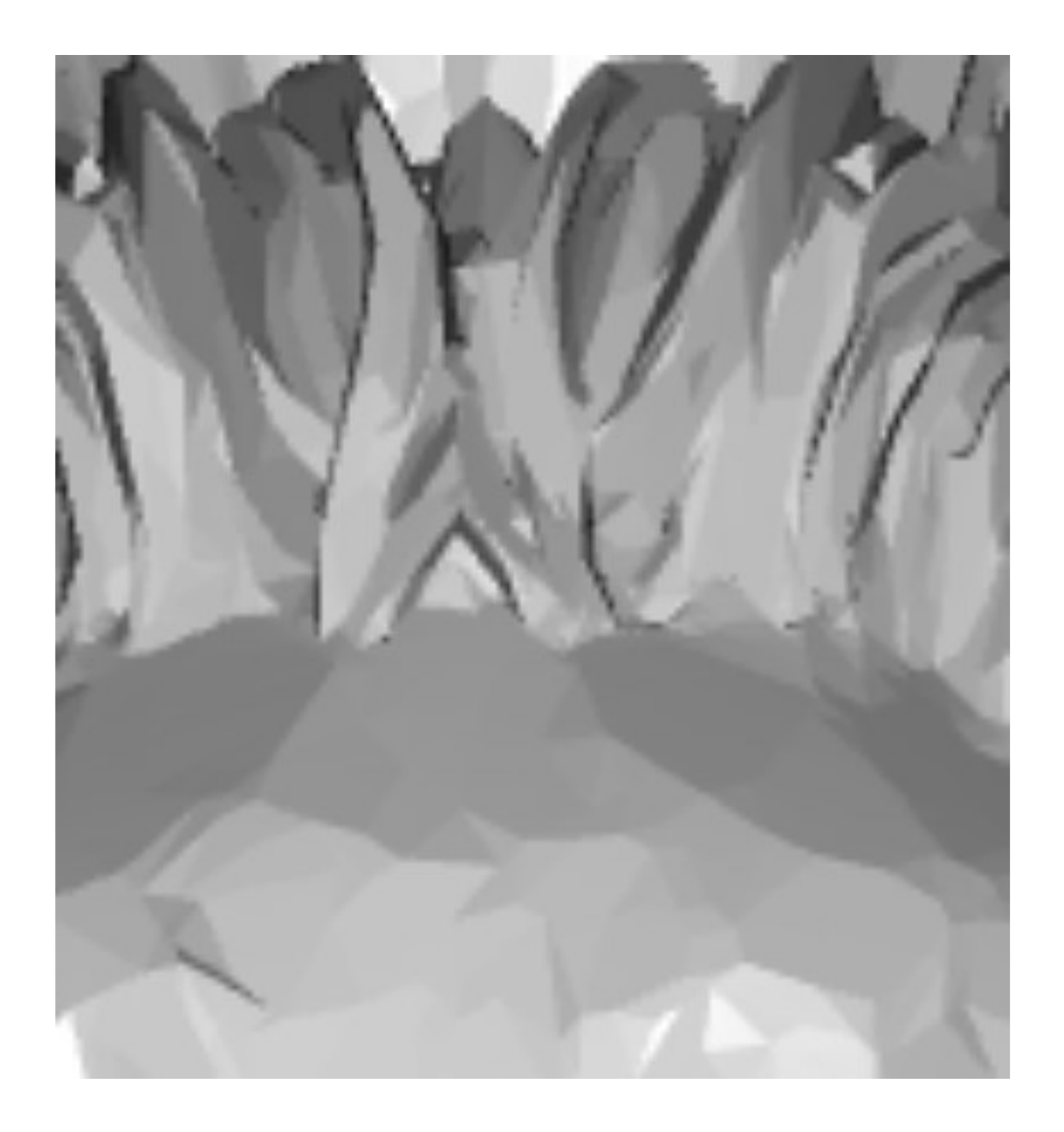}&
		\includegraphics[width=1.72cm]{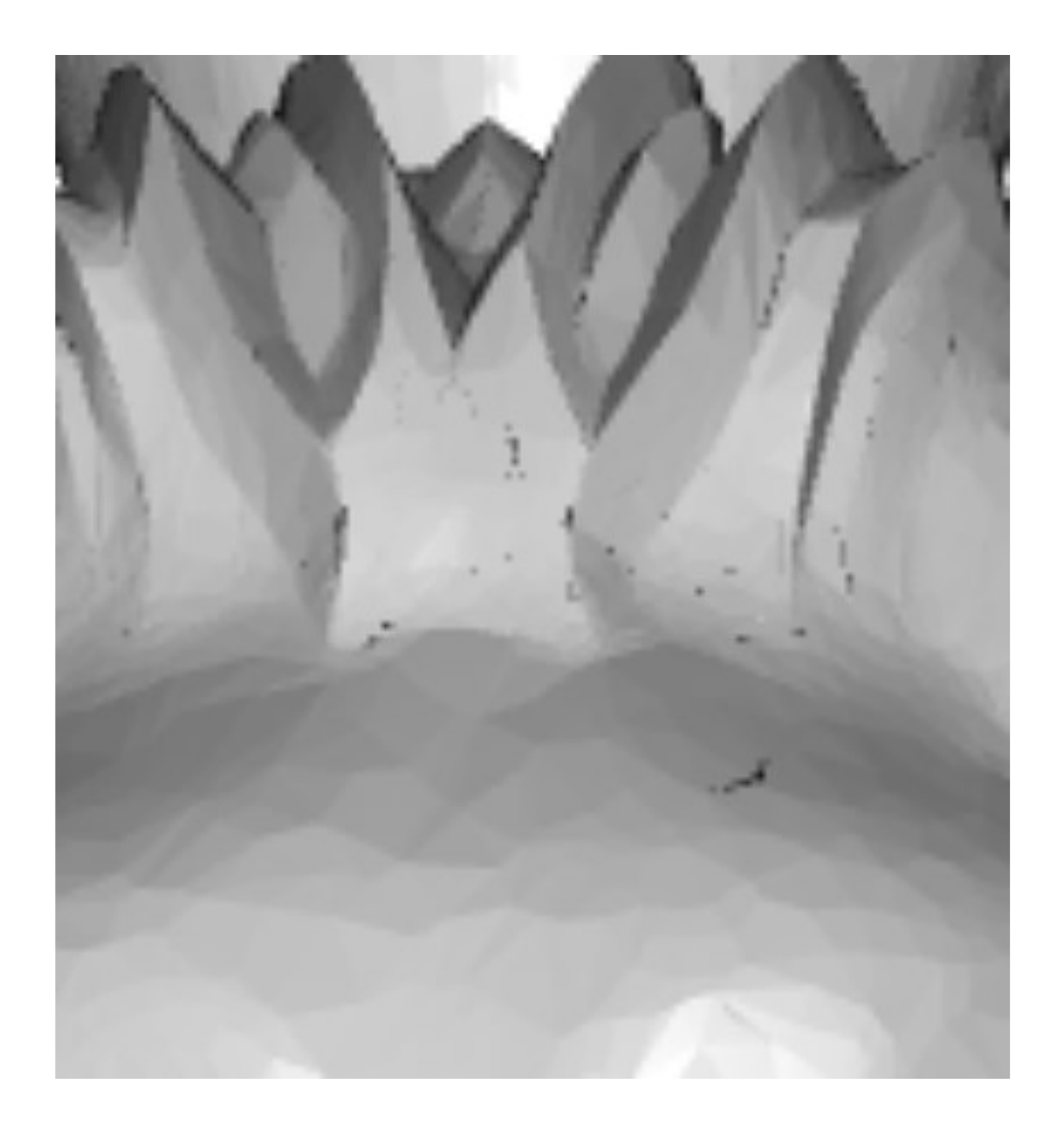}&
		\includegraphics[width=1.72cm]{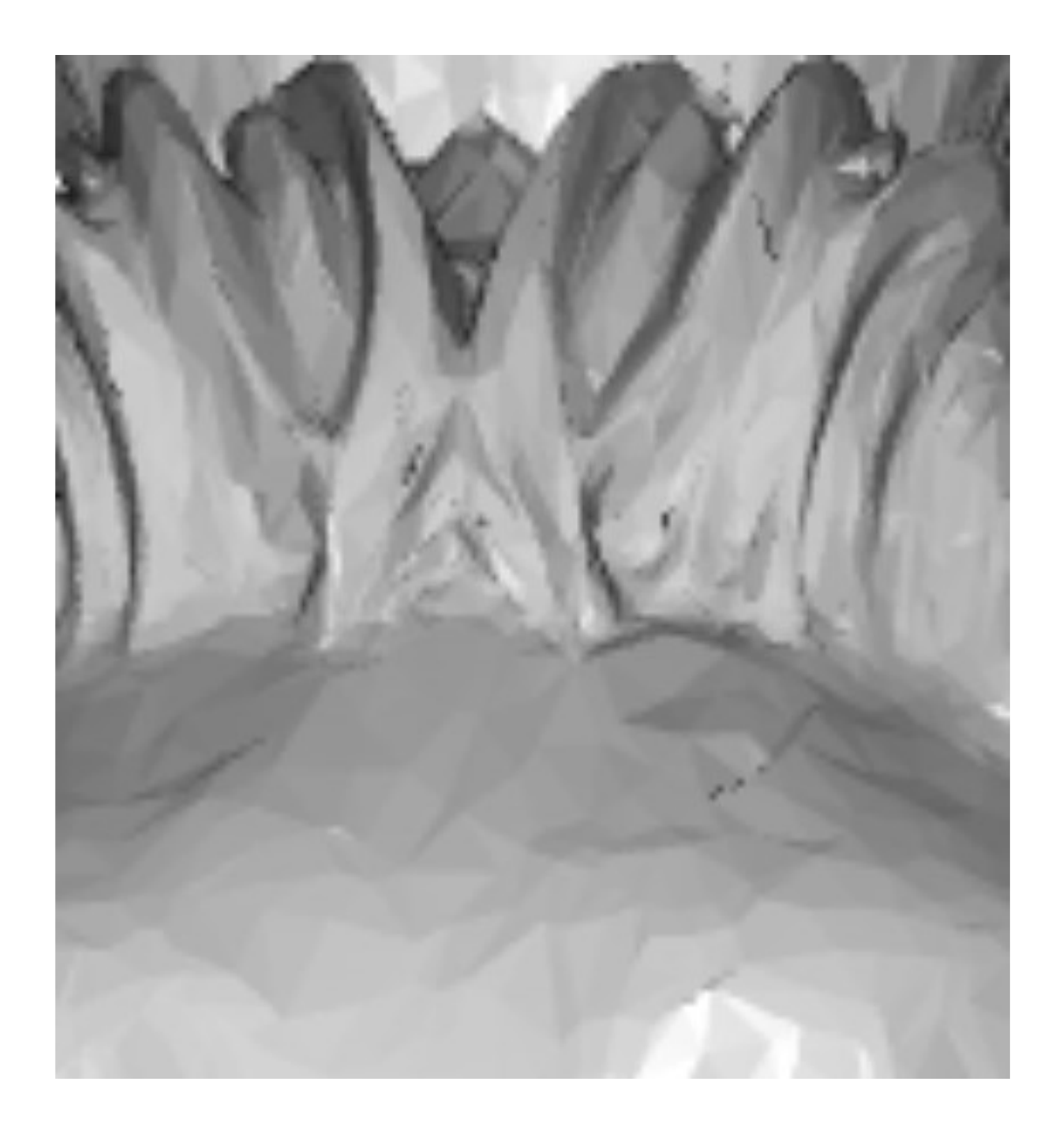}&
		\includegraphics[width=1.72cm]{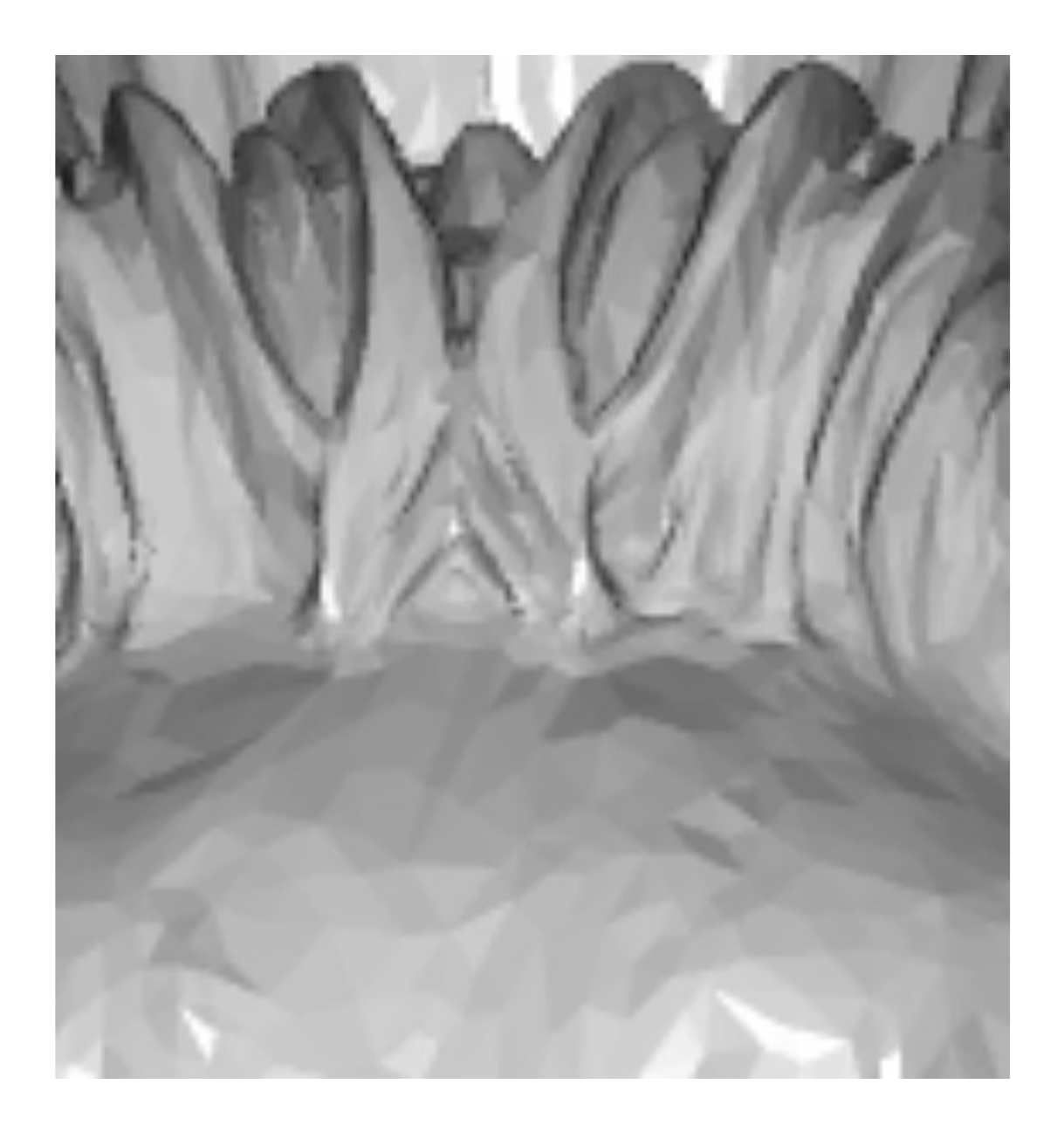}\\
		\includegraphics[width=1.72cm]{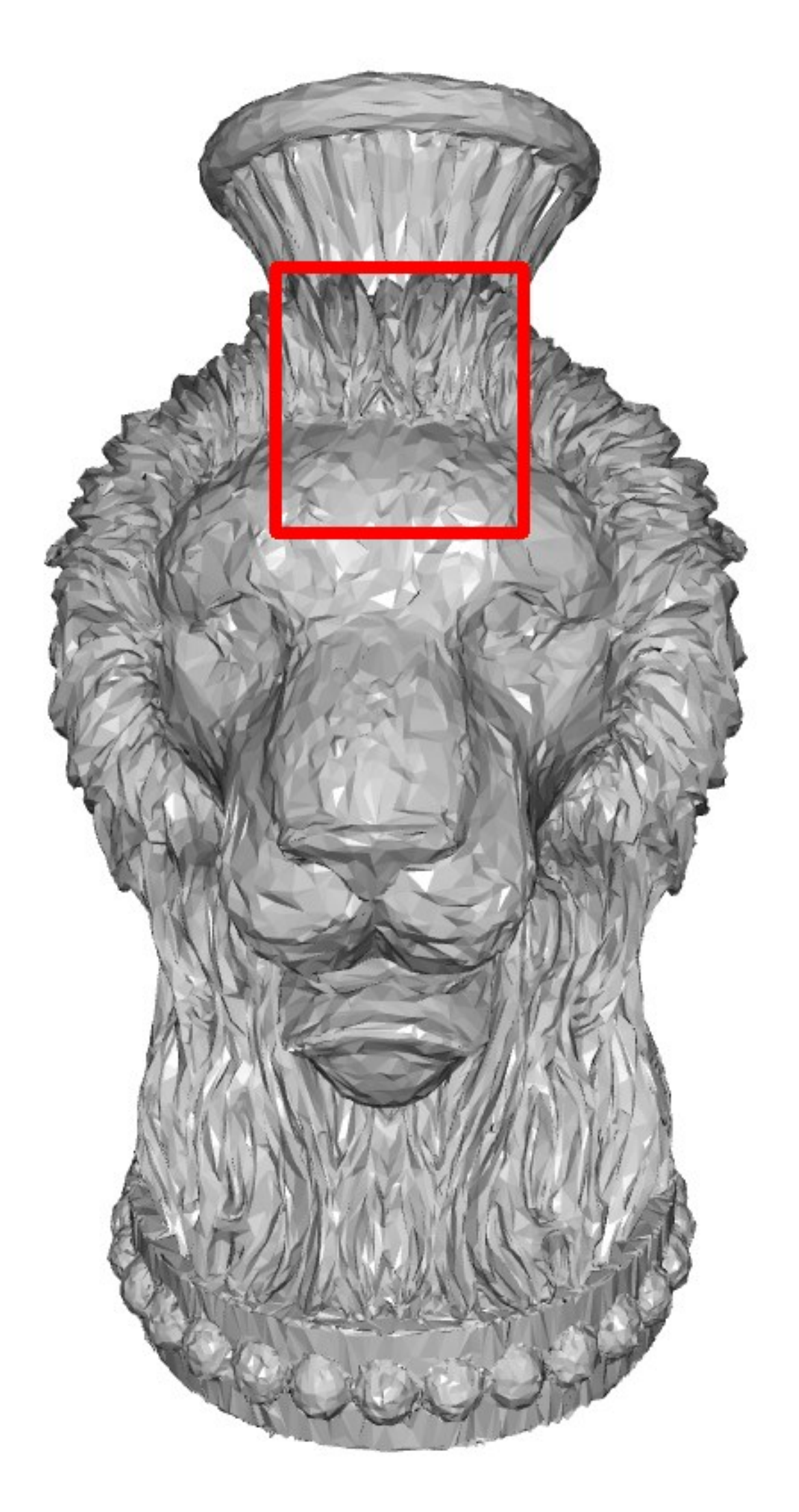}&
		\includegraphics[width=1.72cm]{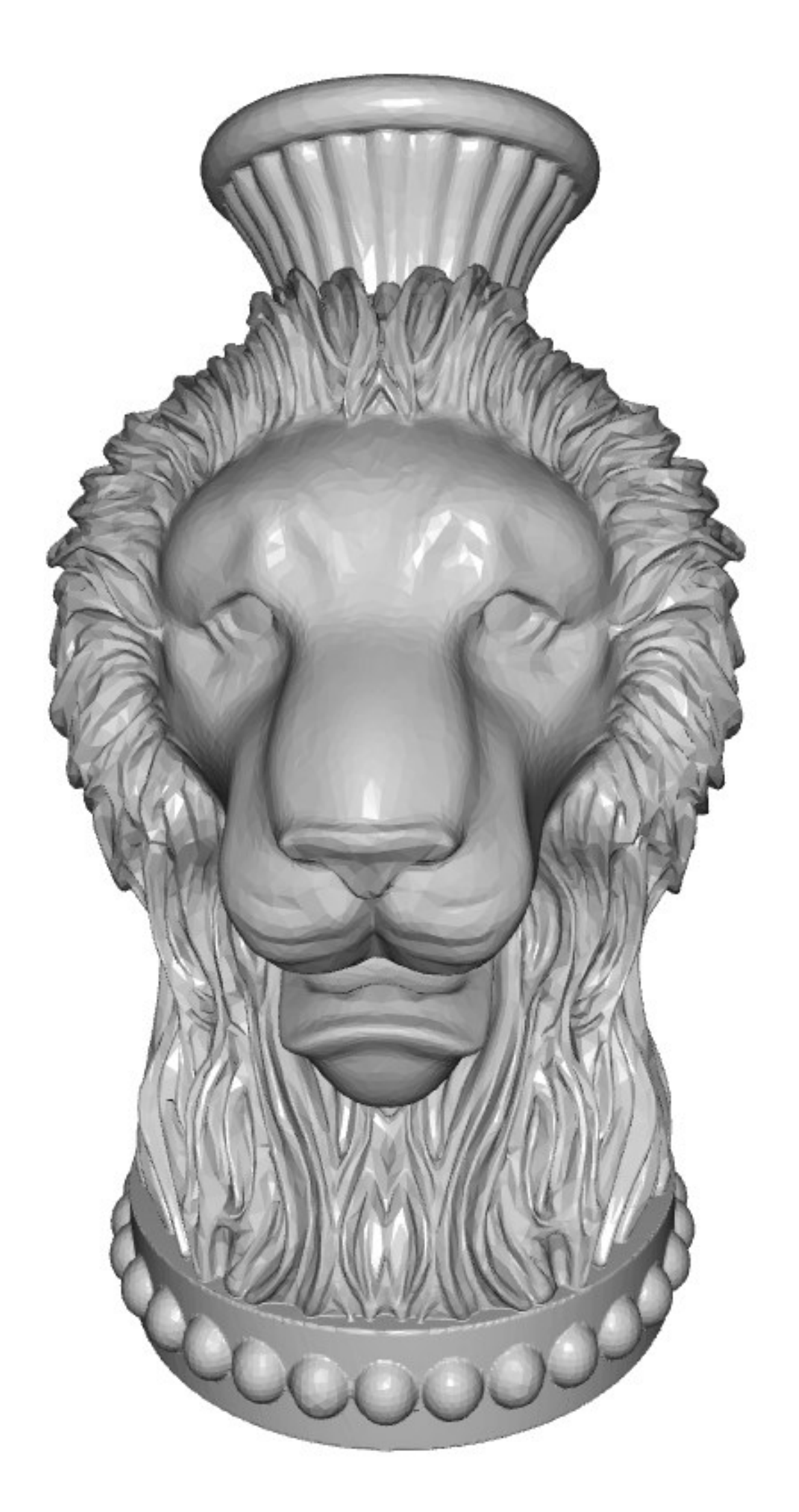}&
		\includegraphics[width=1.72cm]{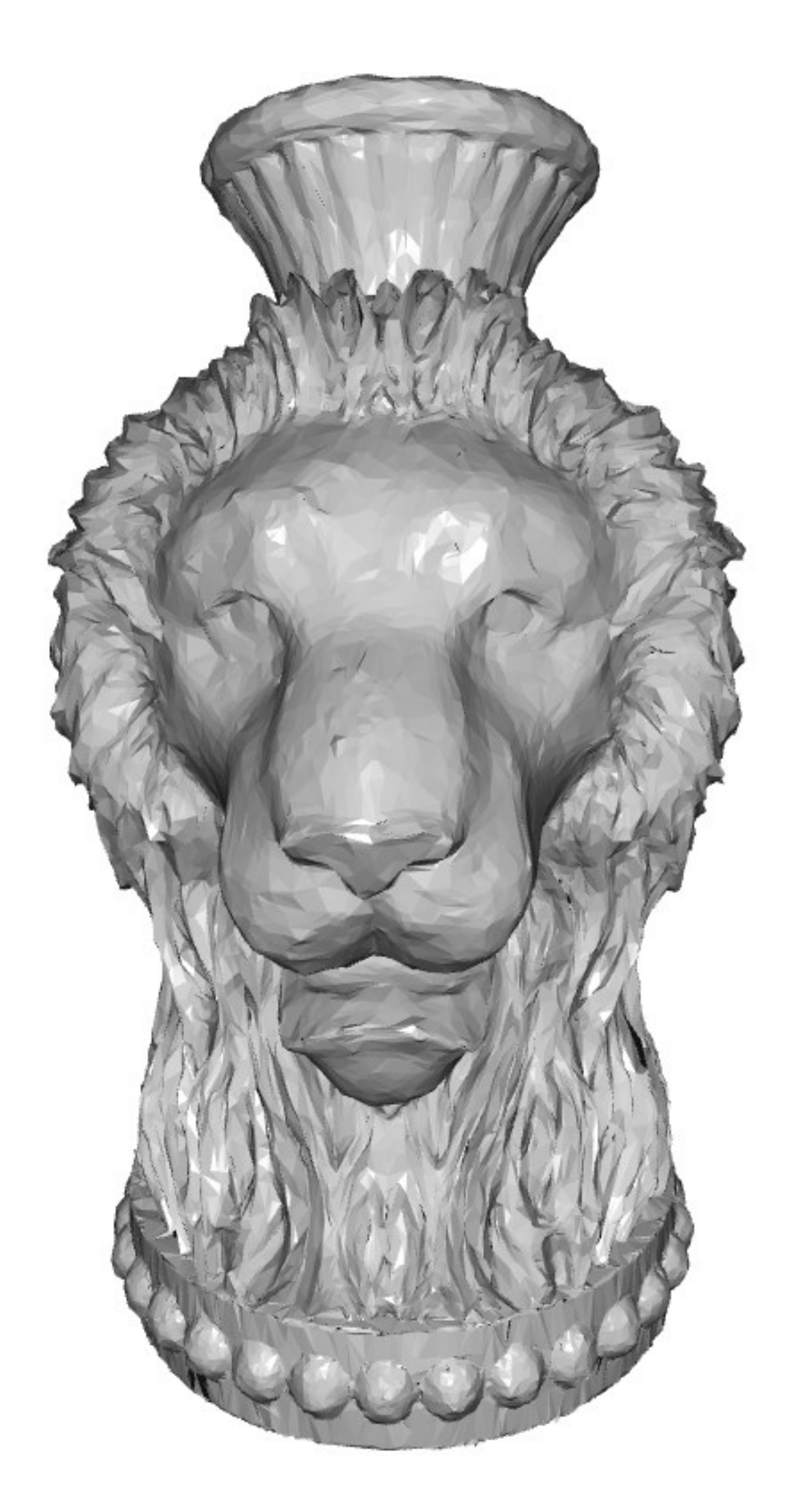}&
		\includegraphics[width=1.72cm]{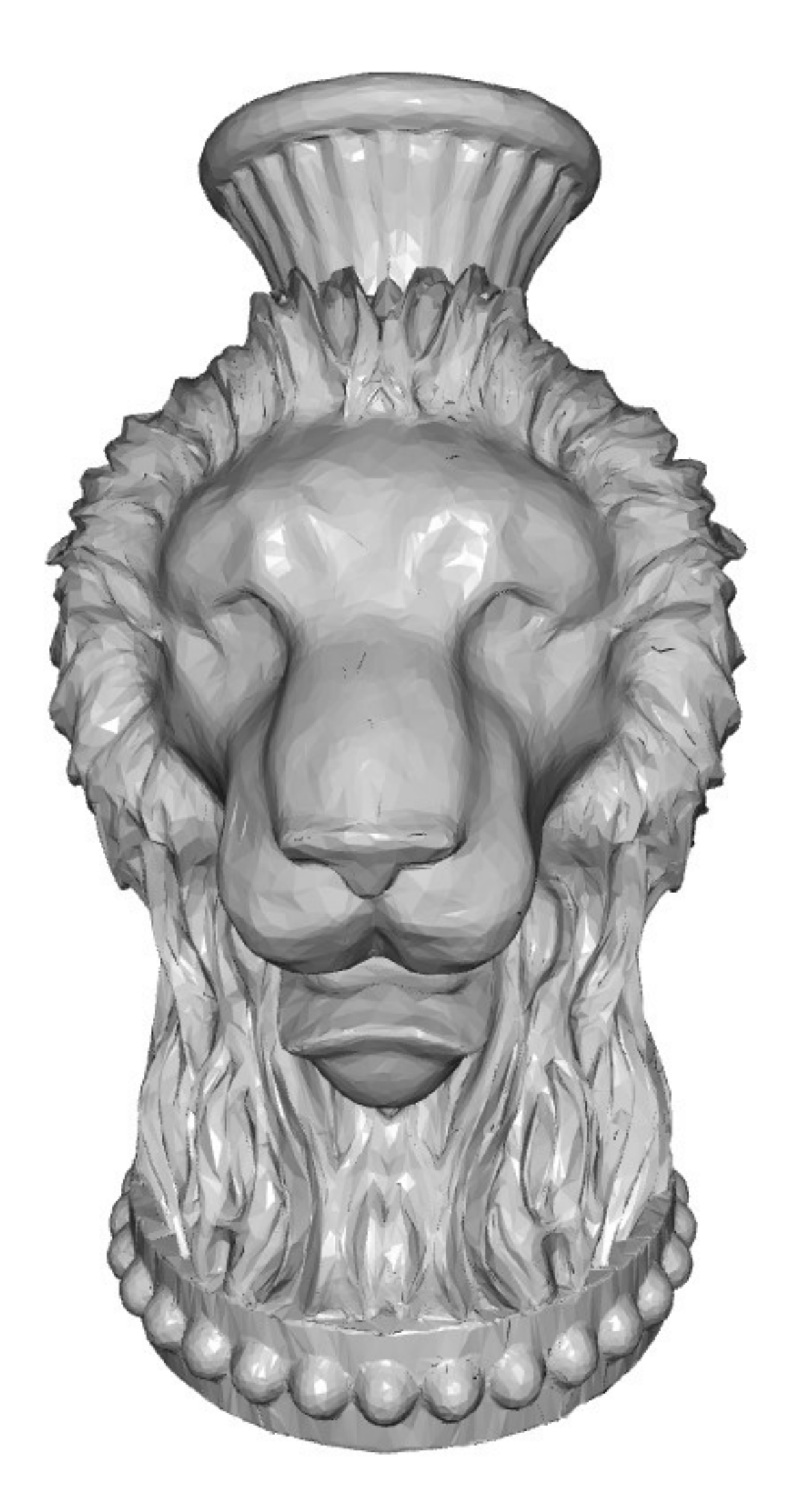}&
		\includegraphics[width=1.72cm]{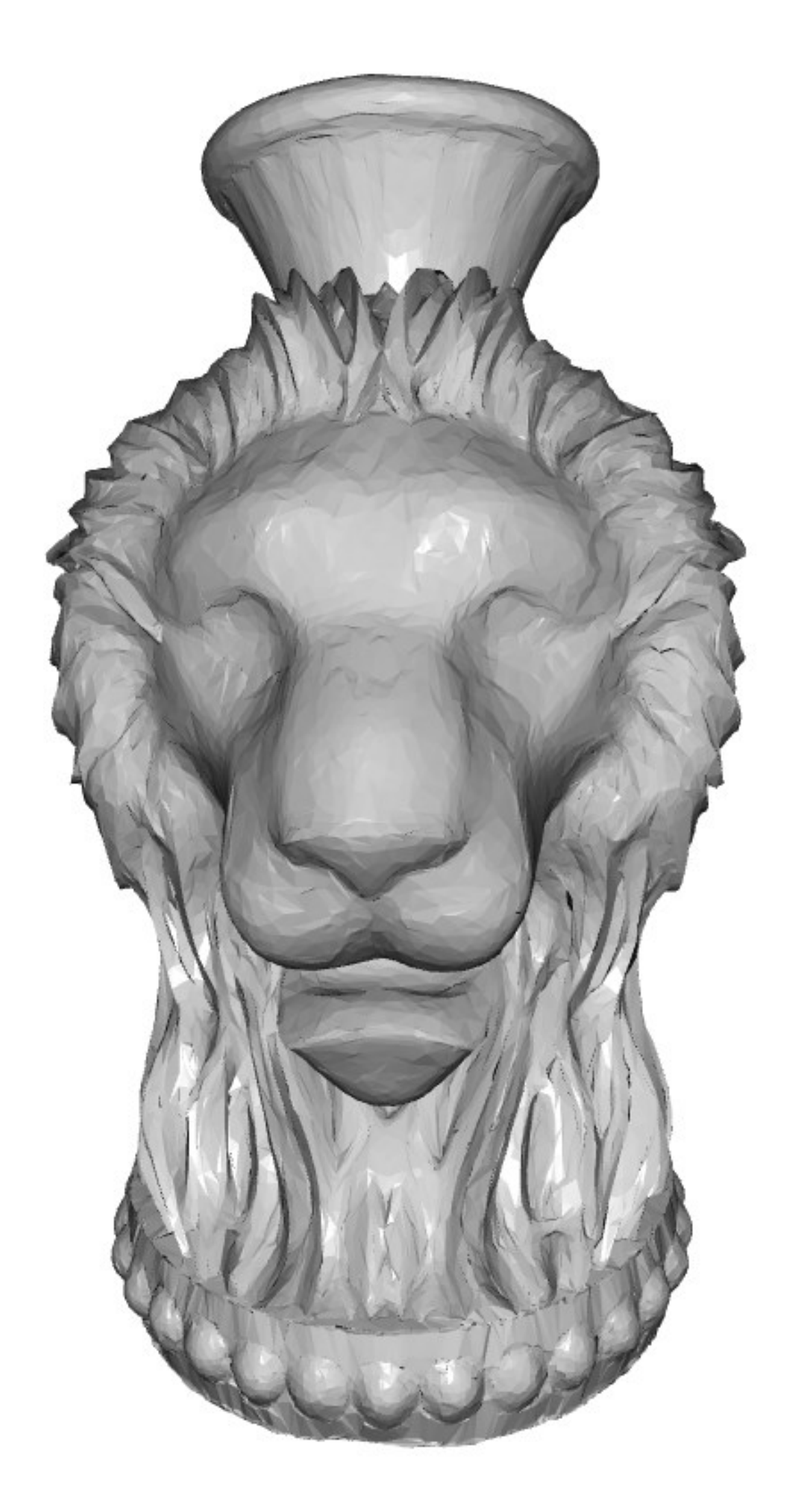}&
		\includegraphics[width=1.72cm]{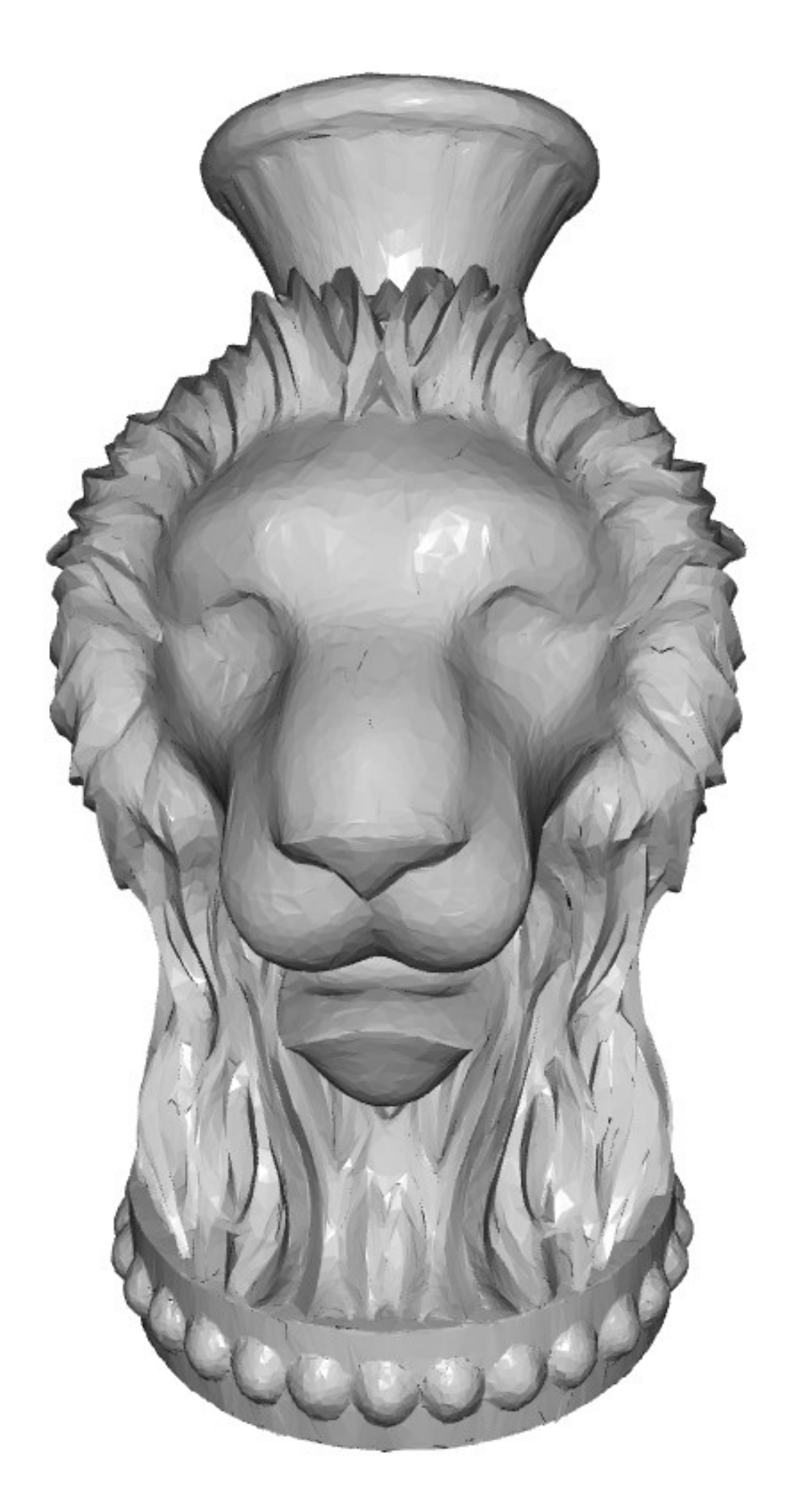}&
		\includegraphics[width=1.72cm]{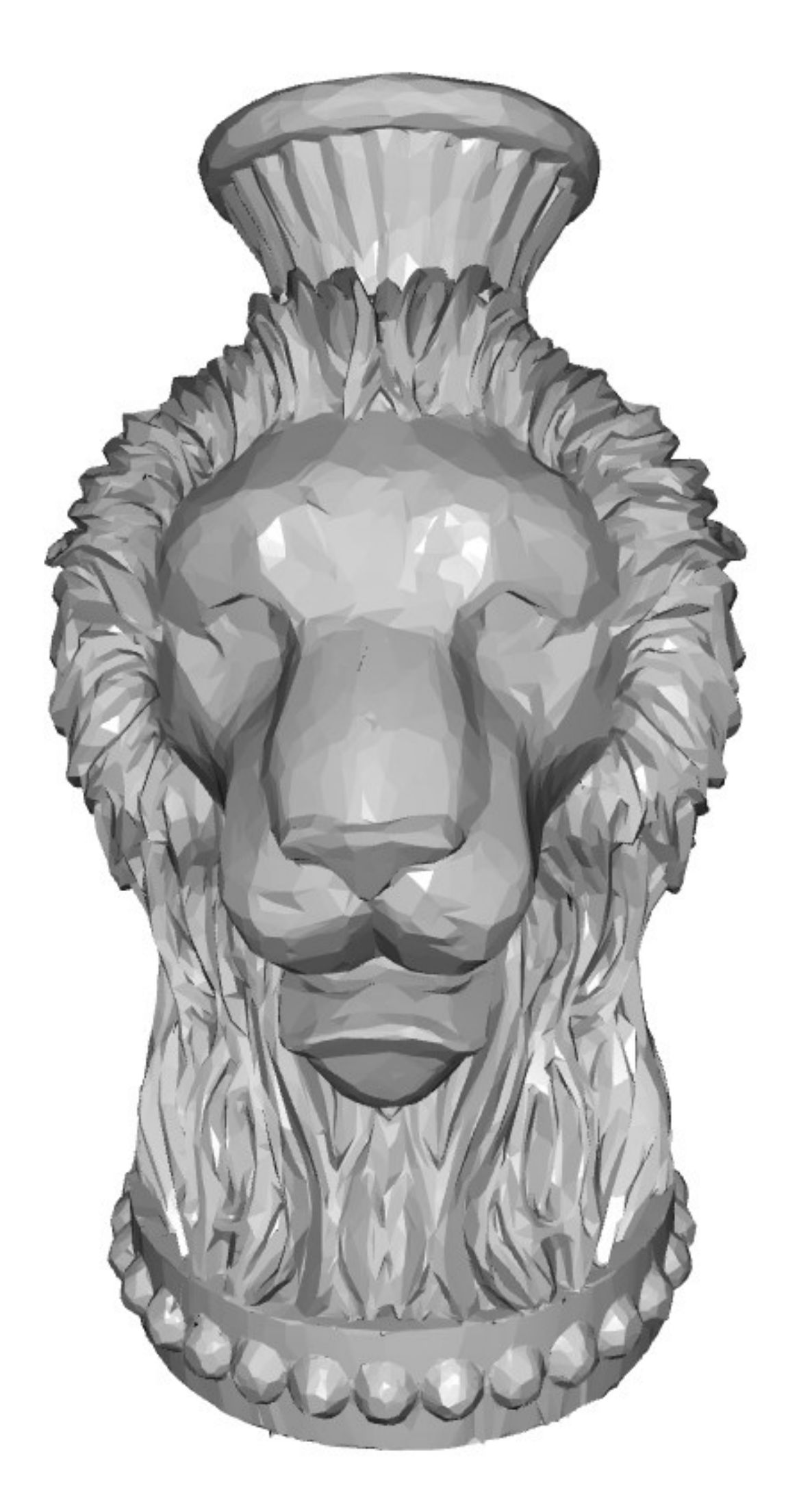}&
		\includegraphics[width=1.72cm]{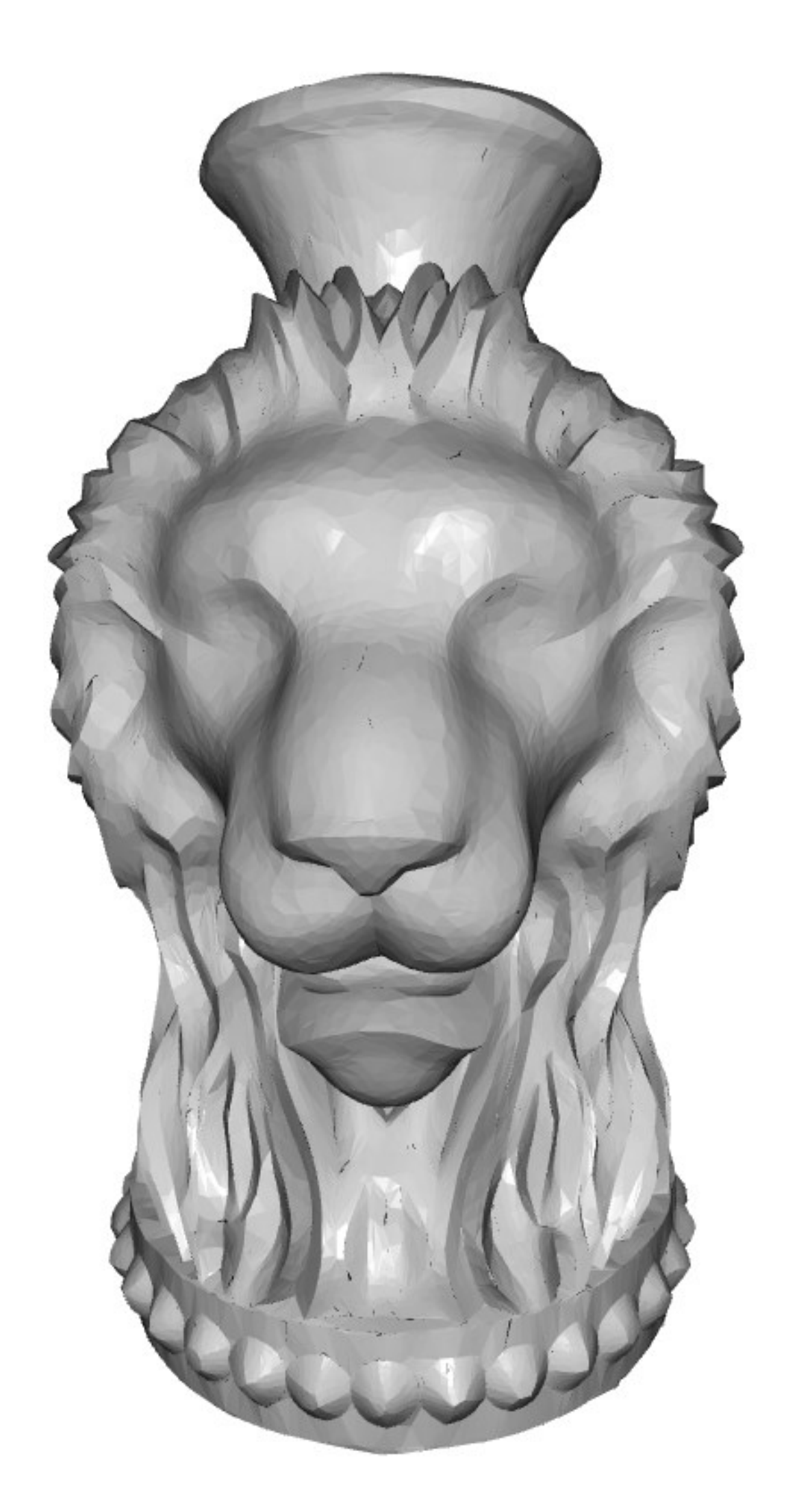}&
		\includegraphics[width=1.72cm]{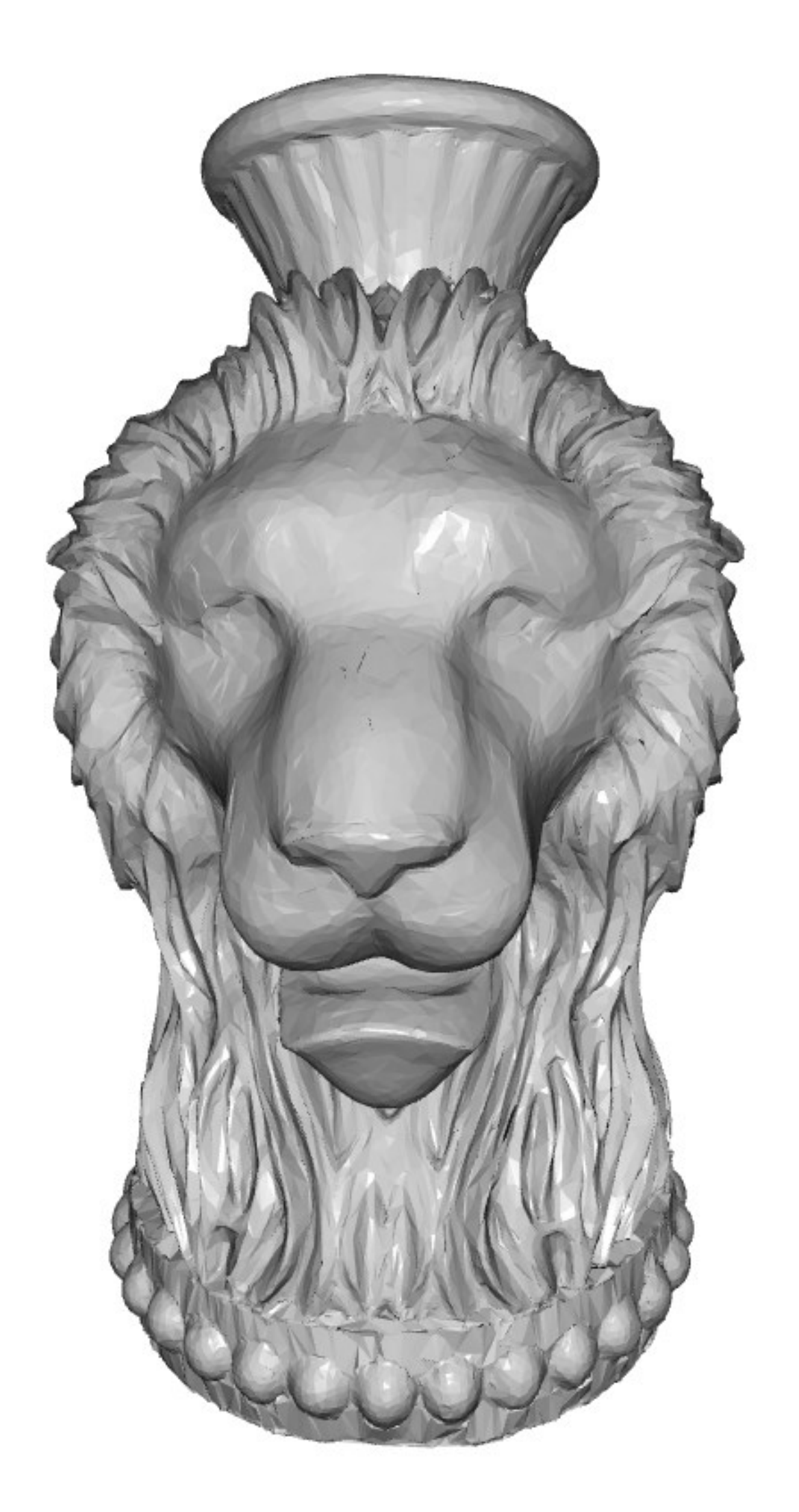}&
		\includegraphics[width=1.72cm]{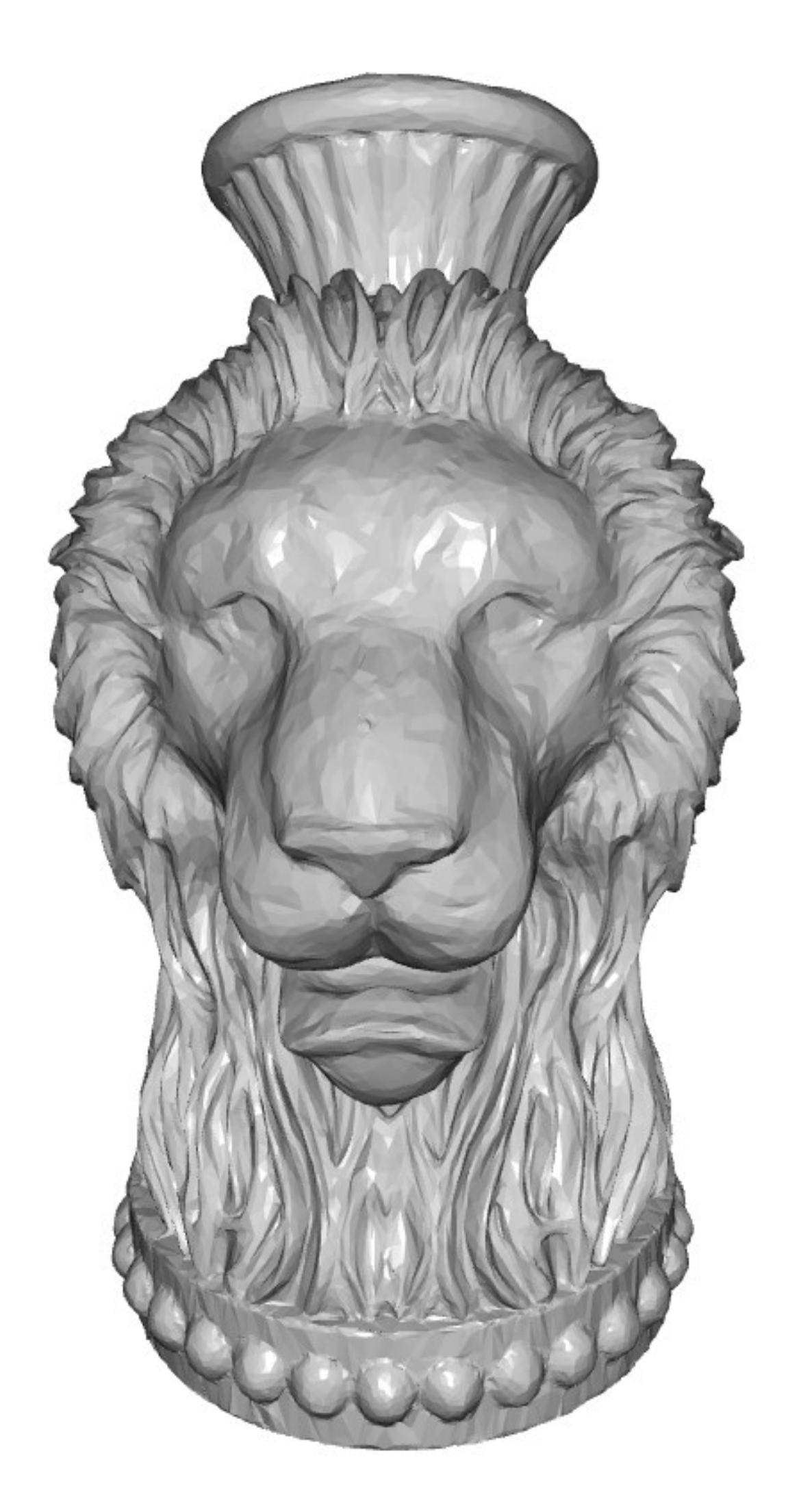}\\	
		\textrm{(a)~Noisy}&\textrm{(b)~Original}&\textrm{(c)~\cite{fleishman2003bilateral}}&\textrm{(d)~\cite{jones2003non}}&\textrm{(e)~\cite{sun2007fast}}&	\textrm{(f)~\cite{zheng2011bilateral}}&\textrm{(g)~\cite{he2013mesh}}&\textrm{(h)~\cite{zhang2015guided}}&\textrm{(i)~\cite{li2018non}}&\textrm{(j)~Ours}\\
		\end{array}
		$}
	\caption{ \label{fig:com_denosing and_features} Comparison between our algorithm and the selected state-of-the-art algorithms on the armadillo $(\sigma_{n}=0.3{e}_l)$, bunny $(\sigma_{n}=0.5{e}_l)$, Max Planck $(\sigma_{n}=0.1{e}_l)$, vaselion $(\sigma_{n}=0.2{e}_l)$.}
\end{figure*}

\begin{figure*}[htbp]
	\centering
	\setlength{\arraycolsep}{0pt}
	\makebox[\textwidth][c]{$
		\begin{array}{c}	
		\includegraphics[width=8cm]{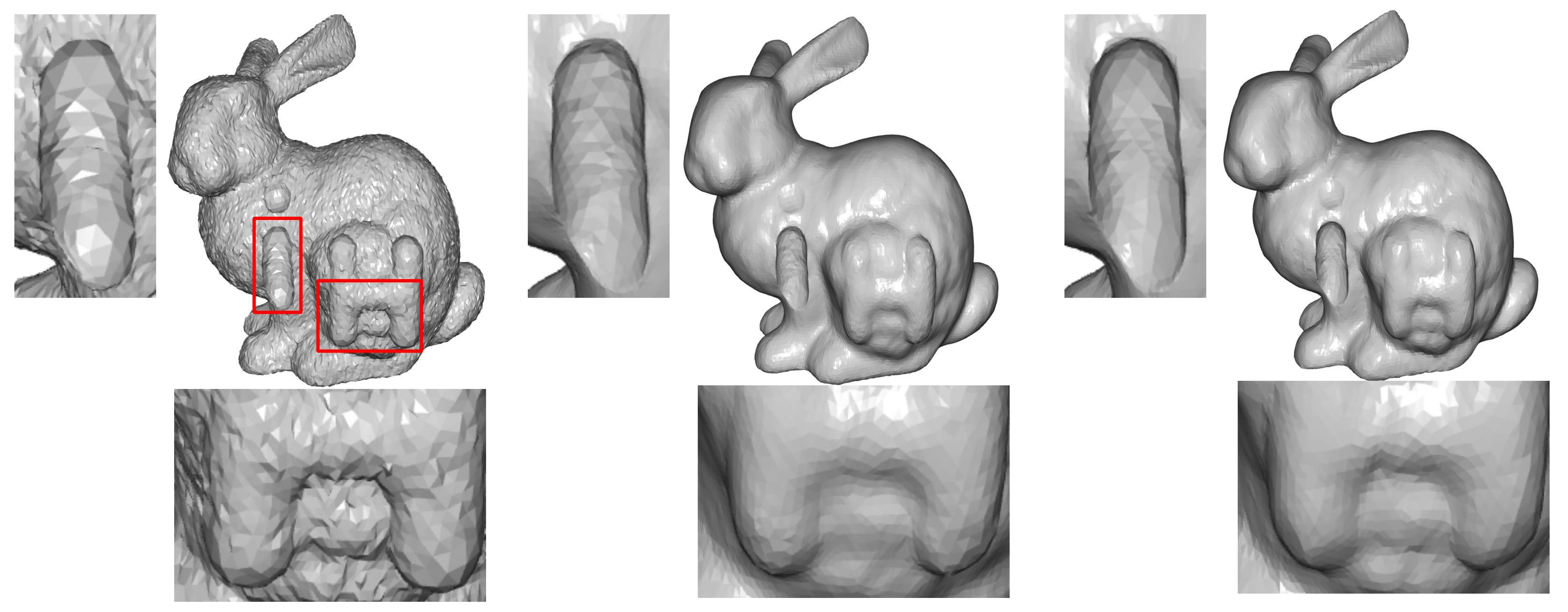}\\
		\begin{array}{ccc}
		\makebox[3cm]{\small Input}&
		\makebox[3cm]{\small ~\cite{pan2020hlo}}&
		\makebox[3cm]{\small Ours}\\
		\end{array}\\
		\includegraphics[width=8cm]{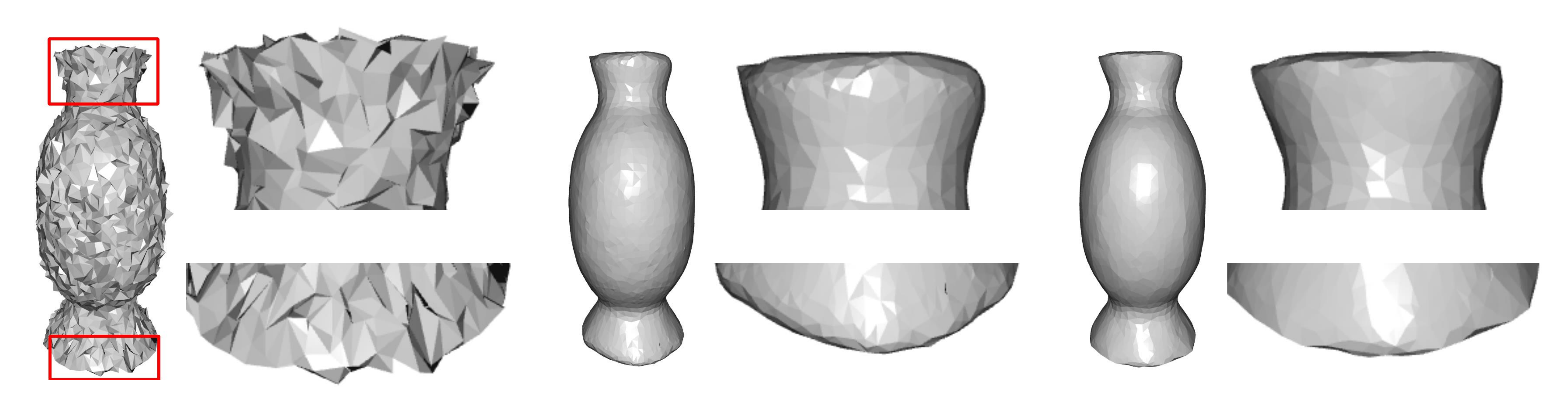}\\
		\begin{array}{ccc}
		\makebox[3cm]{\small Input}&
		\makebox[3cm]{\small ~\cite{pan2020hlo}}&
		\makebox[3cm]{\small Ours}\\
		\end{array}\\
		\end{array}
		$}
	\caption{\label{Fig.compare hlo} Comparison between our algorithm (Number of iterations: iHbunny: 40, Vase: 100) and our previous work~\cite{pan2020hlo}(Number of iterations: iHbunny: 10, Vase: 20) on noisy iHbunny $(\sigma_{n}=0.2{e}_l)$ and noisy Vase $(\sigma_{n}=0.5{e}_l)$.}
\end{figure*}

\begin{figure*}[htbp]
	\centering
	\setlength{\arraycolsep}{0pt}
	\makebox[\textwidth][c]{$
		\begin{array}{ccccccccc}	
		\includegraphics[width=1.72cm]{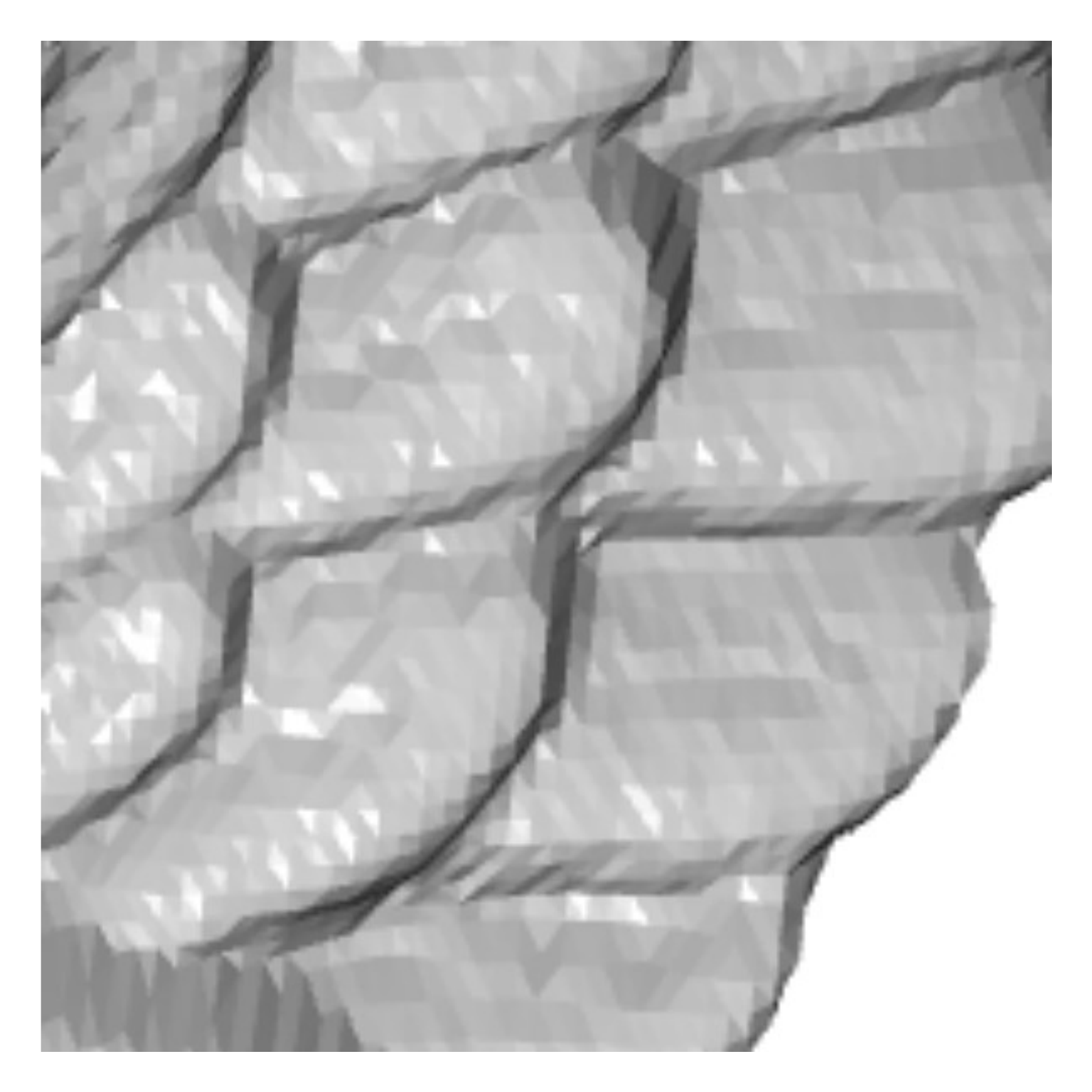}&
		\includegraphics[width=1.72cm]{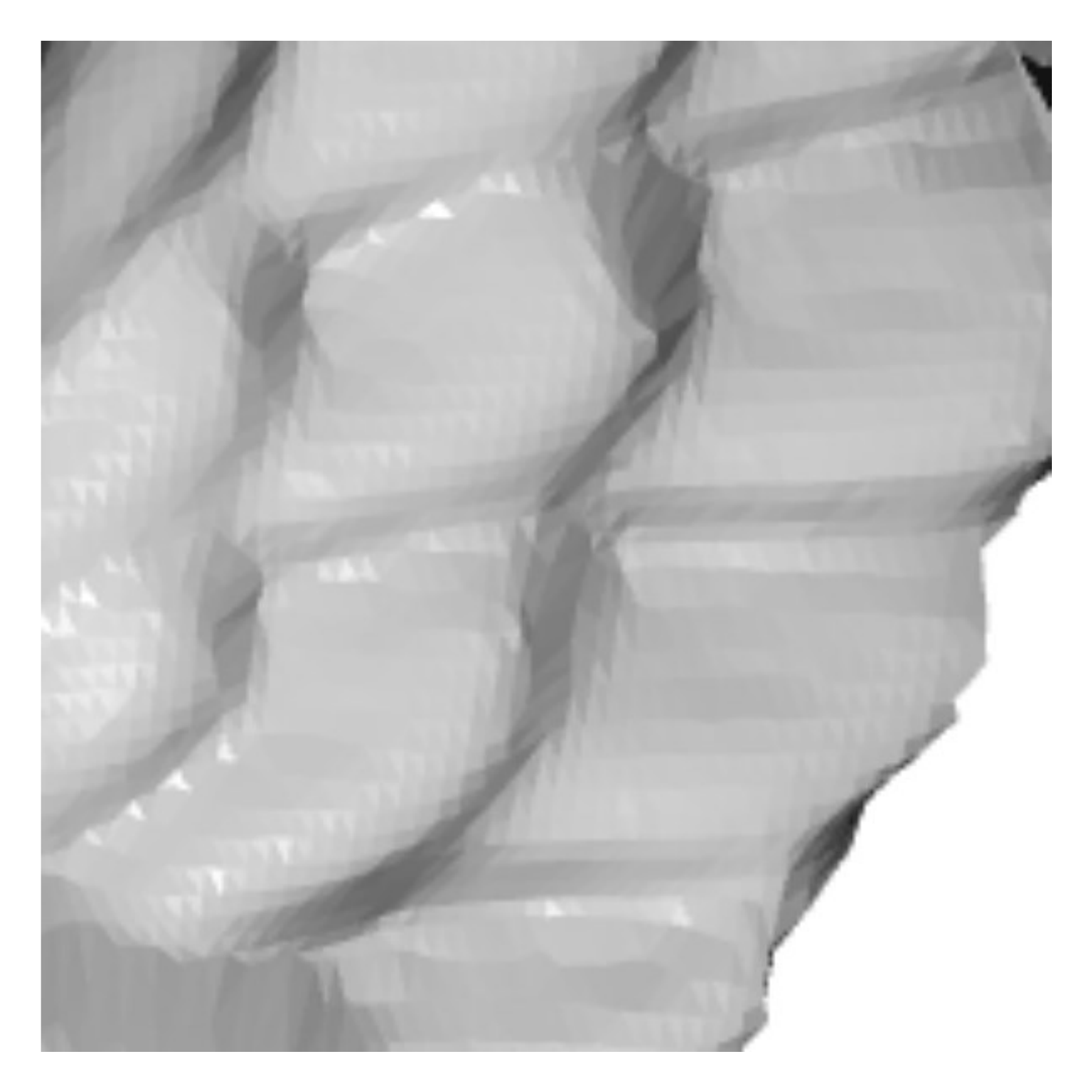}&
		\includegraphics[width=1.72cm]{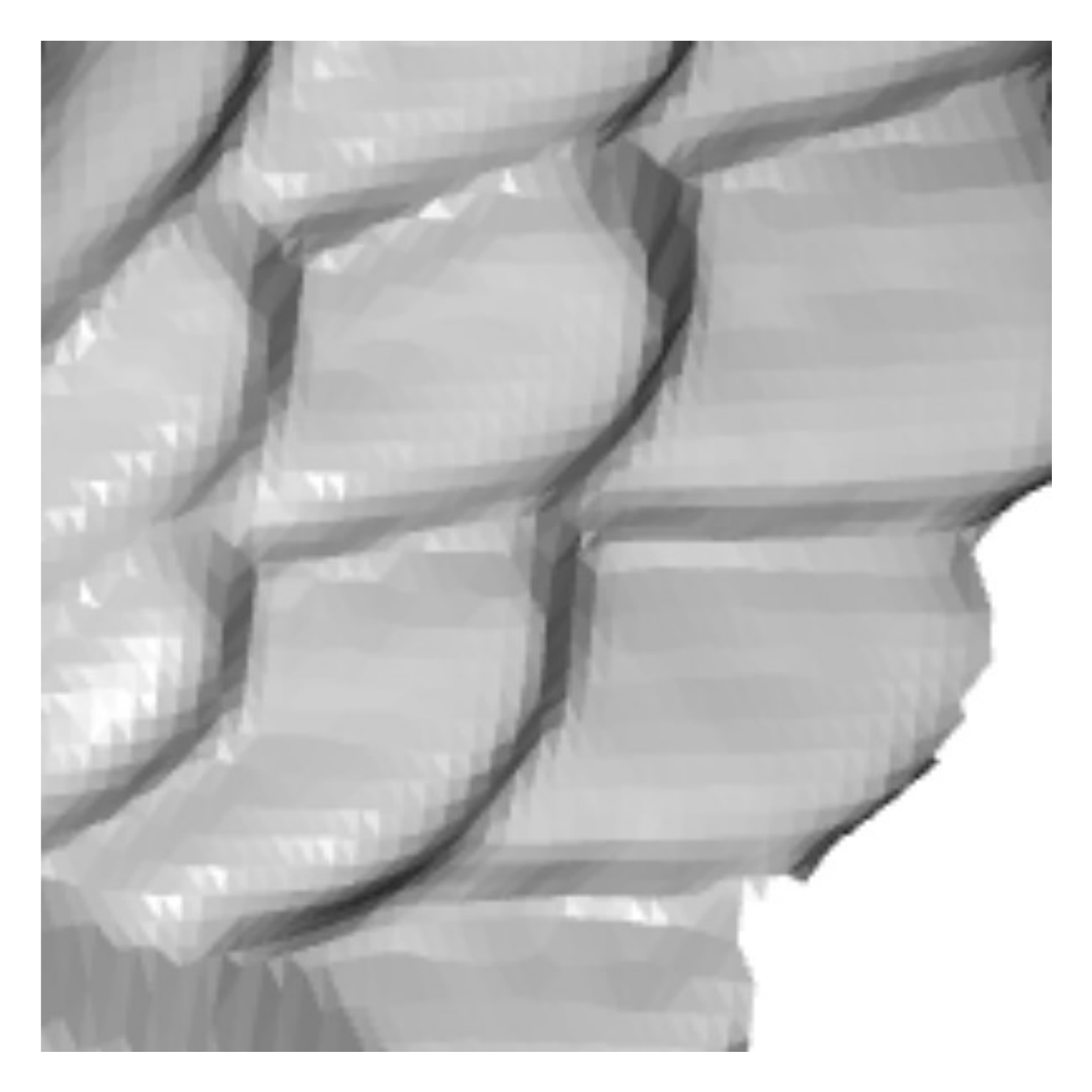}&
		\includegraphics[width=1.72cm]{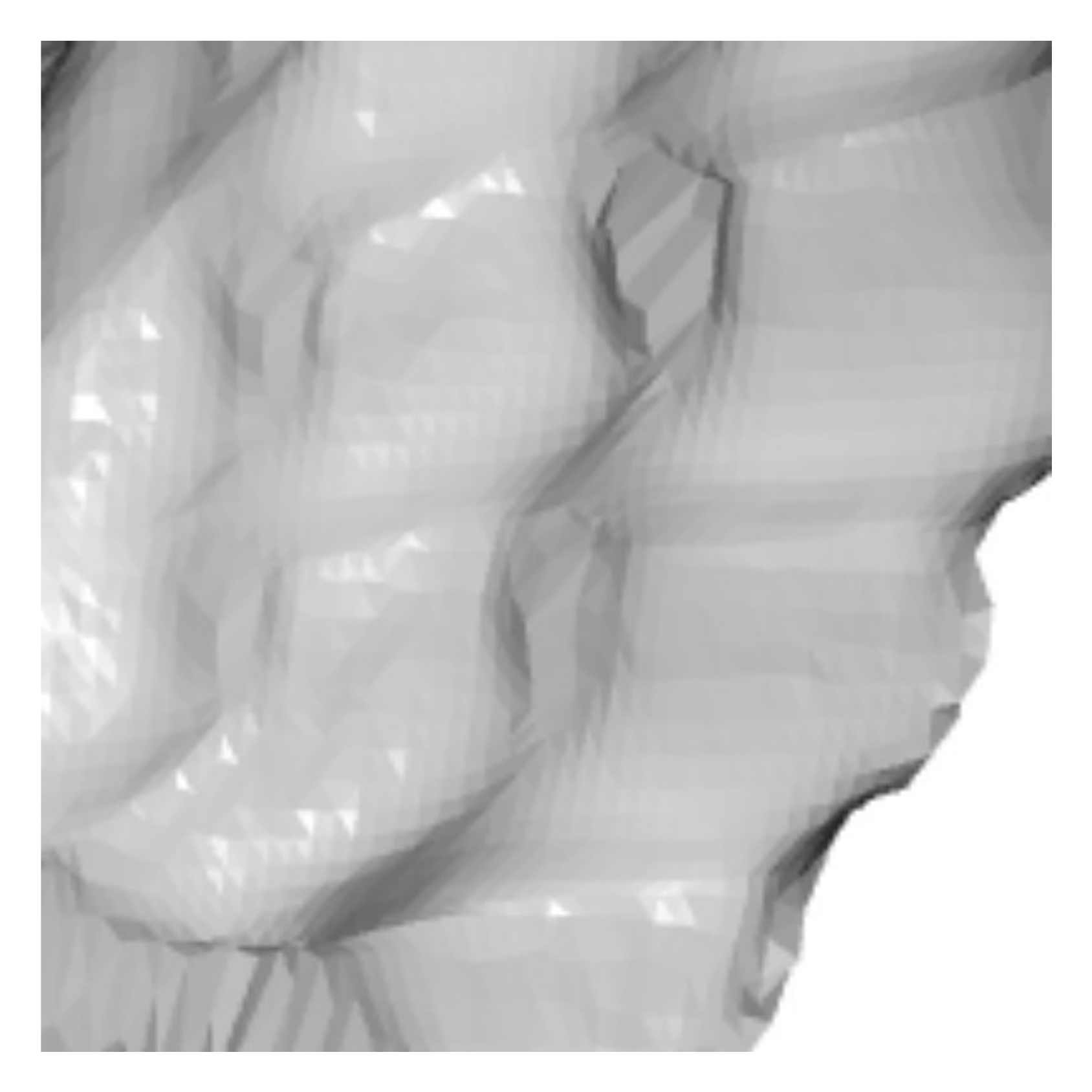}&
		\includegraphics[width=1.72cm]{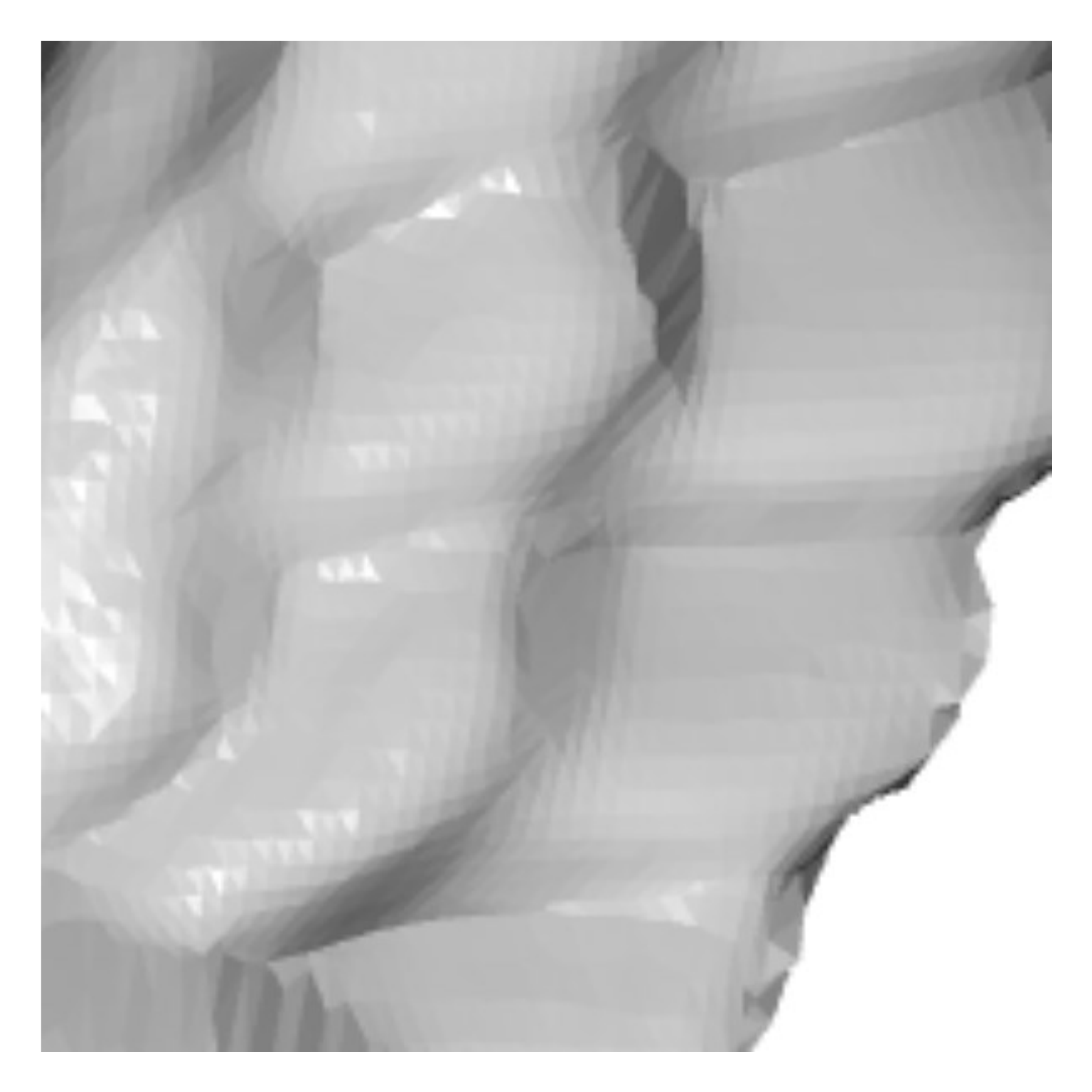}&
		\includegraphics[width=1.72cm]{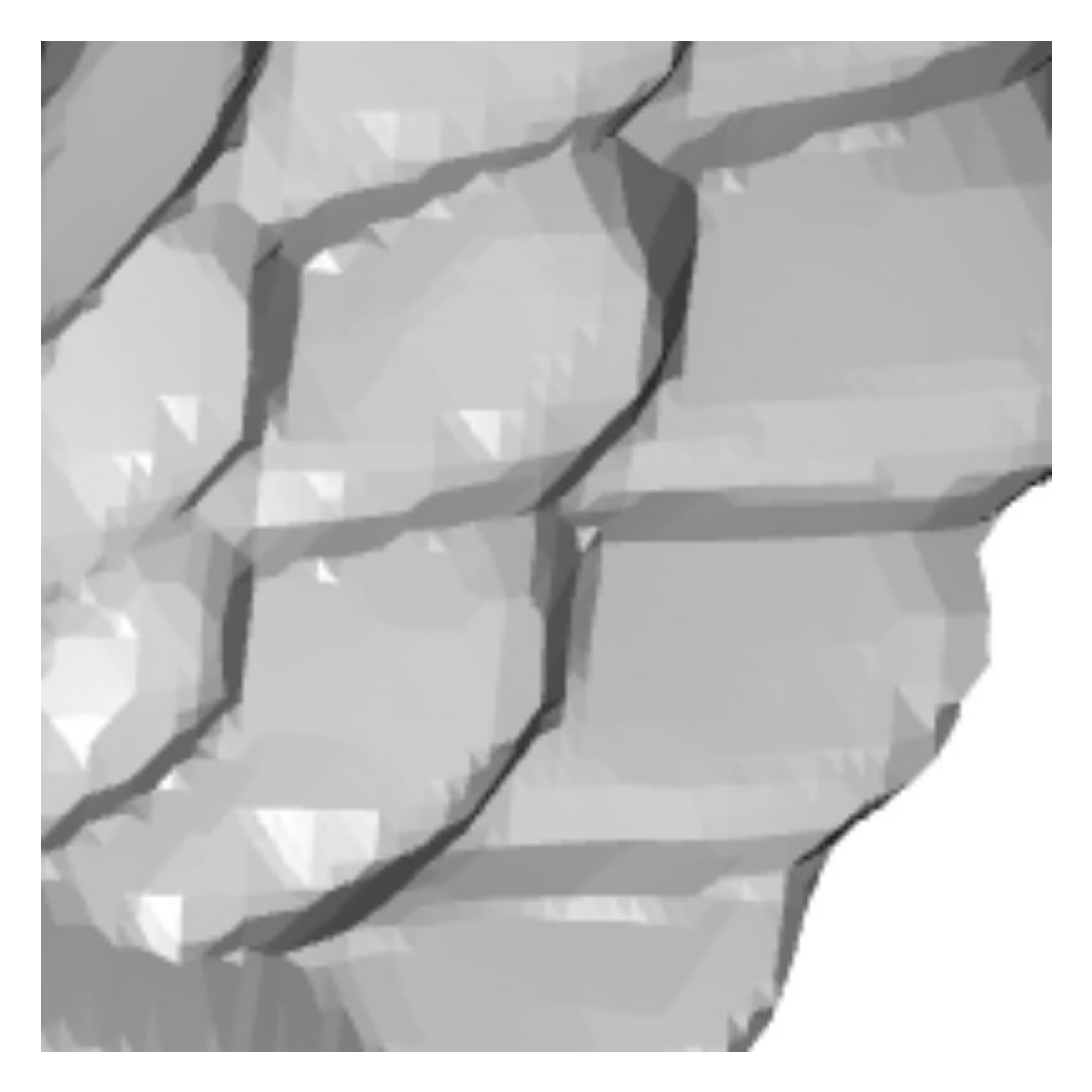}&
		\includegraphics[width=1.72cm]{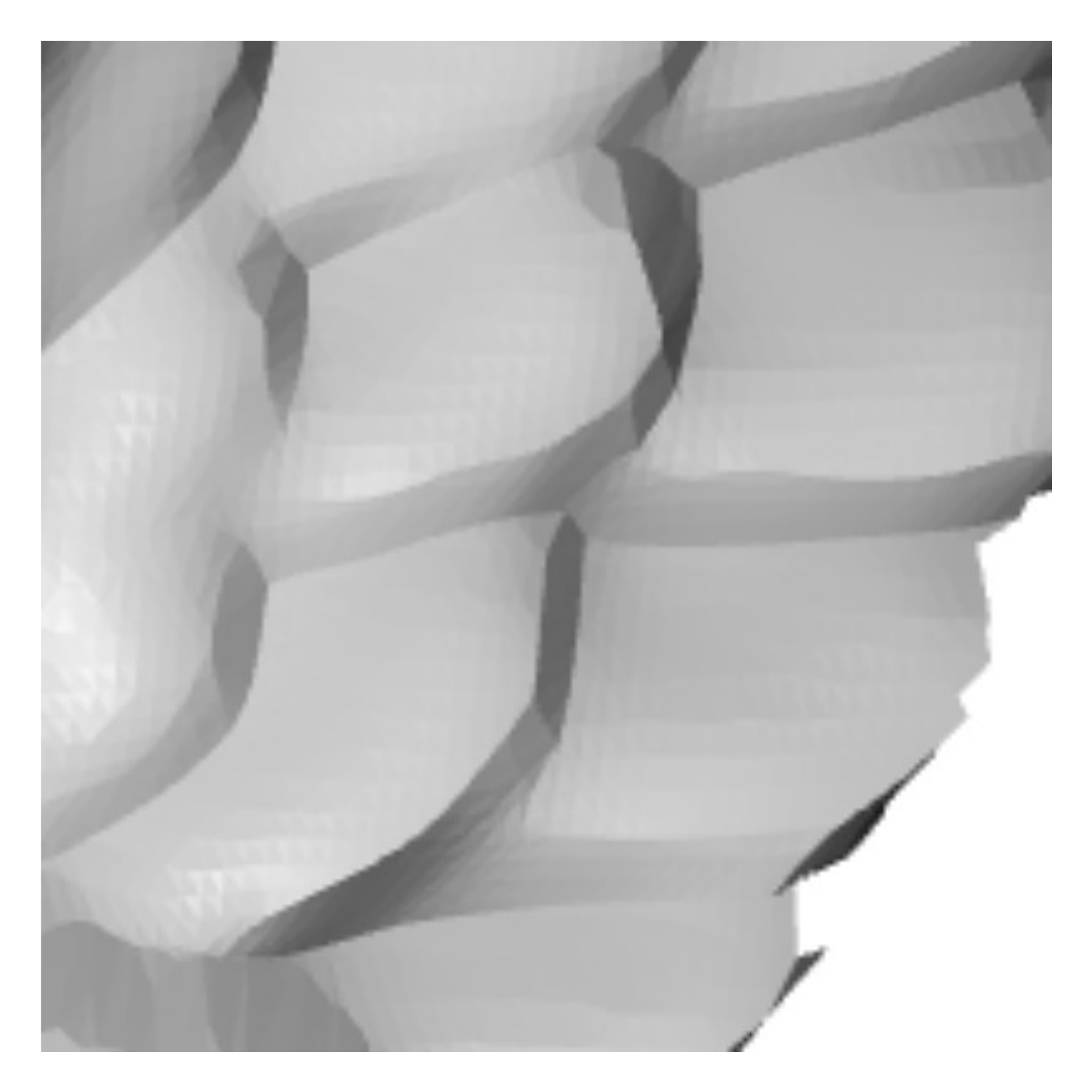}&
		\includegraphics[width=1.72cm]{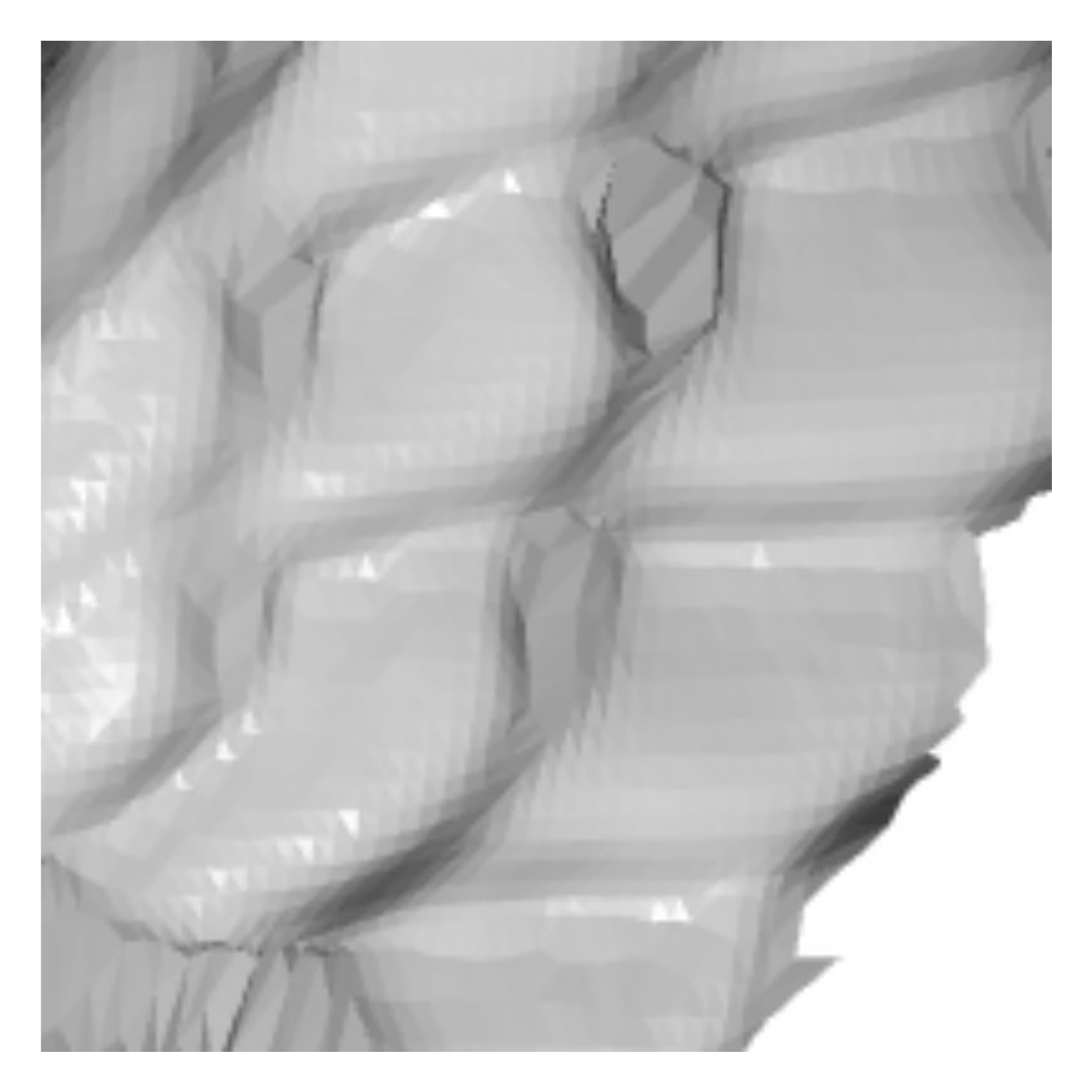}&
		\includegraphics[width=1.72cm]{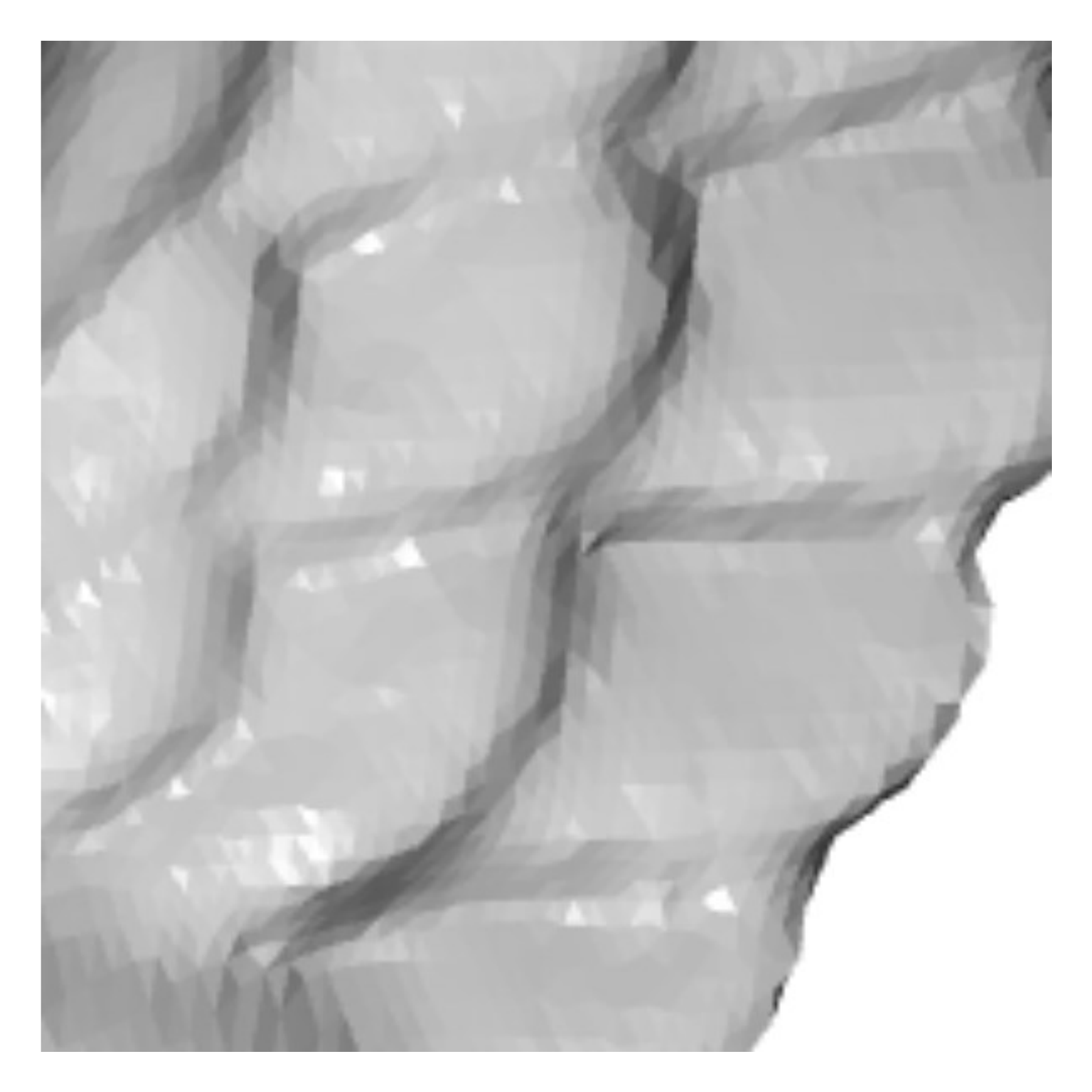}\\
		\includegraphics[width=1.72cm]{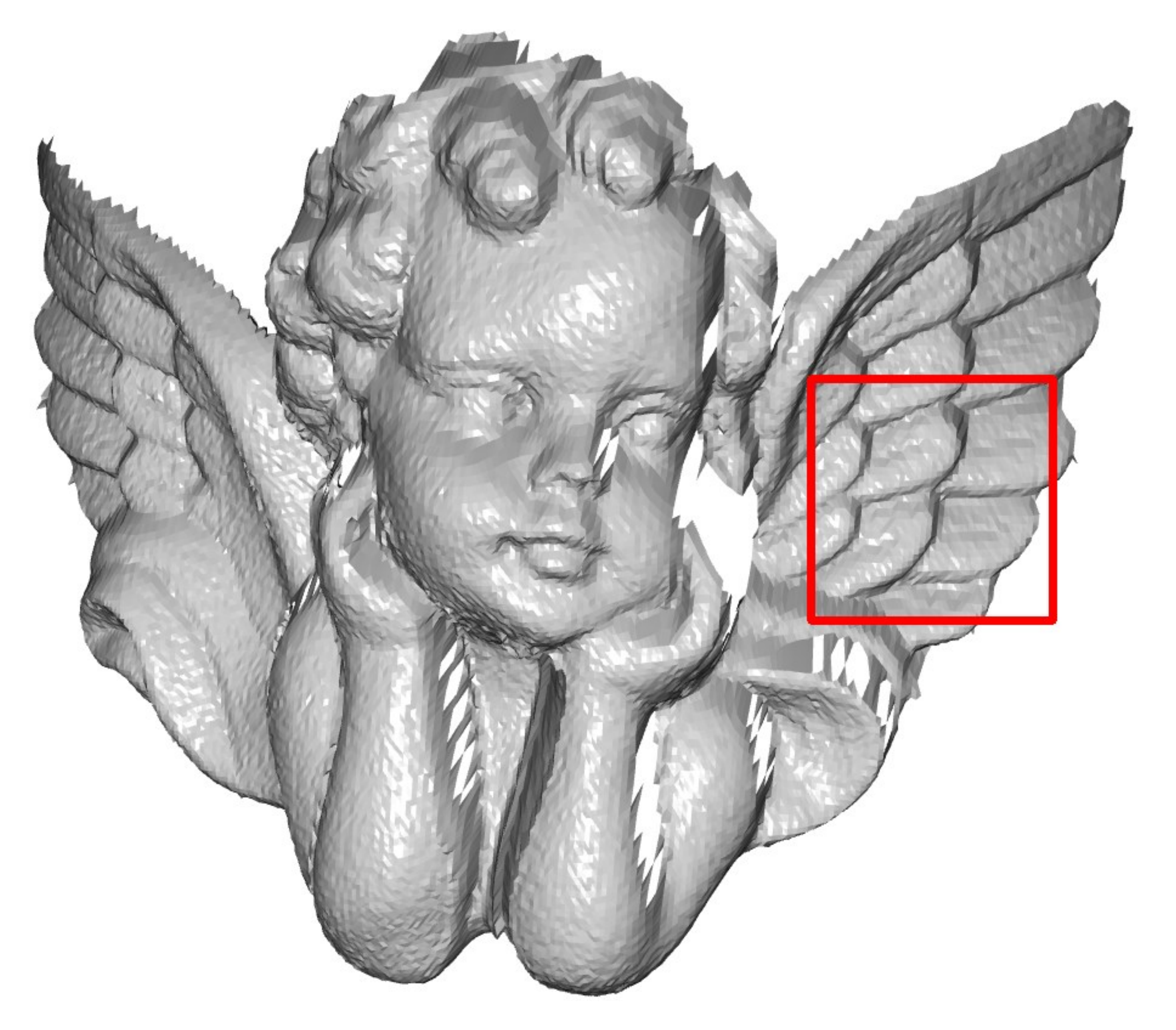}&
		\includegraphics[width=1.72cm]{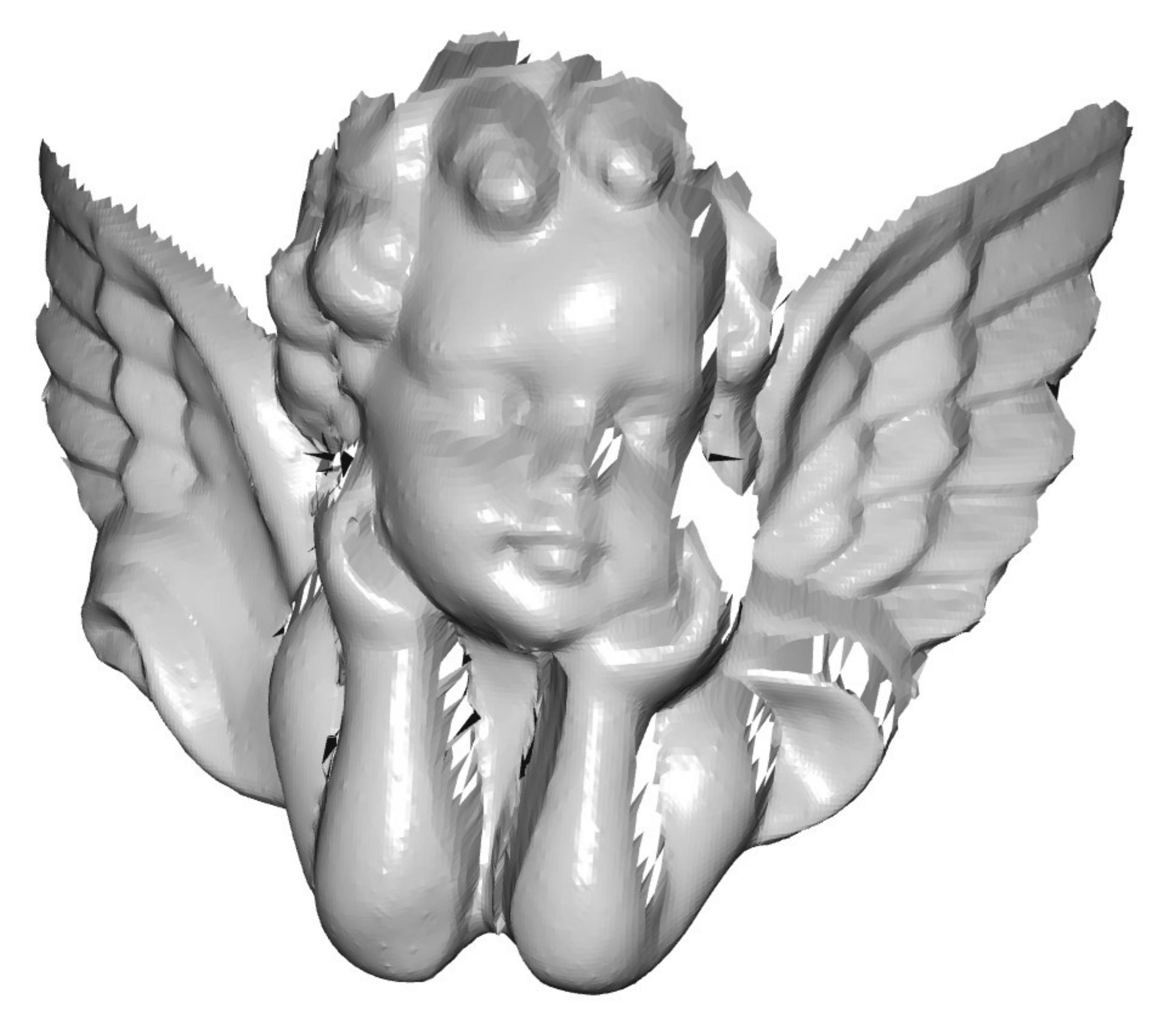}&
		\includegraphics[width=1.72cm]{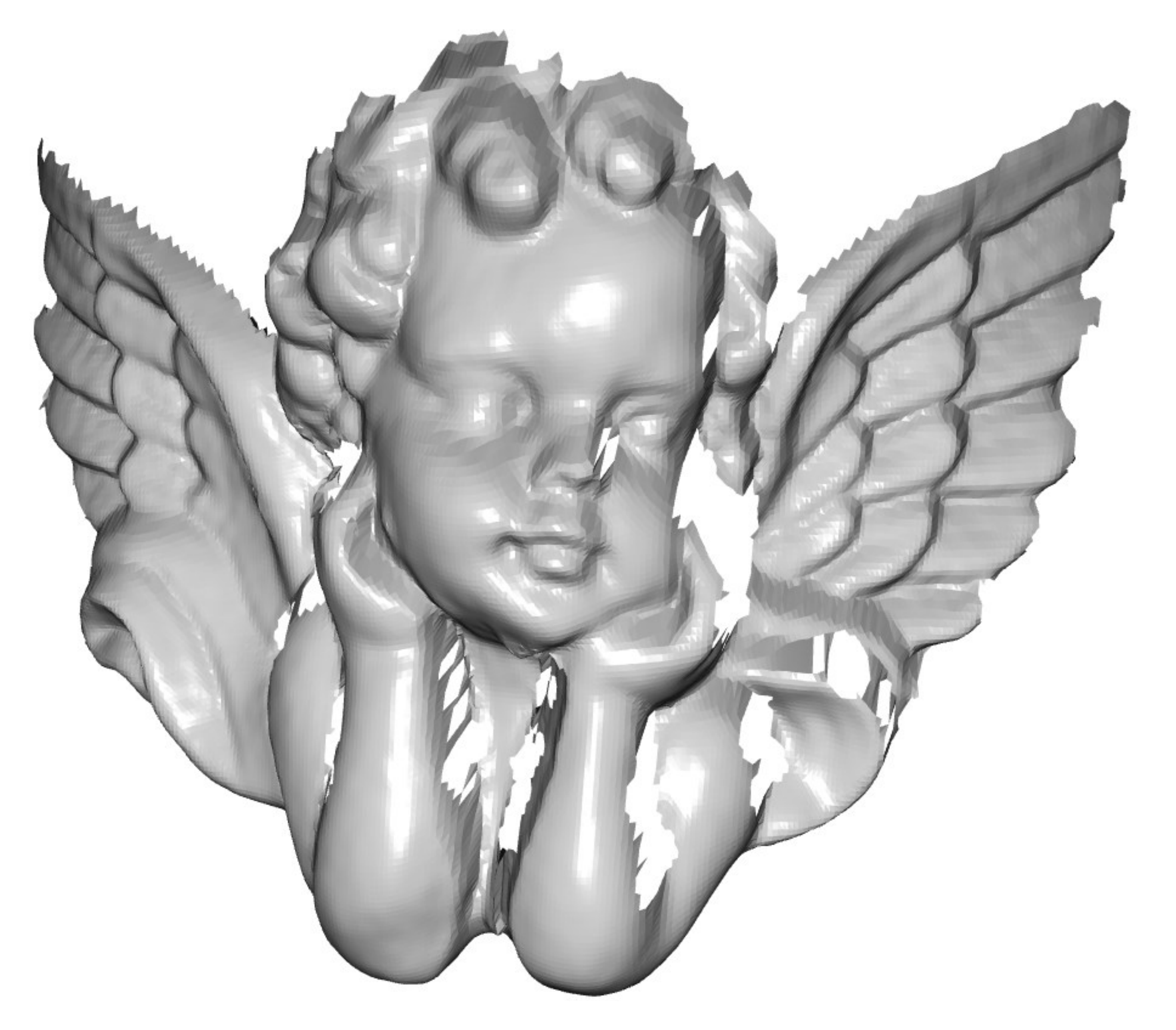}&
		\includegraphics[width=1.72cm]{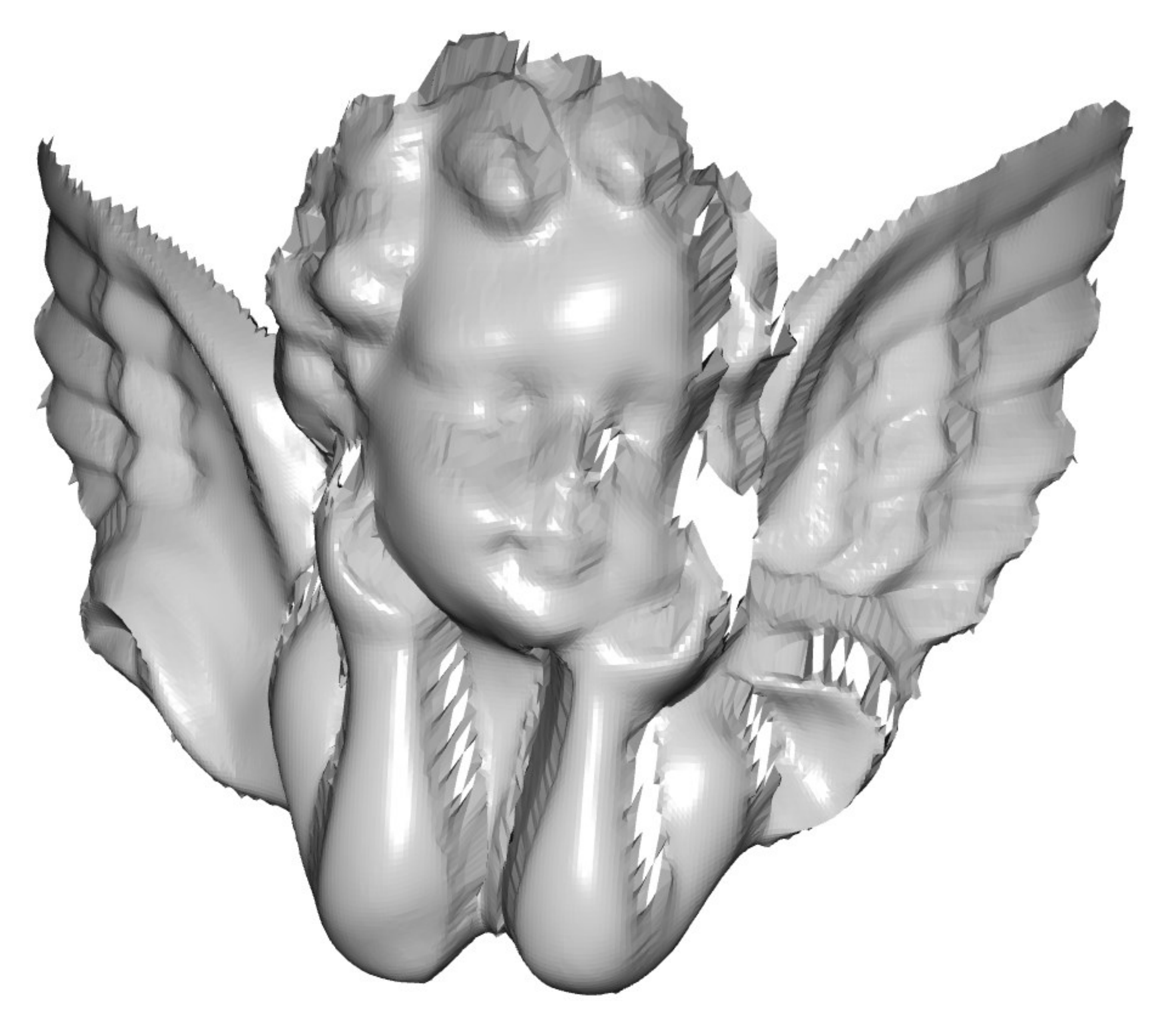}&
		\includegraphics[width=1.72cm]{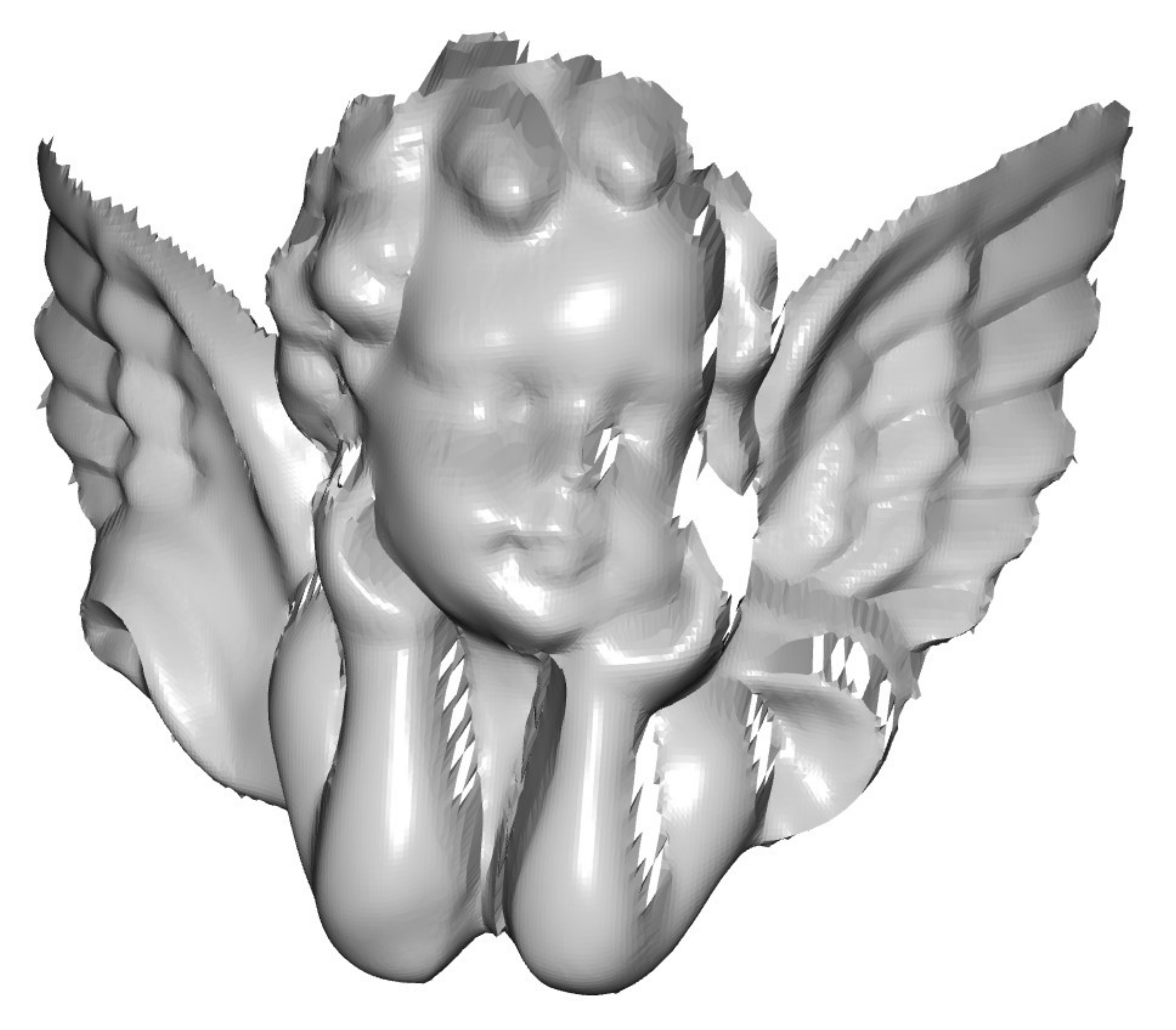}&
		\includegraphics[width=1.72cm]{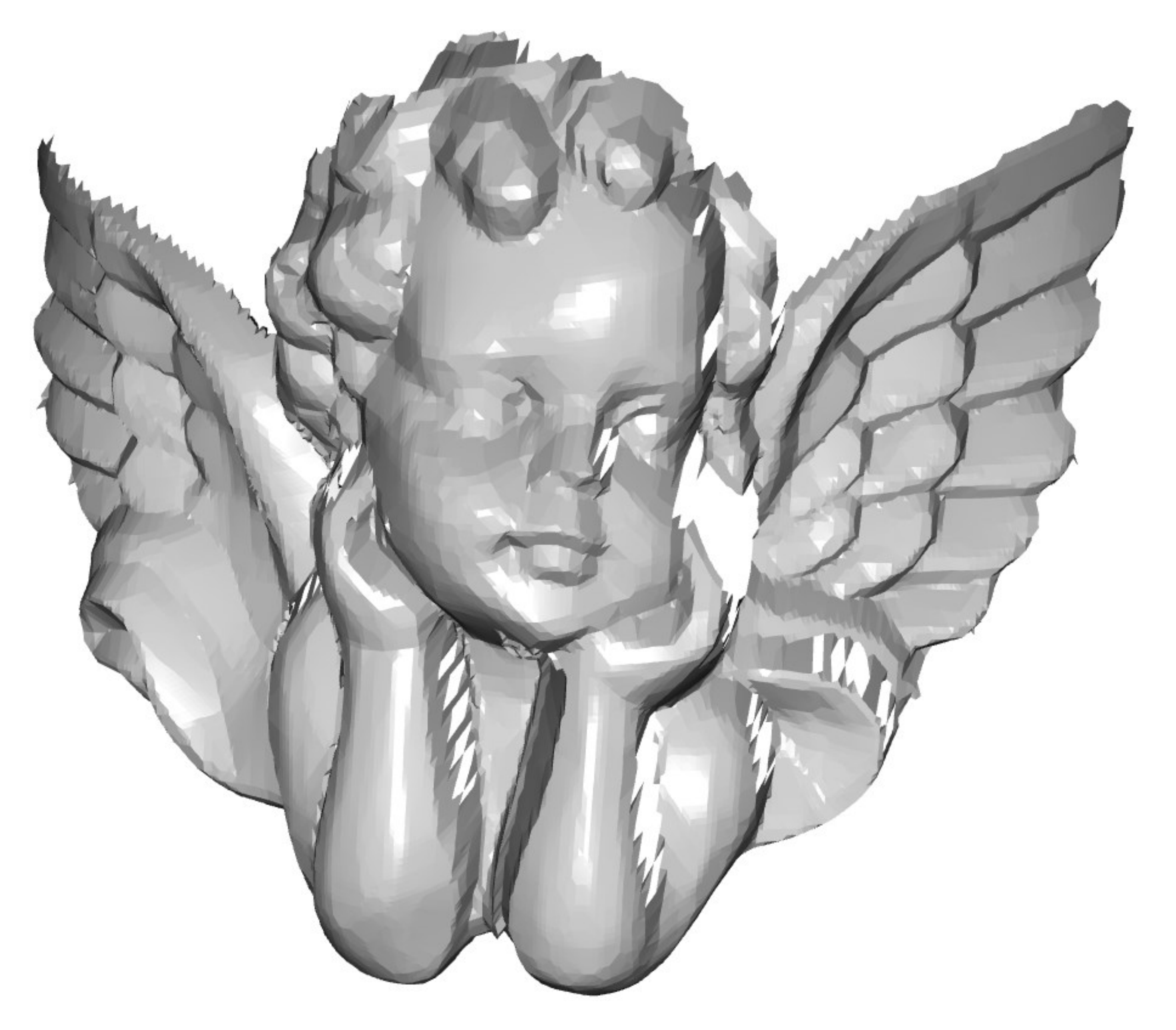}&
		\includegraphics[width=1.72cm]{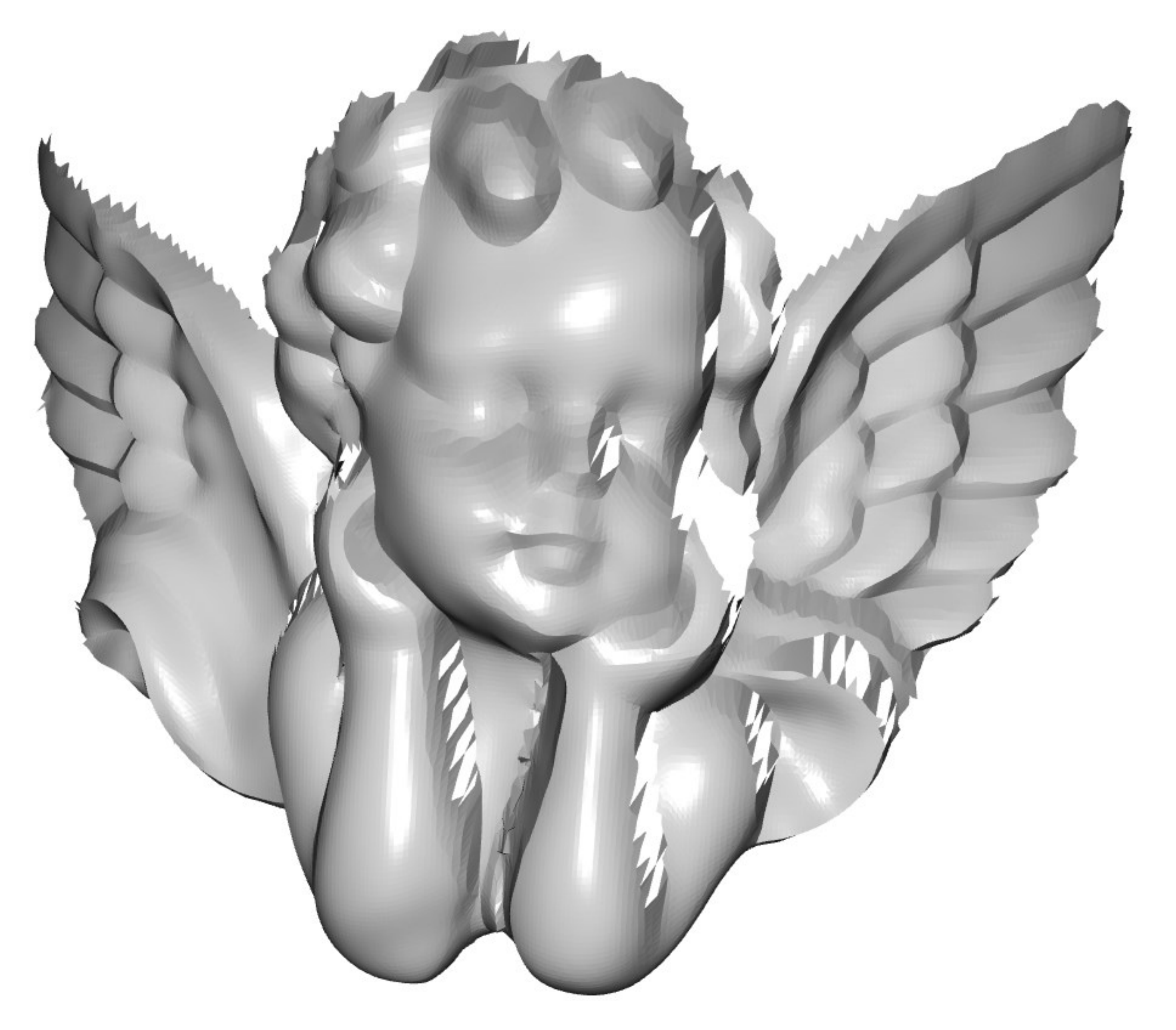}&
		\includegraphics[width=1.72cm]{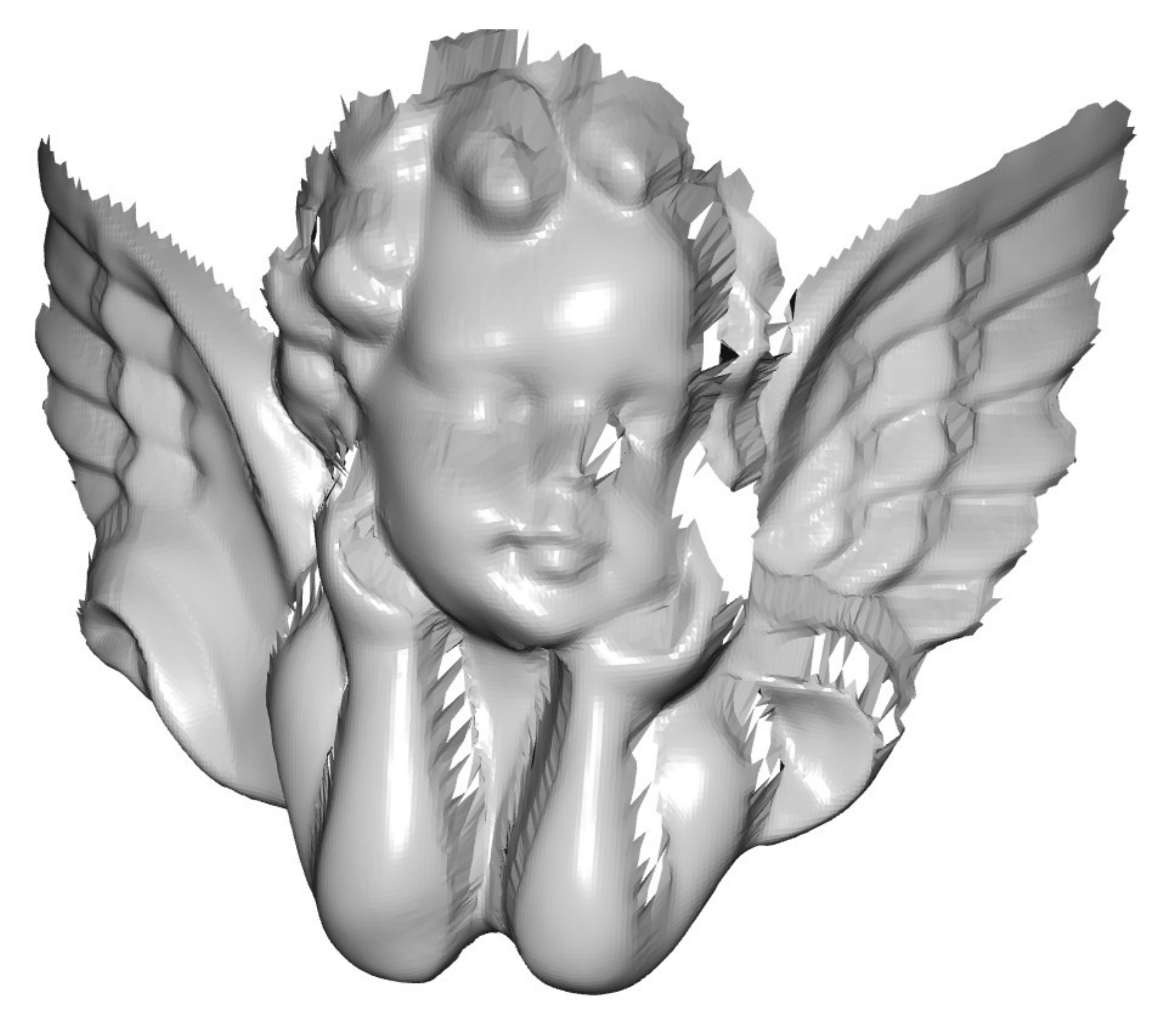}&
		\includegraphics[width=1.72cm]{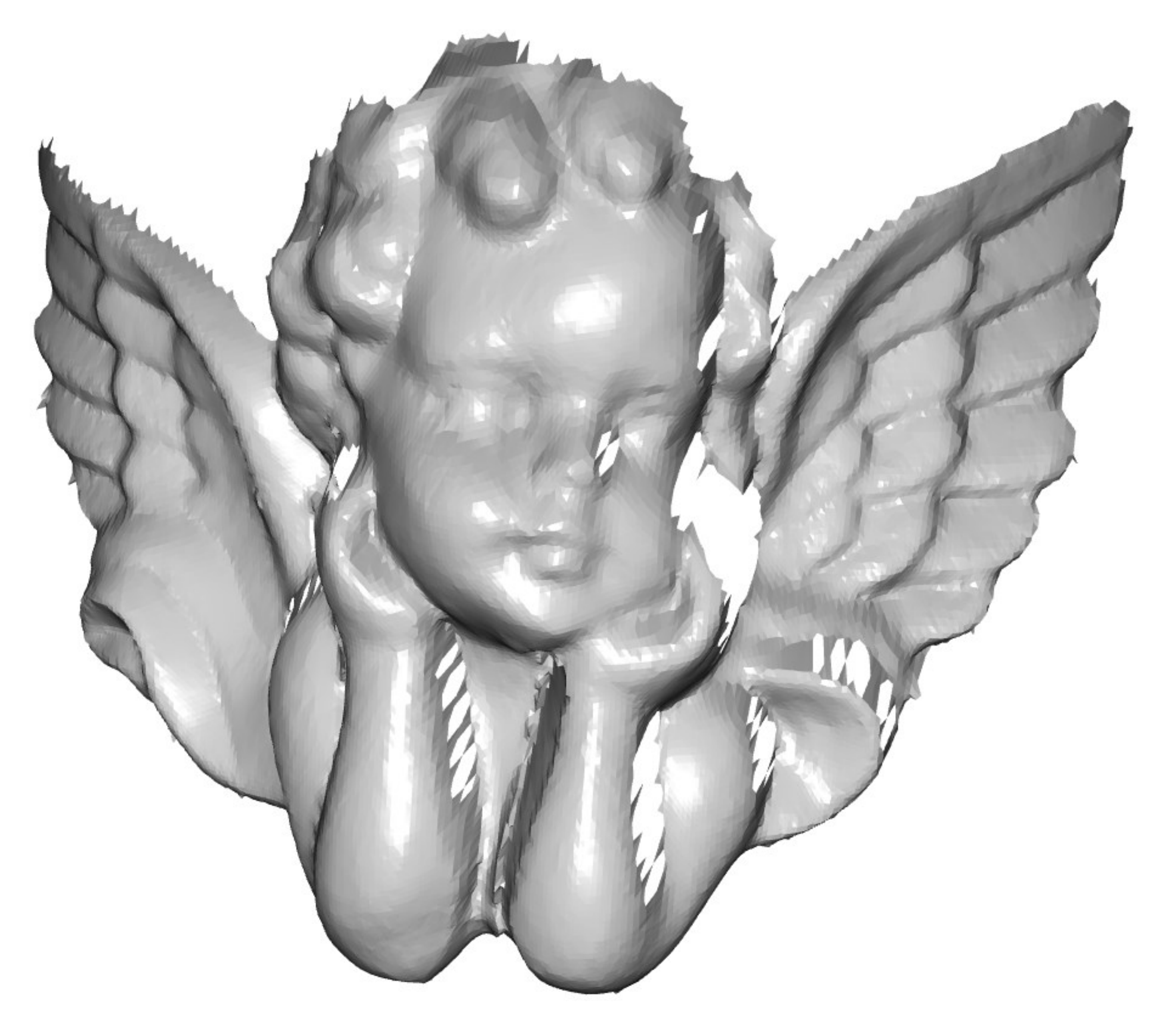}\\
		\includegraphics[width=1.72cm]{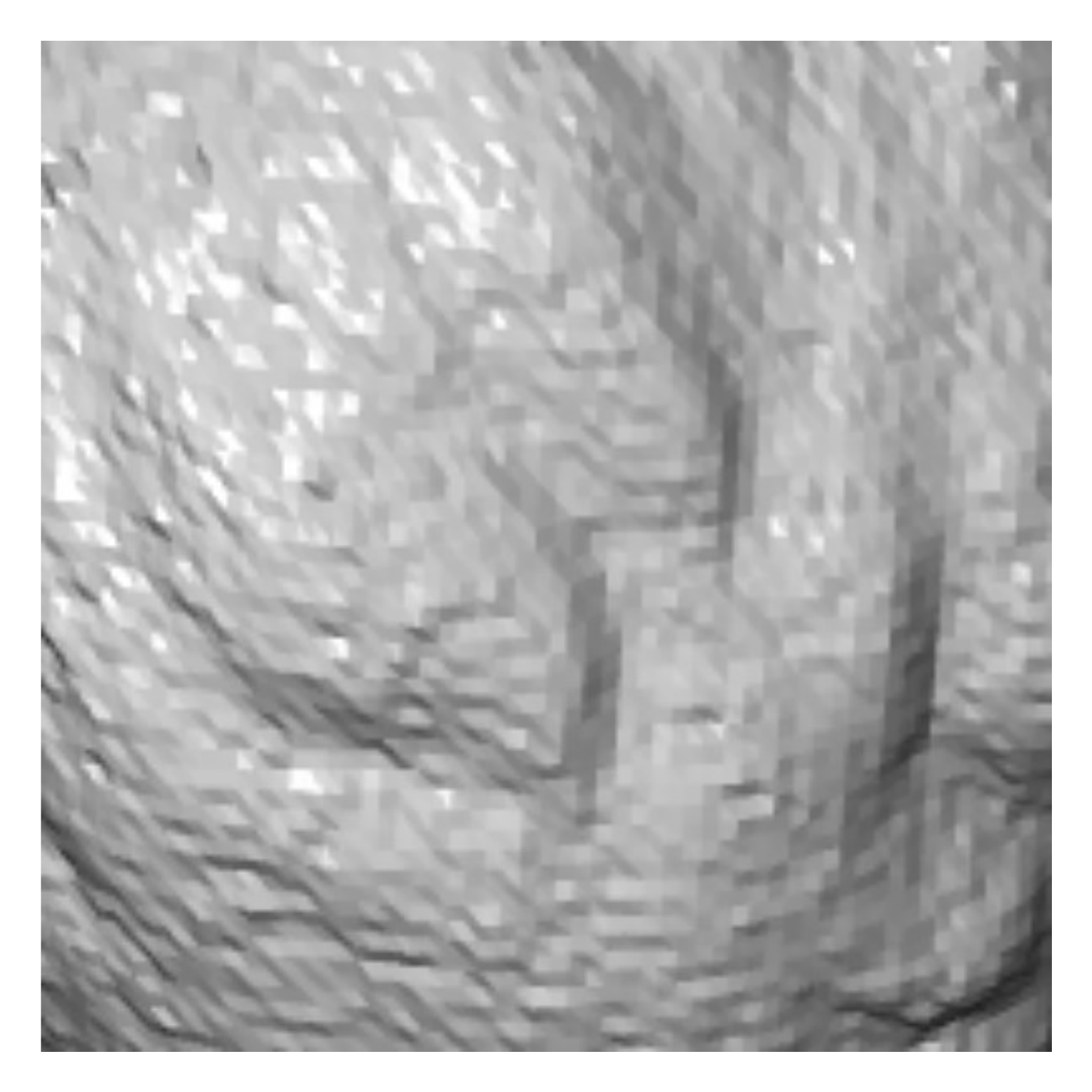}&
		\includegraphics[width=1.72cm]{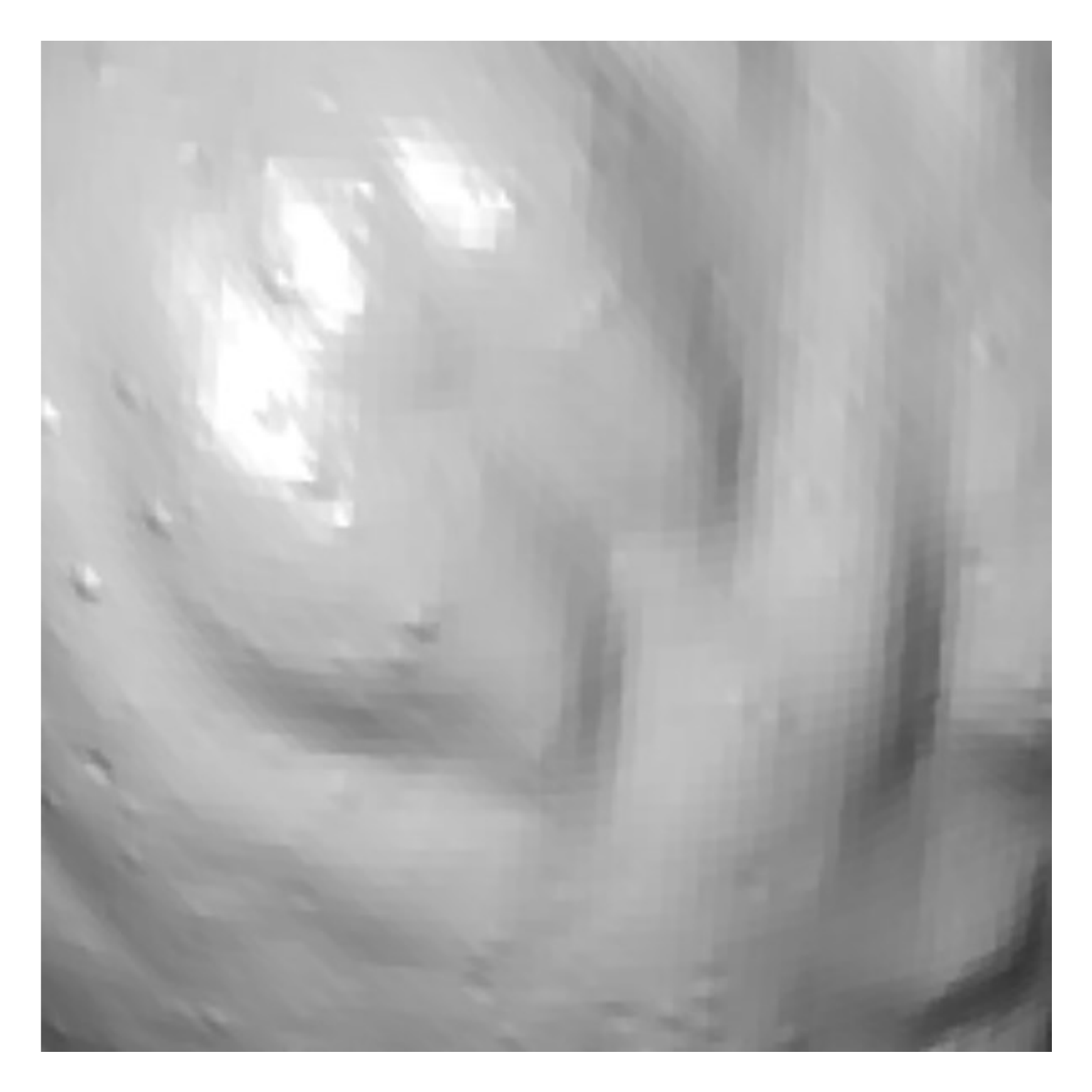}&
		\includegraphics[width=1.72cm]{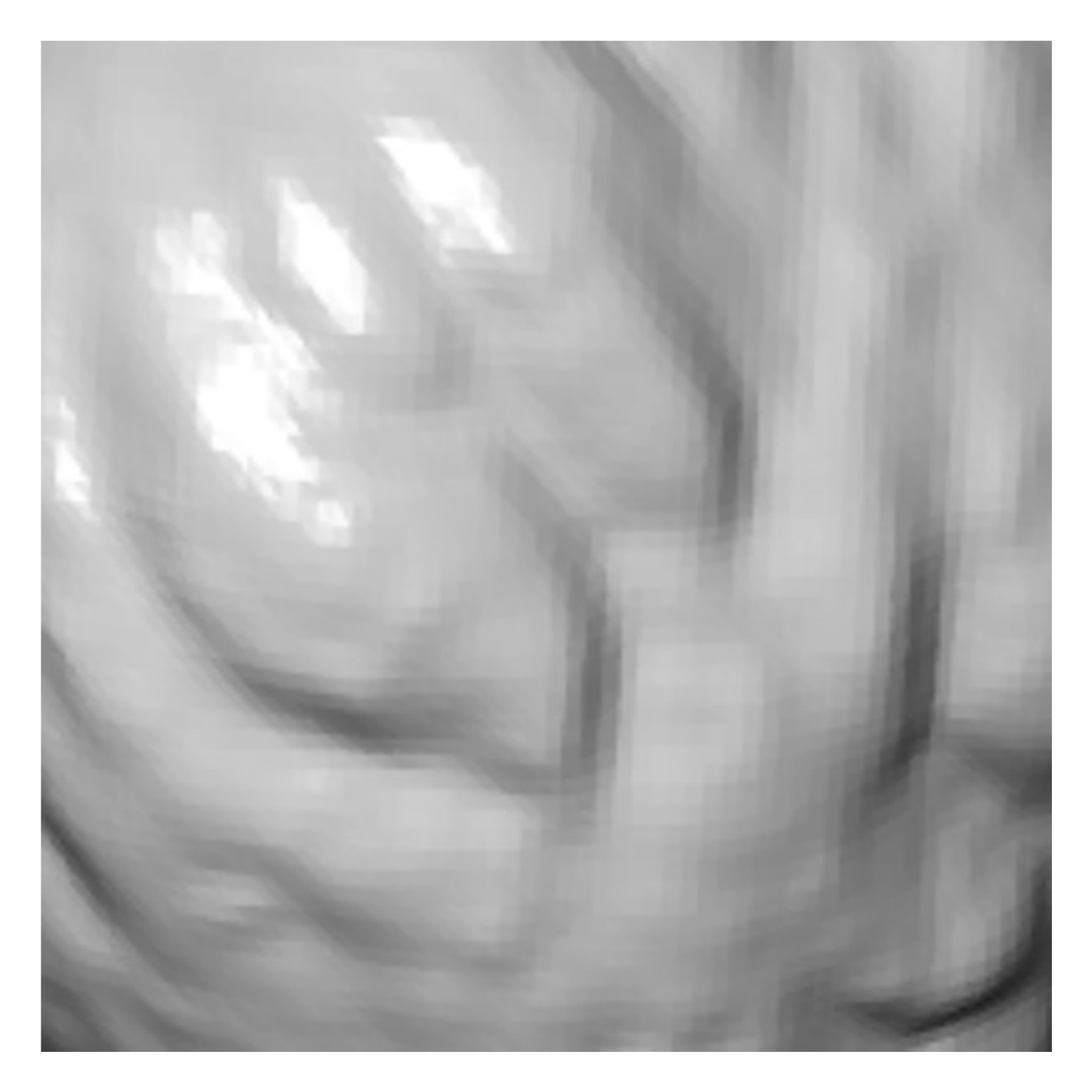}&
		\includegraphics[width=1.72cm]{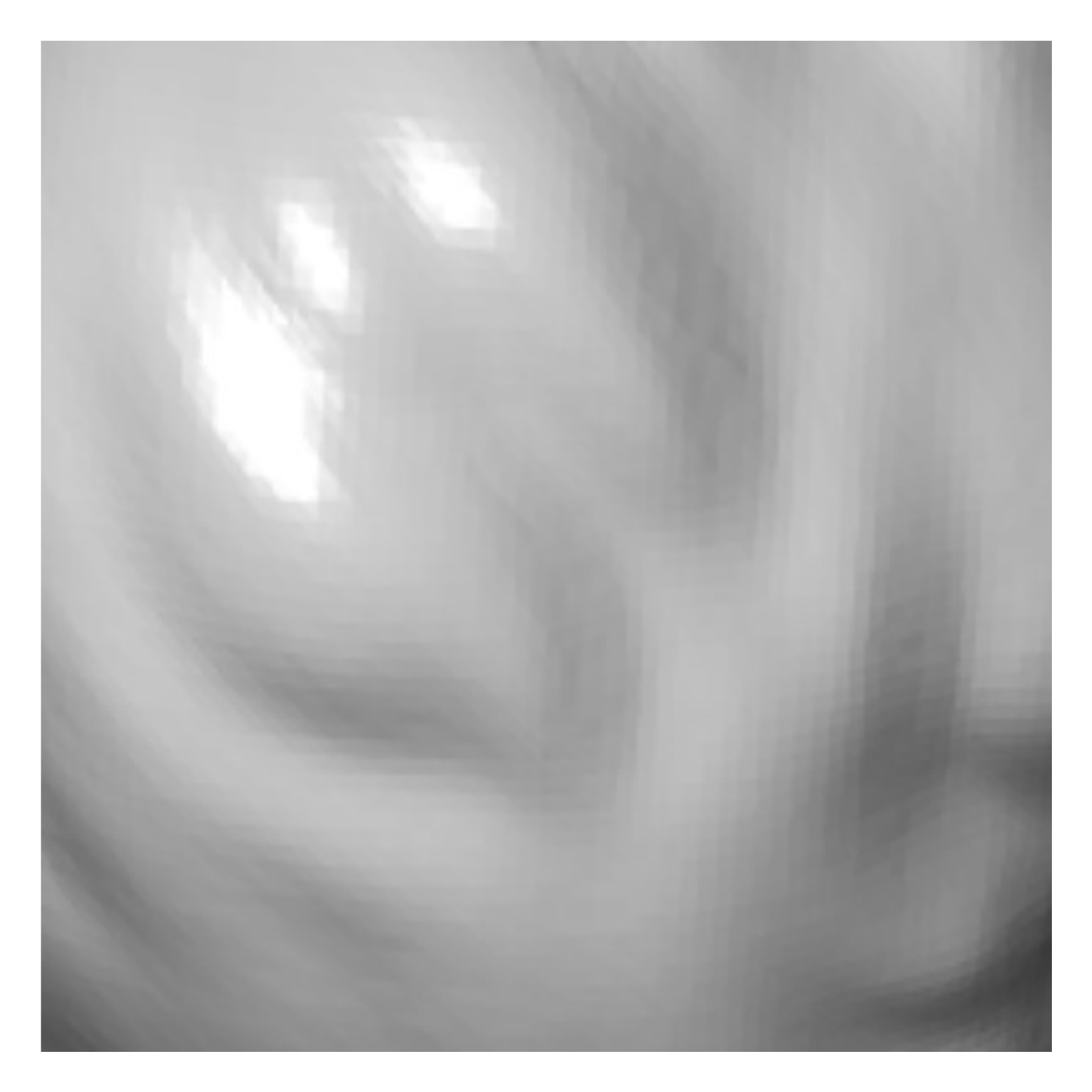}&
		\includegraphics[width=1.72cm]{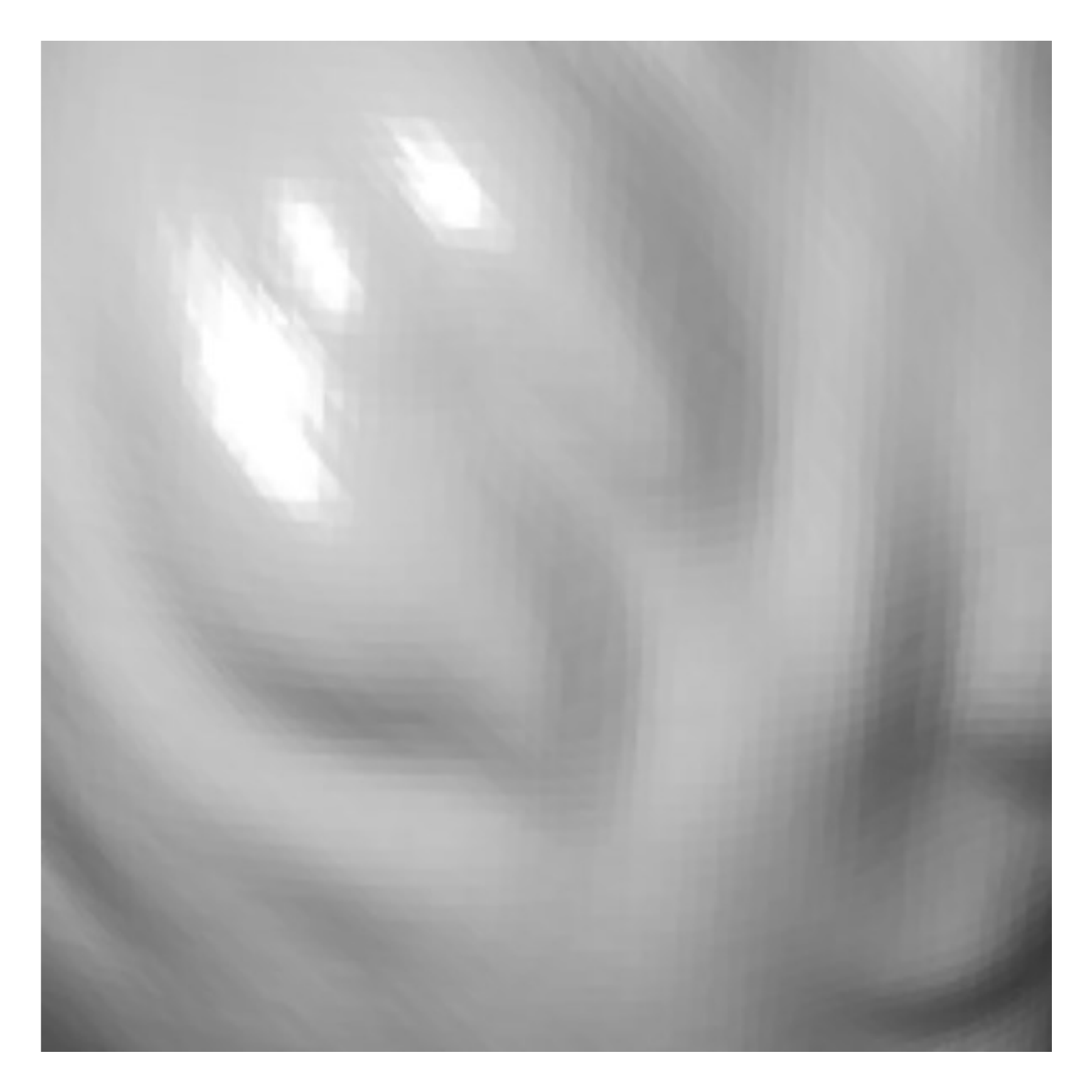}&
		\includegraphics[width=1.72cm]{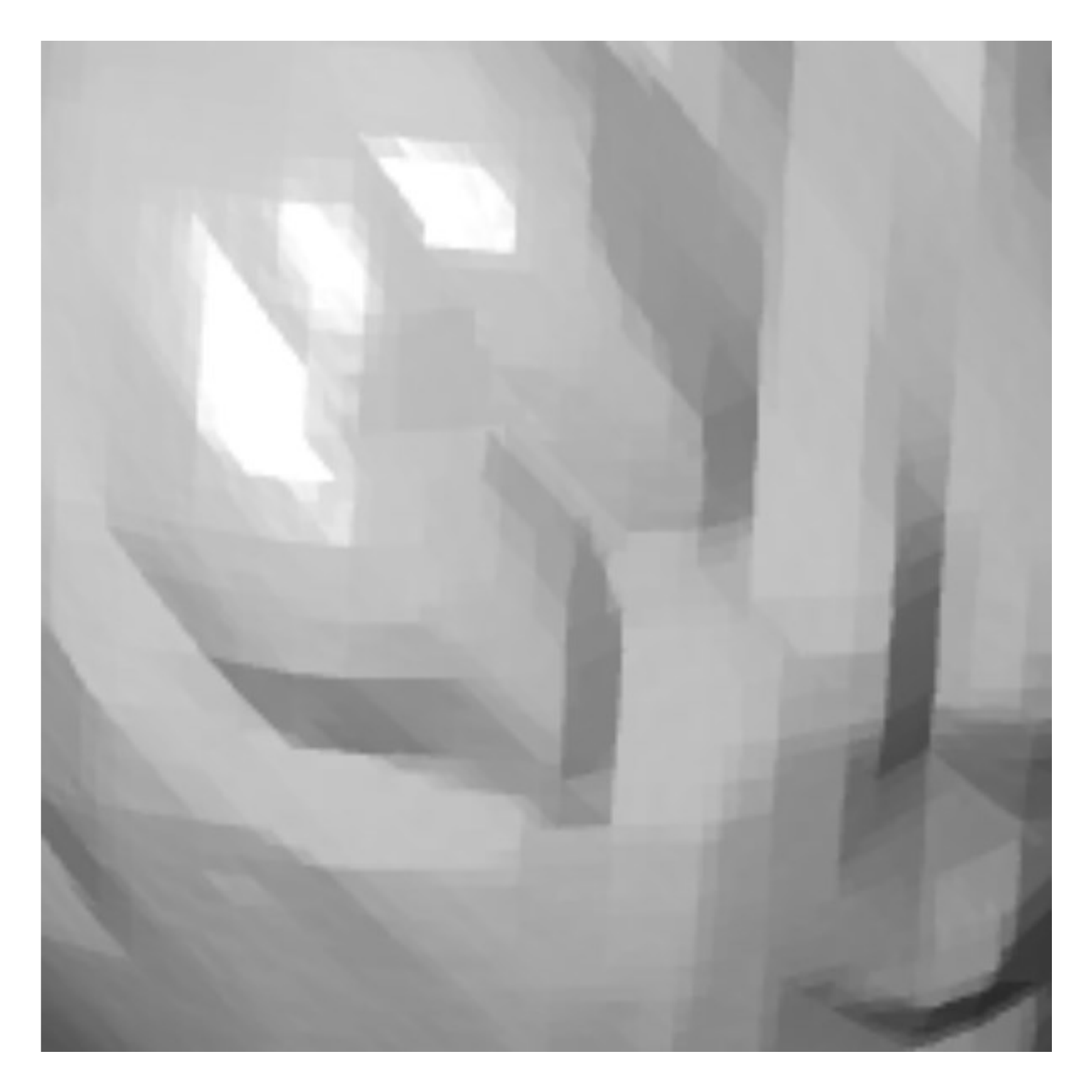}&
		\includegraphics[width=1.72cm]{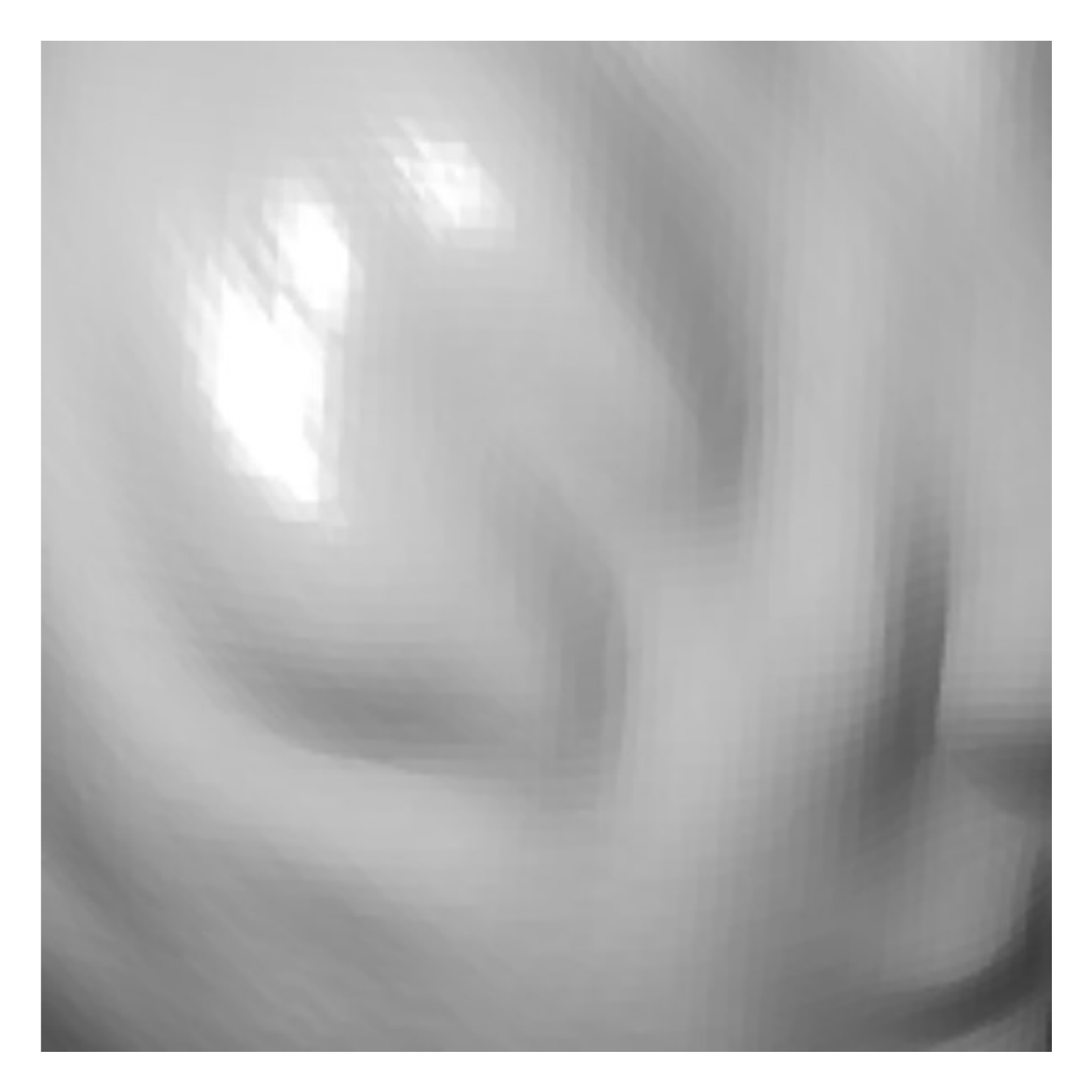}&
		\includegraphics[width=1.72cm]{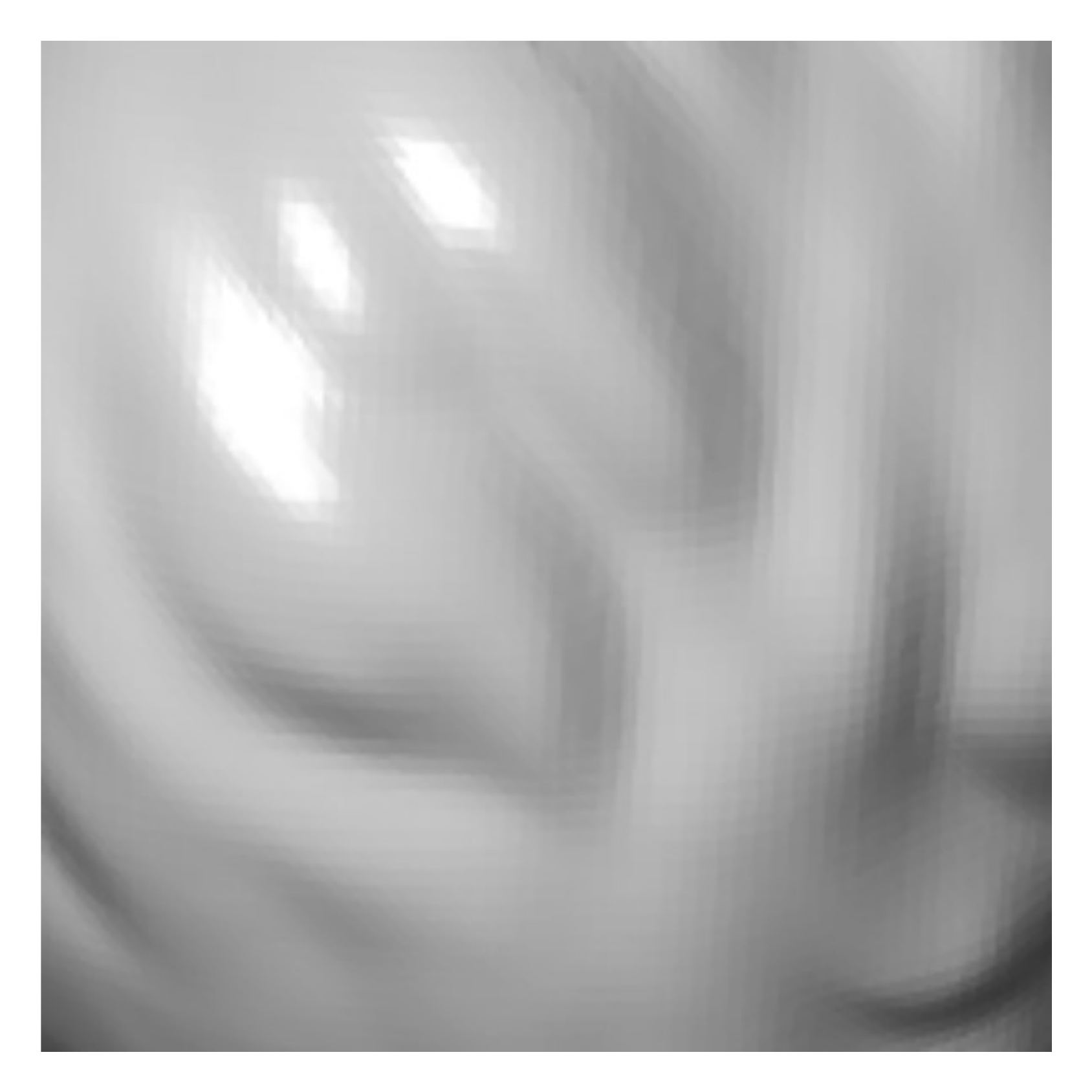}&
		\includegraphics[width=1.72cm]{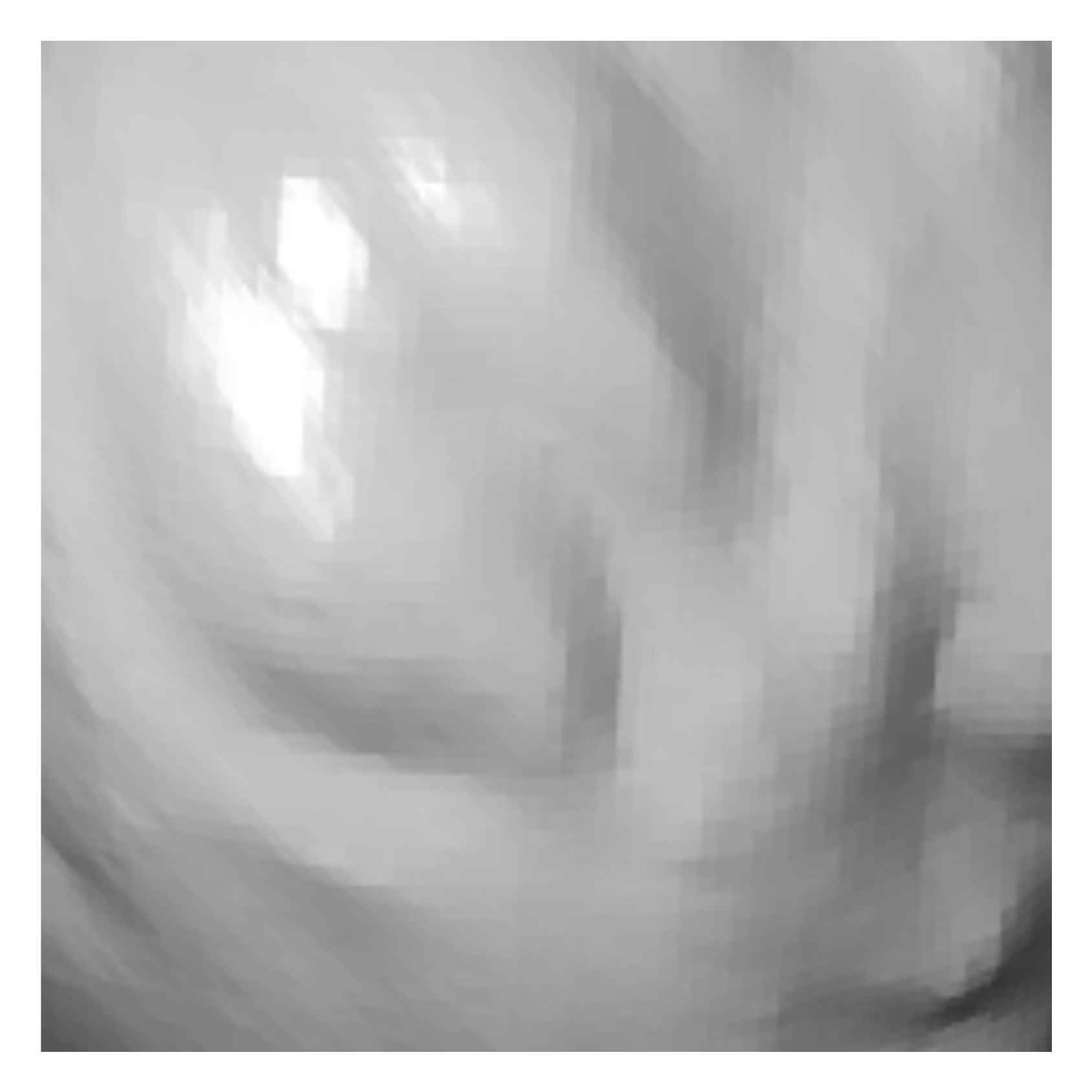}\\
		\includegraphics[width=1.72cm]{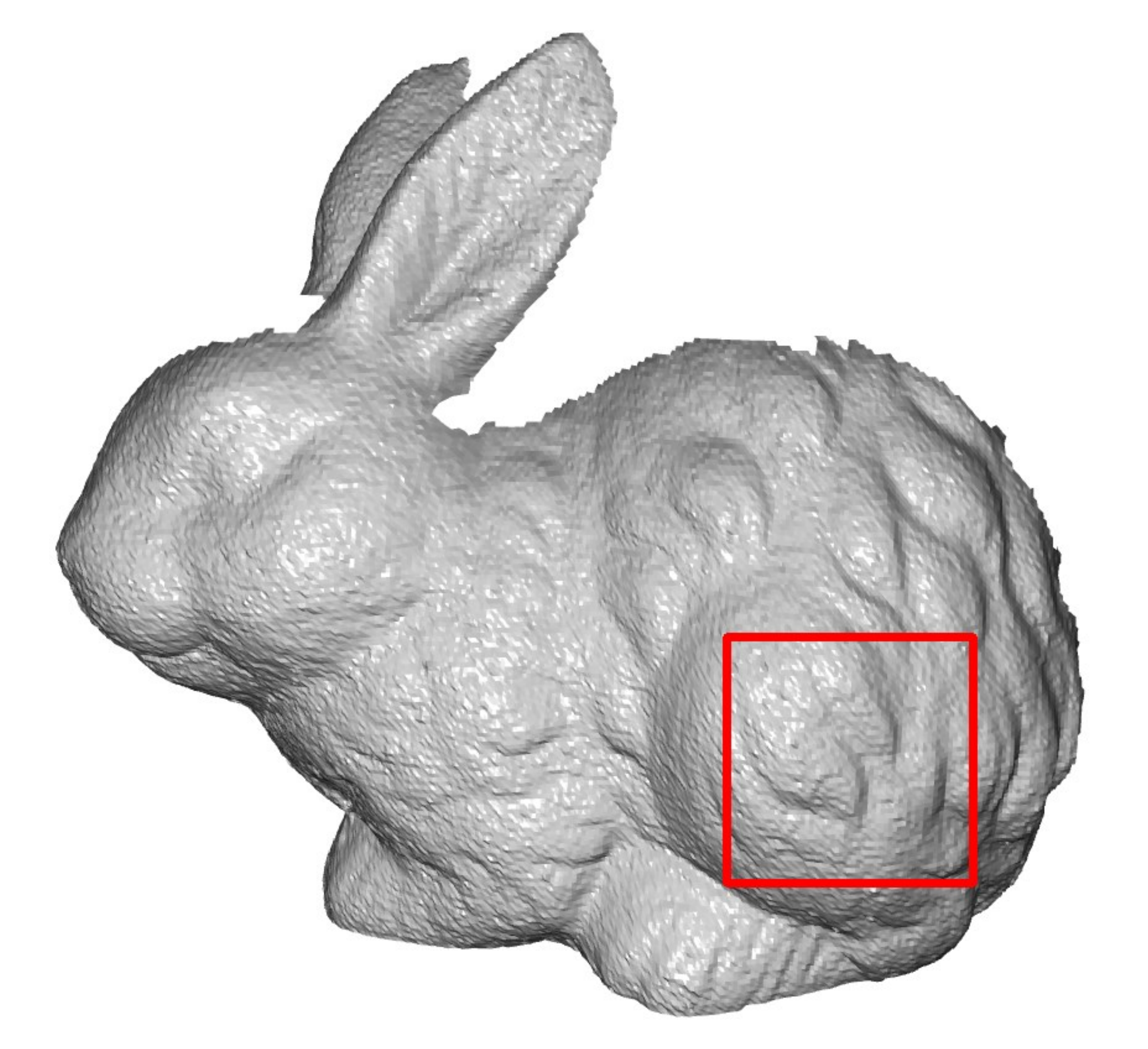}&
		\includegraphics[width=1.72cm]{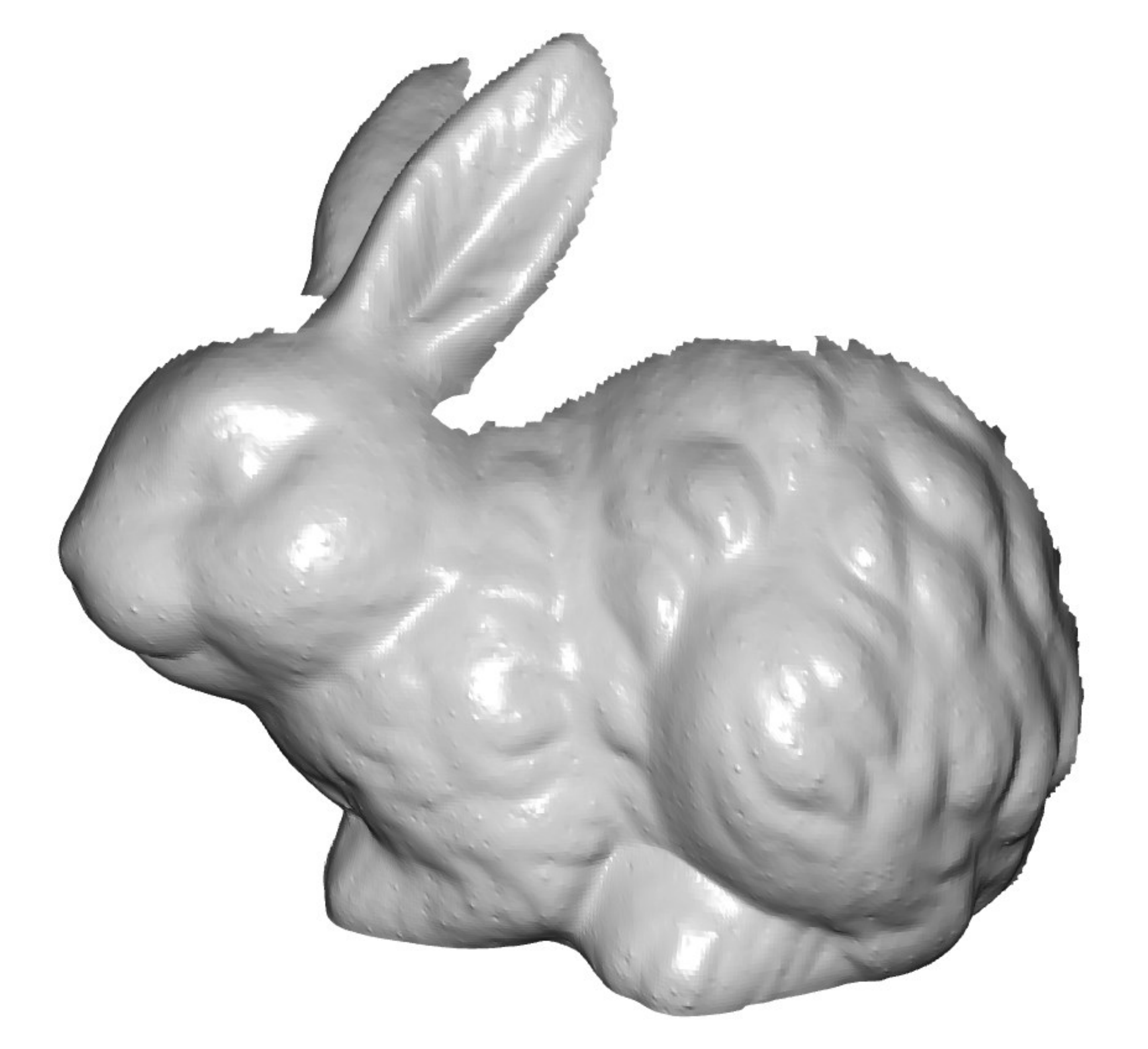}&
		\includegraphics[width=1.72cm]{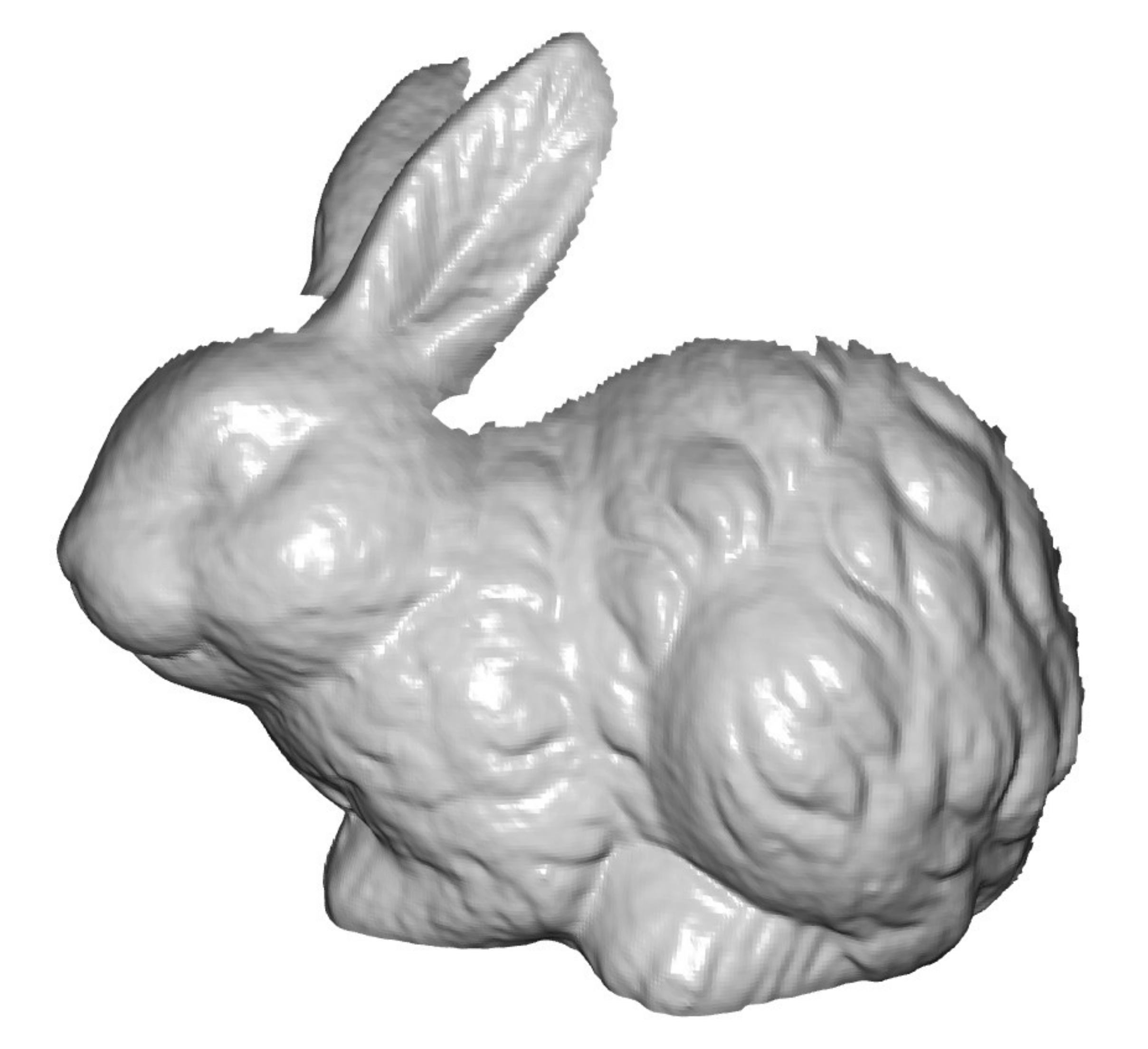}&
		\includegraphics[width=1.72cm]{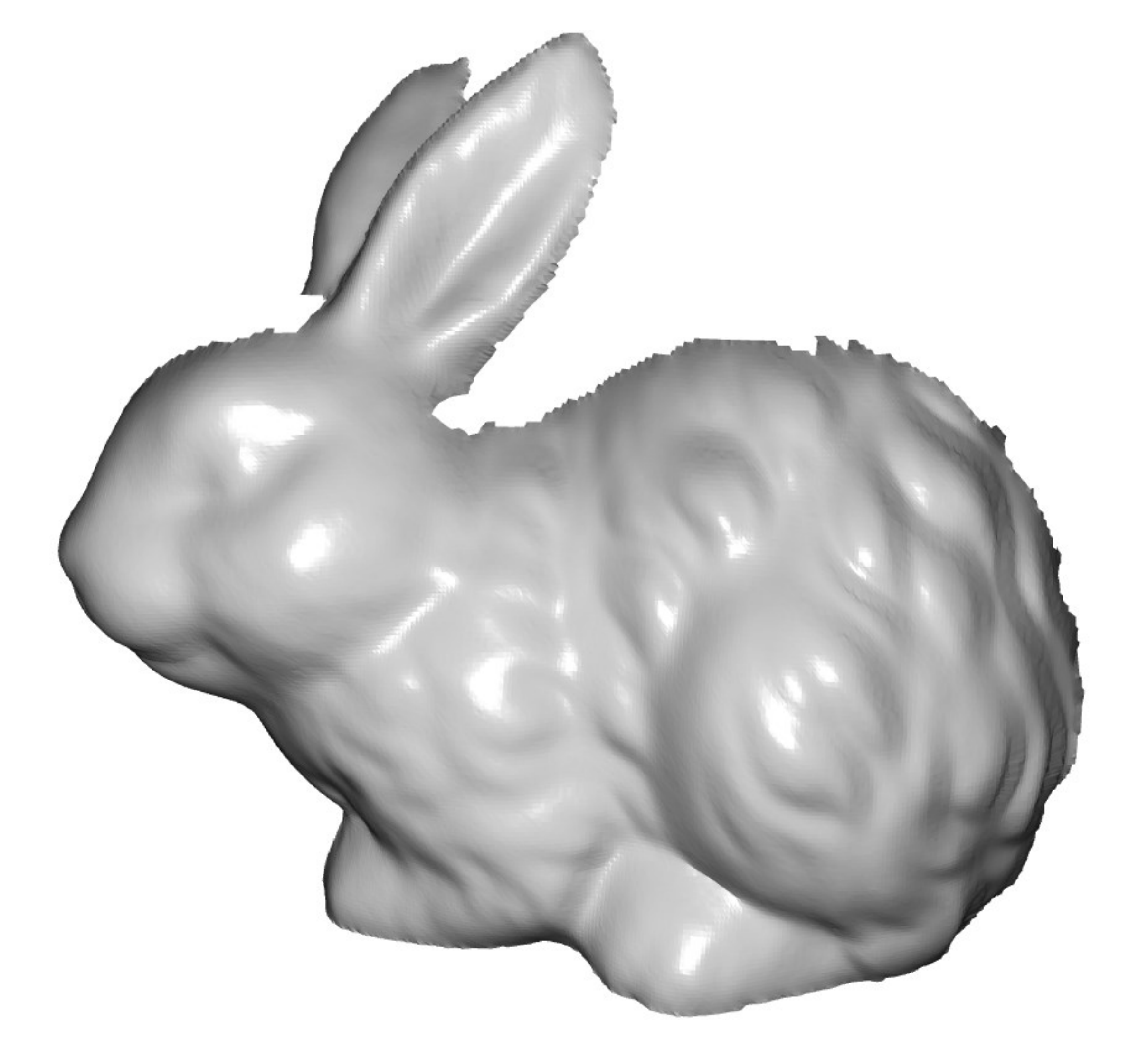}&
		\includegraphics[width=1.72cm]{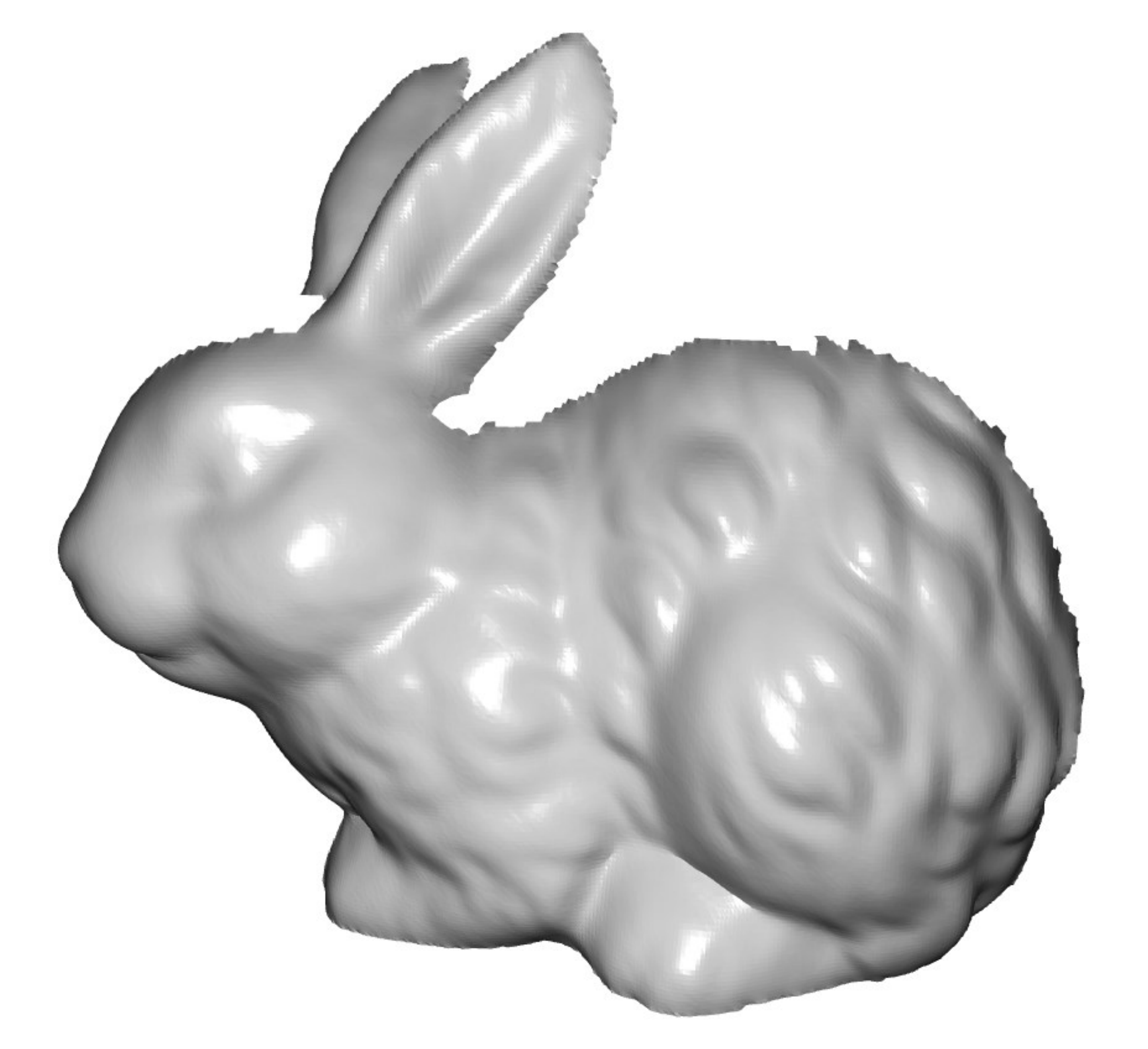}&
		\includegraphics[width=1.72cm]{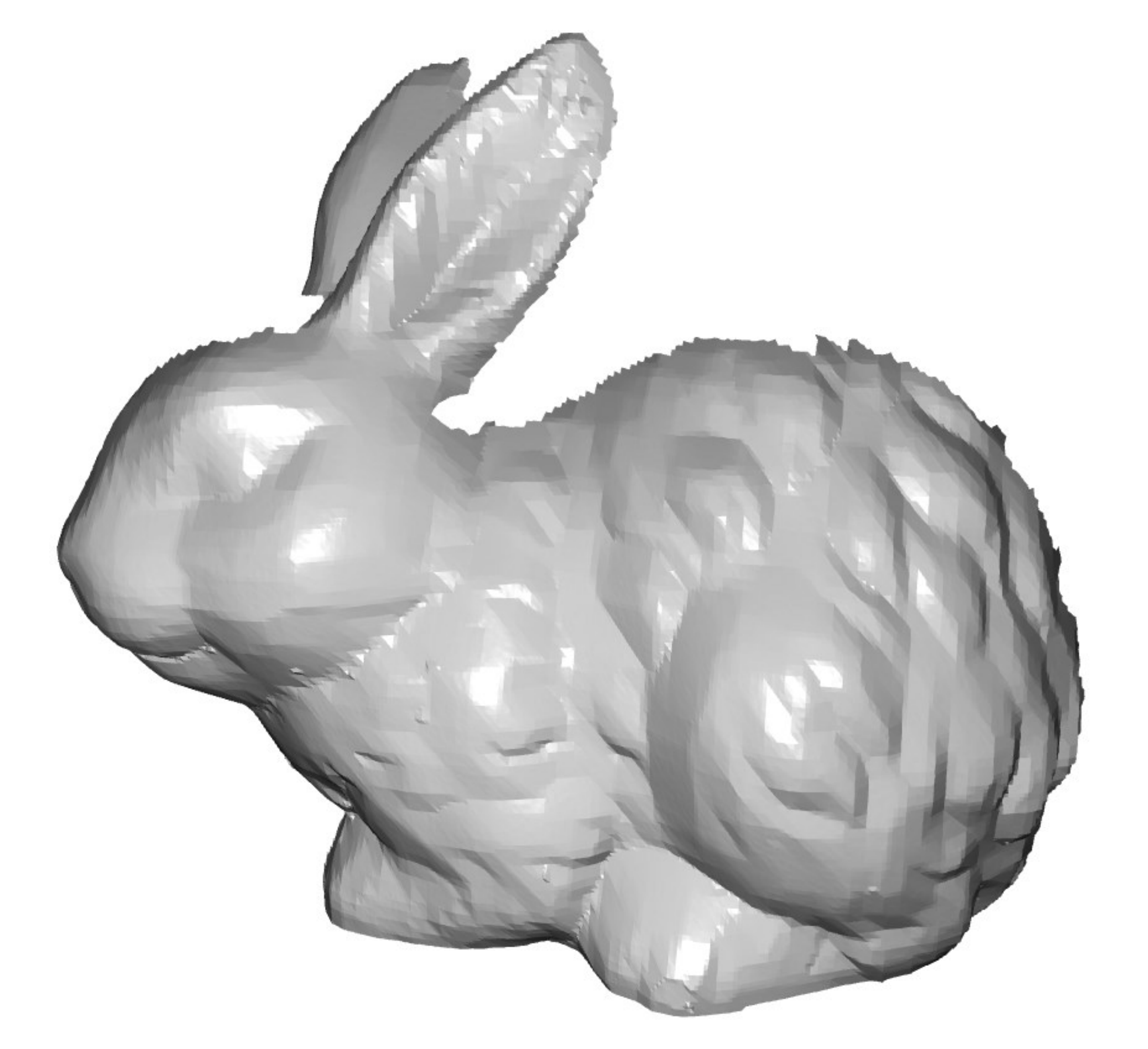}&
		\includegraphics[width=1.72cm]{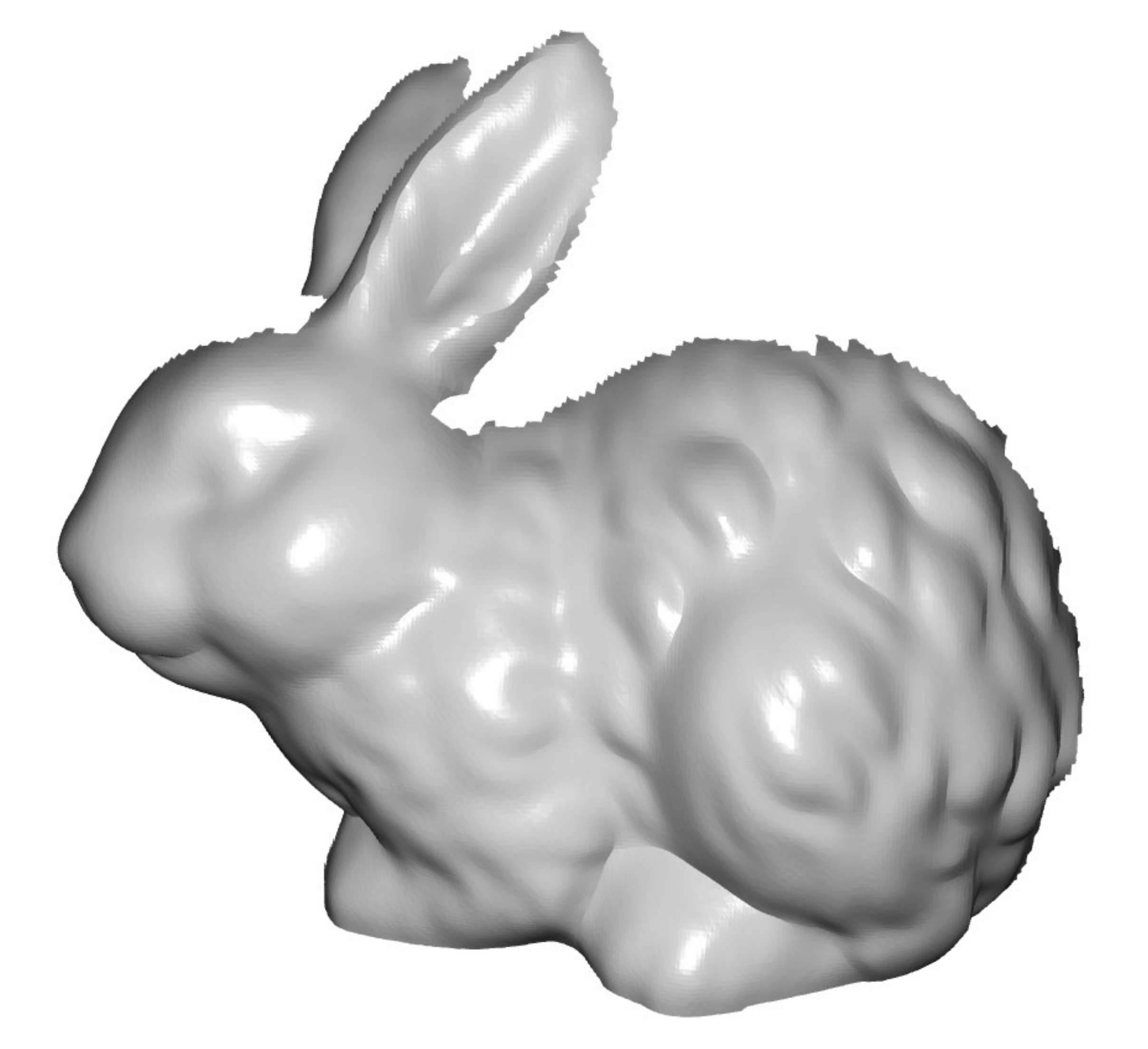}&
		\includegraphics[width=1.72cm]{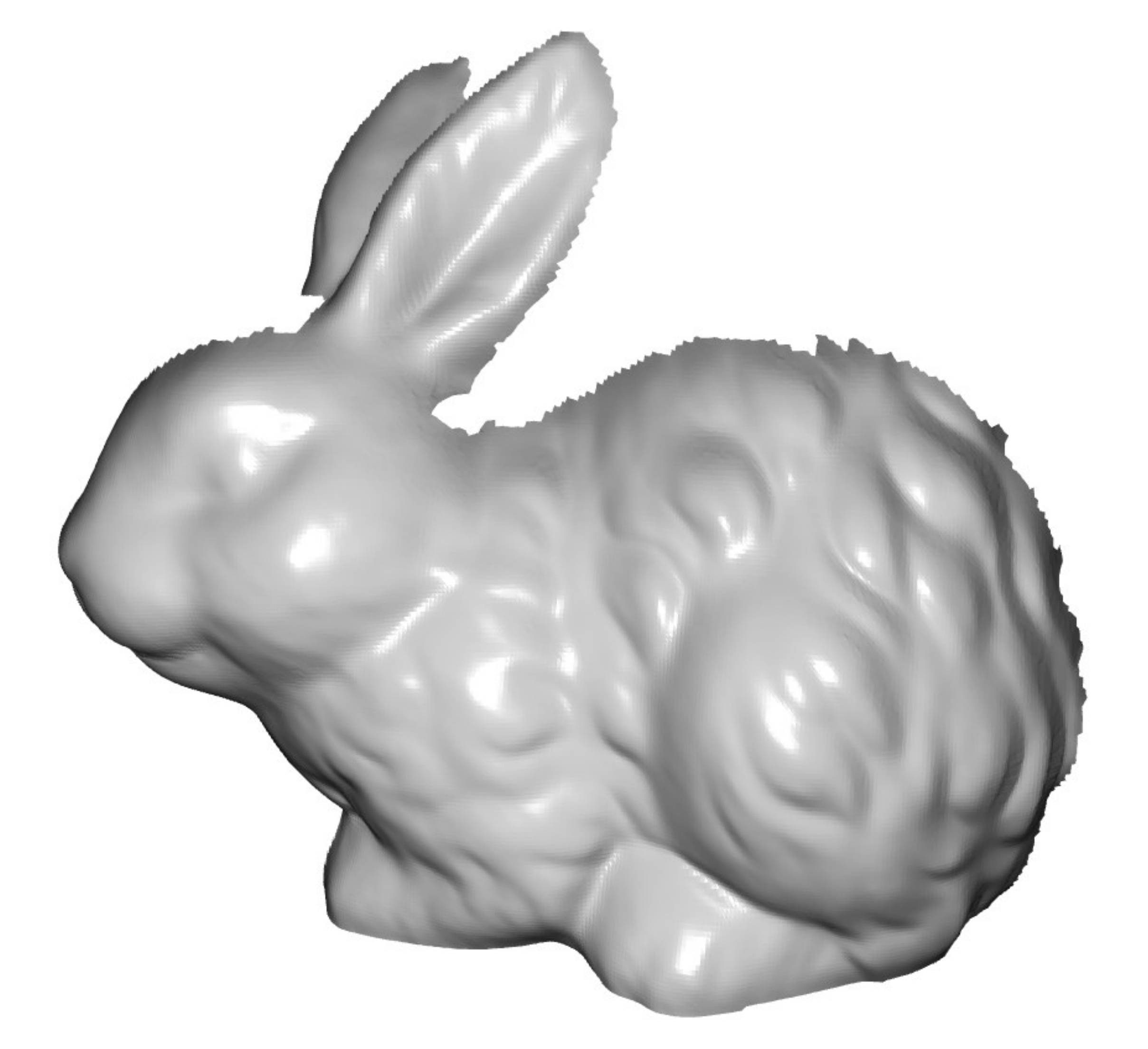}&
		\includegraphics[width=1.72cm]{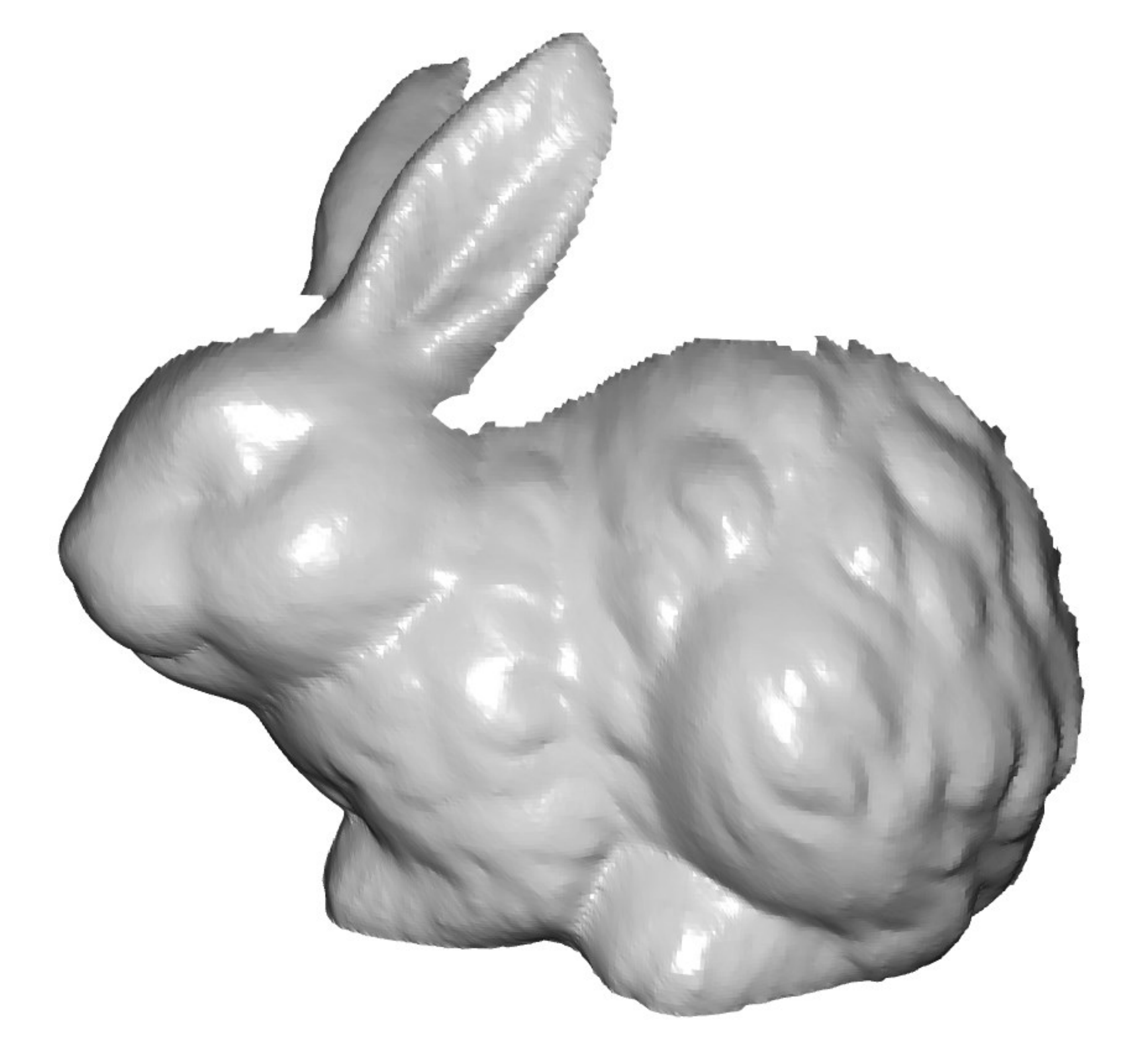}\\
		\textrm{(a)~Noisy}&\textrm{(b)~\cite{fleishman2003bilateral}}&\textrm{(c)~\cite{jones2003non}}&\textrm{(d)~\cite{sun2007fast}}&\textrm{(e)~\cite{zheng2011bilateral}}&\textrm{(f)~\cite{he2013mesh}}&\textrm{(g)~\cite{zhang2015guided}}&\textrm{(h)~\cite{li2018non}}&\textrm{(i)~Ours}	
		\end{array}
		$}
	\caption{\label{fig:com_scan_rabbit_and_angel_100} Comparison between our algorithm and the selected state-of-the-art algorithms on the model of the real scan (angel and rabbit). The parameters are selected as the second row of Table~\ref{table2}}
\end{figure*}

\subsubsection{Result analysis}
In the experiments, we have selected four representative models with rich features: armadillo, bunny, Max Planck, and vaselion. In the armadillo model of Fig. ~\ref{fig:com_denosing and_features} row 1, this represents a type of mesh with more vertices and faces, features, and a complete shape. The face normal filtering methods are mostly two-step filtering algorithms, such as~\cite{sun2007fast,zheng2011bilateral,zhang2015guided}, those algorithms first filter the faces normal, and then updates the vertices. Such methods rely heavily on the threshold of the first step, and is suitable for models with rich planar features. For models with rich geometric features, it is easy to over smoothing. We reduce the Gaussian curvature energy of the overall model by minimizing the Gaussian curvature absolute value of each vertex. Our model's Gaussian curvature energy value is closest to the ground truth model. In the comparison method, our MSAE value is the smallest, which also shows that our algorithm preserves the features the best in the smoothing process (see Table~\ref{table2}). As shown in Fig.~\ref{fig:com_denosing and_features}, our results are the best in terms of both overall shape and local detail.

\begin{table*}[!]
	\centering
	\caption{Quantitative comparison of noise removal and feature preservation performances with other state-of-the-art methods. 
	}\label{table2}
	\resizebox{1.0\textwidth}{!}{ 
		\begin{tabular}{cccccccc}\hline
			Models & Methods &\tabincell{c}{Parameters}&\tabincell{c}{MSAE} & \tabincell{c}{GCE}& \tabincell{c}{KLD}& \tabincell{c}{$\mathscr{D}$$_{mean}$}& \tabincell{c}{$\mathscr{D}$$_{max}$}\\
			\hline
			\tabincell{c}{Armadillo\\$\sigma_{n}=0.3{e}_l$\\(Figure~\ref{fig:com_denosing and_features} row 1)\\$\left|V\right|$: 43243\\ $\left|F\right|$: 86482} &
			\tabincell{l}{ 
				Noisy \\
				Ground truth \\
				~\cite{fleishman2003bilateral} \\
				~\cite{jones2003non} \\
				~\cite{sun2007fast}\\
				~\cite{zheng2011bilateral}\\
				~\cite{he2013mesh} \\
				~\cite{zhang2015guided}\\
				~\cite{li2018non}\\
				Ours\\
			}
			& \tabincell{l}{$(-)$\\$(-)$\\$ (10) $\\$(1,1) $\\$ (0.5,20,10) $\\$ (1,3.5\times10^{-1},20,1.0\times10^{-2},10) $\\$ (1.4\times10^{-2},1.0\times10^{-3},1.0\times10^{3},0.5,2.9\times10^{-3},3.0\times10^{-6}) $\\$ (2,1,0.35,20,1,1.0\times10^{-3},10) $\\$ (3.9\times10^{-1},20,10) $\\($40$)}
			& \tabincell{l}{$22.07$\\$0.00$\\$13.43$\\$12.36$\\$15.06$\\$14.89$\\$13.44$\\$14.88$\\$12.15$\\$\pmb{9.72}$}
			& \tabincell{l}{$13282.30$\\$1958.15$\\$4524.08$\\$4854.35$\\$3765.05$\\$4012.49$\\$2385.80$\\$2167.83$\\$3376.38$\\$\pmb{1665.57}$}
			& \tabincell{l}{$0.79$\\$0.00$\\$0.23$\\$0.13$\\$0.35$\\$0.45$\\$0.52$\\$0.75$\\$0.26$\\$\pmb{0.08}$}
			&\tabincell{l}{$1.62\times10^{-3}$\\$0.00$\\$1.67\times10^{-3}$\\$1.54\times10^{-3}$\\$1.85\times10^{-3}$\\$1.80\times10^{-3}$\\$1.67\times10^{-3}$\\$1.78\times10^{-3}$\\$1.58\times10^{-3}$\\$\pmb{1.27\times10^{-3}}$}
			&\tabincell{l}{$8.59\times10^{-3}$\\$0.00$\\$8.75\times10^{-3}$\\$8.28\times10^{-3}$\\$7.60\times10^{-3}$\\$8.46\times10^{-3}$\\$8.21\times10^{-3}$\\$8.15\times10^{-3}$\\$\pmb{7.61\times10^{-3}}$\\$7.71\times10^{-3}$}
			\\\hline
			\tabincell{l}{Bunny\\$\sigma_{n}=0.5{e}_l$\\(Figure~\ref{fig:com_denosing and_features} row 2)\\$\left|V\right|$: 3301\\ $\left|F\right|$: 6598} &
			\tabincell{l}{ 
				Noisy \\
				Ground truth \\
				~\cite{fleishman2003bilateral} \\
				~\cite{jones2003non} \\
				~\cite{sun2007fast}\\
				~\cite{zheng2011bilateral}\\
				~\cite{he2013mesh} \\
				~\cite{zhang2015guided}\\
				~\cite{li2018non}\\
				Ours\\
			}
			& \tabincell{l}{$(-)$\\$(-)$\\$ (30) $\\$(1,1) $\\$ (0.5,20,30) $\\$ (1,3.5\times10^{-1},20,1.0\times10^{-2},30) $\\$ (1.4\times10^{-2},1.0\times10^{-3},1.0\times10^{3},0.5,5.9\times10^{-3},3.2\times10^{-4}) $\\$ (2,1,0.35,20,1,1.0\times10^{-3},30) $\\$ (3.9\times10^{-1},20,30) $\\($100$)}
			& \tabincell{l}{$31.58$\\$0.00$\\$19.37$\\$18.47$\\$20.05$\\$18.95$\\$15.52$\\$14.63$\\$17.08$\\$\pmb{11.93}$}
			& \tabincell{l}{$1733.72$\\$146.60$\\$602.77$\\$773.32$\\$504.60$\\$548.70$\\$411.53$\\$221.53$\\$509.81$\\$\pmb{99.37}$}
			& \tabincell{l}{$1.65$\\$0.00$\\$0.31$\\$0.26$\\$0.18$\\$0.19$\\$0.20$\\$0.37$\\$0.12$\\$\pmb{0.07}$}
			&\tabincell{l}{$1.82\times10^{-2}$\\$0.00$\\$2.03\times10^{-2}$\\$1.68\times10^{-2}$\\$1.99\times10^{-2}$\\$1.86\times10^{-2}$\\$1.70\times10^{-2}$\\$1.69\times10^{-2}$\\$1.72\times10^{-2}$\\$\pmb{1.60\times10^{-2}}$}
			&\tabincell{l}{$8.30\times10^{-2}$\\$0.00$\\$9.80\times10^{-2}$\\$7.79\times10^{-2}$\\$9.19\times10^{-2}$\\$6.97\times10^{-2}$\\$8.21\times10^{-2}$\\$7.25\times10^{-2}$\\$7.06\times10^{-2}$\\$\pmb{6.86\times10^{-3}}$}
			\\ \hline
			\tabincell{l}{Max Planck\\$\sigma_{n}=0.1{e}_l$\\(Figure~\ref{fig:com_denosing and_features} row 3)\\$\left|V\right|$: 5272\\ $\left|F\right|$: 10540} &
			\tabincell{l}{ 
				Noisy \\
				Ground truth \\
				~\cite{fleishman2003bilateral} \\
				~\cite{jones2003non} \\
				~\cite{sun2007fast}\\
				~\cite{zheng2011bilateral}\\
				~\cite{he2013mesh} \\
				~\cite{zhang2015guided}\\
				~\cite{li2018non}\\
				Ours\\
			}
			&\tabincell{l}{$(-)$\\$(-)$\\$ (10) $\\$(1,1) $\\$ (0.5,20,10) $\\$ (1,3.5\times10^{-1},20,1.0\times10^{-2},10) $\\$ (1.4\times10^{-2},1.0\times10^{-3},1.0\times10^{3},0.5,3.8\times10^{-3},0.4) $\\$ (2,1,0.35,20,1,1.0\times10^{-3},10) $\\$ (3.9\times10^{-1},20,10) $\\($40$)}
			& \tabincell{l}{$9.54$\\$0.00$\\$8.59$\\$7.32$\\$11.83$\\$10.75$\\$8.67$\\$11.63$\\$9.57$\\$\pmb{6.60}$}
			& \tabincell{l}{$455.27$\\$226.34$\\$242.63$\\$221.09$\\$240.25$\\$225.02$\\$170.66$\\$179.21$\\$221.26$\\$\pmb{127.10}$}
			& \tabincell{l}{$0.11$\\$0.00$\\$0.08$\\$0.07$\\$0.12$\\$0.20$\\$0.18$\\$0.46$\\$0.14$\\$\pmb{0.07}$}
			&\tabincell{l}{$5.58\times10^{-1}$\\$0.00$\\$8.19\times10^{-1}$\\$7.25\times10^{-1}$\\$1.10$\\$9.45\times10^{-1}$\\$\pmb{6.63\times10^{-1}}$\\$1.05$\\$8.39\times10^{-1}$\\$1.19$}
			&\tabincell{l}{$2.71$\\$0.00$\\$3.47$\\$3.09$\\$5.30$\\$3.88$\\$\pmb{2.71}$\\$5.58$\\$5.11$\\$5.19$}
			\\  \hline      
			\tabincell{l}{Vaselion\\$\sigma_{n}=0.2{e}_l$\\(Figure~\ref{fig:com_denosing and_features} row 4)\\$\left|V\right|$: 38728\\ $\left|F\right|$: 77452} &
			\tabincell{l}{ 
				Noisy \\
				Ground truth \\
				~\cite{fleishman2003bilateral} \\
				~\cite{jones2003non} \\
				~\cite{sun2007fast}\\
				~\cite{zheng2011bilateral}\\
				~\cite{he2013mesh} \\
				~\cite{zhang2015guided}\\
				~\cite{li2018non}\\
				Ours\\
			}
			&\tabincell{l}{$(-)$\\$(-)$\\$ (10) $\\$(1,1) $\\$ (0.5,20,10) $\\$ (1,3.5\times10^{-1},20,1.0\times10^{-2},10) $\\$ (1.4\times10^{-2},1.0\times10^{-3},1.0\times10^{3},0.5,1.1\times10^{-3},2.0\times10^{-6}) $\\$ (2,1,0.35,20,1,1.0\times10^{-3},10) $\\$ (3.9\times10^{-1},20,10) $\\($40$)}
			& \tabincell{l}{$25.11$\\$0.00$\\$21.13$\\$NaN$\\$24.92$\\$21.78$\\$16.84$\\$25.66$\\$21.55$\\$\pmb{12.72}$}
			& \tabincell{l}{$12933.10$\\$2757.32$\\$7435.05$\\$8127.51$\\$7017.39$\\$6449.90$\\$5536.49$\\$4403.20$\\$6184.21$\\$\pmb{2741.82}$}
			& \tabincell{l}{$0.47$\\$0.00$\\$0.14$\\$0.12$\\$0.15$\\$0.18$\\$0.26$\\$0.41$\\$0.10$\\$\pmb{0.02}$}
			&\tabincell{l}{$1.46\times10^{-3}$\\$0.00$\\$1.72\times10^{-3}$\\$NaN$\\$2.18\times10^{-3}$\\$1.85\times10^{-3}$\\$1.46\times10^{-3}$\\$2.18\times10^{-3}$\\$\pmb{1.91\times10^{-3}}$\\$1.94\times10^{-3}$}
			&\tabincell{l}{$7.65\times10^{-3}$\\$0.00$\\$1.41\times10^{-1}$\\$7.46\times10^{-3}$\\$2.80\times10^{-2}$\\$1.32\times10^{-2}$\\$\pmb{7.20\times10^{-3}}$\\$1.19\times10^{-2}$\\$6.77\times10^{-2}$\\$1.80\times10^{-2}$}
			\\ 
			\hline           
		\end{tabular}
	}
\end{table*}

\begin{figure}[htbp]
	\centering
	\includegraphics[width=4cm]{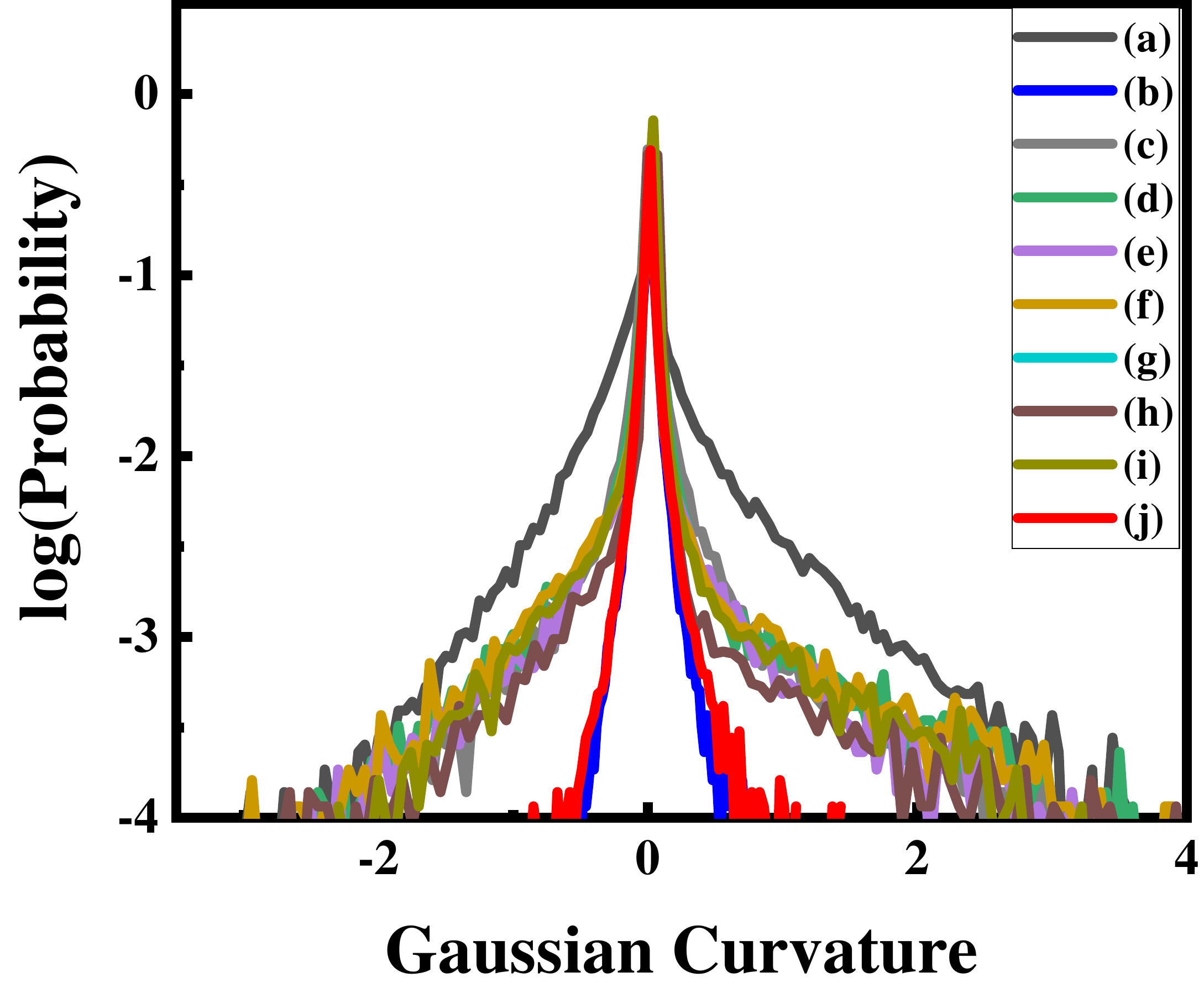}
	\includegraphics[width=4cm]{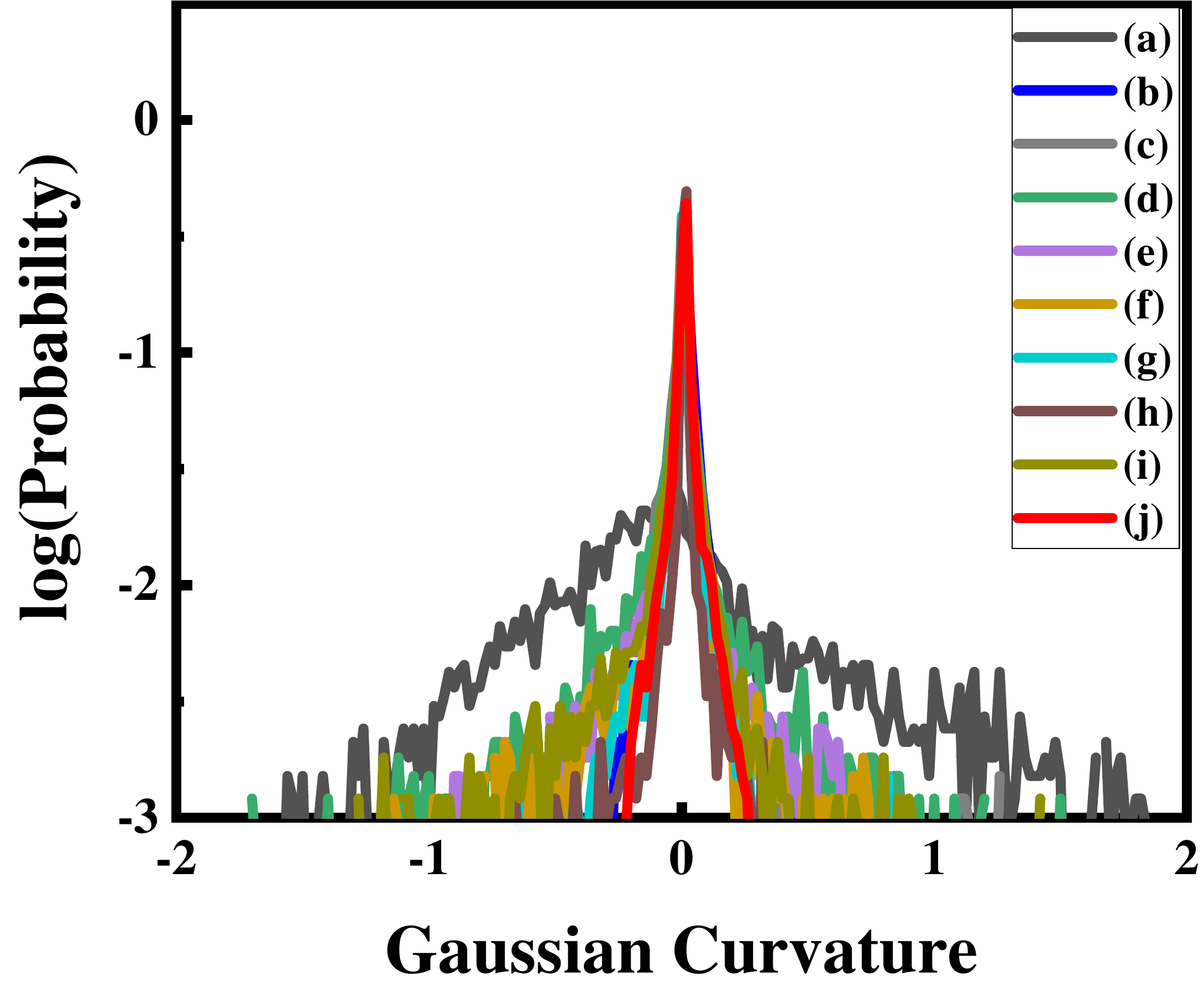}\\
	\makebox[4cm]{Armadillo~($\sigma_{n}=0.3{e}_l$)}
	\makebox[4cm]{Bunny~($\sigma_{n}=0.5{e}_l$)}
	\includegraphics[width=4cm]{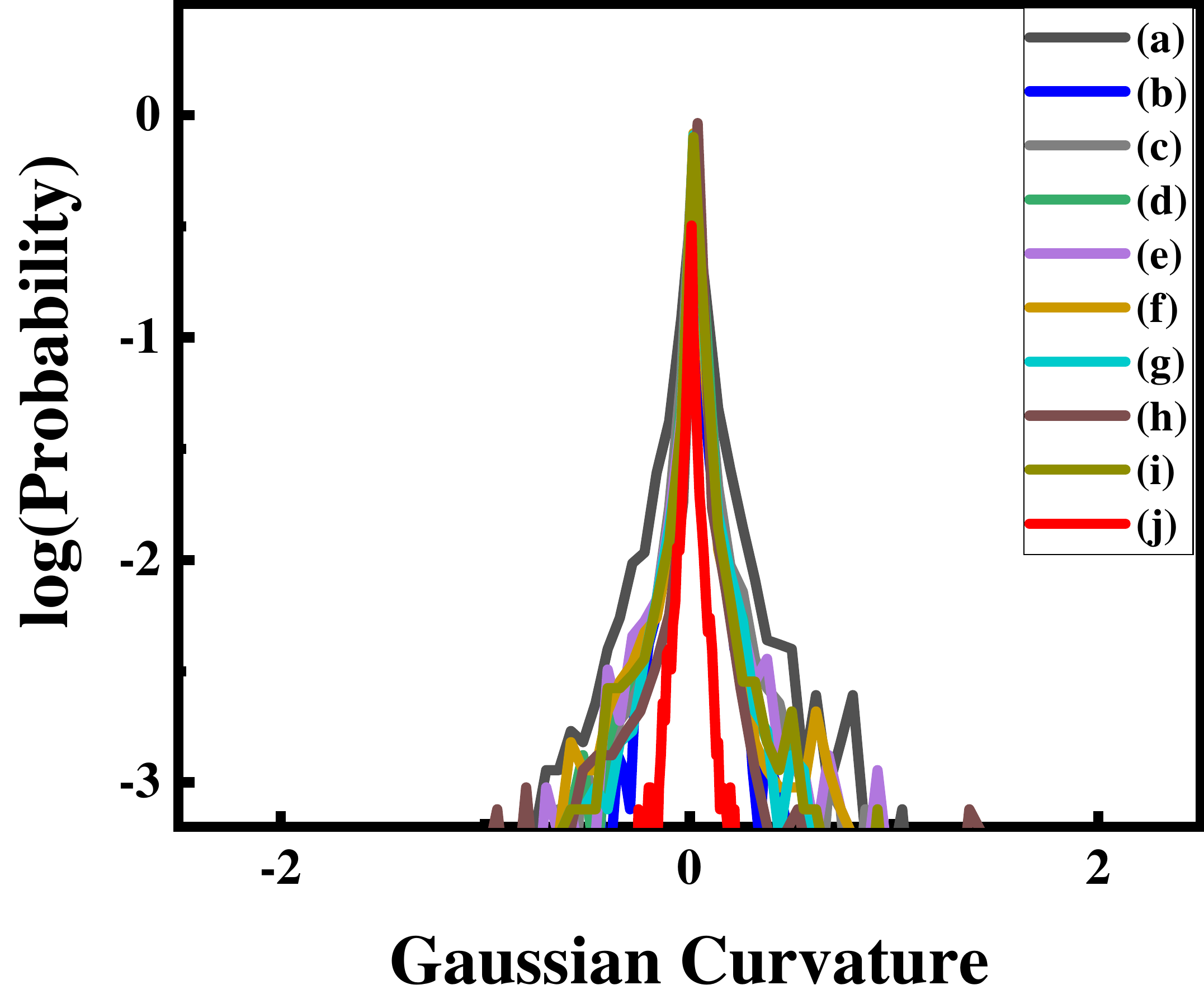}
	\includegraphics[width=4cm]{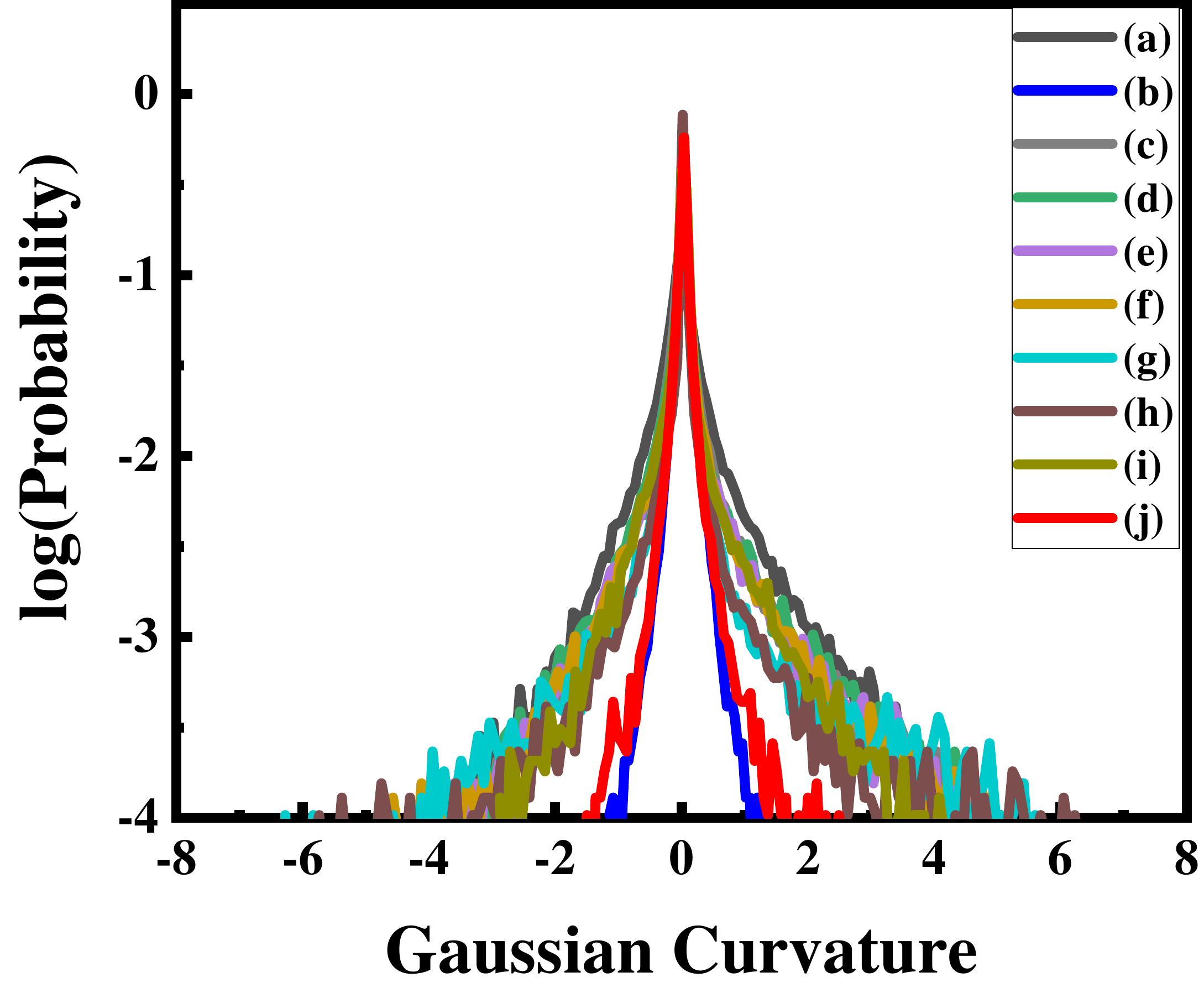}\\
	\makebox[4cm]{Max Planck~($\sigma_{n}=0.1{e}_l$)}
	\makebox[4cm]{Vaselion~($\sigma_{n}=0.2{e}_l$)}\\
	\caption{ \label{fig:com_distibution_8_method} Gaussian curvature distribution map corresponding to Fig.\ref{fig:com_denosing and_features} row 1-4. The left axis is Gaussian curvature probability in Log scale. And the bottom axis is the Gaussian curvature value of the vertices on the model. The curves of the eight different colors are: (a) is the Noisy input, (b) is the Ground truth . (c) is the method ~\cite{fleishman2003bilateral}, (d) is the method ~\cite{jones2003non}, (e) is the method ~\cite{sun2007fast}, (f) is the method ~\cite{zheng2011bilateral}, (g) is the method ~\cite{he2013mesh}, (h) is the method ~\cite{zhang2015guided}, (i) is the method ~\cite{li2018non}, (j) is Ours. Our results are close to the ground truth. The quantitative K-L divergence is summarized in the right column of Table~\ref{table2}.}
\end{figure}

\par Figures~\ref{fig:com_denosing and_features} row 2 demonstrates that our method also works in models with few vertices and faces but rich features with high noise level $(\sigma_{n}=0.5{e}_l)$. In Fig.~\ref{fig:com_denosing and_features} row 2, we can see that for the method of ~\cite{zhang2015guided}, handling low-resolution mesh models can result in local regions using larger neighborhood filtering normals and calculating normals on larger patches, which can result in over smoothing. Figure~\ref{fig:com_denosing and_features} demonstrates the robustness of GCF at different noise levels and vertex numbers.

\par The quantitative comparison results are shown in Table~\ref{table2}. In Fig.~\ref{fig:com_distibution_8_method}, we draw the Gaussian curvature probability distributions of the models which can help better explain why our scheme achieves such a good effect. A wider Gaussian curvature probability distribution indicates a higher level of noise. As can be seen from the armadillo in Fig.~\ref{fig:com_distibution_8_method}, for the noisy model, the Gaussian curvature energy is the largest, so the curve is the widest. For the ground truth model, the Gaussian curvature energy is concentrated near 0, and the distribution curve is relatively narrow. Our method's distribution curve is closest to the ground truth distribution amongst all compared methods. The KLD values of our method in Table~\ref{table2} is the smallest, indicating that the Gaussian curvature probability distribution of our method is the closest to the ground truth distribution. The MSAE values of our method are the smallest, indicating that the output of our method can preserve the model's geometric feature the best. It is also interesting to observe that the GCE values of our method's outputs are also the lowest.
\par We also compared to our previous work~\cite{pan2020hlo} on smoothing and feature preservation. As shown in the Fig.~\ref{Fig.compare hlo}, Gaussian curvature filtering achieves better results on different noise level models.


\subsection{Applied to real scan models}
\par We have seen that our method achieves state-of-the-art effect on synthesized CAD meshes. We have applied the algorithm to process real 3D models. We use two real scanned meshes which contain unknown noise in the experiments and results are shown in Fig.~\ref{fig:com_scan_rabbit_and_angel_100}. Since these two real scanned meshes do not have a ground truth model, they are not shown here as shown in Table~\ref{table2}, but only show qualitative smoothing results.
\subsection{Apply GPU version of GCF on super large meshes}
\par Our algorithm can be implemented by GPU parallel computing, which is especially practical for super large meshes that need to be optimized. Select the fastest algorithm~\cite{sun2007fast} in the comparison method of this paper and compare our algorithm on two super large meshes in terms of computing time. While~\cite{sun2007fast} failed in processing large meshes, our method is successfully applied on large meshes. The results are shown in Table~\ref{table4}:

\begin{table}
	\begin{minipage}{0.8\linewidth}
		$\begin{array}{cccc}
		\subfigure[Kitten]{\includegraphics[width=2cm]{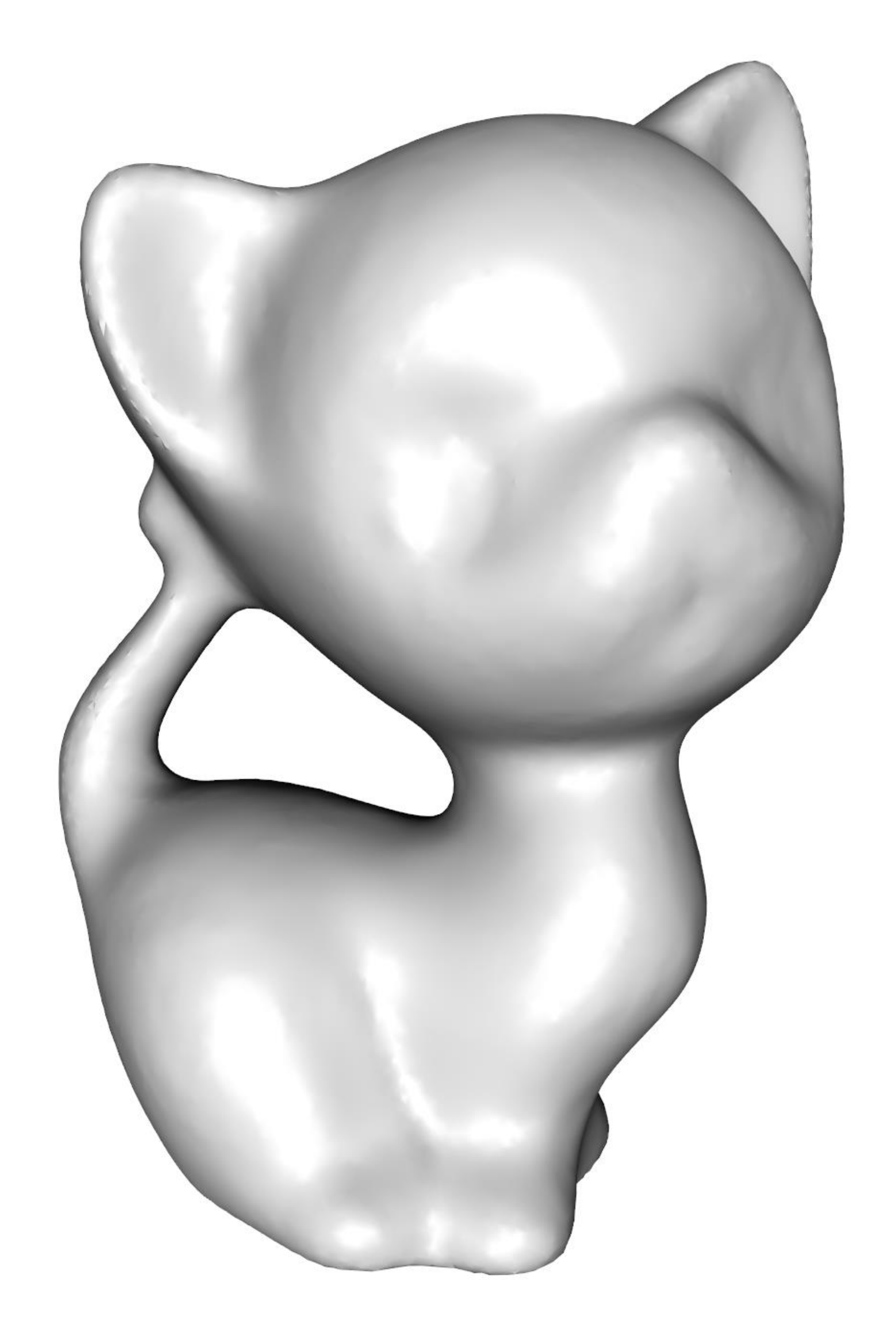}}&
		\subfigure[Dragon]{\includegraphics[width=2cm]{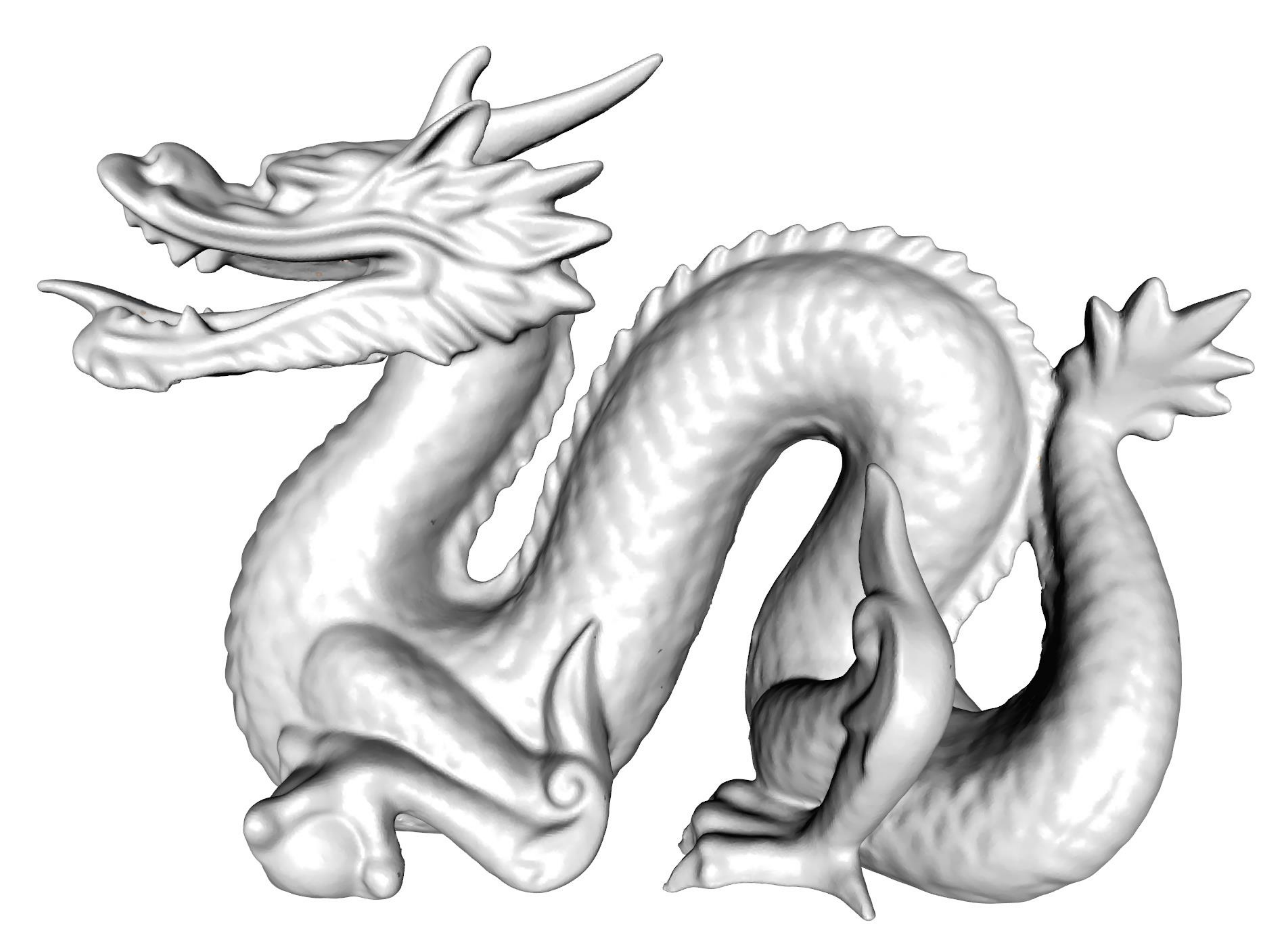}}
		\subfigure[Statues]{\includegraphics[width=2cm]{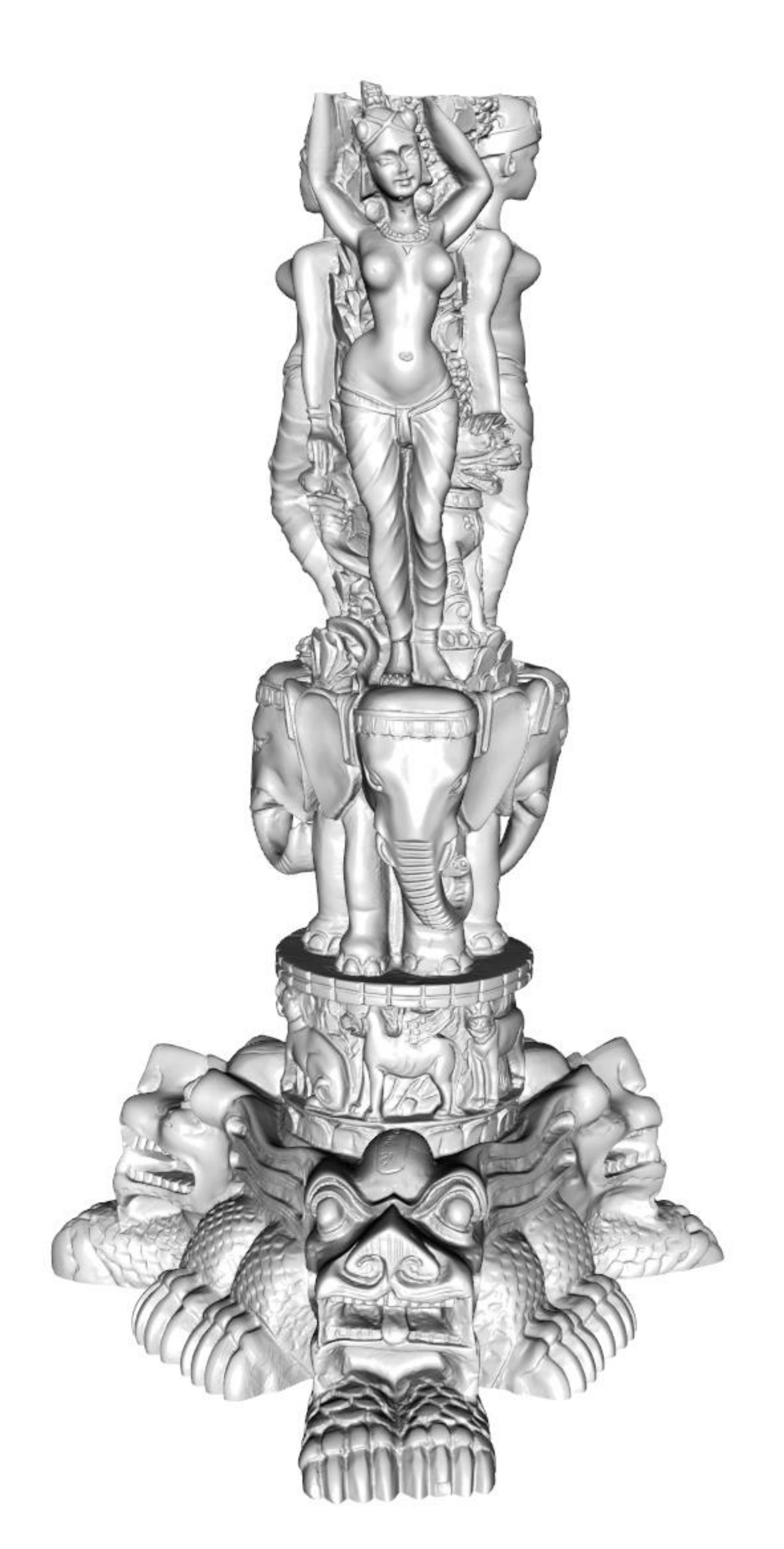}}&
		\subfigure[Lucy]{\includegraphics[width=2cm]{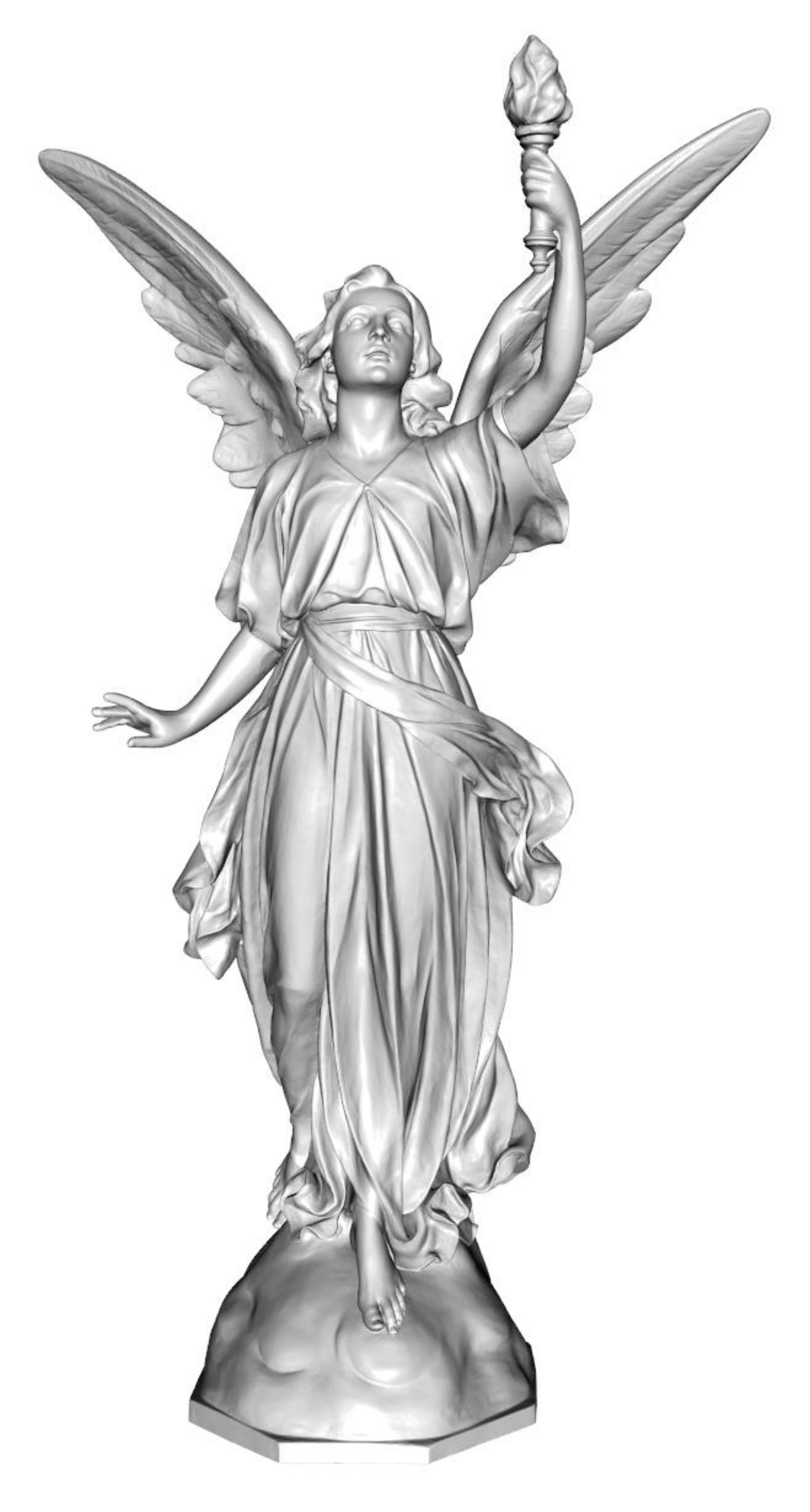}}
		
	\end{array}$
	\end{minipage}\\
	\begin{minipage}{0.8\linewidth} \small 
	\centering \begin{tabular}{c|cccc} 
	\hline \diagbox{Methods}{Times(s)}{Models} & Kitten &Dragon & Statues& Lucy  \\ \hline
	~\cite{sun2007fast} & $16.83$ & $258.8$& $NaN$&$NaN$\\ GCF(GPU) & $3.28$ & $3.39$ & $4.93$& $7.50$\\ \hline
	\end{tabular} 
	\end{minipage}\caption{Time-consuming comparison on meshes with a huge number of vertices.(Kitten:~$\left|V\right|$: 24956, $\left|F\right|$: 49912; Dragon:~$\left|V\right|$: 437645, $\left|F\right|$: 871414; Statues:~$\left|V\right|$: 4999996, $\left|F\right|$: 10000000; Lucy:~$\left|V\right|$: 14027872, $\left|F\right|$: 28055742;). Number of iterations: ~\cite{sun2007fast} iterates 20 times for the normal of the face and 20 iterations for the vertices. Our iterations are 40. } \label{table4} 
	\end{table}

	\subsection{Limitations and future work}
	\par For some corner features like Fig.~\ref{Fig.projection_and_update}, GCF can ensure that these corner features are preserved during the smoothing process. However, In the sharp corner feature shown in Fig.~\ref{Fig.hl_compare}, the apex of the cone tip, GCF treats it as noise. This is a common problem shared by all local methods that do not have a priori assumption about edges. For noisy CAD models with particularly sharp features, our method is inferior to methods based on face normal filtering at preserving sharp features. In future work, we will conduct further research on the possibility of GCF in the field of mesh developability and editing.

	\section{Conclusions}
	In this paper, we propose an iterative filter that optimizes the Gaussian curvature of a triangular mesh. Our method does not need to explicitly calculate the Gaussian curvature. Our method is simple but effective and efficient, as confirmed by the numerical experiments. Our algorithm has only one parameter in the form of the number of iterations, which makes our algorithm easy to use. From the results of multiple sets of experiments, whether it is visual or quantitative analysis, we have verified that our method outperforms the state-of-the-art. With this Gaussian curvature filter, minimizing Gaussian curvature on triangular meshes becomes much easier. This is important for many industries that require Gaussian curvature optimization, such as ship manufacture, car shape design, etc. And we believe that our method can be beneficial to both academical research and practical industries.
	
	
	\section*{Acknowledgments}
	This work was supported in part by the National Natural Science Foundation of China under Grant 61907031. Also partially supported by the Education Department of Guangdong Province, PR China, under project No 2019KZDZX1028.

\bibliographystyle{IEEEtran}  
\bibliography{IEEEabrv,mybibfile}

\end{document}